\documentclass[12pt]{iopart}
\pdfoutput=1 

\usepackage{iopams}
\usepackage{amstext} 
\usepackage{amssymb} 
\usepackage{bbold} 
\usepackage{bbm} 
\usepackage{amsthm} 
\usepackage{stmaryrd} 
\usepackage{stix}
\usepackage{etoolbox} 
\usepackage{xparse}
\usepackage{pgffor}

\usepackage[utf8]{inputenc} 
\usepackage{textgreek} 
\usepackage{relsize} 
\usepackage{microtype} 
\usepackage{ulem} 
\usepackage[compress]{cite} 
\usepackage{csquotes} 
\usepackage{setspace} 
\usepackage{scalerel}

\usepackage{graphicx} 
\usepackage[usenames,dvipsnames]{xcolor} 
\usepackage{tikzit} 
\usepackage{stackengine}
\usepackage{circuitikz} 
\usetikzlibrary{
    arrows,
    shapes,
    decorations,
    intersections,
    backgrounds,
    positioning,
    circuits.ee.IEC
    }

\tikzstyle{inline text}=[text height=1.5ex, text depth=0.25ex, yshift=0.5mm]
\tikzstyle{upground}=[circuit ee IEC, thick, ground, rotate=90, scale=2]
\tikzstyle{downground}=[circuit ee IEC, thick, ground, rotate=-90, scale=1.5]
\tikzstyle{point}=[regular polygon, regular polygon sides=3, draw, scale=0.75, inner sep=-0.5pt, minimum width=9mm, fill=white, regular polygon rotate=180, tikzit fill={rgb,255: red,242; green,255; blue,92}]
\tikzstyle{wide copoint}=[fill=white, draw, shape=isosceles triangle, shape border rotate=90, isosceles triangle stretches=true, inner sep=0pt, minimum width=1.5cm, minimum height=6.12mm]
\tikzstyle{wide point}=[fill=white, draw, shape=isosceles triangle, shape border rotate=-90, isosceles triangle stretches=true, inner sep=0pt, minimum width=1.5cm, minimum height=6.12mm, yshift=-0.0mm]
\tikzstyle{wide dpoint}=[wide point, doubled]
\tikzstyle{copoint}=[regular polygon, regular polygon sides=3, draw, scale=0.75, inner sep=-0.5pt, minimum width=9mm, fill=white, tikzit fill={rgb,255: red,255; green,128; blue,0}, tikzit draw={rgb,255: red,255; green,128; blue,0}]
\tikzstyle{dot}=[inner sep=0mm, minimum width=2mm, minimum height=2mm, draw, shape=circle]
\tikzstyle{black dot}=[dot, fill={gray!30}, text depth=-0.2mm]
\tikzstyle{white dot}=[dot, fill=white, text depth=-0.2mm]
\tikzstyle{small box}=[rectangle, inline text, fill=white, draw, minimum height=5mm, yshift=-0.5mm, minimum width=5mm, font={\small}]
\tikzstyle{small gray box}=[small box, fill={gray!30}]
\tikzstyle{medium box}=[rectangle, inline text, fill=white, draw, minimum height=5mm, yshift=-0.5mm, minimum width=10mm, font={\small}]
\tikzstyle{square box}=[small box]
\tikzstyle{medium gray box}=[small box, fill={gray!30}]
\tikzstyle{semilarge box}=[rectangle, inline text, fill=white, draw, minimum height=5mm, yshift=-0.5mm, minimum width=12.5mm, font={\small}]
\tikzstyle{large box}=[rectangle, inline text, fill=white, draw, minimum height=5mm, yshift=-0.5mm, minimum width=15mm, font={\small}]
\tikzstyle{large gray box}=[small box, fill={gray!30}]
\tikzstyle{dpoint}=[point, doubled]
\tikzstyle{dcopoint}=[copoint, doubled]
\tikzstyle{boldedge}=[doubled, shorten <=-0.17mm, shorten >=-0.17mm]
\tikzstyle{normal}=[line width=0.9pt]
\tikzstyle{doubled}=[line width=1pt]
\tikzstyle{boldedge}=[doubled, shorten <=-0.17mm, shorten >=-0.17mm]
\tikzstyle{small dbox}=[small box, doubled]
\tikzstyle{white ddot}=[white dot, doubled]
\tikzstyle{black ddot}=[black dot, doubled, tikzit fill=black]
\tikzstyle{map}=[draw, shape=NEbox, inner sep=2pt, minimum height=6mm, fill=white]
\tikzstyle{box}=[draw, shape=rectangle, inner sep=2pt, minimum height=6mm, minimum width=6mm, fill=white]
\tikzstyle{dbox}=[draw, doubled, shape=rectangle, inner sep=2pt, minimum height=6mm, minimum width=6mm, fill=white]
\tikzstyle{dmap}=[draw, doubled, shape=NEbox, inner sep=2pt, minimum height=6mm, fill=white]
\tikzstyle{dmapdag}=[draw, doubled, shape=SEbox, inner sep=2pt, minimum height=6mm, fill=white]
\tikzstyle{dmapadj}=[draw, doubled, shape=SEbox, inner sep=2pt, minimum height=6mm, fill=white]
\tikzstyle{dmaptrans}=[draw, doubled, shape=SWbox, inner sep=2pt, minimum height=6mm, fill=white]
\tikzstyle{dmapconj}=[draw, doubled, shape=NWbox, inner sep=2pt, minimum height=6mm, fill=white]
\tikzstyle{map}=[draw, shape=NEbox, inner sep=2pt, minimum height=6mm, fill=white]
\tikzstyle{dashedmap}=[draw, dashed, shape=NEbox, inner sep=2pt, minimum height=6mm, fill=white]
\tikzstyle{mapdag}=[draw, shape=SEbox, inner sep=2pt, minimum height=6mm, fill=white]
\tikzstyle{mapadj}=[draw, shape=SEbox, inner sep=2pt, minimum height=6mm, fill=white]
\tikzstyle{maptrans}=[draw, shape=SWbox, inner sep=2pt, minimum height=6mm, fill=white]
\tikzstyle{mapconj}=[draw, shape=NWbox, inner sep=2pt, minimum height=6mm, fill=white]
\tikzstyle{semilarge map}=[draw, shape=NEbox, inner sep=2pt, minimum height=6mm, fill=white, minimum width=9.5mm]
\tikzstyle{semilarge dmap}=[draw, doubled, shape=NEbox, inner sep=2pt, minimum height=6mm, fill=white, minimum width=9.5mm]
\tikzstyle{kpointdag}=[kpoint adjoint]
\tikzstyle{kpointadj}=[kpoint adjoint]
\tikzstyle{kpointconj}=[kpoint conjugate]
\tikzstyle{kpointtrans}=[kpoint transpose]
\tikzstyle{kpoint common}=[draw, fill=white, inner sep=1pt, minimum height=4mm]
\tikzstyle{kpoint sc}=[shape=cornerpoint, kpoint common]
\tikzstyle{kpoint adjoint sc}=[shape=cornercopoint, kpoint common]
\tikzstyle{kpoint}=[shape=cornerpoint, shorten left=5pt, kpoint common, tikzit fill={rgb,255: red,255; green,128; blue,0}]
\tikzstyle{kpoint adjoint}=[shape=cornercopoint, shorten left=5pt, kpoint common, tikzit fill={rgb,255: red,255; green,128; blue,0}]
\tikzstyle{kpoint conjugate}=[shape=cornerpoint, shorten right=5pt, kpoint common]
\tikzstyle{kpoint transpose}=[shape=cornercopoint, shorten right=5pt, kpoint common]
\tikzstyle{kpoint symm}=[shape=cornerpoint, shorten left=5pt, shorten right=5pt, kpoint common]
\tikzstyle{wide kpoint}=[kpoint, minimum width=1 cm, inner sep=2pt]
\tikzstyle{wide kpointdag}=[kpointdag, minimum width=1 cm, inner sep=2pt]
\tikzstyle{wide kpointconj}=[kpointconj, minimum width=1 cm, inner sep=2pt]
\tikzstyle{wide kpointtrans}=[kpointtrans, minimum width=1 cm, inner sep=2pt]
\tikzstyle{wider kpoint}=[kpoint, minimum width=1.25 cm, inner sep=2pt]
\tikzstyle{wider kpointdag}=[kpointdag, minimum width=1.25 cm, inner sep=2pt]
\tikzstyle{wider kpointconj}=[kpointconj, minimum width=1.25 cm, inner sep=2pt]
\tikzstyle{wider kpointtrans}=[kpointtrans, minimum width=1.25 cm, inner sep=2pt]
\tikzstyle{dkpoint}=[kpoint, doubled, tikzit fill={rgb,255: red,255; green,85; blue,210}]
\tikzstyle{wide dkpoint}=[wide kpoint, doubled, tikzit fill={rgb,255: red,68; green,255; blue,0}]
\tikzstyle{dkpointdag}=[kpoint adjoint, doubled]
\tikzstyle{wide dkpointdag}=[wide kpointdag, doubled]
\tikzstyle{label}=[fill=white, draw=white, shape=circle, tikzit draw={rgb,255: red,10; green,26; blue,255}, tikzit fill={rgb,255: red,0; green,12; blue,255}, font={\small}]
\tikzstyle{squarelabel}=[fill=white, draw=white, shape=rectangle, tikzit draw=black]
\tikzstyle{eslabel}=[tikzit draw={rgb,255: red,255; green,191; blue,191}, tikzit fill={rgb,255: red,255; green,191; blue,191}, font={\tiny}]
\tikzstyle{large dmap}=[draw, doubled, shape=NEbox, inner sep=2pt, minimum height=6mm, fill=white, minimum width=12mm]
\tikzstyle{gray point}=[point, fill={gray!40!white}]
\tikzstyle{gray dpoint}=[gray point, doubled, tikzit draw={rgb,255: red,128; green,128; blue,128}, tikzit fill={rgb,255: red,128; green,128; blue,128}]
\tikzstyle{gray copoint}=[copoint, fill={gray!40!white}, tikzit fill={rgb,255: red,128; green,128; blue,128}]
\tikzstyle{gray dcopoint}=[gray copoint, doubled, tikzit fill={rgb,255: red,128; green,128; blue,128}]
\tikzstyle{circlenew}=[draw=black, shape=circle, inner sep=1pt]
\tikzstyle{blue label}=[text=NavyBlue, tikzit draw={rgb,255: red,0; green,96; blue,167}, tikzit fill={rgb,255: red,35; green,68; blue,255}]
\tikzstyle{big dot}=[fill=white, draw=black, shape=circle, minimum width=6mm, minimum height=6mm]
\tikzstyle{3d box}=[fill=white, draw=black, shape=trapezium, trapezium left angle=-70, trapezium right angle=70, rotate=10]
\tikzstyle{slant red box}=[fill={rgb,255: red,191; green,0; blue,64}, draw={rgb,255: red,191; green,0; blue,64}, shape=rectangle, xslant=0.5, font={\tiny}, text={rgb,255: red,191; green,0; blue,64}, fill opacity=0.5, line width=1pt]
\tikzstyle{slant point}=[regular polygon, regular polygon sides=3, draw, scale=0.75, inner sep=-0.5pt, minimum width=9mm, fill white, regular polygon rotate=180, yslant=-0.3]
\tikzstyle{tiny orange label}=[font={\tiny}, text={rgb,255: red,255; green,128; blue,0}, tikzit draw={rgb,255: red,255; green,128; blue,0}]
\tikzstyle{tiny red label}=[font={\tiny}, text={rgb,255: red,191; green,0; blue,64}, tikzit draw={rgb,255: red,191; green,0; blue,64}, draw=none]
\tikzstyle{red label}=[text={rgb,255: red,191; green,0; blue,64}, tikzit draw={rgb,255: red,191; green,0; blue,64}]
\tikzstyle{slant label black}=[font={\tiny}, xslant=0.5, tikzit draw=black]
\tikzstyle{slant label red}=[font={\tiny}, xslant=0.5, text={rgb,255: red,191; green,0; blue,64}, tikzit draw={rgb,255: red,191; green,0; blue,64}]
\tikzstyle{slant label orange}=[font={\tiny}, xslant=0.5, text={rgb,255: red,255; green,128; blue,0}, tikzit draw={rgb,255: red,255; green,128; blue,0}]
\tikzstyle{slanted point}=[fill={rgb,255: red,191; green,0; blue,64}, draw={rgb,255: red,191; green,0; blue,64}, shape=triangle, regular polygon, regular polygon sides=3, scale=0.75, inner sep=-0.5pt, minimum width=5mm, regular polygon rotate=90, xslant=0.5, fill opacity=0.5, font={\tiny}, line width=1pt, text={rgb,255: red,191; green,0; blue,64}]
\tikzstyle{slanted point black}=[draw=black, shape=triangle, regular polygon, regular polygon sides=3, scale=0.75, inner sep=-0.5pt, minimum width=5mm, regular polygon rotate=90, xslant=0.5, font={\tiny}, line width=0.2pt, text=black, fill=white, tikzit fill=white]
\tikzstyle{red dot}=[fill={rgb,255: red,191; green,0; blue,64}, draw={rgb,255: red,191; green,0; blue,64}, shape=circle, inner sep=0, minimum width=1.5mm, minimum height=1.5mm]
\tikzstyle{black dot}=[fill=black, draw=black, shape=circle, inner sep=0, minimum width=1.5mm, minimum height=1.5mm]
\tikzstyle{orange dot}=[fill={rgb,255: red,255; green,128; blue,0}, draw={rgb,255: red,255; green,128; blue,0}, shape=circle, inner sep=0, minimum width=1.5mm, minimum height=1.5mm]
\tikzstyle{blue dot}=[fill={rgb,255: red,0; green,0; blue,228}, draw={rgb,255: red,0; green,0; blue,228}, shape=circle, inner sep=0, minimum width=1.5mm, minimum height=1.5mm]
\tikzstyle{slant white}=[fill=white, draw=black, shape=rectangle, xslant=0.5, font={\tiny}, line width=1pt]
\tikzstyle{slant small map}=[fill=white, draw=black, xslant=0.5, shape=rectangle, font={\tiny}, line width=1pt, inner sep=0.6mm]
\tikzstyle{slanted copoint black}=[draw=black, shape=triangle, regular polygon, regular polygon sides=3, scale=0.75, inner sep=-0.5pt, minimum width=5mm, regular polygon rotate=-90, xslant=0.5, font={\tiny}, line width=0.2pt, text=black, fill=white, tikzit fill=white]
\tikzstyle{purple dot}=[fill={rgb,255: red,128; green,0; blue,128}, draw={rgb,255: red,128; green,0; blue,128}, shape=circle, inner sep=0, minimum width=1.5mm, minimum height=1.5mm]
\tikzstyle{white dot 2}=[fill=white, draw=black, shape=circle]
\tikzstyle{horizontal point}=[style=point, rotate=-90, tikzit shape=rectangle, tikzit fill={rgb,255: red,191; green,128; blue,64}]
\tikzstyle{pslant orange}=[style=slanted point black, fill={rgb,255: red,255; green,128; blue,0}, draw={rgb,255: red,255; green,128; blue,0}, tikzit fill={rgb,255: red,255; green,128; blue,0}, tikzit draw={rgb,255: red,255; green,128; blue,0}]
\tikzstyle{upground horizontal}=[style=upground, rotate=-90]
\tikzstyle{double horizontal point}=[style=horizontal point, line width=1pt]
\tikzstyle{double point}=[style=point, line width=1pt]
\tikzstyle{double copoint}=[style=copoint, line width=1pt]
\tikzstyle{horizontal copoint}=[style=double copoint, rotate=-90]
\tikzstyle{slant label purple}=[style=slant label black, tikzit draw={rgb,255: red,128; green,0; blue,128}, text={rgb,255: red,128; green,0; blue,128}]
\tikzstyle{orange copoint}=[style=pslant orange, rotate=-180, tikzit fill={rgb,255: red,255; green,128; blue,0}]
\tikzstyle{new style 0}=[style=slant white, draw={rgb,255: red,0; green,0; blue,228}, fill={rgb,255: red,0; green,0; blue,228}, fill opacity=0.5, shape=rectangle]
\tikzstyle{wide slanted point}=[style=wide point, xslant=0.5, fill=white, rotate=-90, minimum width=0.8cm, fill={rgb,255: red,128; green,128; blue,128}, fill opacity=0.5, line width=1pt]
\tikzstyle{black dot white}=[style=black dot, text=white, draw=none, tikzit draw={rgb,255: red,191; green,255; blue,0}, shape=circle]
\tikzstyle{small point}=[style=point, draw=black, minimum width=3mm]

\tikzstyle{new edge style 1}=[-, line width=1pt, shorten <=-0.17mm, shorten >=-0.17mm, tikzit draw={rgb,255: red,204; green,0; blue,3}]
\tikzstyle{diredge}=[-, postaction=decorate, decoration={markings, mark=at position 0.55 with \edgearrow}]
\tikzstyle{bold diredge}=[-, diredge, line width=1pt, tikzit draw={rgb,255: red,128; green,0; blue,128}]
\tikzstyle{grey}=[-, draw={rgb,255: red,188; green,188; blue,188}]
\tikzstyle{classical}=[-, dashed, tikzit draw={rgb,255: red,255; green,128; blue,0}]
\tikzstyle{reddashed}=[-, dashed, draw={rgb,255: red,0; green,128; blue,128}, postaction=decorate, decoration={markings, mark=at position 0.55 with \edgearrow}]
\tikzstyle{reddahednoarrow}=[-, dashed, draw={rgb,255: red,179; green,40; blue,40}]
\tikzstyle{arrow edge}=[-, ->, draw={rgb,255: red,191; green,191; blue,191}, tikzit draw={rgb,255: red,191; green,191; blue,191}, ultra thick]
\tikzstyle{tarrow edge}=[-, ->, draw={rgb,255: red,191; green,191; blue,191}, tikzit draw={rgb,255: red,191; green,191; blue,191}]
\tikzstyle{gray edge}=[-, draw={rgb,255: red,191; green,191; blue,191}, tikzit draw={rgb,255: red,191; green,191; blue,191}, ultra thick]
\tikzstyle{lightgrayedge}=[-, draw={rgb,255: red,207; green,207; blue,207}]
\tikzstyle{green edge}=[-, tikzit draw={rgb,255: red,128; green,128; blue,0}, draw={rgb,255: red,128; green,128; blue,0}]
\tikzstyle{red edge}=[-, draw={rgb,255: red,191; green,0; blue,64}, tikzit draw={rgb,255: red,191; green,0; blue,64}]
\tikzstyle{arrow edge black}=[-, ->]
\tikzstyle{solid blue}=[-, draw={rgb,255: red,0; green,96; blue,167}, tikzit draw={rgb,255: red,0; green,96; blue,167}]
\tikzstyle{classical blue}=[-, draw={rgb,255: red,0; green,96; blue,167}, tikzit draw={rgb,255: red,0; green,96; blue,167}, dashed]
\tikzstyle{fill gray}=[-, fill=gray]
\tikzstyle{bold gray}=[-, line width=1pt, tikzit draw={rgb,255: red,128; green,128; blue,128}]
\tikzstyle{fill pink}=[-, fill={rgb,255: red,193; green,100; blue,94}, fill opacity=0.5, draw={rgb,255: red,134; green,68; blue,65}, line width=1pt, tikzit draw={rgb,255: red,134; green,68; blue,65}, tikzit fill={rgb,255: red,193; green,100; blue,94}]
\tikzstyle{fill carta da zucchero}=[-, fill={rgb,255: red,129; green,158; blue,219}, fill opacity=0.5, line width=0.4mm]
\tikzstyle{fill white}=[-, fill=white]
\tikzstyle{fill purple}=[-, fill={rgb,255: red,113; green,69; blue,128}, fill opacity=0.5, draw={rgb,255: red,79; green,48; blue,90}, tikzit fill={rgb,255: red,113; green,69; blue,128}, tikzit draw={rgb,255: red,79; green,48; blue,90}, line width=1pt]
\tikzstyle{fill green}=[-, fill={rgb,255: red,62; green,128; blue,120}, fill opacity=0.5, draw={rgb,255: red,33; green,68; blue,63}, tikzit fill={rgb,255: red,62; green,128; blue,120}, tikzit draw={rgb,255: red,33; green,68; blue,63}, line width=1pt]
\tikzstyle{bold orange}=[-, draw={rgb,255: red,255; green,128; blue,0}, fill=none, line width=1pt]
\tikzstyle{bold black}=[-, line width=1pt, draw=black, fill=none, tikzit draw=black]
\tikzstyle{bold red}=[-, draw={rgb,255: red,191; green,0; blue,64}, fill=none, line width=1pt]
\tikzstyle{fill light green}=[-, fill={rgb,255: red,166; green,166; blue,112}, fill opacity=0.5, draw={rgb,255: red,121; green,121; blue,81}, line width=1pt]
\tikzstyle{new edge style 0}=[-, fill=yellow, fill opacity=0.5, draw={rgb,255: red,146; green,146; blue,0}, tikzit fill=yellow, tikzit draw={rgb,255: red,146; green,146; blue,0}]
\tikzstyle{bold dashed red}=[-, draw={rgb,255: red,191; green,0; blue,64}, fill=none, line width=1pt, dashed]
\tikzstyle{bold dashed orange}=[-, draw={rgb,255: red,255; green,128; blue,0}, dashed, line width=1pt]
\tikzstyle{bold blue}=[-, draw={rgb,255: red,0; green,0; blue,228}, line width=1pt]
\tikzstyle{arrow red}=[draw={rgb,255: red,191; green,0; blue,64}, ->, line width=1pt]
\tikzstyle{new edge style 2}=[-, draw={rgb,255: red,191; green,0; blue,64}, line width=1pt]
\tikzstyle{boldish}=[-, line width=0.6mm, fill=cyan]
\tikzstyle{white edge}=[-, draw=white]
\tikzstyle{purple edge}=[-, draw={rgb,255: red,128; green,0; blue,128}, line width=1pt]
\tikzstyle{light gray}=[-, fill={rgb,255: red,191; green,191; blue,191}, draw={rgb,255: red,191; green,191; blue,191}, tikzit fill={rgb,255: red,191; green,191; blue,191}, tikzit draw={rgb,255: red,191; green,191; blue,191}, fill opacity=0.3]
\tikzstyle{invisible edge}=[-, fill opacity=0, fill=none]
\tikzstyle{carta da zucchero thin}=[-, style=fill carta da zucchero, line width=0.1pt, fill={rgb,255: red,129; green,158; blue,219}, tikzit fill={rgb,255: red,129; green,158; blue,219}]
\tikzstyle{pink thin}=[-, style=fill pink, line width=0.1pt, fill={rgb,255: red,193; green,100; blue,94}]
\tikzstyle{fill green thin edge}=[-, style=fill green, tikzit fill={rgb,255: red,62; green,128; blue,120}, line width=0.1pt]

\definecolor{evred}{rgb}{0.996, 0.403, 0.537}
\definecolor{evgreen}{rgb}{0.501, 1.0, 0.505}
\definecolor{evblue}{rgb}{0.2, 0.588, 1.0}

\usepackage[frozencache]{minted} 

\setminted[python]{
bgcolor=lightgray!10,
xleftmargin=2mm,
fontsize=\footnotesize,
frame=lines,
linenos,
python3=true}
\setminted[text]{
bgcolor=lightgray!10,
xleftmargin=2mm,
fontsize=\footnotesize}

\theoremstyle{definition}
\newtheorem{definition}{Definition}[section]
\newtheorem{remark}{Remark}[section]
\newtheorem{theorem}{Theorem}[section]
\newtheorem{proposition}[theorem]{Proposition}
\newtheorem{lemma}[theorem]{Lemma}
\newtheorem{observation}[theorem]{Observation}
\newtheorem{corollary}[theorem]{Corollary}
\newtheorem{conjecture}[theorem]{Conjecture}

\usepackage{xifthen}

\newcommand{\opapp}[2]{\ensuremath{#1\left(#2\right)}} 
\newcommand{\opapptxt}[2]{\ensuremath{\text{#1}\left(#2\right)}} 

\newcommand{\nats}{\ensuremath{\mathbb{N}}}

\newcommand{\reals}{\ensuremath{\mathbb{R}}}

\newcommand{\tsuchthat}[2]{\ensuremath{\left\{#1\middle|#2\right\}}} 
\newcommand{\suchthat}[2]{\tsuchthat{\,#1\,}{\,#2\,}} 

\newcommand{\downset}[1]{\ensuremath{#1\!\downarrow}}
\newcommand{\upset}[1]{\ensuremath{#1\!\uparrow}}
\newcommand{\domSym}{\text{dom}}
\newcommand{\dom}[1]{\opapp{\domSym}{#1}}

\newcommand{\restrict}[2]{#1|_{#2}}
\newcommand{\Subsets}[1]{\opapp{\mathcal{P}\!}{#1}} 

\newcommand{\ev}[1]{\text{#1}} 
\newcommand{\discrete}[1]{\opapptxt{discrete}{#1}} 
\newcommand{\indiscrete}[1]{\opapptxt{indiscrete}{#1}} 
\newcommand{\total}[1]{\opapptxt{total}{#1}} 
\newcommand{\seqcomposeSym}{\rightsquigarrow}

\newcommand{\LsetsSym}{\Lambda} 
\newcommand{\Lsets}[1]{\opapp{\LsetsSym}{#1}} 
\newcommand{\Hist}[1]{\opapptxt{Hist}{#1}} 
\newcommand{\ExtHist}[1]{\opapptxt{ExtHist}{#1}} 
\newcommand{\Ext}[1]{\opapptxt{Ext}{#1}} 

\newcommand{\tips}[2]{\opapp{\text{tips}_{#1}}{#2}} 
\newcommand{\Events}[1]{{E}^{#1}} 
\newcommand{\Inputs}[1]{{I}^{#1}} 
\newcommand{\ExtCausFun}[1]{\opapptxt{ExtCausFun}{#1}} 

\newcommand{\TipHists}[2]{\opapp{\text{TipHists}_{#1}}{#2}} 
\newcommand{\histconstrSym}[1]{\sim_{#1}}
\newcommand{\histconstr}[3]{#2\!\histconstrSym{#1}\!\!#3}

\newcommand{\histconstreqcls}[2]{\ensuremath{\left[#1\right]_{\histconstrSym{#2}}}}
\newcommand{\TipEqCls}[2]{\opapp{\text{TipEq}_{#1}}{#2}}
\newcommand{\DistSym}{\mathcal{D}}
\newcommand{\Dist}[1]{\opapp{\DistSym}{#1}}
\newcommand{\CausDist}[1]{\opapptxt{CausDist}{#1}} 
\newcommand{\topdist}[1]{\left\lceil #1 \right\rceil} 
\newcommand{\extdist}[2]{\left\lfloor #1 \right\rfloor_{#2}} 
\newcommand{\StdCov}[1]{\opapptxt{StdCov}{#1}} 
\newcommand{\Covers}[1]{\opapptxt{Covers}{#1}} 
\newcommand{\EmpModels}[1]{\opapptxt{EmpMod}{#1}} 
\newcommand{\Vertices}[1]{\opapp{\mathcal{V}}{#1}} 
\newcommand{\Faces}[1]{\opapp{\mathcal{F}}{#1}} 
\newcommand{\AffSubsp}[1]{\opapp{\mathbb{A}}{#1}} 
\newcommand{\boundary}[1]{\opapp{\partial}{#1}} 
\newcommand{\interior}[1]{#1\backslash\boundary{#1}} 
\newcommand{\Slice}[2]{\opapp{\text{Slice}_{#1}}{#2}}
\newcommand{\NormEqs}[1]{\opapptxt{NormEqs}{#1}} 
\newcommand{\QNormEqs}[1]{\opapptxt{QNormEqs}{#1}} 
\newcommand{\mass}[1]{\opapptxt{mass}{#1}} 
\newcommand{\CCPD}[1]{\opapptxt{CCPD}{#1}} 
\newcommand{\QNCCPD}[1]{\opapp{\text{CCPD}_{\text{QNorm}}}{#1}} 
\newcommand{\PsEmpModels}[1]{\opapptxt{PEmpMods}{#1}} 
\newcommand{\PsEmpModelsVec}[1]{\left\langle \PsEmpModels{#1} \right\rangle} 

\newcommand{\CausEqs}[1]{\opapptxt{CausEqs}{#1}} 
\newcommand{\StdCausEqs}[1]{\opapp{\text{CausEqs}_{std}}{#1}} 
\newcommand{\SolCausEqs}[1]{\opapp{\text{CausEqs}_{sol}}{#1}} 
\newcommand{\Causaltope}[1]{\opapp{\text{Caus}}{#1}} 
\newcommand{\StdCausaltope}[1]{\opapp{\text{Caus}_{std}}{#1}} 
\newcommand{\SolCausaltope}[1]{\opapp{\text{Caus}_{sol}}{#1}} 

\newcommand{\hist}[1]{
    \ensuremath{
        \left\{
            \foreach \i\j [count=\idx] in {#1}{%
                \ifnum\idx=1%
                    \ev{\i}\!:\!\j%
                \else%
                    ,\,\ev{\i}\!:\!\j%
                \fi%
            }
        \right\}
    }
}
\newcommand{\evset}[1]{
    \ensuremath{
        \left\{
            \foreach \i [count=\idx] in {#1}{%
                \ifnum\idx=1%
                    \ev{\i}%
                \else%
                    ,\ev{\i}%
                \fi%
            }
        \right\}
    }
}

\begin{document}

\title{The Geometry of Causality}

\author{Stefano Gogioso$^{1,2}$ and Nicola Pinzani$^{1,3}$}

\address{$^1$Hashberg Ltd, London, UK}
\address{$^2$Department of Computer Science, University of Oxford, Oxford, UK}
\address{$^3$QuIC, Universit\'{e} Libre de Bruxelles, Brussels, BE}
\ead{$^1$stefano.gogioso@cs.ox.ac.uk, $^2$nicola.pinzani@ulb.be}
\vspace{10pt}

\begin{abstract}
    We provide a unified operational framework for the study of causality, non-locality and contextuality, in a fully device-independent and theory-independent setting.
    This paper is the final instalment in a trilogy: spaces of input histories, our dynamical generalisation of causal orders, were introduced in ``The Combinatorics of Causality''; the sheaf-theoretic treatment of causal distributions, known to us as empirical models, was detailed in ``The Topology of Causality''.

    We define causaltopes---our chosen portmanteau of ``causal polytopes''---for arbitrary spaces of input histories and arbitrary choices of input contexts.
    We show that causaltopes are obtained by slicing simpler polytopes of conditional probability distributions with a set of causality equations, which we fully characterise.
    We provide efficient linear programs to compute the maximal component of an empirical model supported by any given sub-causaltope, as well as the associated causal fraction.

    We introduce a notion of causal separability relative to arbitrary causal constraints.
    We provide efficient linear programs to compute the maximal causally separable component of an empirical model, and hence its causally separable fraction, as the component jointly supported by certain sub-causaltopes.

    We study causal fractions and causal separability for several novel examples, including a selection of quantum switches with entangled or contextual control.
    In the process, we demonstrate the existence of ``causal contextuality'', a phenomenon where causal inseparability is clearly correlated to, or even directly implied by, non-locality and contextuality.
\end{abstract}

\maketitle








\section{Introduction}
\label{section:introduction}

The study of quantum correlations has a long history, arguably originating with Bell's demonstration that such correlations violated the assumptions of local realism \cite{bell1964on}.
While the influence of Bell's work in shaping quantum foundations is undeniable, it is also worth noting that---as explained by Pitowsky \cite{pitowsky1989bell,Pitowsky1994}---Bell's inequalities can be considered a special case of Boole's formulation of the ``conditions for possible experience'' \cite{boole1862theory}, dating more than a century earlier.
More specifically, Boole sought to derive conditions imposed by logical dependencies on relative frequencies, in order to provide a description of compatible probabilistic behaviours.

Of course, back in 1862, Boole could not have foreseen that naturally occurring correlations---such as those demonstrated by Bell---would end up violating his ``conditions for possible experience'', requiring the subsequent paradigm shift to non-locality.
In this paper, we follow a line of reasoning non unlike that of Boole, but one which accommodates the possibility of non-local and contextual correlations, by design.
Our principal mathematical achievement is then the explicit description of ``causaltopes'' (portmanteau of ``causal polytopes''), the convex spaces that capture all possible probabilistic behaviours compatible, in much the same sense that Boole might have used, with a given causal structure.

This work is the third instalment in a trilogy.
In the first instalment, titled ``The Combinatorics of Causality'' \cite{gogioso2022combinatorics}, we introduce and investigate spaces of input histories, a family of combinatorial objects which can be used to model an exceptional variety of causal structures, including all definite, indefinite and dynamical causal orders investigated by previous literature on quantum causality.
In the second instalment, titled ``The Topology of Causality'' \cite{gogioso2022topology}, we study the combinatorial aspects of causal functions and we develop a topological, sheaf-theoretic framework to describe causal correlation, in a way which is compatible with---and an extension of---the Abramsky-Brandenburger framework for non-locality and contextuality \cite{abramsky2011sheaf}.

In this third and final instalment, we develop a geometric description complementary to the topological one: causal correlations become the points of causaltopes, convex polytopes obtained by slicing polytopes of conditional probability distributions with certain causality equations.
Our methods can be seen as a generalisation of the geometric techniques used in the device independent study of no-signalling correlations, and they present a finer-grained picture of causal separability than the one painted by the literature on causal inequalities \cite{oreshkov2016causal,oreshkov2012quantum}.
Specifically, we are able to quantify the device-independent explainability of conditional probability distributions relative to arbitrary putative causal structures, incorporating constraints such as space-like separation of parties, or dynamical no-signalling.
The more general, relative nature of our definition of causal separability allows us to define new witnesses for indefinite causal order, by exploiting the experimental legitimacy of imposing some causal constraints even in the presence of indefinite causality.

The advantage of a geometrical perspective is not limited to causal inference: combined with geometric tools from Abramsky-Brandenburger \cite{abramsky2017contextual}, it allows us to quantitatively investigate the correlation between indefinite causality and non-locality/contextuality.
This gives rise to novel methods to certify the non-classicality of causation, of particular interests in scenarios where quantum theory is endowed with the possibility of superposing the causal order of quantum channels.
Differently from previous literature on the topic \cite{araujo2015witnessing,dourdent2021semidevice,zych2019bell}, however, the phenomenology involved in our certification of indefinite causality is entirely theory independent.

\subsection{Background on no-signalling and (indefinite) causality}

The impossibility of superluminal signalling---required for compatibility with special and general relativity---can abstracted away from quantum theory and studied as a theory-independent principle.
Early examples of this abstraction can be found as early as 1990s work by Popescu, Rohrilch and co-authors \cite{elitzur1992quantum,popescu1994quantum,popescu1998causality}, and it has since been the topic of innumerable works in quantum foundations, including (but by no means limited to) \cite{chiribella2010probabilistic,chiribella2011informational,abramsky2011sheaf,abramsky2011cohomology,aolita2012fully,abramsky2014no,coecke2014terminality,chiribella2016quantum,coecke2016terminality}.

The connection between no-signalling correlations and polytopes is first characterised in mid-2000s work by Barrett and co-authors \cite{barrett2005nonlocal,barrett2006maximally}, which paved the way for systematic study of the field.
As the authors of \cite{barrett2005nonlocal} point out, the generalisation of no-signalling condition from bipartite to multipartite is not quite straightforward.
For example, consider the three equations below, imposing tripartite no-signalling constraints on a conditional probability distribution (in the case of binary inputs and outputs):
\begin{equation}
\label{eqn:nosignalling}
\begin{array}{rcl}
\sum_{a}p(a,b,c|x,y,z) &=& \sum_a p(a,b,c|x',y,z) \ \  \forall \ b,c,y,z,x,x' \in \{0,1\} \\
\sum_{b}p(a,b,c|x,y,z) &=& \sum_b p(a,b,c|x,y',z) \ \  \forall \ a,c,x,z,y,y' \in \{0,1\} \\
\sum_{c}p(a,b,c|x,y,z) &=& \sum_c p(a,b,c|x,y,z') \ \  \forall \ a,b,x,y,z,z' \in \{0,1\}
\end{array}
\end{equation}
One could attempt to operationally describe these no-signalling constraints by saying that none of the three parties involved can signal to any other, but this turns out to be incorrect.
Indeed, we could consider the following conditional probability distribution:
\begin{center}
\scalebox{0.8}{
    \begin{tabular}{c|cccc|cccc}
    CAB  & 000 & 001 & 010 & 011 & 100 & 101 & 110 & 111 \\ \hline
    000  & 1/4 & 0   & 0   & 1/4 & 1/4 & 0   & 0   & 1/4 \\
    001  & 1/8 & 1/8 & 1/8 & 1/8 & 1/8 & 1/8 & 1/8 & 1/8 \\
    010  & 1/8 & 1/8 & 1/8 & 1/8 & 1/8 & 1/8 & 1/8 & 1/8 \\
    011  & 1/4 & 0   & 0   & 1/4 & 0   & 1/4 & 1/4 & 0   \\ \hline
    100  & 0   & 1/4 & 1/4 & 0   & 0   & 1/4 & 1/4 & 0   \\
    101  & 1/8 & 1/8 & 1/8 & 1/8 & 1/8 & 1/8 & 1/8 & 1/8 \\
    110  & 1/8 & 1/8 & 1/8 & 1/8 & 1/8 & 1/8 & 1/8 & 1/8 \\
    111  & 1/4 & 0   & 0   & 1/4 & 0   & 1/4 & 1/4 & 0   \\
    \end{tabular}
}
\end{center}
The distribution above is consistent with the requirement of no signalling from any party to any other \textit{individual} party, but it nonetheless fails to satisfy the tripartite no-signalling constraints from Equations \ref{eqn:nosignalling}.
Indeed, while there is no signalling from Charlie to Alice or from Charlie to Bob---a fact that can be easily established by looking at the Charlie-Alice and Charlie-Bob marginals---there is signalling from Charlie to Alice \textit{and} Bob, jointly: the outcomes of the latter are perfectly correlated or anti-correlated, depending on the input of the former.
One of the objectives of our work is to investigate the fine-grained structure of causality intervening in examples such as the above, featuring fewer restrictions than genuine tripartite no-signalling, but more restrictions than any other fixed causal order between tree events.

No-signalling constraints are a theory-independent tool: they are linear constraints prescribing well-definition of conditional probability distribution marginals for independent subsystems.
Of special interest, for practical reasons, is the subset of no-signalling distributions which are quantum-realisable, i.e. those which could be (at least in principle) obtained from quantum experiments.
Quantum-realisable distributions still form a convex spaces, but one which is no longer a polytope \cite{allcock2009recovering,pitowsky2014quantum,de2015simple,koontong2018geometry} and is significantly more difficult to characterise \cite{navascues2007bounding,navascues2008convergent}.
Because our approach is theory-independent, we will not explicitly concern ourselves with the question of quantum-realisability, which we leave to future extensions of the semidefinite programming work by Navascues and co-authors \cite{navascues2007bounding,navascues2008convergent} to characterise.

Much like no-signalling conditions gain finer gradations in the passage from bipartite to multipartite, so does the notion of non-locality.
Indeed, Barrett and co-authors \cite{barrett2005nonlocal} already make the observation---originally by Svetlichny \cite{svetlichny1987distinguishing}, then generalised in \cite{collins2002bell}---that full tripartite locality is too restrictive too capture all examples of interest.
Several proposals for multipartite generalisation of locality have been formulated since, where common ``sources'' are distributed to specific subsets of agents \cite{fritz2012beyond,tavakoli2015quantum,wolfe2016inflation,renou2019genuine,gisin2020constraints,pozaskerstjens2021network}, and it has been shown \cite{fritz2012beyond,branciard2012bilocal,fraser2018causal,renou2019genuine} that---subject to the extra assumption of factorisability of such common sources---quantum non-locality is detectable without reference to free choice of inputs.
This is different from the approach presented in our work, where free choice of inputs is an essential ingredient, a necessary requirement for the witnessing of non-local behaviour.
Furthermore, it is known \cite{renou2019genuine,fritz2012beyond} that the requirement of factorisability for common sources can result in failure for the space of correlations to be convex: instead, our generalisation of non-locality is always captured by convex polytopes, enabling the use of convex geometry and linear programming in its study.

Proposals for a broader causal generalisation of non-locality have also been put forward, mainly in the context of causal Bayesian networks \cite{henson2014theory,fritz2016beyond}.
Fritz \cite{fritz2016beyond}, for example, shows that the relevant assumptions of a Bell-like scenario can be directly derived from the following network:
\begin{center}
\begin{tikzpicture}
	\begin{pgfonlayer}{nodelayer}
		\node [style=big dot] (0) at (-1.25, 1.25) {$a$};
		\node [style=big dot] (1) at (0, -1) {$s$};
		\node [style=big dot] (2) at (1.25, 1.25) {$b$};
		\node [style=big dot] (3) at (-2.5, -1) {$x$};
		\node [style=big dot] (4) at (2.5, -1) {$y$};
	\end{pgfonlayer}
	\begin{pgfonlayer}{edgelayer}
		\draw (3) to (0);
		\draw (1) to (0);
		\draw (1) to (2);
		\draw (4) to (2);
	\end{pgfonlayer}
\end{tikzpicture}
\end{center}
In this particular case, an assignment of value distributions to each vertex factors over subsets with disjoint causal past if and only if it satisfies the traditional no-signalling conditions.
We don't use causal Bayesian networks in our work, instead relying on a combinatorial generalisation of causal orders known as ``spaces of input histories'', introduced in \cite{gogioso2022combinatorics}.

Our generalisation of causality, non-locality and contextuality is not limited to definite causal order, but instead extends to indefinite causal order and dynamical causal constraints.
The theory-independent and device-independent certification of indefinite causal order is the subject of \textit{causal inequalities}, originally developed in the study of process matrices \cite{oreshkov2012quantum}.
Causal inequalities are linear bounds obeyed by distributed protocol supported by a single definite causal order, or more generally by a probabilistic mixture thereof, and several process matrices violating them have been documented in the literature \cite{oreshkov2012quantum,baumeler2014perfect,baumeler2016the}.
For example, the GYNI (Guess Your Neighbour's input) causal inequality \cite{oreshkov2012quantum} can be derived in the context of a bipartite game, where Bob is tasked to either communicate a bit to Alice, or to guess Alice's bit, depending on the value of a classical random variable $\lambda \in \{0,1\}$.
If we assume that Alice causally precedes Bob, or that Bob causally precedes Alice, it can be shown that the probability of success for this game is always bounded by $3/4$.
A recipe to violate this bound is given by the following OCB protocol \cite{oreshkov2012quantum}:
\begin{itemize}
    \item Alice measures the qubit she receives in the Z basis, using the measurement outcome as her classical output $o_\ev{A} \in \{0,1\}$. She then encodes her (freely chosen) classical input $i_\ev{A} \in \{0,1\}$ into the Z basis of a qubit as $i_\ev{A} \mapsto |i_\ev{A}\rangle$, a qubit which she then forwards.
    \item If $\lambda = 0$, Bob measures the qubit he receives in the X basis, using the measurement outcome as his classical output $o_\ev{B} \in \{0,1\}$ (mapping $\langle\!+\!| \mapsto 0$ and $\langle\!-\!| \mapsto 1$). He then encodes his (freely chosen) classical input $i_\ev{B} \in \{0,1\}$ into the Z basis of a qubit as $i_\ev{B} \mapsto |i_\ev{B}\oplus o_\ev{B}\rangle$, a qubit which he then forwards.
    \item If $\lambda = 1$, Bob measures the qubit he receives in the Z basis, using the measurement outcome as his classical output $o_\ev{B} \in \{0,1\}$. He prepares a qubit in the state $|0\rangle$, independently of his classical input, and forward the qubit.
\end{itemize}
This recipe above is turned into a pair of quantum instruments, which are then fed into the following process matrix:
\begin{equation}
W^{A_1 A_2 B_1 B_2} = 1/4 \left[ 1^{A_1 A_2 B_1 B_2} + 1/\sqrt{2}\left(\sigma_z^{A_2}\sigma_z^{B_1} +  \sigma_z^{A_1}\sigma_x^{B_1}\sigma_z^{B_2}\right) \right]
\end{equation}
The probability of success for the GINY game using the setup above is $\left(2+ \sqrt{2}\right)/4 > 3/4$, proving violation of the GINY causal inequality.

Aside from a brief foray into the OCB process discussed above and the BFW empirical model from \cite{baumeler2016the}, the examples presented in this work are all quantum-realisable (at least in the sense of quantum indefinite causality).
More specifically, the examples used to demonstrate the phenomenon of ``contextual causality'' are all based on quantum switches \cite{oreshkov2012quantum,chiribella2013quantum}, with either entangled control or contextual control.
A version of the single switch with entangled control (cf. Subsubsection \ref{subsubsection:ghzswitch-2-1}, p.\pageref{subsubsection:ghzswitch-2-1}) was also recently investigated in \cite{tein2022device}, where a CHSH-like argument is used---in conjunction with the software PANDA \cite{mckinney2010data}---to derive an explicit characterisation for some of the faces of the corresponding causaltope.
Indefinite causality is certified via the inequalities associated to the causaltope faces: this differs from our equation-based approach, and is instead analogous to previous literature on causal inequalities and non-locality.

\section{The geometry of causality}
\label{section:geometry-causality}

The convex space $\Dist{X}$ formed by probability distributions on a finite set $X$ is known as a simplex, and it is arguably the simplest (pun not intended) example of a polytope.
Almost as simple are the polytopes obtained by finite products of simplices (a.k.a. simplexes): these are the convex spaces $\prod_{j \in J}\Dist{X_j}$ describing probability distributions conditional on a finite sets $J$.
These polytopes will---to first approximation---host our empirical models: the main result of this Section will be the explicit characterisation of the exact convex subspace spanned by the empirical models, on arbitrary covers, for arbitrary spaces of input histories.
As it turns out, this convex subspace is always a polytope, which we will refer to as a ``causaltope'', a portmanteau of ``causal polytope''.

A standard way to describe polytopes is to provide inequalities corresponding to their faces.
Unfortunately, the number of faces and vertices for a polytope can grow exponentially in the embedding dimension, making this description hard to handle in the general case: in this sense, our causaltopes are no exception.
Our key observation, however, will be that causaltopes admit an equivalent, more compact characterisation: instead of describing them via the ``causal inequalities'' that define their faces, we obtain them by slicing simpler polytopes in higher dimensions, using ``causal equations'' which we derive from the topology of the underlying spaces of input histories.
This simpler, explicit characterisation allows us to systematically derive interesting properties of causaltopes and to efficiently compute the causal fractions of arbitrary empirical models, both over individual causaltopes and over convex hulls of multiple causaltopes.

\subsection{Polytopes}
\label{subsection:geometry-causality-polytopes}

We start this by recalling definitions and basic facts about polytopes.
Different mathematical literature uses ``polytope'' to refer to slightly different families of objects: here, we use ``polytope'' to refer to a bounded, closed and convex polytope, embedded in a real vector space.

\begin{definition}
For any finite set $J$, we define $\reals^J$ to be the finite-dimensional real vector space formed by functions $J \rightarrow \reals$ under pointwise addition and scalar multiplication.
We adopt the Kronecker delta functions as the standard basis for this space:
\[
    \underline{\delta}_i
    := j \mapsto \left\{\begin{array}{rl}
        1 & \text{ if } i = j\\
        0 & \text{ if } i \neq j
    \end{array}\right.
\]
If $\underline{x} \in \reals^J$, we write $x_j$ for the $j$-th component of $\underline{x}$ in the standard basis, for every $j \in J$:
\[
    x_j := \underline{x}(j) \in \reals
\]
We take $\reals^J$ to be equipped with the inner product for the standard basis:
\[
    \underline{x}^T \underline{y}
    := \sum_{j \in J} x_j y_j
\]
We also take $\reals^J$ to be equipped with the product order:
\[
    \underline{x} \leq \underline{y}
    \stackrel{def}{\Leftrightarrow}
    \forall j \in J.\; x_j \leq y_j
\]
For every $n \in \nats$, we write $\reals^n$ to denote $J := \{1,...,n\}$, where $n=0$ means $J=\emptyset$.
\end{definition}

\begin{remark}
The choice of inner product allows us to define linear equations using matrices, as $A\underline{x} = \underline{b}$, where $\underline{b} \in \reals^R$ and $R$ is a finite set indexing the rows of $A$.
This is equivalent to the following, more verbose formulation:
\[
    A\underline{x} = \underline{b}
    \Leftrightarrow
    \forall r \in R. \sum_{j \in J} A_{rj} x_j = b_r
\]
Furthermore, the choice of product order allows us to define linear inequalities using matrices, as $A\underline{x} \leq \underline{b}$.
This is equivalent to the following, more verbose formulation:
\[
    A\underline{x} \leq \underline{b}
    \Leftrightarrow
    \forall r \in R. \sum_{j \in J} A_{rj} x_j \leq b_r
\]
Sometimes, for additional clarity, it will be convenient to write such systems explicitly, such as:
\[
    \left(
    \begin{array}{ccc}
    A_{11}&\dots &A_{1n}\\
    \vdots&\ddots&\vdots\\
    A_{m1}&\dots &A_{mn}
    \end{array}
    \right)
    \left(
    \begin{array}{c}
    x_1\\\vdots\\x_n
    \end{array}
    \right)
    \leq
    \left(
    \begin{array}{c}
    b_1\\\vdots\\b_m
    \end{array}
    \right)
\]
When doing so, we will implicitly assume that bijections $J \leftrightarrow \{1,...,n\}$ and $R \leftrightarrow \{1,...,m\}$ have been fixed, where $n:=|J|$ and $m:=|R|$ denote number of elements in $J$ and $R$, respectively (i.e. the number of columns and rows, respectively).
\end{remark}

\begin{definition}
A \emph{polytope} is any bounded subset $K \subset \reals^J$ defined by the joint solutions $\underline{x} \in \reals^J$ to a system $A \underline{x} = \underline{b}$ of linear equations and a system $C \underline{x} \leq \underline{d}$ of linear inequalities:
\begin{equation}
    K =
    \suchthat{
        \underline{x} \in \reals^J
    }{
        A \underline{x} = \underline{b}
        \text{ and }
        C \underline{x} \leq \underline{d}
    }
\end{equation}
We refer to $\reals^J$ as the \emph{embedding space} and say that $K$ is \emph{embedded} in $\reals^J$.
\end{definition}

\begin{remark}
The choice of direction $\leq$ for the inequalities is merely a matter of convention, and inequalities $\underline{c}^T\underline{x} \geq d$ in the other direction can be expressed as $-\underline{c}^T\underline{x} \leq -d$.
It is also possible for the system of equations to be empty, i.e. for $A$ to have no rows.
\end{remark}

As our first example, we consider the case for the \emph{standard hypercube} $[0,1]^J \subset \reals^J$, defined by the inequalities $x_j \leq 1$ and $x_j \geq 0$ for all $j \in J$:
\[
    \left(\begin{array}{ccc}
    1 &&\\
    &\ddots&\\
    &&1\\
    -1&&\\
    &\ddots&\\
    &&-1
    \end{array}\right)
    \left(
    \begin{array}{c}
    x_1\\\vdots\\x_n
    \end{array}
    \right)
    \leq
    \left(\begin{array}{c}
    1\\\vdots\\1\\0\\\vdots\\0
    \end{array}\right)
\]
More generally, for every $\underline{u} \in \reals^J$ with $u_j > 0$ for all $j \in J$, we can define the \emph{standard hypercuboid} $\prod_j[0,u_j] \subset \reals^J$, where the upper-bounding inequalities $x_j \leq 1$ for the unit hypercube are replaced by $x_j \leq u_j$:
\[
    \left(\begin{array}{ccc}
    1 &&\\
    &\ddots&\\
    &&1\\
    -1&&\\
    &\ddots&\\
    &&-1
    \end{array}\right)
    \left(
    \begin{array}{c}
    x_1\\\vdots\\x_n
    \end{array}
    \right)
    \leq
    \left(\begin{array}{c}
    u_1\\\vdots\\u_n\\0\\\vdots\\0
    \end{array}\right)
\]
Finally, we can consider the \emph{standard simplex} $\Delta^J \subset \reals^J$, which is defined by the same lower-bounds $x_j \geq 0$ as the previous examples, but with the single upper-bound $\sum_{j \in J} x_j \leq 1$:
\[
    \left(\begin{array}{ccc}
    1&\dots&1\\
    -1&&\\
    &\ddots&\\
    &&-1\\
    \end{array}\right)
    \left(
    \begin{array}{c}
    x_1\\\vdots\\x_n
    \end{array}
    \right)
    \leq
    \left(\begin{array}{c}
    1\\0\\\vdots\\0
    \end{array}\right)
\]
The convex sub-space of $\reals^J$ defined by $x_j \geq 0$ for all $j \in J$ is known as the \emph{positive cone} of $\reals^J$ and denoted by $(\reals^+)^J$. It is an example of an ``unbounded'' polytope.

In the three examples above, no equations were involved, the inequalities were consistent and no pairs of inequalities combined to form an equation: the polytopes are non-empty, bounded, closed regular subsets of $\reals^n$.
This is the case covered by the more common ``half-space description'' of a polytope, as the intersection of a unique ``essential'' family of half-spaces, defined by linear inequalities.
The half-space description has the advantage of being canonical, but it is too restrictive for our purposes: if we wish to ``slice'' our polytopes by imposing additional equations in this description, we have to explicitly transform the defining inequalities to ones on the resulting affine subspace, losing all information about the original embedding.

The description we adopted is closer to the formulation of convex polytopes used by linear programming, and it affords us the freedom of imposing equations without changing the existing inequalities.
However, we pay for this additional freedom with lack of canonicity:
\begin{itemize}
    \item Some of the equations or inequalities could be redundant.
    \item Equations are not necessary: $\underline{a}^T\underline{x} = b$ can be replaced by $\underline{a}^T\underline{x} \leq b$ and $\underline{a}^T\underline{x} \geq b$.
    \item Inequalities can pair up into equations, as above.
\end{itemize}
That said, all interesting examples of polytopes in this work are defined by a mix of equations and inequalities, making the linear programming formulation more appealing.
Indeed, our calculation of causal fractions will take the form of linear programs.

As an example of a polytope which is naturally defined by both equations and inequalities, we consider the space $\Dist{J} \subset \reals^J$ of probability distributions over the set $J$.
The probability $x_j$ of each $j \in J$ must be a non-negative real number (i.e. $x_j \geq 0$) and probabilities must sum to unity (i.e. $\sum_{j \in J} x_j = 1$):
\[
    \Big(\;1 \;\;\dots\;\; 1\;\Big)
    \left(
    \begin{array}{c}
    x_1\\\vdots\\x_n
    \end{array}
    \right)
    =
    1
    \hspace{2cm}
    \left(\begin{array}{ccc}
    -1&&\\
    &\ddots&\\
    &&-1\\
    \end{array}\right)
    \left(
    \begin{array}{c}
    x_1\\\vdots\\x_n
    \end{array}
    \right)
    \leq
    \left(\begin{array}{c}
    0\\\vdots\\0
    \end{array}\right)
\]
More generally, we wish to consider the space $\prod_{y \in Y} \Dist{J^{(y)}}$ of probability distributions conditional on some finite set $Y$, where we allow $J^{(y)}$ to also depend on the choice of $y \in Y$.
The construction of this space is obtained from the following result.

\begin{observation}
Let $Y$ be a non-empty finite set and let $K^{(y)} \subset R^{J^{(y)}}$ be a family of polytopes indexed by $y \in Y$, each defined by its own system of linear equations $A^{(y)}\underline{x}^{(y)}=\underline{b}^{(y)}$ and linear inequalities $C^{(y)}\underline{x}^{(y)}\leq\underline{d}^{(y)}$.
The \emph{product polytope} $\prod_{y \in Y} K^{(y)}$ is embedded in $\prod_{y \in Y}\reals^{J^{(y)}} = \reals^{\sqcup_{y \in Y} J^{(y)}}$, where the disjoint union $\sqcup_{y \in Y} J^{(y)}$ is formally defined as follows:
\begin{equation}
    \bigsqcup_{y \in Y} J^{(y)}
    :=\suchthat{(y,j)}{y \in Y, j \in J^{(y)}}
\end{equation}
The product polytope $\prod_{y \in Y} K^{(y)}$ is defined by the following equations and inequalities:
\begin{equation}
    \forall y \in Y.\;
    A^{(y)} \underline{x}^{(y)} = \underline{b}^{(y)}
\end{equation}
\begin{equation}
    \forall y \in Y.\;
    C^{(y)} \underline{x}^{(y)} \leq \underline{d}^{(y)}
\end{equation}
Above, we have indexed the coordinates of vectors $\underline{x} \in \reals^{\sqcup_{y \in Y} J^{(y)}}$ as $x^{(y)}_{j}$, for $y \in Y$ and $j \in J^{(y)}$, and we have defined:
\[ 
    \underline{x}^{(y)}
    := \left(x^{(y)}_{j}\right)_{j \in J^{(y)}}
    \in \reals^{J^{(y)}}
\]
\end{observation}

According to the observation above, and compatibly with the definition of conditional probability distributions, the product polytope $\prod_{y \in Y} \Dist{J^{(y)}} \subset \reals^{\sqcup_{y \in Y} J^{(y)}}$ is defined by the following families of equations and inequalities:
\[
    \forall y \in Y.\;
    \sum_{j \in J^{(y)}} x^{(y)}_j = 1
    \hspace{2cm}
    \forall y \in Y.\;
    \forall j \in J^{(y)}.\;
    x^{(y)}_j \geq 0
\]
Thinking concretely in terms of vectors and matrices, the product polytope is defined by combining the equations and inequalities of its factors in a block-diagonal way:
\[
    \scalebox{0.75}{$
        \left(\begin{array}{cccc}
        A^{(1)} & 0 & \dots & 0\\
        0 & A^{(2)} & \dots & 0\\
        0 & 0 & \ddots & 0 \\
        0 & 0 & \dots & A^{(m)}
        \end{array}\right)
        \left(\begin{array}{c}
        \underline{x}^{(1)}\\\vdots\\\underline{x}^{(m)}
        \end{array}\right)
        =
        \left(\begin{array}{c}
        \underline{b}^{(1)}\\\vdots\\\underline{b}^{(m)}
        \end{array}\right)
    $}
    \hspace{1cm}
    \scalebox{0.75}{$
        \left(\begin{array}{cccc}
        C^{(1)} & 0 & \dots & 0\\
        0 & C^{(2)} & \dots & 0\\
        0 & 0 & \ddots & 0 \\
        0 & 0 & \dots & C^{(m)}
        \end{array}\right)
        \left(\begin{array}{c}
        \underline{x}^{(1)}\\\vdots\\\underline{x}^{(m)}
        \end{array}\right)
        \leq
        \left(\begin{array}{c}
        \underline{d}^{(1)}\\\vdots\\\underline{d}^{(m)}
        \end{array}\right)
    $}
\]
Above, we have implicitly fixed a bijection $Y \leftrightarrow \{1,...,m\}$.
As a special case, we obtain an explicit description for the space $\prod_{y \in Y} \Dist{J^{(y)}} \subset \reals^{\sqcup_{y \in Y} J^{(y)}}$, where $I$ is the identity matrix on $\reals^J$ and we write $\underline{1}^T := (1,\;\dots\;,1)$:
\[
    \scalebox{0.75}{$
        \left(\begin{array}{cccc}
        \underline{1}^T & 0 & \dots & 0\\
        0 & \underline{1}^T & \dots & 0\\
        0 & 0 & \ddots & 0 \\
        0 & 0 & \dots & \underline{1}^T
        \end{array}\right)
        \left(\begin{array}{c}
        \underline{x}^{(1)}\\\vdots\\\underline{x}^{(m)}
        \end{array}\right)
        =
        \left(\begin{array}{c}
        1\\\vdots\\1
        \end{array}\right)
    $}
    \hspace{1cm}
    \scalebox{0.75}{$
        \left(\begin{array}{cccc}
        -I & 0 & \dots & 0\\
        0 & -I & \dots & 0\\
        0 & 0 & \ddots & 0 \\
        0 & 0 & \dots & -I
        \end{array}\right)
        \left(\begin{array}{c}
        \underline{x}^{(1)}\\\vdots\\\underline{x}^{(m)}
        \end{array}\right)
        \leq
        \left(\begin{array}{c}
        0\\\vdots\\0
        \end{array}\right)
    $}
\]

Every polytope $K \subset \reals^J$ admits an alternative, equivalent characterisation as a convex hull, consisting of all convex-linear combinations of a finite set of points in $\reals^J$.
If the points are required to be convex-linearly independent, then the set of points is unique, defining the vertices of the polytope.

\begin{observation}
Let $K \subset \reals^J$ be a polytope.
There exists a unique (necessarily finite) set of convex-linearly independent points $\Vertices{K} \subset \reals^J$ for which $K$ is the convex hull, i.e. such that the points in $K$ are exactly all possible convex combinations of points in $\Vertices{K}$:
\begin{equation}
    K = \suchthat{
        \sum_{\underline{v} \in \Vertices{K}} d\left(\underline{v}\right) \underline{v} 
    }{
        d \in \Dist{\Vertices{K}}
    }
\end{equation}
We refer to the points in $\Vertices{K}$ as the \emph{vertices} of $K$.
By convex-linear independence of $\Vertices{K}$ we mean that no $\underline{v} \in \Vertices{K}$ can be obtained as a convex combination of points in $\Vertices{K}\backslash\{\underline{v}\}$.
\end{observation}

For the standard hypercube $[0,1]^J$, the vertices are the $2^{|J|}$ vectors in $\reals^J$ with coordinates given by all possible combinations of 0s and 1s:
\[
    \Vertices{[0,1]^J}
    = \{0,1\}^J
    = \suchthat{\underline{x} \in \reals^J}{\forall j \in J.\; x_j \in \{0,1\}}
\]
Analogously, the vertices of the standard hypercuboid $\prod_{j\in J}[0, u_j]$ are the $2^{|J|}$ vectors in $\reals^J$ with coordinates given by all possible combinations of 0s and $u_j$s:
\[
    \Vertices{\prod_{j \in J} [0, u_j]}
    = \prod_{j \in J} \{0, u_j\}
    = \suchthat{\underline{x} \in \reals^J}{\forall j \in J.\; x_j \in \{0, u_j\}}
\]
The vertices of the standard simplex $\Delta^J$ are a subset of those for the standard hypercube $[0,1]^J$, consisting of the $|J|+1$ vertices where at most one coordinate is 1:
\[
    \Vertices{\Delta^{J}}
    = \suchthat{\underline{x} \in \Vertices{[0,1]^J}}{\sum_{j \in J} x_j \leq 1}
\]
In other words, the vertices of the standard simplex $\Delta^J$ are the Kronecker delta functions together with the zero vector:
\[
    \Vertices{\Delta^{J}}
    = \{\underline{0}\}\cup\suchthat{\underline{\delta}_i}{i \in J}
\]
The vertices of the polytope of probability distributions $\Dist{J}$ are the $|J|$ non-zero vertices of $\Delta^J$, i.e. the Kronecker delta functions:
\[
    \Vertices{\Dist{J}}
    = \suchthat{\underline{\delta}_i}{i \in J}
\]
We can use Kronecker deltas to embed functions $f \in \prod_{y \in Y} J^{(y)}$ into the corresponding indicator vectors $\underline{\delta}_f \in \reals^{\sqcup_{y \in Y} J^{(y)}}$:
\begin{equation}
    \label{equation:kronecker-delta-extended}
    \left(\underline{\delta}_f\right)^{(y)}_j
    :=
    \left\{\begin{array}{rl}
    1 & \text{ if } f(y) = j\\
    0 & \text{ if } f(y) \neq j
    \end{array}\right.    
\end{equation}
The $\prod_{y \in Y} |J^{(y)}|$ indicator vectors defined above are exactly the vertices of the polytope $\prod_{y \in Y}\Dist{J^{(y)}}$ of conditional probability distributions:
\[
    \Vertices{\prod_{y \in Y}\Dist{J^{(y)}}}
    = \suchthat{\underline{\delta}_f}{f \in  \prod_{y \in Y} J^{(y)}}
\]
The unconditional case of $\Dist{J}$ is recovered by taking a singleton $Y = \{\ast\}$ and identifying functions $f: Y \rightarrow J$ with the corresponding elements $f(\ast) \in J$, so that the vectors $\underline{\delta}_{\ast \mapsto i}$ from Equation \ref{equation:kronecker-delta-extended} are identified with the Kronecker delta functions $\underline{\delta}_{i}$.

Polytopes $K \subset \reals^J$ are compact manifolds with boundary, embedded into a an affine subspace of $\reals^J$: when we talk about an $n$-dimensional polytope, we mean a polytope which is an $n$-dimensional manifold with boundary.
Because the topological dimension of $K$ might be strictly lower than that of the embedding space $\reals^J$, topological notions such as boundary and interior must be defined within a suitable subspace.
Luckily, there is a canonical choice for this subspace, obtainable from the equations and inequalities defining the polytope.

\begin{observation}
There is a unique minimal affine subspace $\AffSubsp{K} \subset \reals^J$ which contains a given polytope $K \subset \reals^J$, and it can be explicitly identified by the following procedure:
\begin{enumerate}
    \item Replace every pair of inequalities in the form $\underline{c}^T\underline{x} \leq d$ and $-\underline{c}^T\underline{x} \leq -d$ with the corresponding equation $\underline{c}^T\underline{x} = d$.
    \item Turn the system of equations into reduced row echelon form (RREF), removing zero rows.
\end{enumerate}
The affine subspace $\AffSubsp{K}$ is the space of solutions for the resulting system of equations, and there is a bijection between affine subspaces and systems of equation in RREF without zero rows.
The polytope $K$ is a regular closed subset of $\AffSubsp{K}$: the topological dimension of $K$ is the dimension of $\AffSubsp{K}$, which is equal to $|J|$ minus the number of non-zero rows in the system of equations in RREF.
\end{observation}

The standard hypercube $K=[0,1]^{J}$, the standard hypercuboids $K=\prod_{j \in J}[0, u_j]$ and the standard simplex $K=\Delta^J$ all have dimension $|J|$, with $\AffSubsp{K} = \reals^J$.
The polytope $\Dist{J}$ of probability distributions has dimension $|J|-1$, with the following minimal affine subspace:
\[
\AffSubsp{\Dist{J}}
= \suchthat{\underline{x} \in \reals^J}{\sum_{j \in J} x_j = 1}
\]
The polytope $\prod_{y \in Y}\Dist{J^{(y)}}$ of conditional probability distributions has dimension $\sum_{y \in Y}(|J^{(y)}|-1) = \left(\sum_{y \in Y}|J^{(y)}|\right)-|Y|$, with the following minimal affine subspace:
\[
\AffSubsp{\prod_{y \in Y}\Dist{J^{(y)}}}
= \suchthat{\underline{x} \in \reals^{\sqcup_{y \in Y} J^{(y)}}}{\forall y \in Y.\;\sum_{j \in J^{(y)}} x^{(y)}_j = 1}
\]
Note that the system of equations previously presented for this last example was already in RREF, without any zero rows.

Having established a suitable embedding subspace, we turn our attention to the boundary of polytopes.
It is an important fact that the boundary of an $n$-dimensional polytope $K$ is the union of a finite number of $(n-1)$-dimensional polytopes, its \emph{faces}: these can be obtained combinatorially, as the convex hull of certain subsets of the vertices of $K$, or geometrically, by intersecting $K$ with certain affine subspaces.

\begin{definition}
Let $K \subset \reals^J$ be a polytope and let $\AffSubsp{K} \subset \reals^J$ be the minimal affine subspace which contains $K$.
The \emph{boundary} of polytope $K$ is defined to be the topological boundary $\boundary{K}$ of $K$ within $\AffSubsp{K}$.
Similarly, the \emph{interior} of polytope $K$ is defined to be the topological interior $\interior{K}$ of $K$ within $\AffSubsp{K}$.
\end{definition}

\begin{observation}
The boundary of an $n$-dimensional polytope $K$, where $n \geq 1$, is the union $\boundary{K} = \bigcup \Faces{K}$ of a unique finite set $\Faces{K}$ of $(n-1)$-dimensional polytopes, where the vertices of each $F \in \Faces{K}$ are themselves vertices of $K$, i.e. $\Vertices{F} \subset \Vertices{K}$.
We refer to the polytopes in $\Faces{K}$ as the \emph{faces} of $K$.
\end{observation}

\begin{observation}
Let $K \subset \reals^J$ be an $n$-dimensional polytope, where $n \geq 1$.
Every face $F \in \Faces{K}$ can be obtained by intersecting $K$ with a unique affine subspace of $\AffSubsp{K}$.
Equivalently, it can be obtained by turning some inequality $\underline{c}^T\underline{x} \leq d$ defining $K$ into an equation $\underline{c}^T\underline{x} = d$ (excluding inequalities which already pair up to form an equation). 
\end{observation}

The standard hypercube $[0, 1]^J$ has $2|J|$ faces, each face corresponding to either one of the $|J|$ inequalities $x_j \geq 0$ or one of the $|J|$ inequalities $x_j \leq 1$:
\begin{itemize}
    \item The face corresponding to $x_j \geq 0$ has vertices $\suchthat{\underline{x} \in \Vertices{[0,1]^J}}{x_j = 0}$.
    \item The face corresponding to $x_j \leq 1$ has vertices $\suchthat{\underline{x} \in \Vertices{[0,1]^J}}{x_j = 1}$.
\end{itemize}
Analogously, the standard hypercuboid $\prod_{j \in J}[0, u_j]$ has $2|J|$ faces, each face corresponding to either one of the $|J|$ inequalities $x_j \geq 0$ or one of the $|J|$ inequalities $x_j \leq u_j$:
\begin{itemize}
    \item The face corresponding to $x_j \geq 0$ has vertices $\suchthat{\underline{x} \in \Vertices{\prod_{j' \in J}[0, u_{j'}]}}{x_j = 0}$.
    \item The face corresponding to $x_j \leq u_j$ has vertices $\suchthat{\underline{x} \in \Vertices{\prod_{j' \in J}[0, u_{j'}]}}{x_j = u_j}$.
\end{itemize}
The standard simplex $\Delta^J$ has $|J|+1$ faces, $|J|$ faces corresponding to the inequalities $x_j \geq 0$ and one face corresponding to the inequality $\sum_{j \in J} x_j \leq 1$:
\begin{itemize}
    \item The face corresponding to $x_j \geq 0$ has vertices $\suchthat{\underline{x} \in \Vertices{\Delta^J}}{x_j = 0}$, which can be equivalently written as $\Vertices{\Delta^J}\backslash\{\underline{\delta}_j\}$.
    \item The face corresponding to $\sum_{j \in J} x_j \leq 1$ has vertices $\suchthat{\underline{x} \in \Vertices{\Delta^J}}{\sum_{j \in J} x_j = 1}$, which can be equivalently written as $\Vertices{\Delta^J}\backslash\{\underline{0}\}$.
\end{itemize}
The polytope of probability distributions $\Dist{J}$ has $|J|$ faces.
Each face corresponds to the inequality $x_j \geq 0$, for some $j \in J$, and it has the following vertices:
\[
\suchthat{\underline{x} \in \Vertices{\Dist{J}}}{x_j = 0}
=\suchthat{\underline{\delta}_i}{i \in J, i \neq j}
\]
The polytope of conditional probability distributions $\prod_{y \in Y}\Dist{J^{(y)}}$ has $\sum_{y \in Y}|J^{(y)}|$ faces.
Each face corresponds to the inequality $x^{(y)}_j \geq 0$, for some $y \in Y$ and $j \in J^{(y)}$, and it has the following vertices:
\[
\suchthat{\underline{x} \in \Vertices{\prod_{y \in Y}\Dist{J^{(y)}}}}{x^{(y)}_j = 0}
=\suchthat{\underline{\delta}_f}{f \in \prod_{y' \in Y}J^{(y')}, f(y) \neq j}
\]

\begin{remark}
Because each face $F \in \Faces{K}$ of an $n$-dimensional polytope $K$ is an $(n-1)$-dimensional polytope, we can iterate the construction and look at the faces of $F$, which are $(n-2)$-dimensional polytopes, then at their faces in turn, and so on until we reach the $0$-dimensional faces, i.e. the individual vertices of the original polytope $K$.
Together with $K$ at the top and the empty set at the bottom, this assemblage of faces forms a finite lattice under inclusion, known as the \emph{face poset} of $K$: the meet operation in the lattice is the set-theoretic intersection of faces, while the join operation is the convex closure of their set-theoretic union.
If we replace each polytope in the face poset with the subset of $\Vertices{K}$ of which it is convex hull, we obtain a sub-lattice of the powerset $\Subsets{\Vertices{K}}$ isomorphic to the face poset.
If we further forget the coordinates of the points, we are left with an \emph{abstract simplicial complex}, capturing the combinatorial properties of $K$ without reference to its geometric realisation.
However, neither the face poset nor abstract simplicial complexes play a role in this work, so we will not elaborate on them any further.
\end{remark}

\subsection{Constrained conditional probability distributions}
\label{subsection:geometry-causality-ccpd}

It is an important observation that complicated polytopes, such as the causality polytopes studied by this work, can be obtained by ``slicing'' simple higher-dimensional polytopes, intersecting them with affine subspaces.
In this section we recap a selection of results about polytopes conditional probability distributions under additional linear constraints.
These are all known results from linear programming---or simple consequences thereof---but we nevertheless include brief proofs for the benefit of readers less familiar with the field.

\begin{definition}
Let $K \subset \reals^J$ be a polytope.
We say that a polytope $K' \subset \reals^J$ is obtained by \emph{slicing} from $K$ if it takes the form $K' = K \cap W$ for some affine subspace $W \subset \reals^J$.
If we wish to specify the subspace, we say that $K'$ is obtained by \emph{slicing $K$ with $W$}.
We adopt the following notation for it:
\begin{equation}
    \Slice{W}{K} := K \cap W
\end{equation}
\end{definition}

\begin{proposition}
\label{proposition:slices-are-polytopes}
Let $K \subset \reals^J$ be a polytope, defined by a system of equations $A \underline{x} = \underline{b}$ and inequalities $C \underline{x} = \underline{d}$.
Let $W \subset \reals^J$ be an affine subspace, defined by a system of equations $A' \underline{x} = \underline{b}'$.
Then $\Slice{W}{K} \subset \reals^J$ is a polytope, defined by the equations and inequalities for $K$ together with the equations for $W$:
\[
\left(
\begin{array}{c}
A \\\hline A'
\end{array}
\right)
\underline{x}
=
\left(
\begin{array}{c}
b \\\hline b'
\end{array}
\right)
\hspace{3cm}
C \underline{x} \leq \underline{d}
\]
\end{proposition}
\begin{proof}
See \ref{proof:proposition:slices-are-polytopes}
\end{proof}

\begin{proposition}
\label{proposition:slicing-closed-under-iteration}
Let $K \subset \reals^J$ be a polytope and let $V, W \subset \reals^J$ be affine subspaces.
Slicing $\Slice{W}{K}$ with $V$ is the same as slicing $K$ with $W\cap V$:
\[
\Slice{V}{\Slice{W}{K}} = \Slice{V\cap W}{K}
\]
We say that ``slicing is closed under iteration''.
\end{proposition}
\begin{proof}
See \ref{proof:proposition:slicing-closed-under-iteration}
\end{proof}

As a simple example, we consider the polytope of probability distributions $\Dist{J}$.
This polytope is obtained by slicing the standard hypercube $[0,1]^J$ with a single ``normalisation equation'', stating that the entries of the probability distribution must sum to 1:
\[
\Dist{J} = \Slice{W}{[0,1]^J}
\hspace{2cm}
W := \suchthat{\underline{x} \in \reals^J}{\sum_{j \in J} x_j = 1}
\]
More generally, the polytope of conditional probability distributions on some $Y$ is obtained by slicing the standard hypercube with a family of ``normalisation equations'', stating that the entries of the probability distribution conditional to each $y \in Y$ must sum to 1.

\begin{definition}
Let $Y$ be a finite non-empty set and let $\underline{J}=\left(J^{(y)}\right)_{y \in Y}$ be a family of finite non-empty sets.
The corresponding \emph{normalisation equations} are defined as follows:
\[
    \forall y \in Y.\;
    \sum_{j \in J^{(y)}} x^{(y)}_j = 1
\]
We write $\NormEqs{\underline{J}}$ for the affine subspace of $\reals^{\sqcup_{y \in Y} J^{(y)}}$ defined by the equations.
\end{definition}

\begin{proposition}
\label{proposition:probdist-from-slicing-hypercube}
Let $Y$ be a finite non-empty set and let $\underline{J}=\left(J^{(y)}\right)_{y \in Y}$ be a family of finite non-empty sets.
The polytope of conditional probability distributions $\prod_{y \in Y} \Dist{J^{(y)}}$ is obtained by slicing the standard hypercube $[0,1]^{\sqcup_{y \in Y} J^{(y)}} = \prod_{y \in Y} [0,1]^{J^{(y)}}$ with the normalisation equations:
\[
\prod_{y \in Y} \Dist{J^{(y)}}
= [0,1]^{\sqcup_{y \in Y} J^{(y)}} \cap \NormEqs{\underline{J}}
\]
In particular, we have $\AffSubsp{\prod_{y \in Y} \Dist{J^{(y)}}} = \NormEqs{\underline{J}}$.
\end{proposition}
\begin{proof}
See \ref{proof:proposition:probdist-from-slicing-hypercube}
\end{proof}

The characterisation of the polytope $\prod_{y \in Y} \Dist{J^{(y)}}$ of conditional probability distributions by slicing of the standard hypercube provides an interesting introductory example for the technique, but it is not directly useful in this work: in defining our ``causaltopes'', we will take the polytope $\prod_{y \in Y} \Dist{J^{(y)}}$ itself as a starting point.
However, an alternative characterisation of $\prod_{y \in Y} \Dist{J^{(y)}}$, by slicing of the same standard hypercube, affords us the opportunity to introduce a new concept of ``quasi-normalisation equations'', which will play an important role in our causal decompositions.

\begin{definition}
Let $Y$ be a finite non-empty set, with a total order $Y = \{y_1,...,y_n\}$ fixed on it.
Let $\underline{J}=\left(J^{(y)}\right)_{y \in Y}$ be a family of finite non-empty sets.
The corresponding \emph{quasi-normalisation equations} are defined as follows:
\[
    \forall i \in \{1,...,n-1\}.\;
    \sum_{j \in J^{(y_i)}} x^{(y_i)}_j = \sum_{j \in J^{(y_{i+1})}} x^{(y_{i+1})}_j
\]
We write $\QNormEqs{\underline{J}}$ for the affine subspace of $\reals^{\sqcup_{y \in Y} J^{(y)}}$ defined by the equations.
The polytope of \emph{quasi-normalised conditional distributions} is defined by slicing the standard hypercube $[0,1]^{\sqcup_{y \in Y}} J^{(y)}$ with the quasi-normalisation equations:
\[
    [0,1]^{\sqcup_{y \in Y} J^{(y)}} \cap \QNormEqs{\underline{J}}
\]
\end{definition}

\begin{proposition}
\label{proposition:qnorm-eqs-independent-of-order}
Let $Y$ be a finite non-empty set and let $\underline{J}=\left(J^{(y)}\right)_{y \in Y}$ be a family of finite non-empty sets.
The affine subspace $\QNormEqs{\underline{J}}$ is independent of the specific choice of total order for $Y$, and hence so is the polytope of quasi-normalised conditional distributions.
\end{proposition}
\begin{proof}
See \ref{proof:proposition:qnorm-eqs-independent-of-order}
\end{proof}

\begin{proposition}
\label{proposition:mass-of-qnorm-dist}
For each quasi-normalised conditional distribution $\underline{u}$, there exists a unique $\mass{\underline{u}} \in [0,1]$ and a distribution $\underline{e} \in \prod_{y \in Y} \Dist{J^{(y)}}$ such that:
\[
    \underline{u} = \mass{\underline{u}} \underline{e}
\]
We refer to $\mass{\underline{u}}$ as the \emph{mass} of the quasi-normalised distribution $\underline{u}$.
If $\mass{\underline{u}} > 0$, the distribution $\underline{e}$ is furthermore unique. 
\end{proposition}
\begin{proof}
See \ref{proof:proposition:mass-of-qnorm-dist}
\end{proof}

As previously mentioned, the ``causaltopes'' defined in this work are obtained by slicing polytopes of conditional probability distributions with linear subspaces.
Because slicing is the same as applying constraints, we refer to these as ``constrained'' conditional probability distributions.

\begin{definition}
Let $Y$ be a finite non-empty set and let $\underline{J}=\left(J^{(y)}\right)_{y \in Y}$ be a family of finite non-empty sets.
A polytope of \emph{constrained} conditional probability distributions is one in the following form, for some linear subspace $W \subseteq \reals^{\sqcup_{y \in Y} J^{(y)}}$:
\begin{equation}
\CCPD{W,\underline{J}} :=
\Slice{W}{\prod_{y \in Y} \Dist{J^{(y)}}}  
\subseteq \prod_{y \in Y} \Dist{J^{(y)}}
\end{equation}
The corresponding polytope of constrained quasi-normalised conditional distributions takes the following form:
\begin{equation}
\begin{array}{rcl}
\QNCCPD{W,\underline{J}} &:=&
\Slice{W}{[0,1]^{\sqcup_{y \in Y} J^{(y)}} \cap \QNormEqs{\underline{J}}}
\\
&=&\Slice{W \cap \QNormEqs{\underline{J}}}{[0,1]^{\sqcup_{y \in Y} J^{(y)}}}
\end{array}
\end{equation}
\end{definition}

Polytopes of constrained (quasi-normalised) conditional probability distributions have the same inclusion hierarchy as the ``minimal'' linear subspaces that define them:

\begin{proposition}
\label{proposition:hierarchy-of-ccpd-polytopes}
Let $Y$ be a finite non-empty set and let $\underline{J}=\left(J^{(y)}\right)_{y \in Y}$ be a family of finite non-empty sets.
Let $\CCPD{V,\underline{J}}$ and $\CCPD{U,\underline{J}}$ be polytopes of constrained conditional probability distributions.
Write $\langle \CCPD{V,\underline{J}} \rangle \subseteq V$ and $\langle \CCPD{U,\underline{J}} \rangle \subseteq U$ for the linear subspaces spanned by linear combinations of vectors in $\CCPD{V,\underline{J}}$ and $\CCPD{U,\underline{J}}$ respectively.
The following statements hold:
\begin{enumerate}
    \item if $V \subset U$ then $\CCPD{V,\underline{J}} \subseteq \CCPD{U,\underline{J}}$
    \item $\CCPD{\langle\CCPD{V,\underline{J}}\rangle,\underline{J}} = \CCPD{V,\underline{J}}$
    \item if $\QNCCPD{V,\underline{J}} \subseteq \QNCCPD{U,\underline{J}}$ then $\CCPD{V,\underline{J}} \subseteq \CCPD{U,\underline{J}}$
    \item if $\CCPD{V,\underline{J}} \subseteq \CCPD{U,\underline{J}}$ then $\langle\CCPD{V,\underline{J}}\rangle \subseteq \langle\CCPD{V,\underline{J}}\rangle$
    \item if $\CCPD{V,\underline{J}} \subseteq \CCPD{U,\underline{J}}$ then $\QNCCPD{V,\underline{J}} \subseteq \QNCCPD{U,\underline{J}}$
\end{enumerate}
\end{proposition}
\begin{proof}
See \ref{proof:proposition:hierarchy-of-ccpd-polytopes}
\end{proof}

Polytopes of constrained conditional probability distributions are closed under meet (i.e. under intersection).

\begin{proposition}
\label{proposition:ccpd-polytopes-meet}
Let $Y$ be a finite non-empty set and let $\underline{J}=\left(J^{(y)}\right)_{y \in Y}$ be a family of finite non-empty sets.
Let $\CCPD{V,\underline{J}}$ and $\CCPD{U,\underline{J}}$ be polytopes of constrained conditional probability distributions.
Then:
\[
\CCPD{V,\underline{J}} \cap \CCPD{U,\underline{J}}
= \CCPD{V\cap U,\underline{J}}
\]
\end{proposition}
\begin{proof}
See \ref{proof:proposition:ccpd-polytopes-meet}
\end{proof}

As one would expect, constrained conditional probability distributions can be recovered from constrained quasi-normalised conditional probability distributions by applying the normalisation equations.

\begin{proposition}
\label{proposition:ccpd-from-norm-of-qnccpd}
Let $Y$ be a finite non-empty set and let $\underline{J}=\left(J^{(y)}\right)_{y \in Y}$ be a family of finite non-empty sets.
For every linear subspace $W \subseteq \reals^{\sqcup_{y \in Y} J^{(y)}}$, we always have:
\[
\CCPD{W,\underline{J}}
= \Slice{\NormEqs{\underline{J}}}{\QNCCPD{W, \underline{J}}}
\]
For every $y \in Y$, we define a normalisation equation for the distribution conditional to $y$:
\[
\NormEqs{\underline{J}}^{(y)}
:= \suchthat{
    \underline{x} \in \reals^{\sqcup_{y \in Y} J^{(y)}}
}{
    \sum_{j \in J^{(y)}} x_j^{(y)} = 1
}
\]
Then for any individual choice of $y \in Y$ we also have:
\[
\CCPD{W,\underline{J}}
= \Slice{\NormEqs{\underline{J}}^{(y)}}{\QNCCPD{W, \underline{J}}}
\]
\end{proposition}
\begin{proof}
See \ref{proof:proposition:ccpd-from-norm-of-qnccpd}
\end{proof}

Given two nested polytopes $\CCPD{V,\underline{J}} \subset \CCPD{U,\underline{J}}$ of constrained conditional probability distributions, a key task in our work will be to find the largest ``fraction'' of a distribution $\underline{u} \in \CCPD{U,\underline{J}}$ which is ``supported'' by $\CCPD{V,\underline{J}}$, i.e. to find a decomposition $\underline{u} = \underline{v} + \underline{w}$ where $\underline{v} \in \QNCCPD{V,\underline{J}}$, $\underline{w} \in \QNCCPD{U,\underline{J}}$ and the mass of $\underline{v}$ is as large as possible.
The following result---a consequence of our polytopes being defined by linear constraints---significantly simplifies this task, by removing the need to explicitly enforce $\underline{w} \in \QNCCPD{U,\underline{J}}$ in our linear programs.

\begin{proposition}
\label{proposition:difference-in-super-polytope}
Let $Y$ be a finite non-empty set and let $\underline{J}=\left(J^{(y)}\right)_{y \in Y}$ be a family of finite non-empty sets.
Let $\CCPD{V,\underline{J}} \subset \CCPD{U,\underline{J}}$ be two nested polytopes of constrained conditional probability distributions and let $\underline{u} \in \CCPD{U,\underline{J}}$.
If $\underline{v} \in \CCPD{V,\underline{J}}$ is such that $\underline{v} \leq \underline{u}$, then necessarily:
\[
\underline{u}-\underline{v} \in \CCPD{U,\underline{J}}
\]
\end{proposition}
\begin{proof}
See \ref{proof:proposition:difference-in-super-polytope}
\end{proof}

\begin{definition}
\label{definition:component-supported-fraction}
Let $Y$ be a finite non-empty set and let $\underline{J}=\left(J^{(y)}\right)_{y \in Y}$ be a family of finite non-empty sets.
Let $\underline{u} \in \CCPD{U,\underline{J}}$ be a constrained conditional probability distribution.
For any sub-polytope $\CCPD{V,\underline{J}} \subseteq \CCPD{U,\underline{J}}$, we give the following definitions:
\begin{itemize}
    \item A \emph{component} of $\underline{u}$ in $\CCPD{V,\underline{J}}$ is any $\underline{v} \in \QNCCPD{V,\underline{J}}$ such that $\underline{v} \leq \underline{u}$.
    \item A \emph{maximal} component of $\underline{u}$ in $\CCPD{V,\underline{J}}$ is one of maximal mass.
    \item The \emph{supported fraction} of $\underline{u}$ in $\CCPD{V,\underline{J}}$ is the mass of a maximal component of $\underline{u}$ in $\CCPD{V,\underline{J}}$.
\end{itemize}
Colloquially, we say that $\underline{u}$ is $X\%$ supported by $\CCPD{V,\underline{J}}$ to mean that the supported fraction of $\underline{u}$ in $\CCPD{V,\underline{J}}$ is $\frac{X}{100}$.
\end{definition}

\begin{observation}
\label{observation:max-component-lp}
Let $Y$ be a finite non-empty set and let $\underline{J}=\left(J^{(y)}\right)_{y \in Y}$ be a family of finite non-empty sets.
Let $\underline{u} \in \CCPD{U,\underline{J}}$ be a constrained conditional probability distribution and let $\CCPD{V,\underline{J}} \subseteq \CCPD{U,\underline{J}}$ be a sub-polytope.
Let $V$ be defined explicitly by a system of linear equations:
\[
V = \suchthat{\underline{x} \in \reals^{\sqcup_{y \in Y} J^{(y)}}}{A \underline{x} = \underline{0}}
\]
The maximal components $\underline{v}$ of $\underline{u}$ in $\CCPD{V,\underline{J}}$ are the solutions to the following linear program (LP):
\begin{equation}
\begin{array}{rl}
\text{maximise}
&\mass{\underline{v}}
\\
\text{ subject to:}&
\underline{v} \in \QNormEqs{\underline{J}}
\\&
\underline{v} \in V
\\&
\underline{v} \geq 0
\\&
\underline{v} \leq \underline{u}
\end{array}    
\end{equation}
The following form makes the mass and quasi-normalisation equations explicit, for any choice of total order $\{y_1,...,y_n\}$ on $Y$:
\begin{equation}
\begin{array}{rl}
\text{maximise}
&\sum\limits_{j \in J^{(y_1)}} v_j^{(y_1)}
\\
\text{ subject to:}&
\forall i \in \{1,...,n-1\}.\;
\sum\limits_{j \in J^{(y_i)}} v^{(y_i)}_j =\hspace{-2mm} \sum\limits_{j \in J^{(y_{i+1})}} v^{(y_{i+1})}_j
\\&
A \underline{v} = \underline{0}
\\&
\underline{v} \geq 0
\\&
\underline{v} \leq \underline{u}
\end{array}    
\end{equation}
\end{observation}

More generally, we will be interested to find the largest fraction of a (constrained) conditional probability distribution which is supported ``jointly'' by multiple sub-polytopes.
This is the same as being supported by the convex hull of the sub-polytopes, but with the caveat that the convex hull of polytopes of constrained conditional probability distributions need not be itself a polytope of constrained conditional probability distributions.
In particular, we have no way to apply Definition 
\ref{definition:component-supported-fraction} or Observation \ref{observation:max-component-lp} to such a convex hull.

\begin{definition}
\label{definition:decomposition-supported-fraction}
Let $Y$ be a finite non-empty set and let $\underline{J}=\left(J^{(y)}\right)_{y \in Y}$ be a family of finite non-empty sets.
Let $\underline{u} \in \CCPD{U,\underline{J}}$ be a constrained conditional probability distribution.
For any family $\left(\CCPD{V^{(z)},\underline{J}}\right)_{z \in Z}$ of sub-polytopes $\CCPD{V^{(z)},\underline{J}} \subseteq \CCPD{U,\underline{J}}$, we give the following definitions:
\begin{itemize}
    \item A \emph{decomposition} of $\underline{u}$ in $\left(\CCPD{V^{(z)},\underline{J}}\right)_{z \in Z}$ is any family $\left(\underline{v}^{(z)}\right)_{z \in Z}$ of distributions $\underline{v}^{(z)} \in \QNCCPD{V^{(z)},\underline{J}}$, the \emph{components}, such that $\sum_{z \in Z}\underline{v}^{(z)} \leq \underline{u}$.
    \item The \emph{mass} of a decomposition $\left(\underline{v}^{(z)}\right)_{z \in Z}$ is the sum of the masses of the individual components:
    \[
        \mass{\left(\underline{v}^{(z)}\right)_{z \in Z}} := \sum_{z \in Z} \mass{\underline{v}^{(z)}}
    \]
    \item A \emph{maximal decomposition} of $\underline{u}$ in $\left(\CCPD{V^{(z)},\underline{J}}\right)_{z \in Z}$ is one of maximal mass.
    \item The \emph{supported fraction} of $\underline{u}$ in $\left(\CCPD{V^{(z)},\underline{J}}\right)_{z \in Z}$ is the mass of a maximal decomposition of $\underline{u}$ in $\left(\CCPD{V^{(z)},\underline{J}}\right)_{z \in Z}$.
\end{itemize}
Colloquially, we say that $\underline{u}$ is $X\%$ supported by $\left(\CCPD{V^{(z)},\underline{J}}\right)_{z \in Z}$ to mean that the supported fraction of $\underline{u}$ in $\left(\CCPD{V^{(z)},\underline{J}}\right)_{z \in Z}$ is $\frac{X}{100}$.
\end{definition}

\begin{corollary}
\label{corollary:difference-in-super-polytope-decomp}
Let $Y$ be a finite non-empty set and let $\underline{J}=\left(J^{(y)}\right)_{y \in Y}$ be a family of finite non-empty sets.
Let $\CCPD{U,\underline{J}}$ be a polytope of constrained conditional probability distributions and let $\left(\CCPD{V^{(z)},\underline{J}}\right)_{z \in Z}$ be a family of sub-polytopes $\CCPD{V^{(z)},\underline{J}} \subseteq \CCPD{U,\underline{J}}$.
Let $\underline{u} \in \CCPD{U,\underline{J}}$ be a constrained conditional probability distribution.
If $\left(\underline{v}^{(z)}\right)_{z \in Z}$ is a decomposition of $\underline{u}$ in $\left(\CCPD{V^{(z)},\underline{J}}\right)_{z \in Z}$, then necessarily:
\[
\underline{u}-\sum_{z \in Z} \underline{v}^{(z)} \in \CCPD{U,\underline{J}}
\]
\end{corollary}
\begin{proof}
See \ref{proof:corollary:difference-in-super-polytope-decomp}
\end{proof}

\begin{observation}
\label{observation:max-decomposition-lp}
Let $Y$ be a finite non-empty set and let $\underline{J}=\left(J^{(y)}\right)_{y \in Y}$ be a family of finite non-empty sets.
Let $\underline{u} \in \CCPD{U,\underline{J}}$ be a constrained conditional probability distribution and let $\left(\CCPD{V^{(z)},\underline{J}}\right)_{z \in Z}$ be a family of sub-polytopes $\CCPD{V^{(z)},\underline{J}} \subseteq \CCPD{U,\underline{J}}$.
Let each $V^{(z)}$ be defined explicitly by a system of linear equations:
\[
V^{(z)} = \suchthat{\underline{x} \in \reals^{\sqcup_{y \in Y} J^{(y)}}}{A^{(z)} \underline{x} = \underline{0}}
\]
The maximal components $\left(\underline{v}^{(z)}\right)_{z \in Z}$ of $\underline{u}$ in $\left(\CCPD{V^{(z)},\underline{J}}\right)_{z \in Z}$ are the solutions to the following linear program (LP):
\begin{equation}
\begin{array}{rl}
\text{maximise}
&\mass{\left(\underline{v}^{(z)}\right)_{z \in Z}}
\\
\text{ subject to:}&
\forall z \in Z.\;
\underline{v}^{(z)} \in \QNormEqs{\underline{J}}
\\&
\forall z \in Z.\;
\underline{v}^{(z)} \in V^{(z)}
\\&
\underline{v}^{(z)} \geq 0
\\&
\sum\limits_{z \in Z}\underline{v}^{(z)} \leq \underline{u}
\end{array}    
\end{equation}
Making the mass and linear systems explicit, we get:
\begin{equation}
\begin{array}{rl}
\text{maximise}
&\sum\limits_{z \in Z}\mass{\underline{v}^{(z)}}
\\
\text{ subject to:}&
\forall z \in Z.\;
\underline{v}^{(z)} \in \QNormEqs{\underline{J}}
\\&
\forall z \in Z.\;
A^{(z)} \underline{v}^{(z)} = 0
\\&
\underline{v}^{(z)} \geq 0
\\&
\sum\limits_{z \in Z}\underline{v}^{(z)} \leq \underline{u}
\end{array}    
\end{equation}
\end{observation}

\subsection{Causaltopes}
\label{subsection:geometry-causality-ctopes}

Recall from \cite{gogioso2022topology} that an empirical model $e$ on a cover $\mathcal{C}$ is a compatible family $(e_\lambda)_{\lambda \in \mathcal{C}}$ for the presheaf of causal distributions $\CausDist{\Theta, \underline{O}}$:
\[
e_\lambda \in \CausDist{\lambda, \underline{O}} = \Dist{\ExtCausFun{\lambda, \underline{O}}}
\]
The empirical model assigns a probability distribution on the extended causal functions $\ExtCausFun{\lambda, \underline{O}}$ to each context $\lambda \in \mathcal{C}$, a lowerset of input histories upon which outputs for events can be simultaneously and consistently defined.
As such, empirical models on a given cover inherit the convex structure of the individual sets of distributions, by taking context-wise convex combinations:
\[
(x \cdot e + (1-x) \cdot e')_\lambda := x \cdot e_\lambda + (1-x) \cdot e'_\lambda
\]
where $e, e' \in \EmpModels{\mathcal{C}, \underline{O}}$, $x \in [0,1]$ and $\lambda \in \mathcal{C}$ ranges over the contexts specified by the cover $\mathcal{C}$.

The characterisation of empirical models as conditional probability distributions over extended causal functions is inconvenient, because empirical data is typically expressed in the form of conditional probability distributions over joint outputs at events.
The following results provide an equivalent characterisation of empirical models in this sense.

\begin{theorem}
\label{theorem:topdist-extdist}
Let $\Theta$ be a space of input histories and let $\underline{O} = (O_\omega)_{\omega \in \Events{\Theta}}$ be a family of non-empty sets of outputs.
For any $k \in \Ext{\Theta}$, the following function is a bijection:
\begin{equation}
\begin{array}{rcl}
    \ExtCausFun{\downset{k\!}, \underline{O}}
    &\longrightarrow&
    \prod\limits_{\omega \in \dom{k}}O_\omega
    \\
    \Ext{f}
    &\mapsto&
    \Ext{f}(k)
\end{array}
\end{equation}
As a consequence, the following function is a convex-linear bijection:
\begin{equation}
\begin{array}{rcl}
    \CausDist{\downset{k\!}, \underline{O}}
    &\longrightarrow&
    \Dist{\prod\limits_{\omega \in \dom{k}}O_\omega}
    \\
    d
    &\mapsto&
    \topdist{d}
    := \sum\limits_{\Ext{f}} d(\Ext{f}) \delta_{\Ext{f}(k)}
\end{array}
\end{equation}
We refer to $\topdist{d}$ as the \emph{top-element distribution} for $d \in \CausDist{\downset{k\!}, \underline{O}}$.
We furthermore adopt the following notation for its inverse:
\begin{equation}
\begin{array}{rcl}
    \CausDist{\downset{k\!}, \underline{O}}
    &\longleftarrow&
    \Dist{\prod\limits_{\omega \in \dom{k}}O_\omega}
    \\
    \extdist{p}{k}
    &\mapsfrom&
    p
\end{array}
\end{equation}
\end{theorem}
\begin{proof}
See \ref{proof:theorem:topdist-extdist}
\end{proof}

In order to extend Theorem \ref{theorem:topdist-extdist} to arbitrary covers, we need to deal with the fact that lowersets might, in general, contain inconsistent histories, which might result in different outputs for the same event.
If $\lambda \in \Lsets{\Theta}$ is a generic lowerset, we write $\dom{\lambda} := \bigcup_{h \in \lambda} \dom{h}$ for the set of events that appear in the domain of some history in $\lambda$.
Recall that $\TipEqCls{\lambda}{\omega}$ is the set of equivalence classes for the relation $\histconstrSym{\omega}$, which constraints certain input histories to yield the same output value at $\omega$ in all causal functions:
\[
    \TipEqCls{\lambda}{\omega}
    :=
    \suchthat{\histconstreqcls{h}{\omega}}{h \in \TipHists{\lambda}{\omega}}
\]
When $\lambda$ is tight, all equivalence classes $\histconstreqcls{h}{\omega}$ contain a single history, i.e. $\histconstreqcls{h}{\omega} = \{h\}$.
When $\lambda = \downset{k}$ is a downset---as is the case for all contexts in the standard and fully solipsistic covers---each $\omega$ has a single equivalence class associated to it, i.e. $|\TipEqCls{\lambda}{\omega}| = 1$; this is because $k$ witnesses the consistency of all input histories $h \in \downset{k}$ below it.

By definition, the extended causal functions $\Ext{f} \in \ExtCausFun{\lambda, \underline{O}}$ assign the same output value $\Ext{f}(h')_\omega$ to all $h' \in \histconstreqcls{h}{\omega}$, for each equivalence class $\histconstreqcls{h}{\omega} \in \TipEqCls{\lambda}{\omega}$.
This observation paves the way for our desired generalisation of Theorem \ref{theorem:topdist-extdist}: instead of producing a single output for an event $\omega$, we now need to produce independent outputs for each equivalence class in $\TipEqCls{\lambda}{\omega}$.

\begin{theorem}
\label{theorem:topdist-extdist-generalised}
Let $\Theta$ be a space of input histories and let $\underline{O} = (O_\omega)_{\omega \in \Events{\Theta}}$ be a family of non-empty sets of outputs.
For any $\lambda \in \Lsets{\Theta}$, the following function is a bijection:
\begin{equation}
\begin{array}{rcl}
    \ExtCausFun{\lambda, \underline{O}}
    &\longrightarrow&
    \prod\limits_{\omega \in \dom{\lambda}} \left(O_\omega\right)^{\TipEqCls{\lambda}{\omega}}
    \\
    \Ext{f}
    &\mapsto&
    \topdist{\Ext{f}}
\end{array}
\end{equation}
where took $\dom{\lambda} := \bigcup_{h \in \lambda} \dom{h}$ and we defined:
\begin{equation}
    \topdist{\Ext{f}}
    :=
    \omega
    \mapsto
    \left(\Ext{f}(h)_\omega\right)_{\histconstreqcls{h}{\omega} \in \TipEqCls{\lambda}{\omega}}  
\end{equation}
As a consequence, the following function is a convex-linear bijection:
\begin{equation}
\begin{array}{rcl}
    \CausDist{\lambda, \underline{O}}
    &\longrightarrow&
    \Dist{\prod\limits_{\omega \in \dom{\lambda}} \left(O_\omega\right)^{\TipEqCls{\lambda}{\omega}}}
    \\
    d
    &\mapsto&
    \topdist{d}
    := \sum\limits_{\Ext{f}} d(\Ext{f}) \delta_{\topdist{\Ext{f}}}
\end{array}
\end{equation}
We furthermore adopt the following notation for its inverse:
\begin{equation}
\begin{array}{rcl}
    \CausDist{\lambda, \underline{O}}
    &\longleftarrow&
    \Dist{\prod\limits_{\omega \in \dom{\lambda}}\left(O_\omega\right)^{\TipEqCls{\lambda}{\omega}}}
    \\
    \extdist{p}{\lambda}
    &\mapsfrom&
    p
\end{array}
\end{equation}
\end{theorem}
\begin{proof}
See \ref{proof:theorem:topdist-extdist-generalised}
\end{proof}

Applying the bijection above to each context $\lambda \in \mathcal{C}$ in a given cover provides us with an operationally meaningful ambient polytope in which to formulate our empirical models.
By construction, this ambient space enforces causality constraints within each context, but doesn't yet enforce any constraints across contexts: hence, we refer to its points as ``pseudo-empirical models''.

\begin{definition}
Let $\Theta$ be a space of input histories, let $\underline{O} = (O_\omega)_{\omega \in \Events{\Theta}}$ be a family of non-empty sets of outputs and let $\mathcal{C} \in \Covers{\Theta}$ be any cover.
The polytope of \emph{pseudo-empirical models} on $\mathcal{C}$ is defined to be the following polytope of conditional probability distributions:
\begin{equation}
    \PsEmpModels{\mathcal{C}, \underline{O}}
    :=
    \prod_{\lambda \in \mathcal{C}}
    \Dist{\prod\limits_{\omega \in \dom{\lambda}} \left(O_\omega\right)^{\TipEqCls{\lambda}{\omega}}}
\end{equation}
We adopt the following shorthand for the embedding vector space, which is spanned by all linear combinations of the pseudo-empirical models:
\begin{equation}
    \PsEmpModelsVec{\mathcal{C}, \underline{O}}
    :=
    \reals^{
        \bigsqcup\limits_{\lambda \in \mathcal{C}}
        \prod\limits_{\omega \in \dom{\lambda}}
        \left(O_\omega\right)^{\TipEqCls{\lambda}{\omega}}
    }
\end{equation}
\end{definition}

Pseudo-empirical models are conditional probability distributions, and hence we adopt the notation from the previous subsection to describe them:
\[
\underline{u} = \left(\underline{u}^{(\lambda)}\right)_{\lambda \in \mathcal{C}}
\]
For a given $\lambda$, the components $u_o^{(\lambda)}$ are indexed by functions/families $o \in \prod_{\omega \in \dom{\lambda}} \left(O_\omega\right)^{\TipEqCls{\lambda}{\omega}}$, the components of which are in turn indexed as follows:
\[
o = \left(o_{\omega, \histconstreqcls{h}{\omega}}\right)_{\omega \in \dom{\lambda},\; \histconstreqcls{h}{\omega} \in \TipEqCls{\lambda}{\omega}}
\]
The restriction of pseudo-empirical model to empirical models goes by the definition of suitable ``causality equations'', constraining the marginals of the probability distributions associated to various lowersets in the cover.
Before we can phrase the equations, however, we must define what we mean by restriction of the probability distributions $\underline{u}^{(\lambda)}$ from the cover lowersets $\lambda \in \mathcal{C}$ to arbitrary sub-lowersets $\mu \subseteq \lambda$.

\begin{lemma}
\label{lemma:output-history-injection}
Let $\Theta$ be a space of input histories and let $\lambda, \mu \in \Lsets{\Theta}$ be two lowersets.
If $\mu \subseteq \lambda$, then the following is a well-defined function:
\[
\begin{array}{rrcl}
\TipEqCls{\mu \subseteq \lambda}{\omega}:
&\TipEqCls{\mu}{\omega}
&\rightarrow&
\TipEqCls{\lambda}{\omega}
\\
&\histconstreqcls{h}{\omega}
&\mapsto&
\histconstreqcls{h}{\omega}
\end{array}
\]
\end{lemma}
\begin{proof}
See \ref{proof:lemma:output-history-injection}
\end{proof}

\begin{definition}
\label{definition:output-history-restriction}
Let $\Theta$ be a space of input histories, let $\underline{O} = (O_\omega)_{\omega \in \Events{\Theta}}$ be a family of non-empty sets of outputs and let $\mathcal{C} \in \Covers{\Theta}$ be any cover.
For every $\lambda \in \mathcal{C}$ and every $\mu \in \Lsets{\Theta}$ such that $\mu \subseteq \lambda$, the \emph{output history restriction} from $\lambda$ to $\mu$ is defined as follows:
\begin{equation}
\label{equation:output-history-restriction}
\begin{array}{rccc}
    \rho_{\lambda, \mu}:
    &\prod\limits_{\omega \in \dom{\lambda}} \left(O_\omega\right)^{\TipEqCls{\lambda}{\omega}}
    &\longrightarrow
    &\prod\limits_{\omega \in \dom{\mu}} \left(O_\omega\right)^{\TipEqCls{\mu}{\omega}}
    \\
    &o
    &\mapsto
    &\left(
        \left(\omega, \histconstreqcls{h}{\omega}\right)
        \mapsto
        o_{\omega,\histconstreqcls{h}{\omega}}
    \right)
\end{array}
\end{equation}
Formally, $o_{\omega,\histconstreqcls{h}{\omega}}$ stands for $o_{\omega,\TipEqCls{\mu \subseteq \lambda}{\omega}\left(\histconstreqcls{h}{\omega}\right)}$.
The restriction extends convex-linearly to a \emph{output history distribution restriction} between the corresponding spaces of probability distributions:
\begin{equation}
\label{equation:output-history-restriction-distr}
\begin{array}{rccc}
    \Dist{\rho_{\lambda, \mu}}:
    &\Dist{\prod\limits_{\omega \in \dom{\lambda}} \left(O_\omega\right)^{\TipEqCls{\lambda}{\omega}}}
    &\longrightarrow
    &\Dist{\prod\limits_{\omega \in \dom{\mu}} \left(O_\omega\right)^{\TipEqCls{\mu}{\omega}}}
    \\
    &d
    &\mapsto
    &\left(
    o \mapsto
    \hspace{-2mm} \sum\limits_{o' \text{ s.t. } \rho_{\lambda, \mu}(o')=o}
    \hspace{-2mm} d(o')
    \right)
\end{array}
\end{equation}
\end{definition}

We are now in a position to define the causality equations.
The initial definition we provide considers all possible pairs of restrictions in a space of input histories: it is the simplest way to formulate these required constraints, but it is also highly redundant.
We then prove two results reducing the number of necessary equations, under certain assumptions.

\begin{definition}
\label{definition:caus-eqs}
Let $\Theta$ be a space of input histories, let $\underline{O} = (O_\omega)_{\omega \in \Events{\Theta}}$ be a family of non-empty sets of outputs and let $\mathcal{C} \in \Covers{\Theta}$ be any cover.
The \emph{causality equations} are indexed by all $\mu \in \Lsets{\Theta}$ and all $\lambda, \lambda' \in \mathcal{C}$ such that $\mu \subseteq \lambda$ and $\mu \subseteq \lambda'$.
For one such triple $\mu, \lambda, \lambda'$, we equate the output history distribution restrictions from $\lambda$ to $\mu$ and from $\lambda'$ to $\mu$:
\begin{equation}
    \CausEqs{\mathcal{C}, \underline{O}}_{\mu, \lambda, \lambda'}
    :=
    \suchthat{
        \underline{u}
        \in
        \PsEmpModelsVec{\mathcal{C}, \underline{O}}
    }{
        \restrict{\underline{u}^{(\lambda)}}{\mu}
        =\restrict{\underline{u}^{(\lambda')}}{\mu}
    }
\end{equation}
where we have adopted the following shorthand for the restriction:
\begin{equation}
    \restrict{\underline{u}^{(\lambda)}}{\mu}
    := \Dist{\rho_{\lambda, \mu}}\left(\underline{u}^{(\lambda)}\right)
\end{equation}
We write $\CausEqs{\mathcal{C}, \underline{O}}$ for the linear subspace of $\PsEmpModelsVec{\mathcal{C}, \underline{O}}$ spanned jointly by all causality equations:
\begin{equation}
    \CausEqs{\mathcal{C}, \underline{O}}
    :=
    \bigcap_{\mu \in \Lsets{\Theta}}
    \bigcap_{\lambda \in \mathcal{C}\cap\upset{\mu} }
    \bigcap_{\lambda' \in \mathcal{C}\cap\upset{\mu} }
    \CausEqs{\mathcal{C}, \underline{O}}_{\mu, \lambda, \lambda'}
\end{equation}
where $\upset{\mu}$ is the upset of $\mu$ in the partial order $\Lsets{\Theta}$ formed by lowersets under inclusion.
\end{definition}

\begin{proposition}
\label{proposition:caus-eqs-chain}
Let $\Theta$ be a space of input histories, let $\underline{O} = (O_\omega)_{\omega \in \Events{\Theta}}$ be a family of non-empty sets of outputs and let $\mathcal{C} \in \Covers{\Theta}$ be any cover.
For every $\mu \in \Lsets{\Theta}$, let $\lambda_{\mu,1}, ..., \lambda_{\mu, n_\mu}$ be a total order on the $\lambda \in \mathcal{C}$ such that $\mu \subseteq \lambda$.
Then we have:
\begin{equation}
    \CausEqs{\mathcal{C}, \underline{O}}
    =
    \bigcap_{\mu \in \Lsets{\Theta}, n_\mu \geq 1}
    \bigcap_{i=1}^{n_\mu-1}
    \CausEqs{\mathcal{C}, \underline{O}}_{\mu, \lambda_{\mu,i}, \lambda_{\mu,i+1}}
\end{equation}
\end{proposition}
\begin{proof}
See \ref{proof:proposition:caus-eqs-chain}
\end{proof}

\begin{proposition}
\label{proposition:caus-eqs-chain-std}
Let $\Theta$ be a space of input histories, let $\underline{O} = (O_\omega)_{\omega \in \Events{\Theta}}$ be a family of non-empty sets of outputs and let $\StdCov{\Theta} \in \Covers{\Theta}$ be the standard cover.
For every $h \in \Ext{\Theta}$, let $k_{h,1}, ..., k_{h, n_h}$ be a total order on the $k \in \Ext{\Theta}$ such that $k \leq h$, i.e. such that $\downset{h} \subseteq \downset{k}$.
Then we have:
\begin{equation}
    \CausEqs{\mathcal{C}, \underline{O}}
    =
    \bigcap_{h \in \Ext{\Theta}, n_h \geq 1}
    \bigcap_{i=1}^{n_h-1}
    \CausEqs{\mathcal{C}, \underline{O}}_{\downset{h}, \downset{k_{h,i}}, \downset{k_{h,i+1}}}
\end{equation}
\end{proposition}
\begin{proof}
See \ref{proof:proposition:caus-eqs-chain-std}
\end{proof}

Finally, we are in a position to define a ``causal polytope''---or ``causaltope''---as the space of pseudo-empirical models which satisfy the causality equations, and to prove that it coincides with the convex space of empirical models.

\begin{definition}
Let $\Theta$ be a space of input histories, let $\underline{O} = (O_\omega)_{\omega \in \Events{\Theta}}$ be a family of non-empty sets of outputs and let $\mathcal{C} \in \Covers{\Theta}$ be any cover.
The associated \emph{causaltope}, a portmanteau of \emph{causal polytope}, is defined to be the following space of constrained conditional probability distributions:
\begin{equation}
    \begin{array}{rcl}
    \Causaltope{\mathcal{C}, \underline{O}}
    &:=&
    \CCPD{
        \CausEqs{\mathcal{C}, \underline{O}},
        \left(\prod\limits_{\omega \in \dom{\lambda}} \left(O_\omega\right)^{\TipEqCls{\lambda}{\omega}}\right)_{\lambda \in \mathcal{C}}
    }
    \\
    &=&
    \CausEqs{\mathcal{C}, \underline{O}} \cap
    \prod\limits_{\lambda \in \mathcal{C}}\Dist{\prod\limits_{\omega \in \dom{\lambda}} \left(O_\omega\right)^{\TipEqCls{\lambda}{\omega}}}
    \end{array}
\end{equation} 
\end{definition}

\begin{observation}
When $\mathcal{C}=\StdCov{\Theta}$ is the standard cover, we refer to the associated causaltope as a \emph{standard causaltope}, taking the following simplified form:
\[
\StdCausaltope{\Theta, \underline{O}}
:= 
\CCPD{
    \StdCausEqs{\mathcal{C}, \underline{O}},
    \left(\prod\limits_{\omega \in \dom{k}} O_\omega\right)_{k \in \max\Ext{\Theta}}
}
\]
We write $\StdCausEqs{\Theta, \underline{O}}$ for the causal equations on the standard cover.
\end{observation}

\begin{observation}
When $\mathcal{C}=\StdCov{\Theta}$ is the solipsistic cover, we refer to the associated causaltope as a \emph{solipsistic causaltope}, taking the following simplified form:
\[
\SolCausaltope{\Theta, \underline{O}}
:= 
\CCPD{
    \SolCausEqs{\mathcal{C}, \underline{O}},
    \left(\prod\limits_{\omega \in \dom{k}} O_\omega\right)_{k \in \max\Theta}
}
\]
We write $\SolCausEqs{\Theta, \underline{O}}$ for the causal equations on the solipsistic cover.
\end{observation}


\begin{theorem}
\label{theorem:causaltopes-emp-models}
Let $\Theta$ be a space of input histories, let $\underline{O} = (O_\omega)_{\omega \in \Events{\Theta}}$ be a family of non-empty sets of outputs and let $\mathcal{C} \in \Covers{\Theta}$ be any cover.
Then the following is a convex-linear bijection:
\begin{equation}
\begin{array}{rcl}
    \EmpModels{\mathcal{C}, \underline{O}}
    &\leftrightarrow&
    \Causaltope{\mathcal{C}, \underline{O}}
    \\
    \underline{e}
    &\mapsto&
    \left(\topdist{e_\lambda}\right)_{\lambda \in \mathcal{C}}
\end{array}
\end{equation}
\end{theorem}
\begin{proof}
See \ref{proof:theorem:causaltopes-emp-models}
\end{proof}

In the previous Section, empirical models were defined topologically, as compatible families in presheaves of causal distributions.
Theorem \ref{theorem:causaltopes-emp-models} provides an equivalent geometric characterisation for empirical models, as constrained conditional probability distributions on joint outputs.
From this moment onwards, we will freely confuse between the topological and geometric picture, referring to the points of causaltopes as ``empirical models''.

\subsection{Causal inseparability for standard empirical models}
\label{subsection:causal-inseparability}

For standard empirical models on spaces satisfying the free choice condition, our geometric characterisation coincides with the definition typically used by previous literature on indefinite causality: empirical models are probability distributions $u^{(i)}_{o} \in [0,1]$ on joint outputs $o \in \prod_{\omega \in \Events{\Theta}} O_\omega$ conditional to joint inputs $i \in \prod_{\omega \in \Events{\Theta}}\Inputs{\Theta}_\omega = \max\Ext{\Theta}$.

We start by observing that the hierarchy of spaces on common events and inputs reflects into an inclusion hierarchy for the associated standard causaltopes.

\begin{proposition}
\label{proposition:subspace-hierarchy-causaltopes}
Let $\Theta \leq \Theta'$ be a space of input histories such that $\Events{\Theta'} = \Events{\Theta}$ and $\Inputs{\Theta'} = \Inputs{\Theta}$ and $\max\Ext{\Theta}=\max\Ext{\Theta'}$ (e.g. because they both satisfy the free choice condition).
Let $\underline{O} = (O_\omega)_{\omega \in \Events{\Theta}}$ be a family of non-empty sets of outputs.
The standard causaltope for $\Theta$ is always contained into the standard causaltope for $\Theta'$:
\[
\StdCausEqs{\Theta, \underline{O}}
\subseteq\StdCausEqs{\Theta', \underline{O}}
\]
\[
\StdCausaltope{\Theta, \underline{O}}
\subseteq\StdCausaltope{\Theta', \underline{O}}
\]
\end{proposition}
\begin{proof}
See \ref{proof:proposition:subspace-hierarchy-causaltopes}
\end{proof}

An immediate consequence of the above result is that the standard empirical models for any space $\Theta$ satisfying the free choice condition are also standard empirical models for the associated indiscrete space $\Hist{\indiscrete{\Events{\Theta}}, \underline{\Inputs{\Theta}}}$.
At the opposite extreme, the standard causaltope for the discrete space $\Hist{\discrete{\Events{\Theta}}, \underline{\Inputs{\Theta}}}$---i.e. the no-signalling polytope studied by previous literature on non-locality---is always contained in the causaltope of $\Theta$.

\begin{observation}
For any non-empty set $E$ of events and any family of non-empty input sets $\underline{I} = (I_e)_{e \in E}$, the standard causaltope for the indiscrete space $\Hist{\indiscrete{E}, \underline{I}}$ is the polytope of pseudo-empirical models:
\[
\StdCausaltope{\Theta_{ind}, \underline{O}}
= \PsEmpModels{\StdCov{\Theta_{ind}}, \underline{O}}
\]
where we have defined the shorthand $\Theta_{ind} := \Hist{\indiscrete{E}, \underline{I}}$.
\end{observation}

\begin{observation}
Let $\Theta$ be a space of input histories and let $\underline{O} = (O_\omega)_{\omega \in \Events{\Theta}}$ be a family of non-empty sets of outputs.
The standard causaltope for the discrete space $\Hist{\discrete{\Events{\Theta}}, \underline{\Inputs{\Theta}}}$ is always contained in the causaltope for $\Theta$:
\[
\StdCausaltope{\Hist{\discrete{\Events{\Theta}}, \underline{\Inputs{\Theta}}}, \underline{O}}
\subseteq \StdCausaltope{\Theta, \underline{O}}
\]
For any non-empty set $E$ of events and any family of non-empty input sets $\underline{I} = (I_e)_{e \in E}$, we refer to $\StdCausaltope{\Hist{\discrete{E}, \underline{I}}, \underline{O}}$ as the \emph{no-signalling causaltope}.
\end{observation}

The following definitions provide ``causally flavoured'' variants of Definition \ref{definition:component-supported-fraction} and Definition \ref{definition:decomposition-supported-fraction}.
In particular, we adopt a special name for the fraction supported by the no-signalling causaltope.

\begin{definition}
\label{definition:component-supported-fraction-causal}
Let $\Theta$ be a space of input histories and let $\underline{O} = (O_\omega)_{\omega \in \Events{\Theta}}$ be a family of non-empty sets of outputs.
Let $\underline{u} \in \StdCausaltope{\Theta, \underline{O}}$ be a standard empirical model.
For any $\Theta' \leq \Theta$ such that $\Events{\Theta'} = \Events{\Theta}$ and $\Inputs{\Theta'} = \Inputs{\Theta}$ and $\max\Ext{\Theta}=\max\Ext{\Theta'}$, we give the following definitions:
\begin{itemize}
    \item A \emph{component} of $\underline{u}$ in $\Theta'$ is a component of $\underline{u}$ in the sub-polytope of constrained conditional probability distributions $\StdCausaltope{\Theta', \underline{O}}$ according to Definition \ref{definition:component-supported-fraction}.
    \item A \emph{maximal} component of $\underline{u}$ in $\Theta'$ is one of maximal mass.
    \item The \emph{causal fraction} of $\underline{u}$ in $\Theta'$ is the mass of a maximal component of $\underline{u}$ in $\Theta'$.
\end{itemize}
Colloquially, we say that $\underline{u}$ is $X\%$ supported by $\Theta'$ to mean that the supported fraction of $\underline{u}$ in $\Theta'$ is $\frac{X}{100}$.
The \emph{no-signalling fraction} of $\underline{u}$ is the causal fraction of $\underline{u}$ in the discrete space $\Hist{\discrete{\Events{\Theta}}, \underline{\Inputs{\Theta}}}$.
\end{definition}

\begin{definition}
\label{definition:decomposition-supported-fraction-causal}
Let $\Theta$ be a space of input histories and let $\underline{O} = (O_\omega)_{\omega \in \Events{\Theta}}$ be a family of non-empty sets of outputs.
Let $\underline{u} \in \StdCausaltope{\Theta, \underline{O}}$ be a standard empirical model.
For any family $\left(\Theta^{(z)}\right)_{z \in Z}$ of sub-spaces $\Theta^{(z)} \leq \Theta$ such that $\Events{\Theta^{(z)}} = \Events{\Theta}$ and $\Inputs{\Theta^{(z)}} = \Inputs{\Theta}$ and $\max\Ext{\Theta^{(z)}}=\max\Ext{\Theta}$, we give the following definitions:
\begin{itemize}
    \item A \emph{(causal) decomposition} of $\underline{u}$ over the sub-spaces $\left(\Theta^{(z)}\right)_{z \in Z}$ is a decomposition of $\underline{u}$ in $\left(\StdCausaltope{\Theta^{(z)}, \underline{O}}\right)_{z \in Z}$ according to Definition \ref{definition:decomposition-supported-fraction}.
    \item A \emph{maximal (causal) decomposition} of $\underline{u}$ over the sub-spaces $\left(\Theta^{(z)}\right)_{z \in Z}$ is one of maximal mass.
    \item The \emph{causal fraction} of $\underline{u}$ over the sub-spaces $\left(\Theta^{(z)}\right)_{z \in Z}$ is the mass of a maximal decomposition of $\underline{u}$ in $\left(\Theta^{(z)}\right)_{z \in Z}$.
\end{itemize}
Colloquially, we say that $\underline{u}$ is $X\%$ supported by $\left(\Theta^{(z)}\right)_{z \in Z}$ to mean that the causal fraction of $\underline{u}$ over the sub-spaces $\left(\Theta^{(z)}\right)_{z \in Z}$ is $\frac{X}{100}$.
\end{definition}

Now consider a standard empirical model $\underline{u}$ for a causally incomplete space $\Theta$.
A key question in the study of indefinite causality is whether the $\underline{u}$ admits an explanation in terms of ``dynamical'' definite causal structure, i.e. whether it is 100\% jointly supported by some causally complete sub-spaces of $\Theta$.
This leads to the definition of the qualitative notion of ``causal (in)separability'' and the associated quantitative notion of ``causally (in)separable fraction''.

\begin{definition}
\label{definition:causal-separability}
Let $\Theta$ be a space of input histories and let $\underline{O} = (O_\omega)_{\omega \in \Events{\Theta}}$ be a family of non-empty sets of outputs.
Let $\underline{u} \in \StdCausaltope{\Theta, \underline{O}}$ be a standard empirical model.
We give the following definitions:
\begin{itemize}
    \item The \emph{causally separable fraction} of $\underline{u}$ is the causal fraction of $\underline{u}$ over the causal completions of $\Theta$, i.e. over the maximal causally complete subspaces of $\Theta$. If $\Theta$ is causally complete, the causally separable fraction is always 1.
    \item The \emph{causally inseparable fraction} of $\underline{u}$ is 1 minus its causally separable fraction.
    \item The empirical model $\underline{u}$ is \emph{causally separable} if it has causally separable fraction 1, and it is \emph{causally inseparable} otherwise.
\end{itemize}
We say that $\underline{u}$ has $X\%$ causally (in)separable fraction if its causally (in)separable fraction is $\frac{X}{100}$.
\end{definition}

The definition of causal inseparability provided above is relative to an ambient space $\Theta$ for the empirical model $\underline{u}$: we take the causal constraints of $\Theta$ as a given, and ask whether the empirical model is causally (in)separable conditional to those constraints.
However, it is important to note that this notion of causal inseparability is strictly finer than that used by previous literature on indefinite causal order \cite{oreshkov2012quantum,oreshkov2016causal,abbott2016multipartite}, where $\Theta$ is fixed to the indiscrete space $\Hist{\indiscrete{E}, \underline{I}}$.

\begin{remark}
The recursive definition of causal inseparability used by previous literature \cite{oreshkov2012quantum,oreshkov2016causal,abbott2016multipartite} allows for the possibility of conditioning the empirical model for subsequent events on the output value at the first event.
Such a conditional definition might appear more general at first sight, but it is actually equivalent to the one given above, because the randomness of the output at the first event can always be lifted to a local hidden variable (cf. Theorem 4.51 p.77 \cite{gogioso2022topology}).
For each fixed value of the random variable, the output at the first event depends deterministically on the input: hence, any formal dependence of the empirical model for the remaining parties on the output for the first even is lifted to a dependence on its input, compatibly with our definition.
\end{remark}

The increased generality of the notion we provide allows non-trivial causal constraints to be captured, and is therefore more suitable for the analysis of practical experimental setups, where much of the causal structure between events is typically known.
Additionally, our approach has two distinct practical advantages:
\begin{itemize}
    \item Fixing a non-indiscrete ambient space $\Theta$ restricts the number of causal completions that need to be considered to compute the causally separable fraction of an empirical model. This can be extremely significant, even for few events: for the 6-party contextually-controlled classical switch example presented later on, the reduction is from the 16511297126400 (or more) causal completions of the indiscrete space to the known 16 causal completions for the ambient space provided.
    \item Fixing a non-indiscrete ambient space $\Theta$ makes it easier to identify concrete examples of indefinite causal order: there are empirical models---such as the two entangled quantum switches---which witness indefinite causality relative to the given ambient space, i.e. if part of the causal structure between events is known, but are causally separable relative to the indiscrete space, i.e. if no causal structure between the events is assumed. 
\end{itemize}

We conclude this Subsection by proving a result showing that causal inseparability is bounded above by the ``separable'' non-local fraction (cf. Definition 4.38 p.70 of ``The Topology of Causality'' \cite{gogioso2022topology}), and formulating a related conjecture.
Several concrete examples of causal separability correlating to contextuality, more or less tightly, are presented in Subsection \ref{subsection:geometry-causality-examples} below.

\begin{proposition}
\label{proposition:csep-frac-bounded-below-by-csep-local-frac}
Let $\Theta$ be a space of input histories and let $e$ be a standard empirical model for $\Theta$.
The causally separable fraction of $e$ is bounded below by its separable local fraction.
In particular, if $e$ is local in a causally separable way, then $e$ is causally separable.
\end{proposition}
\begin{proof}
See \ref{proof:proposition:csep-frac-bounded-below-by-csep-local-frac}
\end{proof}

\begin{conjecture}
Let $\Theta$ be a space of input histories and let $e$ be a standard empirical model for $\Theta$.
If $e$ is both local and causally separable, then it is local in a causally separable way.
\end{conjecture}

\subsection{Causal equations for standard causaltopes}
\label{subsection:std-causal-eqs}

To gain practical familiarity with standard causaltopes and the causal equations that define them, we investigate the standard causaltopes for the causally complete spaces with 3 events and binary inputs, listed in the companion work ``Classification of causally complete spaces on 3 events with binary inputs'' \cite{gogioso2022classification}.
Figure \ref{fig:nosig-eqs} (p.\pageref{fig:nosig-eqs}) shows the causality equations and quasi-normalisation equations for the no-signalling causaltope.
Discussing these equations is, in fact, enough to understand the standard causaltopes for all other spaces in the hierarchy: the discrete space $\Theta_0$, at the bottom, contains all possible extended input histories, and hence the causality equations for other spaces are always a subset of those for the discrete space.

\begin{figure}[p!]
    \centering
    \includegraphics[width=13cm]{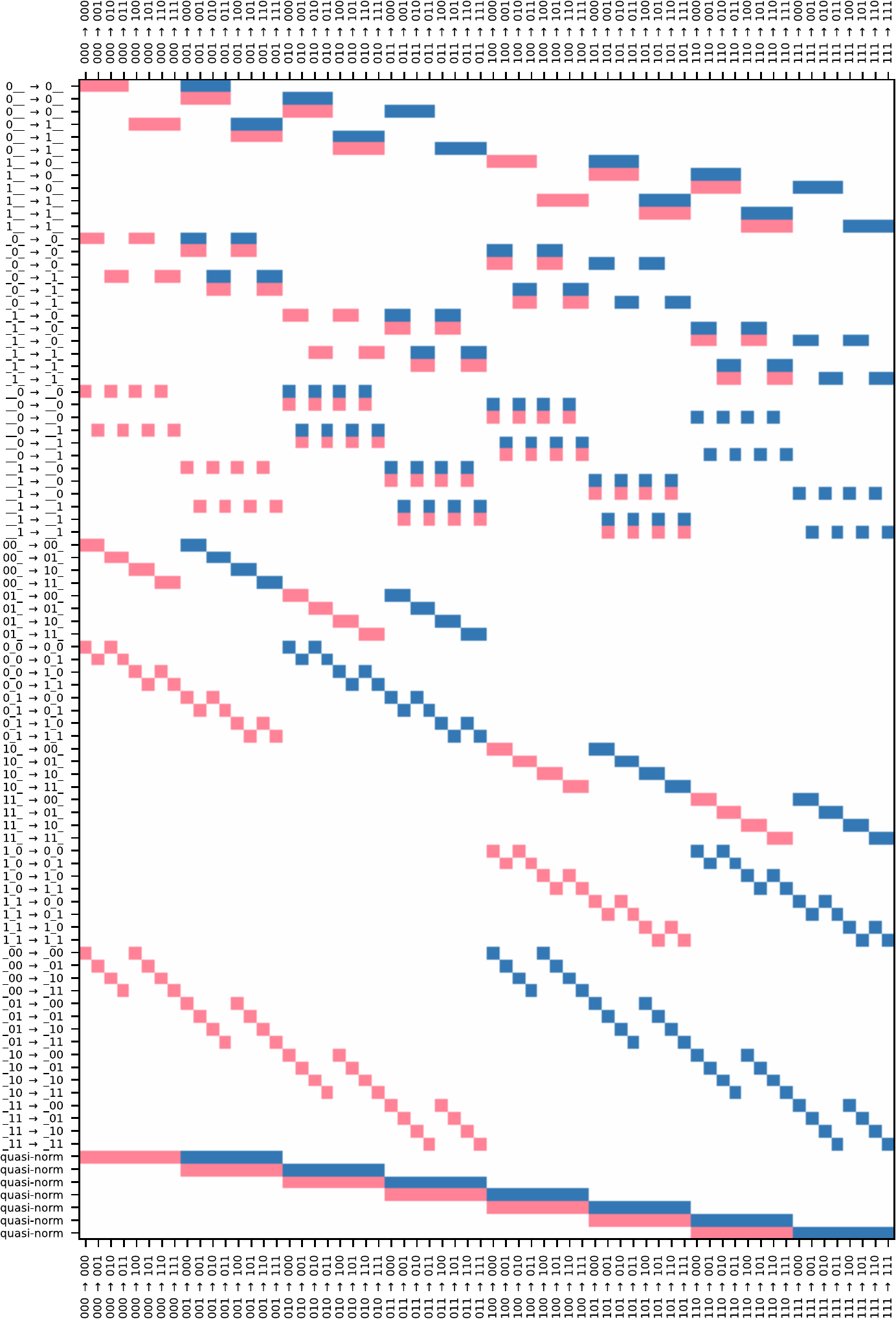}
    \caption{
    Causal equations for the no-signalling causaltope on 3 events and binary inputs.
    Columns correspond to the 64 components of standard pseudo-empirical models, linearised and indexed as $i_A i_B i_C \rightarrow o_A o_B o_C$.
    Rows correspond to the causality equations and quasi-normalisation equations for the no-signalling causaltope.
    Causality equations are indexed as $i_A i_B i_C \rightarrow o_A o_B o_C$, with underscores indicating marginals: see main text for discussion.
    Equation entries are colour-coded, for ease of reading: white is 0, red is 1, blue is -1.
    }
\label{fig:nosig-eqs}
\end{figure}

In Figure \ref{fig:nosig-eqs} (p.\pageref{fig:nosig-eqs}), the columns correspond to the 64 components of standard pseudo-empirical models, which have been linearised.
Specifically, a generic column is indexed as $i_A i_B i_C \rightarrow o_A o_B o_C$, and it corresponds to the component $u^{(\downset{k_{in}})}_{k_{out}}$ of an empirical model $\underline{u}$ for the following input/output histories:
\[
k_{in} := \hist{A/i_A, B/i_B, C/i_C}
\hspace{2cm}
k_{out} := \hist{A/o_A, B/o_B, C/o_C}
\]
For example, consider the following pseudo-empirical model:
\begin{center}
\scalebox{0.8}{
\begin{tabular}{l|rrrrrrrr}
\hfill
ABC & 000 & 001 & 010 & 011 & 100 & 101 & 110 & 111\\
\hline
000 & 0.250 & 0.250 & 0.000 & 0.000 & 0.000 & 0.000 & 0.250 & 0.250\\
001 & 0.000 & 0.000 & 0.250 & 0.250 & 0.250 & 0.250 & 0.000 & 0.000\\
010 & 0.188 & 0.188 & 0.062 & 0.062 & 0.062 & 0.062 & 0.188 & 0.188\\
011 & 0.062 & 0.062 & 0.188 & 0.188 & 0.188 & 0.188 & 0.062 & 0.062\\
100 & 0.188 & 0.188 & 0.062 & 0.062 & 0.062 & 0.062 & 0.188 & 0.188\\
101 & 0.062 & 0.062 & 0.188 & 0.188 & 0.188 & 0.188 & 0.062 & 0.062\\
110 & 0.250 & 0.062 & 0.000 & 0.188 & 0.000 & 0.188 & 0.250 & 0.062\\
111 & 0.188 & 0.000 & 0.062 & 0.250 & 0.062 & 0.250 & 0.188 & 0.000\\
\end{tabular}
}
\end{center}
Below is a heat-map representation of the same empirical model, with rows indexed by $i_A i_B i_C$ and columns indexed by $o_A o_B o_C$:
\begin{center}
    \includegraphics[width=6cm]{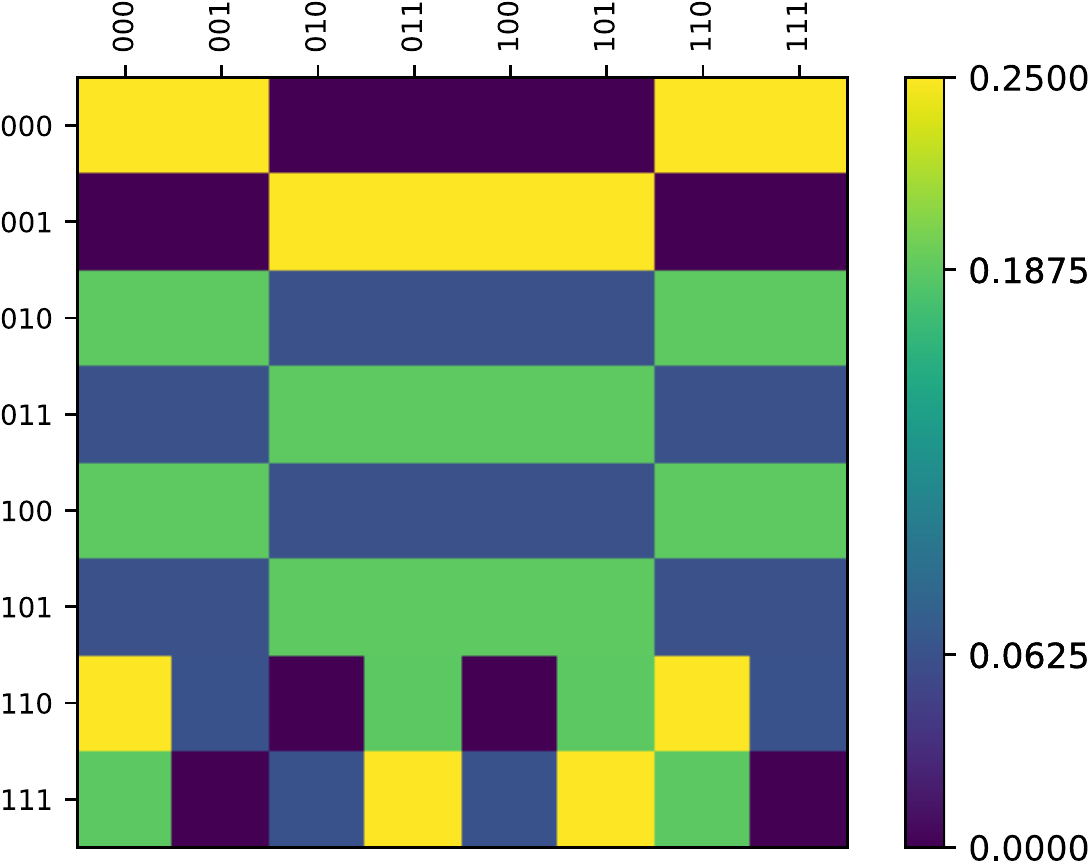}
\end{center}
Below is a the same heat-map, but with components linearised and indexed as $i_A i_B i_C \rightarrow o_A o_B o_C$:
\begin{center}
    \includegraphics[width=\textwidth]{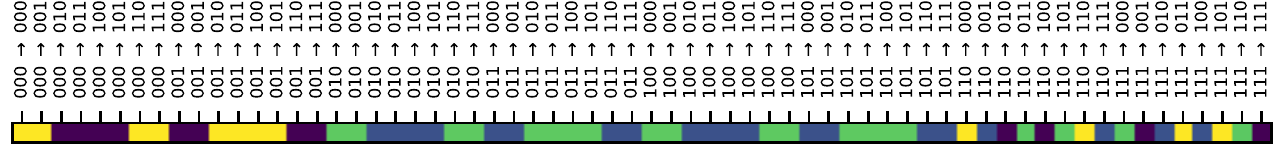}
\end{center}

Each row of Figure \ref{fig:nosig-eqs} (p.\pageref{fig:nosig-eqs}) corresponds to a causality equation or a quasi-normalisation equation (labelled ``quasi-norm'', at the bottom).
Equation entries are colour-coded, for ease of reading: white is 0, red is 1, blue is -1.
Causality equations take the form specified by Proposition \ref{proposition:caus-eqs-chain-std}: they are indexed as $i_A i_B i_C \rightarrow o_A o_B o_C$, with underscores indicating which events are marginalised over.
We go through the various blocks in turn, explaining how they arise.

The first block of equations is indexed by $i_A \_ \_ \rightarrow o_A \_ \_$.
For fixed $i_A, o_A \in \{0,1\}$, we have the equations $\StdCausEqs{\Theta_0, \{0,1\}}_{\downset{h}, \downset{k_{h,i}}, \downset{k_{h,i+1}}}$ for extended input history $h := \hist{A/i_A}$ and maximal extended input history $k_{h, i}$ taking the form $\hist{A/o_A, B/o_B, C/o_C}$ for some $o_B, o_C \in \{0,1\}$: since there are 4 possible such $k_{h, i}$, we have 3 equations (i.e. we have $i=1,2,3$).
Each row equates the sum of the red components (red means coefficient +1) with the sum of the blue components (blue means coefficient -1), ignoring all white components (white means coefficient 0).
\begin{center}
    \includegraphics[width=11cm]{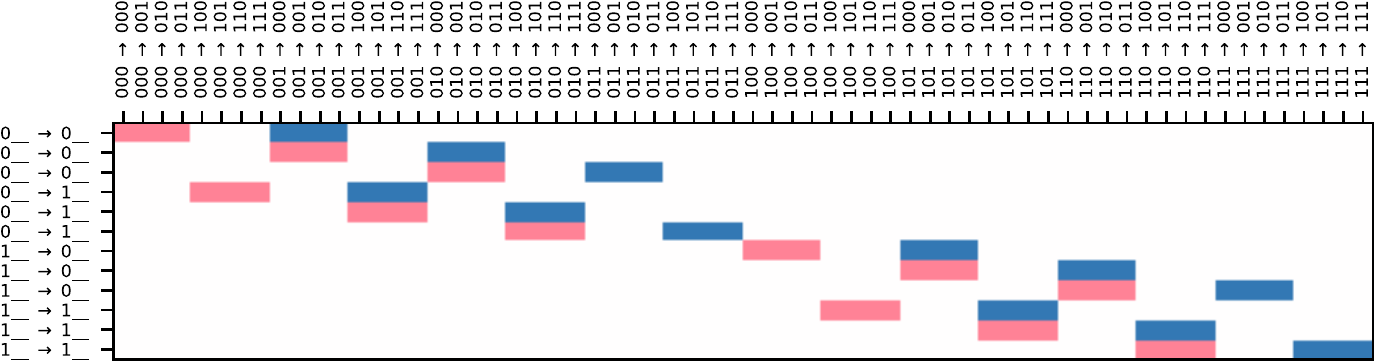}
\end{center}
The second block of equations is indexed by $\_ i_B \_ \rightarrow \_ o_B \_$.
For fixed $i_B, o_B \in \{0,1\}$, we have the equations $\StdCausEqs{\Theta_0, \{0,1\}}_{\downset{h}, \downset{k_{h,i}}, \downset{k_{h,i+1}}}$ for extended input history $h := \hist{B/i_B}$ and maximal extended input history $k_{h, i}$ taking the form $\hist{A/o_A, B/o_B, C/o_C}$ for some $o_A, o_C \in \{0,1\}$: since there are 4 possible such $k_{h, i}$, we have 3 equations (i.e. we have $i=1,2,3$).
\begin{center}
    \includegraphics[width=11cm]{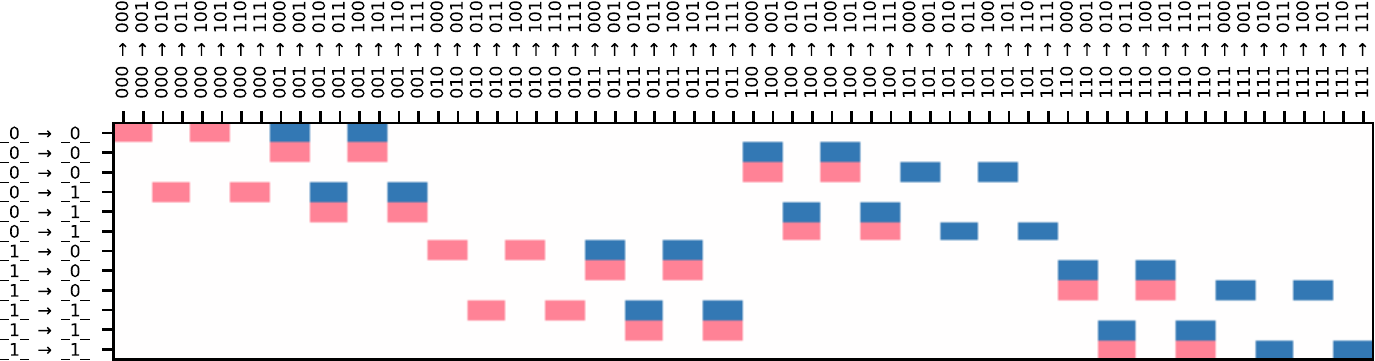}
\end{center}
The third block of equations is indexed by $\_ \_ i_C \rightarrow \_ \_ o_C$.
For fixed $i_C, o_C \in \{0,1\}$, we have the equations $\StdCausEqs{\Theta_0, \{0,1\}}_{\downset{h}, \downset{k_{h,i}}, \downset{k_{h,i+1}}}$ for extended input history $h := \hist{C/i_C}$ and maximal extended input history $k_{h, i}$ taking the form $\hist{A/o_A, B/o_B, C/o_C}$ for some $o_A, o_B \in \{0,1\}$: since there are 4 possible such $k_{h, i}$, we have 3 equations (i.e. we have $i=1,2,3$).
\begin{center}
    \includegraphics[width=11cm]{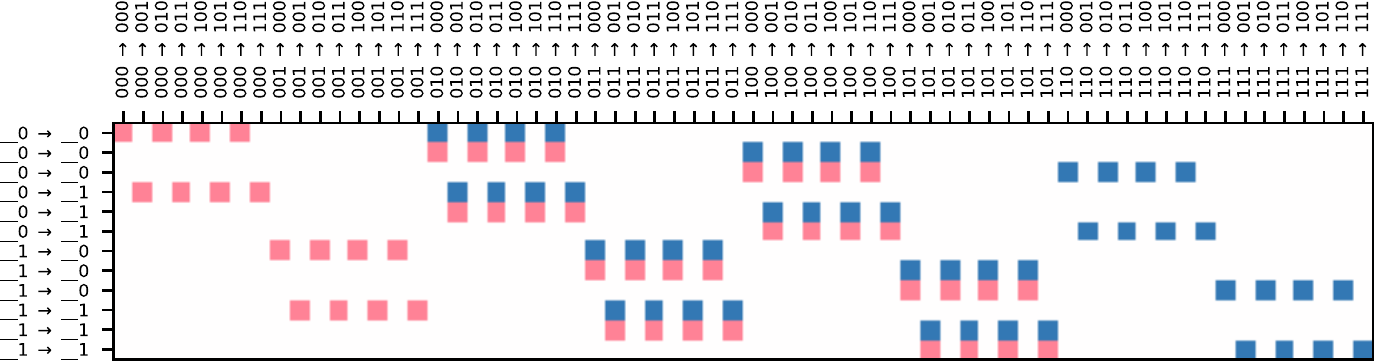}
\end{center}
The fourth block of equations is indexed by $i_A i_B \_ \rightarrow o_A o_B \_$.
For fixed $i_A, i_B, o_A, o_B \in \{0,1\}$, we have the equation $\StdCausEqs{\Theta_0, \{0,1\}}_{\downset{h}, \downset{k_{h,i}}, \downset{k_{h,i+1}}}$ for extended input history $h := \hist{A/i_A, B/i_B}$ and maximal extended input history $k_{h, i}$ taking the form $\hist{A/o_A, B/o_B, C/o_C}$ for some $o_C \in \{0,1\}$: since there are 2 possible such $k_{h, i}$, we have 1 equation (i.e. we have $i=1$ only).
\begin{center}
    \includegraphics[width=11cm]{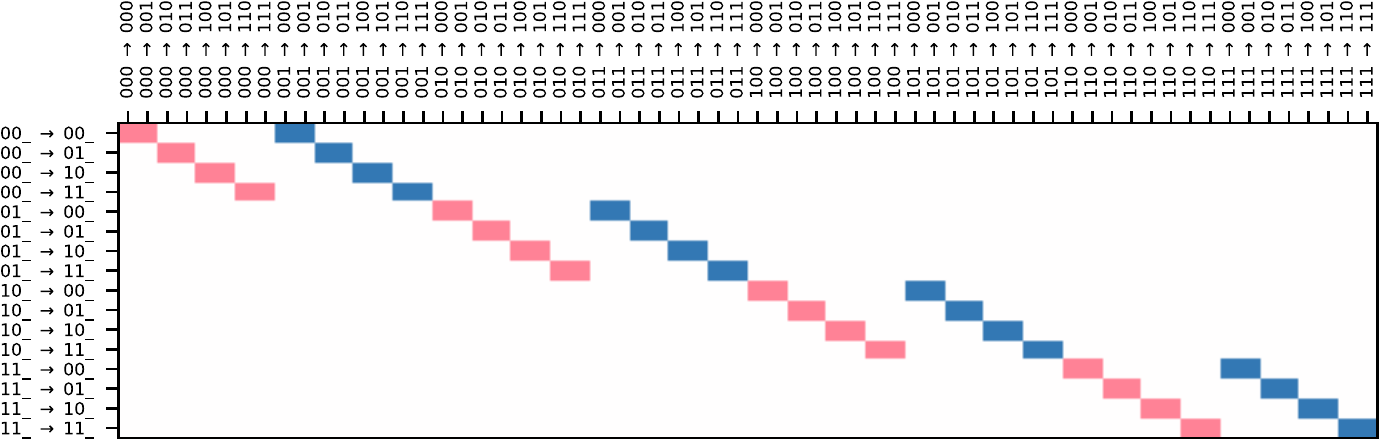}
\end{center}
The fifth block of equations is indexed by $i_A \_ i_C \rightarrow o_A \_ o_C$.
For fixed $i_A, i_C, o_A, o_C \in \{0,1\}$, we have the equation $\StdCausEqs{\Theta_0, \{0,1\}}_{\downset{h}, \downset{k_{h,i}}, \downset{k_{h,i+1}}}$ for extended input history $h := \hist{A/i_A, C/i_C}$ and maximal extended input history $k_{h, i}$ taking the form $\hist{A/o_A, B/o_B, C/o_C}$ for some $o_B \in \{0,1\}$: since there are 2 possible such $k_{h, i}$, we have 1 equation (i.e. we have $i=1$ only).
\begin{center}
    \includegraphics[width=11cm]{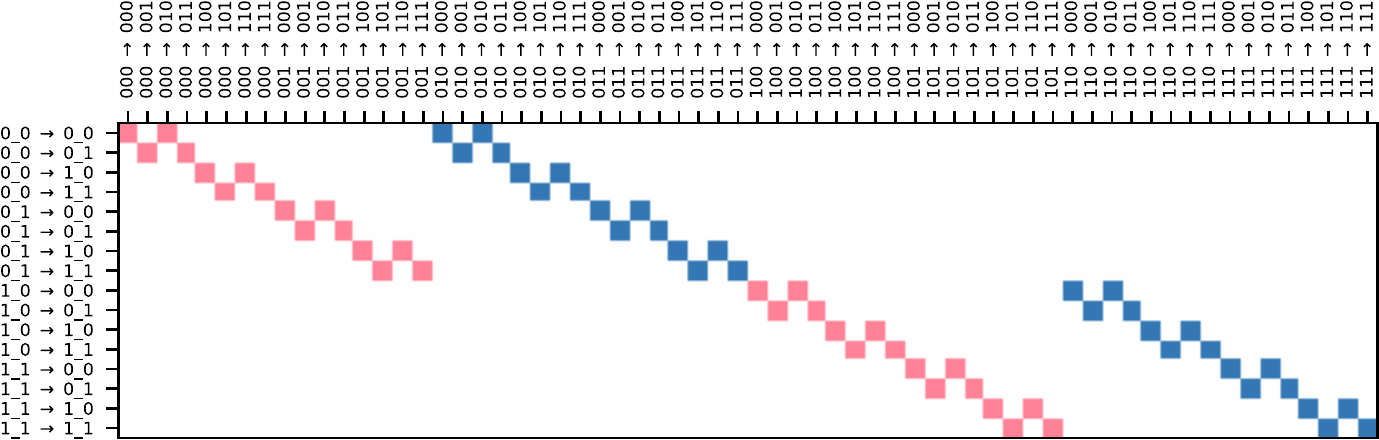}
\end{center}
The sixth block of equations is indexed by $\_ i_B i_C \rightarrow \_ o_B o_C$.
For fixed $i_B, i_C, o_B, o_C \in \{0,1\}$, we have the equation $\StdCausEqs{\Theta_0, \{0,1\}}_{\downset{h}, \downset{k_{h,i}}, \downset{k_{h,i+1}}}$ for extended input history $h := \hist{B/i_B, C/i_C}$ and maximal extended input history $k_{h, i}$ taking the form $\hist{A/o_A, B/o_B, C/o_C}$ for some $o_A \in \{0,1\}$: since there are 2 possible such $k_{h, i}$, we have 1 equation (i.e. we have $i=1$ only).
\begin{center}
    \includegraphics[width=11cm]{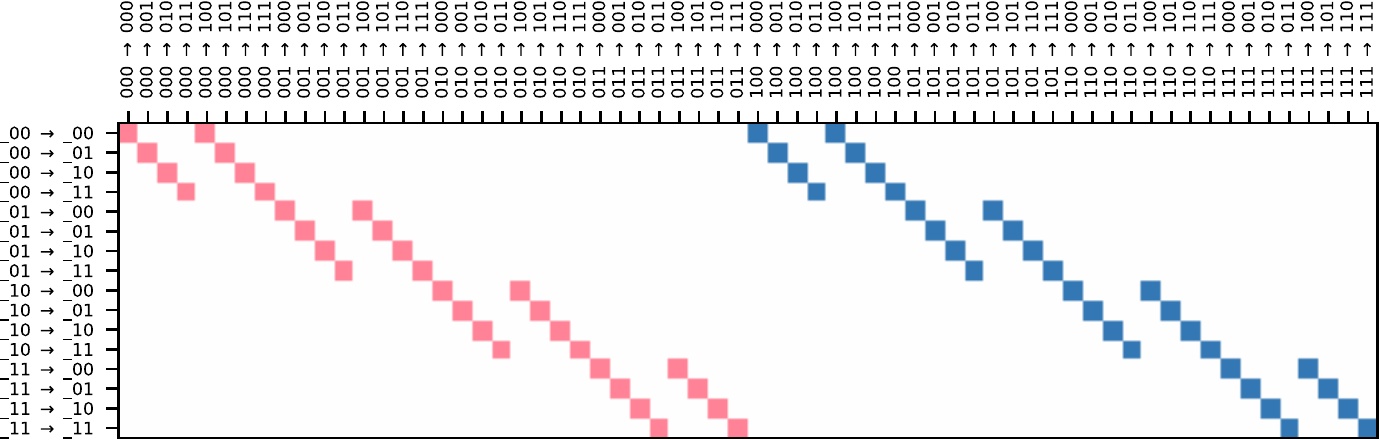}
\end{center}
The seventh and final block of equations consists of the quasi-normalisation equations, equating the sum of the components in successive rows of a pseudo-empirical model.
\begin{center}
    \includegraphics[width=11cm]{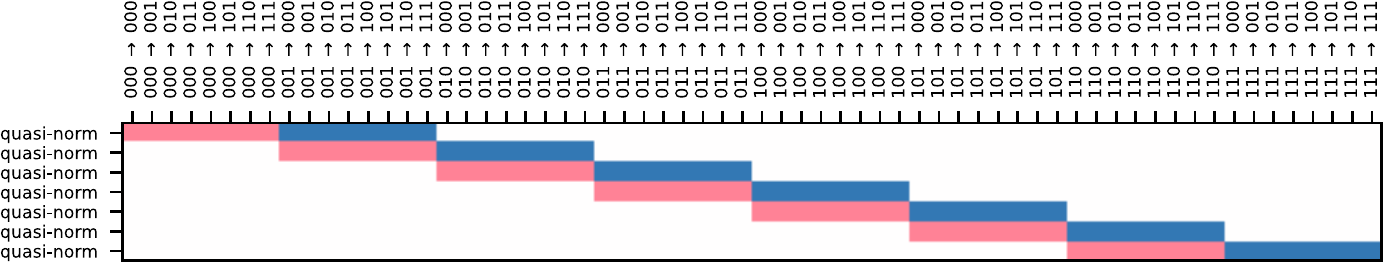}
\end{center}
As an example of how the equations above are combined to defined the causaltopes for other spaces, we consider the set of equations for the space $\Hist{\total{A, B, C}, \{0,1\}}$, shown below.
The first block of equations from the no-signalling causaltope provides the marginals for histories $\hist{A/i_A}$, the fourth block provides the marginals for histories $\hist{A/i_A, B/i_B}$ and the seventh block provides the quasi-normalisation equations.
\begin{center}
    \includegraphics[width=11cm]{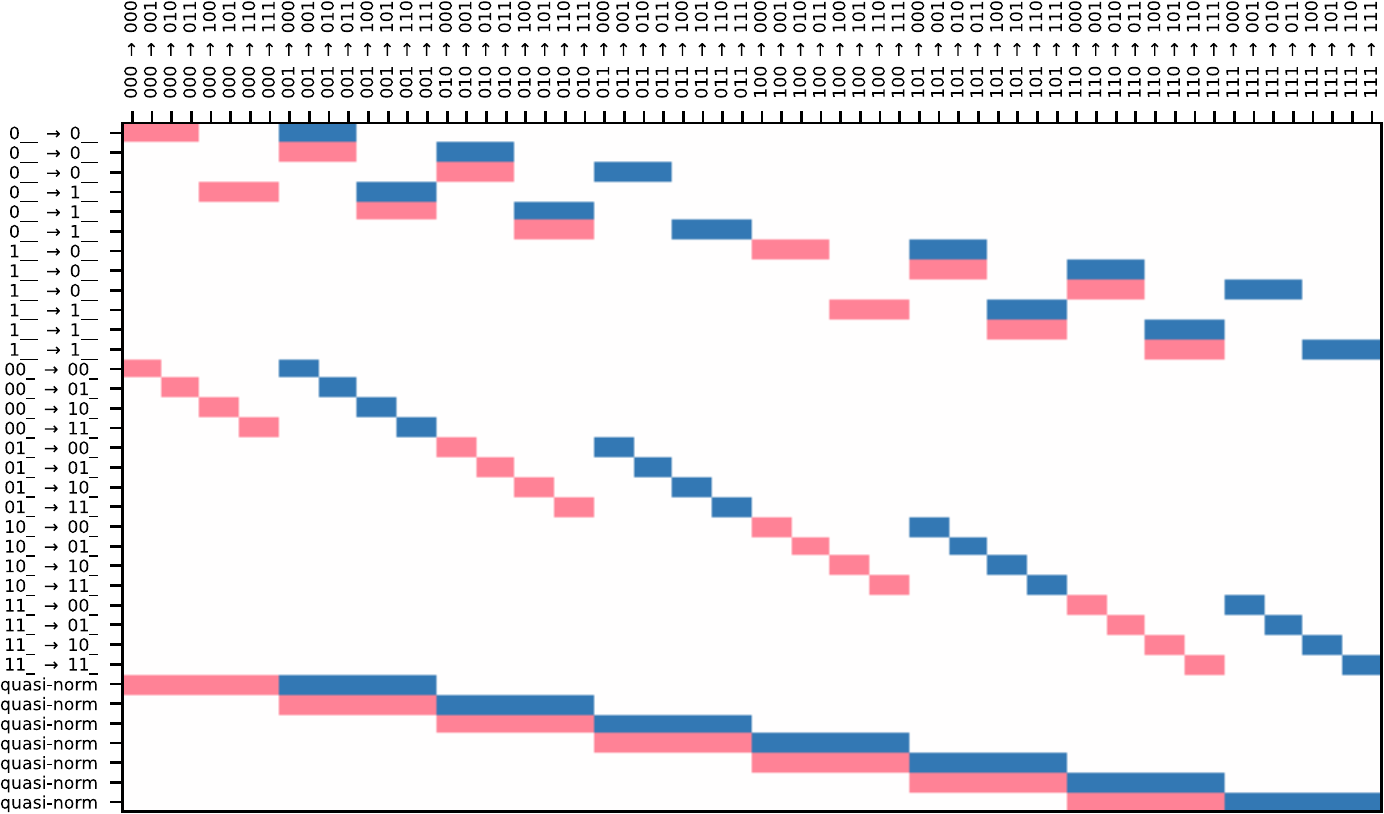}
\end{center}
Analogously, we can consider the set of equations for the space $\Hist{\total{A, C, B}, \{0,1\}}$, shown below.
The first block of equations from the no-signalling causaltope provides the marginals for histories $\hist{A/i_A}$, the fifth block provides the marginals for histories $\hist{A/i_A, C/i_C}$ and the seventh block provides the quasi-normalisation equations.
\begin{center}
    \includegraphics[width=12cm]{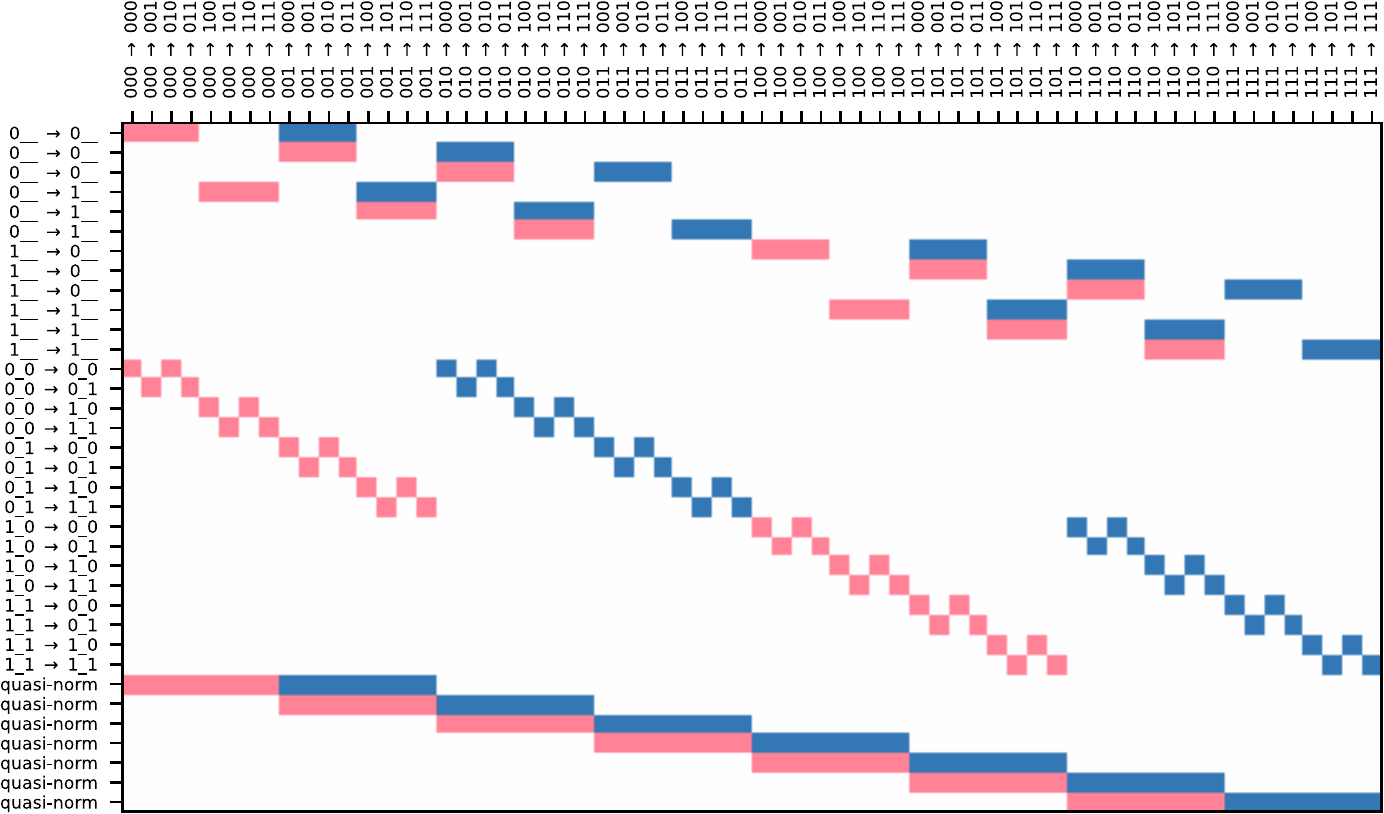}
\end{center}
If we splice the causality equations for the two total orders, taking $\hist{A/0, B/i_B}$ from the first and $\hist{A/1, C/i_C}$ for the second, we obtain the causality equations for the switch space where $\ev{A}$ controls the order of $\ev{B}$ and $\ev{C}$, shown below.
\begin{center}
    \includegraphics[width=11cm]{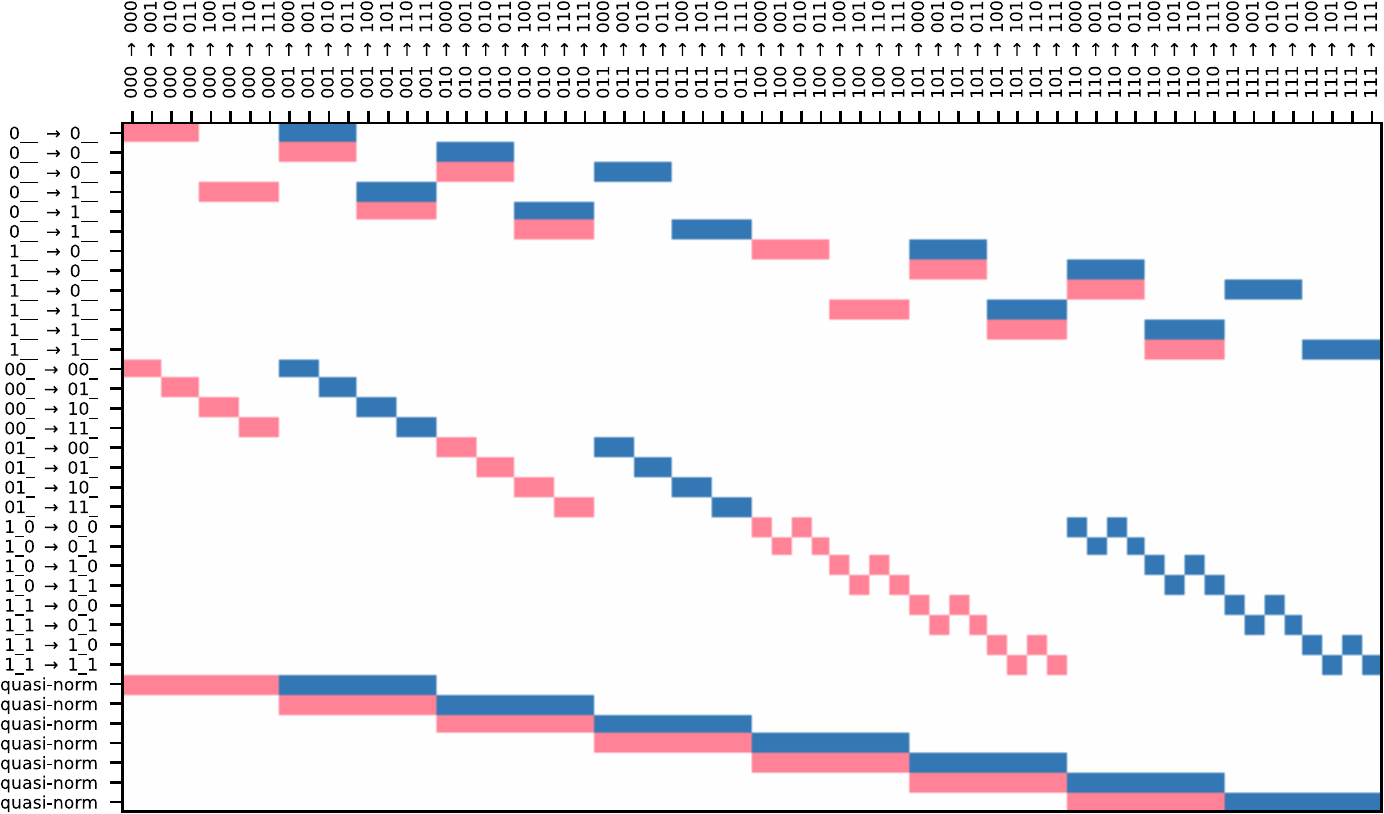}
\end{center}

The equations presented above are still redundant: in order to compute causaltope dimensions, or compare causaltopes, we first need to put them in reduced row echelon form (RREF).
At this point, it becomes clear why we included the quasi-normalisation equations in the mix: the full set of equations defining the polytope must be considered when computing the RREF.
Furthermore, by omitting the final normalisation equation we are considering the cone over the causaltope, which is where the linear programs for component/decomposition take place.

Because the empirical models on 3 events with binary inputs/output have $8 \times 8 = 64$ components, the embedding space for our causaltopes is 64-dimensional; however, 1 dimension/degree of freedom is taken away by the normalisation equation.
Hence, the dimension of a causaltope can be easily computed as 63 minus the number of non-zero rows in the system of causality and quasi-normalisation equations in RREF.

For example, below is the system of causality and quasi-normalisation equations in RREF for the no-signalling causaltope.
The RREF system has 37 non-zero rows, and hence the causaltope has dimension $63-37=26$.
\begin{center}
    \includegraphics[width=11cm]{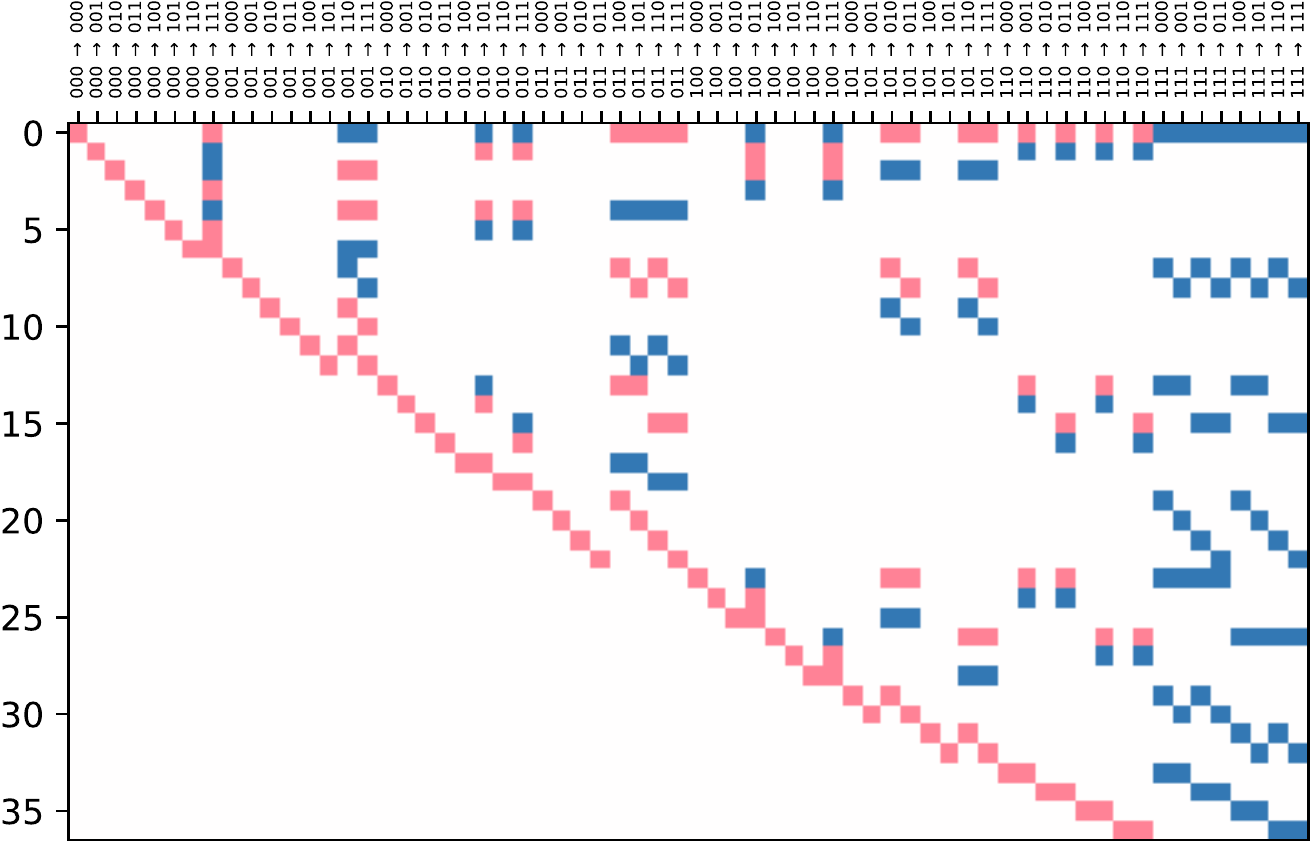}
\end{center}
As a second example, below are the systems of causality and quasi-normalisation equations in RREF for the total orders $\total{A, B, C}$ and $\total{A, C, B}$.
The RREF systems have 21 non-zero rows each, and hence the two causaltopes have dimension $63-21=42$.
\begin{center}
    \includegraphics[width=11cm]{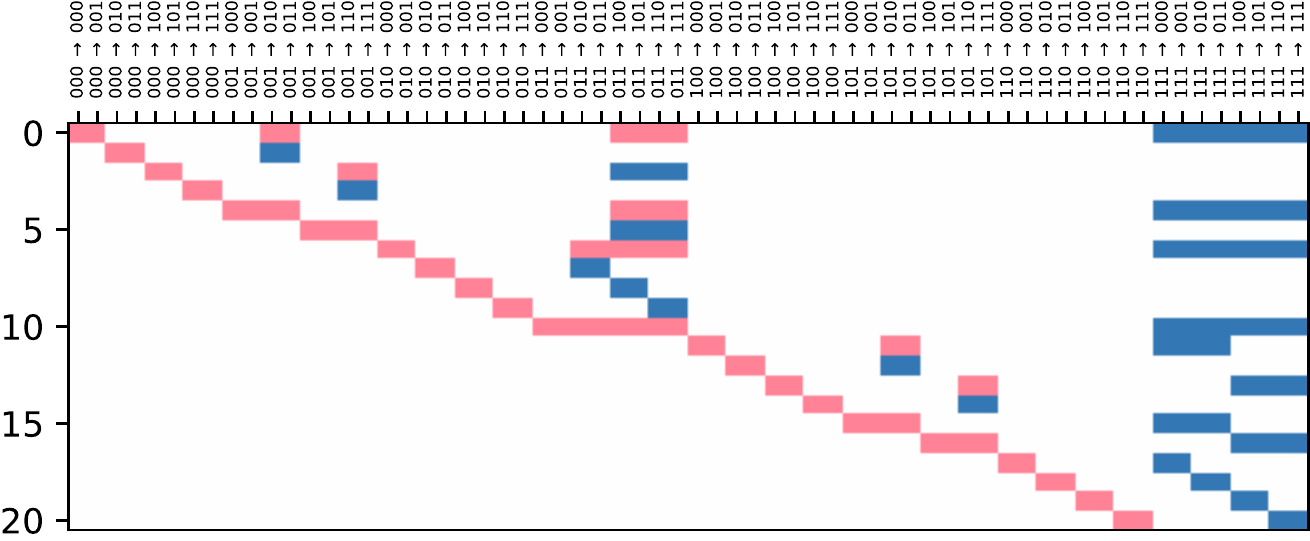}
\end{center}
\begin{center}
    \includegraphics[width=11cm]{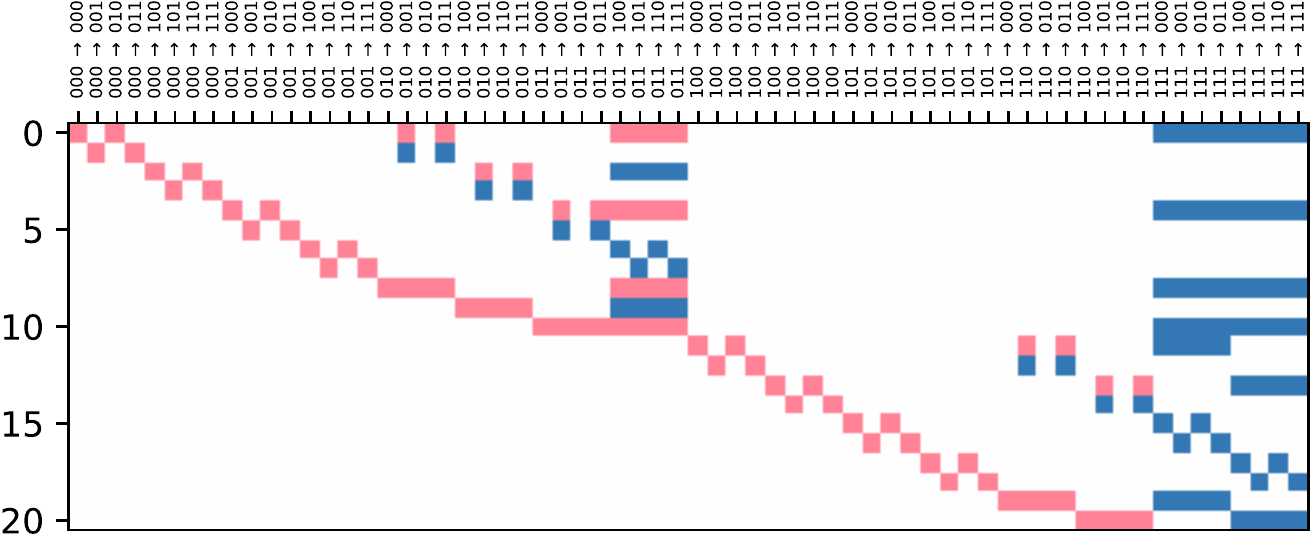}
\end{center}
As a third example, below is the system of causality and quasi-normalisation equations in RREF for the switch order previously considered.
The RREF system has 21 non-zero rows, and hence the causaltope has dimension $63-21=42$.
\begin{center}
    \includegraphics[width=11cm]{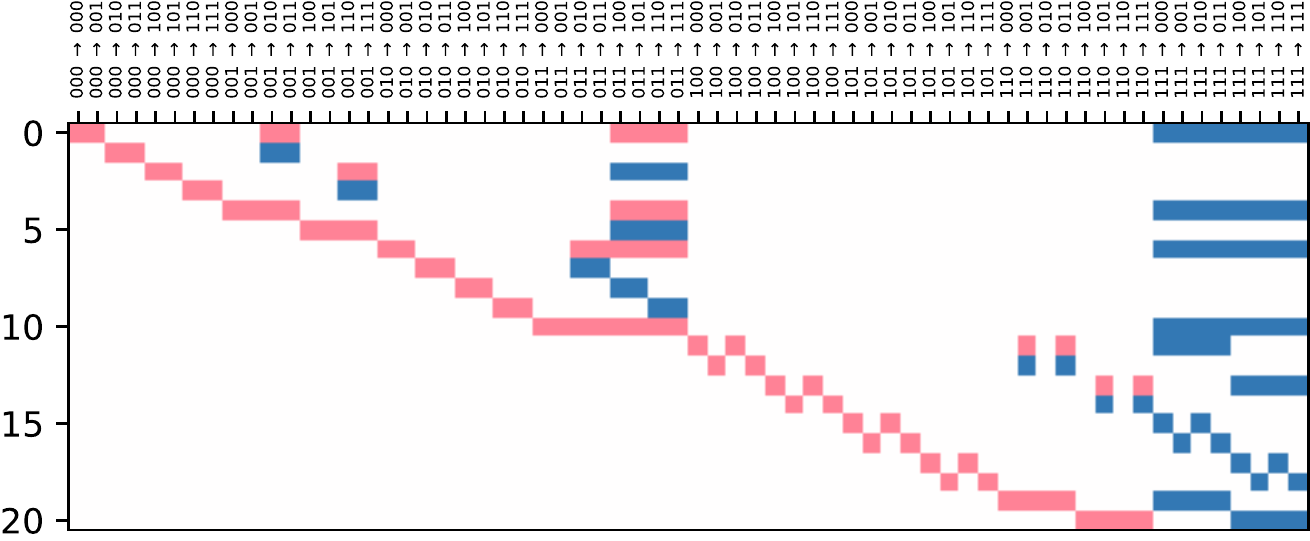}
\end{center}
Figure \ref{fig:hierarchy-spaces-ABC-causaltope-dim} (p.\pageref{fig:hierarchy-spaces-ABC-causaltope-dim}) shows the dimensions for the standard causaltopes of all causally complete spaces on 3 events with binary inputs/outputs: all spaces in the same equivalence class under event-input permutation the same-dimensional causaltopes, so it suffices to show the dimension for each one of the 102 equivalence classes.

The edges colour indicates the minimum increase in causaltope dimension from a subspace: dark blue edges indicate that the standard causaltopes for spaces in the equivalence class coincide with the standard causaltope for some sub-space.
For example, the standard causaltope for spaces in equivalence classes 1 and 3 is the no-signalling causaltope: hence, it comes as no surprise that the causal functions for these spaces are exactly the no-signalling ones; see Figure~5 (p.42) of \cite{gogioso2022combinatorics}.
However, spaces in equivalence class 2 have 27-dimensional causaltopes, despite having the no-signalling functions as their causal functions: standard empirical models for some space $\Theta$ in equivalence class 2 which are not no-signalling, i.e. which are not 100\% supported by the no-signalling causaltope, must necessarily be contextual/non-local in $\Theta$.
This is, indeed, the case for our ``causal fork'' example below.

\begin{figure}[h]
    \centering
    \includegraphics[width=0.9\textwidth]{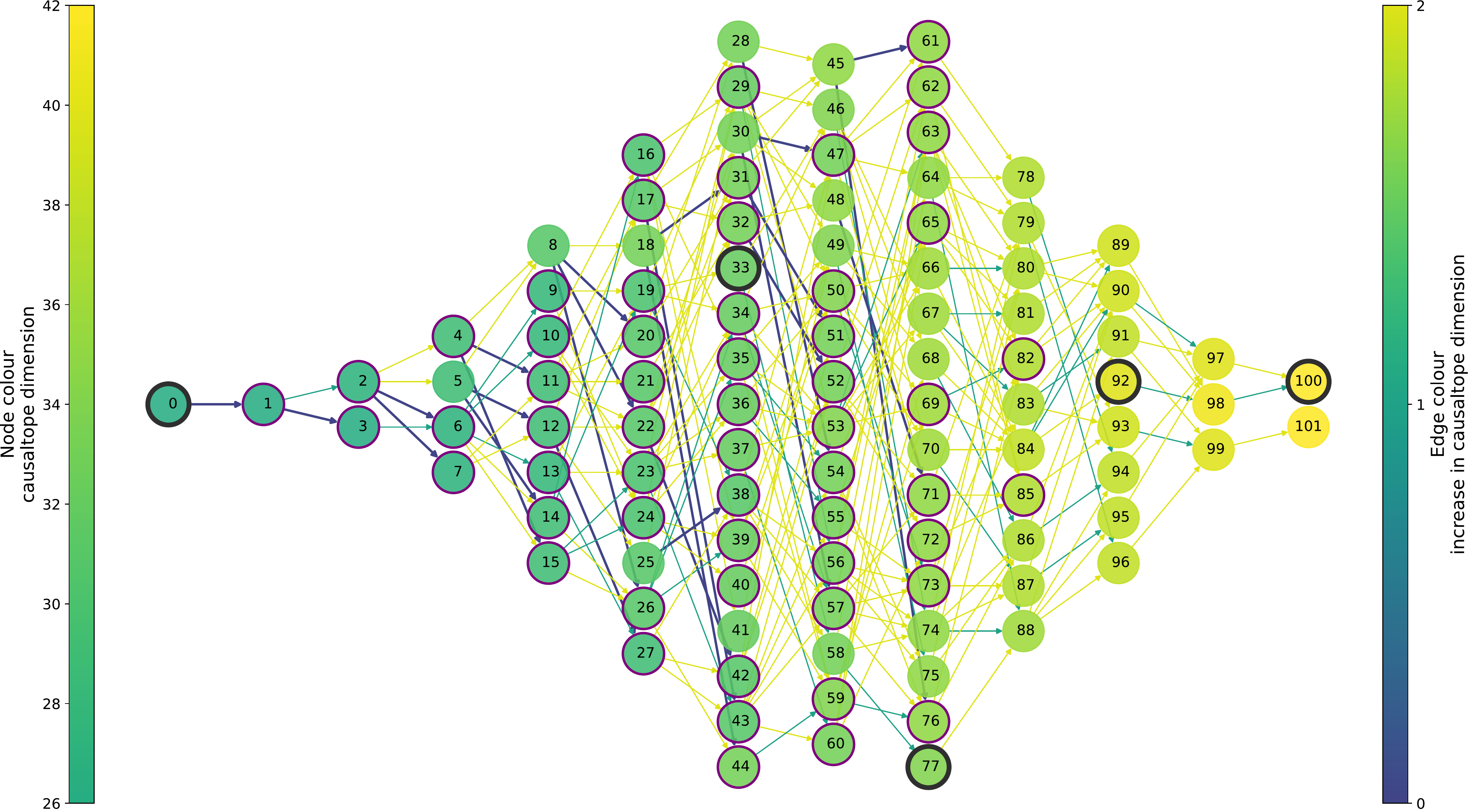}
    \caption{
    The hierarchy of standard causaltopes on causally complete spaces with 3 events and binary inputs, cf. Figure~5 (p.42) of \cite{gogioso2022combinatorics}.
    Node colour indicates the dimension of the causaltopes in the equivalence class, while edge colour indicates the increase in dimension from the causaltopes for the tail spaces to the causaltopes for the head spaces.
    }
\label{fig:hierarchy-spaces-ABC-causaltope-dim}
\end{figure}

\newpage

\subsection{Examples of standard empirical models}
\label{subsection:geometry-causality-examples}

In this subsection, we discuss the empirical models for a selection of examples of interest. All empirical models are for the standard cover, so that any non-classicality arises from non-locality rather than other forms of contextuality (causally-induced or otherwise).

All models have binary inputs and outputs $I_\omega = O_\omega = \{0, 1\}$ at each event, unless otherwise specified.
For convenience, we will describe our scenarios in terms of agents performing operations at the events, always following the same convention: Alice acts at event \ev{A}, Bob acts at event \ev{B}, Charlie acts at event \ev{C}, Diane acts at event \ev{D}, Eve acts at event \ev{E} and Felix acts at event \ev{F}.

\subsubsection{A Classical Switch Empirical Model.}

In this example, Alice classically controls the order of Bob and Charlie, as follows:
\begin{itemize}
    \item Alice flips one of two biased coins, depending on her input: when her input is 0, her output is 75\% 0 and 25\% 1; when her input is 1, her output is 25\% 0 and 75\% 1 instead.
    \item Bob and Charlie are in a quantum switch, controlled in the Z basis and with $|0\rangle$ as a fixed input: on output $a \in \{0, 1\}$, Alice feeds state $|a\rangle$ into the control system of the switch, determining the relative causal order of Bob and Charlie.
    \item Bob and Charlie both apply the same quantum instrument: they measure the incoming qubit they receive in the Z basis, obtaining their output, and then encode their input into the Z basis of the outgoing qubit.
    \item Both the control qubit and the outgoing qubit of the switch are discarded: even without Alice controlling the switch in the Z basis, discarding the control qubit would be enough to make the control classical.
\end{itemize}
The description above results in the following empirical model on 3 events:
\begin{center}
\begin{tabular}{l|rrrrrrrr}
\hfill
ABC & 000 & 001 & 010 & 011 & 100 & 101 & 110 & 111\\
\hline
000 & $3/4$ & $0$ & $0$ & $0$ & $1/4$ & $0$ & $0$ & $0$\\
001 & $3/4$ & $0$ & $0$ & $0$ & $0$ & $0$ & $1/4$ & $0$\\
010 & $0$ & $3/4$ & $0$ & $0$ & $1/4$ & $0$ & $0$ & $0$\\
011 & $0$ & $3/4$ & $0$ & $0$ & $0$ & $0$ & $1/4$ & $0$\\
100 & $1/4$ & $0$ & $0$ & $0$ & $3/4$ & $0$ & $0$ & $0$\\
101 & $1/4$ & $0$ & $0$ & $0$ & $0$ & $0$ & $3/4$ & $0$\\
110 & $0$ & $1/4$ & $0$ & $0$ & $3/4$ & $0$ & $0$ & $0$\\
111 & $0$ & $1/4$ & $0$ & $0$ & $0$ & $0$ & $3/4$ & $0$\\
\end{tabular}
\end{center}
To better understand the table above, we focus on the second row, corresponding to input 001:
\begin{enumerate}
    \item Alice's input is 0, so her output is 75\% 0 and 25\% 1. This means that the probabilities of outputs $0\_\_$ in row 001 of the empirical model must sum to 75\%, and the probabilities of output $1\_\_$ must sum to 25\%.
    \item Conditional to Alice's output being 0, the output is $000$ with 100\% probability:
    \begin{enumerate}
        \item Bob goes first and receives the input state $|0\rangle$ for the switch: he measures the state in the Z basis, obtaining output 0 with 100\% probability.
        Because his input is 0, he then prepares the state $|0\rangle$, which he forwards into the switch.
        \item Charlie goes second and receives the state $|0\rangle$ prepared by Bob: he measures the state in the Z basis, obtaining output 0 with 100\% probability.
        Because his input is 1, he then prepares the state $|1\rangle$, which he forwards into the switch.
        \item Charlie's state $|1\rangle$ comes out of the switch, and is discarded.
    \end{enumerate}
    \item Conditional to Alice's output being 1, the output is $110$ with 100\% probability:
    \begin{enumerate}
        \item Charlie goes first and receives the input state $|0\rangle$ for the switch: he measures the state in the Z basis, obtaining output 0 with 100\% probability.
        Because his input is 1, he then prepares the state $|1\rangle$, which he forwards into the switch.
        \item Bob goes second and receives the state $|1\rangle$ prepared by Charlie: he measures the state in the Z basis, obtaining output 1 with 100\% probability.
        Because his input is 0, he then prepares the state $|0\rangle$, which he forwards into the switch.
        \item Bob's state $|0\rangle$ comes out of the switch, and is discarded.
    \end{enumerate}
\end{enumerate}
This empirical model is causally separable.
A maximum fraction of 75\% is supported by the switch space where Alice choosing 0 makes Bob precede Charlie and a maximum fraction of 25\% is supported by the switch space where Alice choosing 0 makes Charlie precede Bob, with a fraction of 0\% supported by both spaces (i.e. no overlap).
Below we show the two spaces, the corresponding causal fraction, and the (renormalised) component of the empirical model supported by each space:
\begin{center}
    \begin{tabular}{cc}
    \includegraphics[height=3cm]{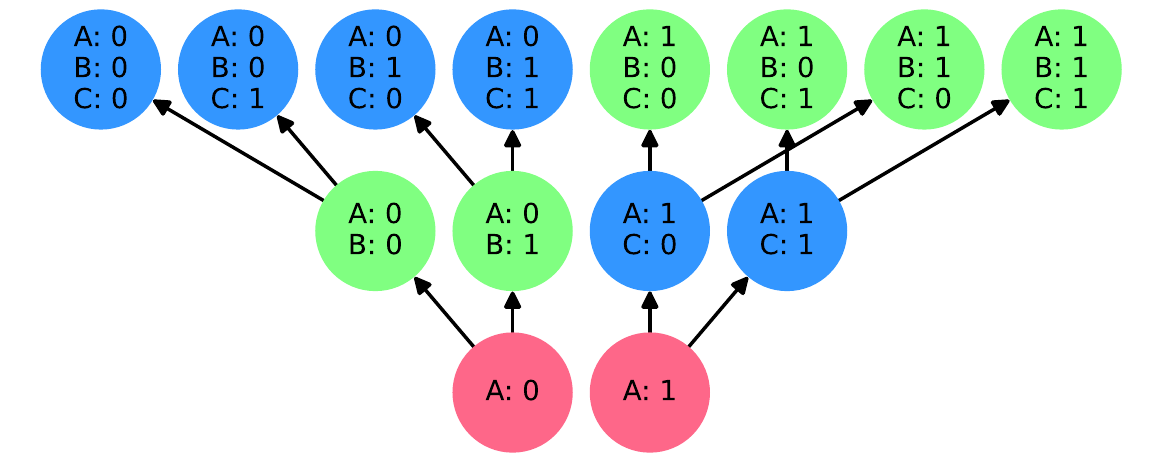}
    &
    \includegraphics[height=3cm]{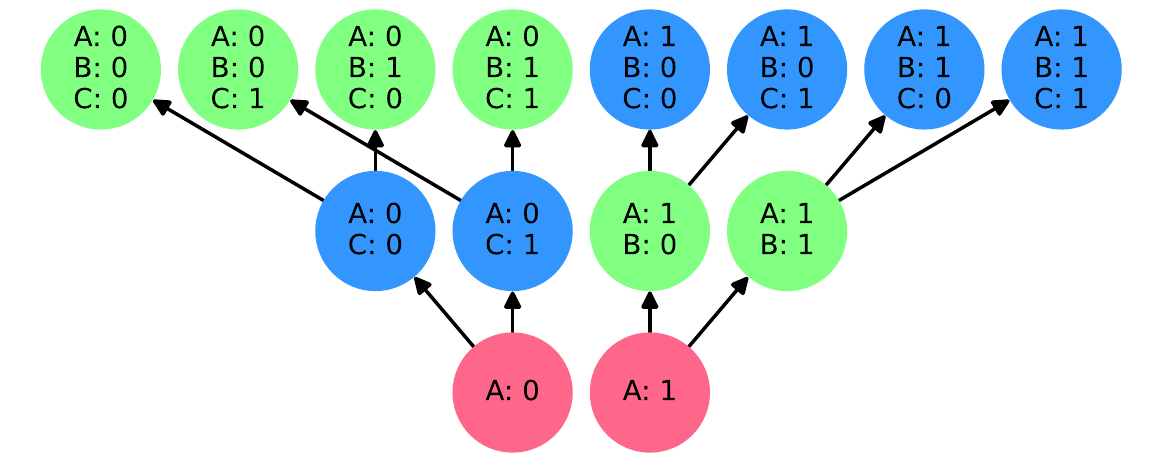}
    \\
    Causal fraction: $75\%$
    &
    Causal fraction: $25\%$
    \vspace{2mm}
    \\
    \scalebox{0.75}{
    \begin{tabular}{l|rrrrrrrr}
    \hfill
    ABC & 000 & 001 & 010 & 011 & 100 & 101 & 110 & 111\\
    \hline
    000 & 1 & 0 & 0 & 0 & 0 & 0 & 0 & 0\\
    001 & 1 & 0 & 0 & 0 & 0 & 0 & 0 & 0\\
    010 & 0 & 1 & 0 & 0 & 0 & 0 & 0 & 0\\
    011 & 0 & 1 & 0 & 0 & 0 & 0 & 0 & 0\\
    100 & 0 & 0 & 0 & 0 & 1 & 0 & 0 & 0\\
    101 & 0 & 0 & 0 & 0 & 0 & 0 & 1 & 0\\
    110 & 0 & 0 & 0 & 0 & 1 & 0 & 0 & 0\\
    111 & 0 & 0 & 0 & 0 & 0 & 0 & 1 & 0\\
    \end{tabular}
    }
    &
    \scalebox{0.75}{
    \begin{tabular}{l|rrrrrrrr}
    \hfill
    ABC & 000 & 001 & 010 & 011 & 100 & 101 & 110 & 111\\
    \hline
    000 & 0 & 0 & 0 & 0 & 1 & 0 & 0 & 0\\
    001 & 0 & 0 & 0 & 0 & 0 & 0 & 1 & 0\\
    010 & 0 & 0 & 0 & 0 & 1 & 0 & 0 & 0\\
    011 & 0 & 0 & 0 & 0 & 0 & 0 & 1 & 0\\
    100 & 1 & 0 & 0 & 0 & 0 & 0 & 0 & 0\\
    101 & 1 & 0 & 0 & 0 & 0 & 0 & 0 & 0\\
    110 & 0 & 1 & 0 & 0 & 0 & 0 & 0 & 0\\
    111 & 0 & 1 & 0 & 0 & 0 & 0 & 0 & 0\\
    \end{tabular}
    }
    \end{tabular}
\end{center}

\subsubsection{A Causal Fork Empirical Model.}

In this example, Charlie produces one of the four Bell basis states and forwards one qubit each to Alice and Bob, who measure it in either the Z or X basis:
\begin{enumerate}
    \item On input $c \in \{0,1\}$ Charlie prepares the 2-qubit state $|0c\rangle$. He then performs a XX parity measurement, resulting in one of $|\Phi^{\pm}\rangle$ states (if his input was 0) or one of $|\Psi^{\pm}\rangle$ states (if his input was 1), all with 50\% probability. He forwards this state to Alice and Bob, one qubit each.
    \item Alice and Bob perform a Z basis measurement on input 0 and an X basis measurement on input 1, and use the measurement outcome as their output.
\end{enumerate}
The following figure summarises the experiment:
\begin{center}
\scalebox{1.25}{    
\begin{tikzpicture}
	\begin{pgfonlayer}{nodelayer}
		\node [style=none] (0) at (-1, -2) {};
		\node [style=none] (1) at (1, -2) {};
		\node [style=none] (3) at (-2.25, 3.25) {};
		\node [style=none] (5) at (-2.25, 3.25) {};
		\node [style=none] (6) at (-2.25, 1) {};
		\node [style=none] (10) at (-4.025, 4.5) {};
		\node [style=none] (11) at (-0.525, 4.5) {};
		\node [style=none] (12) at (-0.525, 1.75) {};
		\node [style=none] (13) at (-4.025, 1.75) {};
		\node [style=red label] (14) at (-4.75, 3.25) {$A$};
		\node [style=none] (16) at (2.225, 3.25) {};
		\node [style=none] (18) at (2.225, 3.25) {};
		\node [style=none] (19) at (2.225, 1) {};
		\node [style=none] (23) at (4, 4.5) {};
		\node [style=none] (24) at (0.5, 4.5) {};
		\node [style=none] (25) at (0.5, 1.75) {};
		\node [style=none] (26) at (4, 1.75) {};
		\node [style=red label] (27) at (4.725, 3.25) {$B$};
		\node [style=small box] (29) at (-2.25, 3.25) {$X/Z$};
		\node [style=small box] (30) at (2.25, 3.25) {$X/Z$};
		\node [style=none] (31) at (-2.5, 3) {};
		\node [style=none] (32) at (-2.25, 3.5) {};
		\node [style=none] (33) at (2.5, 3) {};
		\node [style=none] (34) at (2.25, 3.5) {};
		\node [style=none] (35) at (-2.25, 5.25) {};
		\node [style=none] (36) at (2.25, 5.25) {};
		\node [style=none] (37) at (-3.25, 1.25) {};
		\node [style=none] (38) at (3.25, 1.25) {};
		\node [style=none] (39) at (2.725, -1.5) {};
		\node [style=none] (40) at (-1.775, -1.5) {};
		\node [style=none] (41) at (-1.775, -4) {};
		\node [style=none] (42) at (2.725, -4) {};
		\node [style=small box] (43) at (1, -2.5) {$\Phi^{\pm}/\Psi^{\pm}$};
		\node [style=none] (44) at (1, -3) {};
		\node [style=none] (45) at (-1, -3) {};
		\node [style=none] (46) at (1.25, -2.75) {};
		\node [style=eslabel] (47) at (1.25, -4.75) {$\{ \Phi,\Psi \}$};
		\node [style=none] (48) at (0.75, -2.25) {};
		\node [style=none] (49) at (0, -0.75) {$$};
		\node [style=red label] (50) at (-2.5, -2.75) {$C$};
		\node [style=eslabel] (51) at (-2.25, 5.5) {$\{0,1\}$};
		\node [style=eslabel] (53) at (-3.25, 1) {$\{X,Z\}$};
		\node [style=eslabel] (54) at (3.25, 1) {$\{X,Z\}$};
		\node [style=eslabel] (55) at (2.225, 5.5) {$\{0,1\}$};
		\node [style=none] (56) at (1.25, -4.25) {};
		\node [style=eslabel] (57) at (0, -0.5) {$\{+,-\}$};
	\end{pgfonlayer}
	\begin{pgfonlayer}{edgelayer}
		\draw [style=carta da zucchero thin] (12.center)
			 to (13.center)
			 to (10.center)
			 to (11.center)
			 to cycle;
		\draw [style=carta da zucchero thin] (25.center)
			 to (26.center)
			 to (23.center)
			 to (24.center)
			 to cycle;
		\draw [style=classical, in=-90, out=90] (32.center) to (35.center);
		\draw [style=classical, in=90, out=-90] (31.center) to (37.center);
		\draw [style=classical, in=-90, out=90] (34.center) to (36.center);
		\draw [style=classical, in=90, out=-90] (33.center) to (38.center);
		\draw [style=carta da zucchero thin] (41.center)
			 to (42.center)
			 to (39.center)
			 to (40.center)
			 to cycle;
		\draw [style=bold black] (45.center) to (0.center);
		\draw (44.center) to (43);
		\draw (43) to (1.center);
		\draw [style=bold black, bend right=90, looseness=1.25] (45.center) to (44.center);
		\draw [style=classical, in=270, out=90] (48.center) to (49.center);
		\draw [style=bold black] (6.center) to (5.center);
		\draw [style=bold black] (19.center) to (18.center);
		\draw [style=bold black, in=270, out=90] (0.center) to (6.center);
		\draw [style=bold black, in=270, out=90] (1.center) to (19.center);
		\draw [style=classical] (56.center) to (46.center);
	\end{pgfonlayer}
\end{tikzpicture}
}
\end{center}
The description above results in the following empirical model on 3 events:
\begin{center}
\begin{tabular}{l|rrrrrrrr}
\hfill
ABC & 000 & 001 & 010 & 011 & 100 & 101 & 110 & 111\\
\hline
000 & $1/4$ & $1/4$ & $0$ & $0$ & $0$ & $0$ & $1/4$ & $1/4$\\
001 & $0$ & $0$ & $1/4$ & $1/4$ & $1/4$ & $1/4$ & $0$ & $0$\\
010 & $1/8$ & $1/8$ & $1/8$ & $1/8$ & $1/8$ & $1/8$ & $1/8$ & $1/8$\\
011 & $1/8$ & $1/8$ & $1/8$ & $1/8$ & $1/8$ & $1/8$ & $1/8$ & $1/8$\\
100 & $1/8$ & $1/8$ & $1/8$ & $1/8$ & $1/8$ & $1/8$ & $1/8$ & $1/8$\\
101 & $1/8$ & $1/8$ & $1/8$ & $1/8$ & $1/8$ & $1/8$ & $1/8$ & $1/8$\\
110 & $1/4$ & $0$ & $0$ & $1/4$ & $0$ & $1/4$ & $1/4$ & $0$\\
111 & $1/4$ & $0$ & $0$ & $1/4$ & $0$ & $1/4$ & $1/4$ & $0$\\
\end{tabular}
\end{center}
To better understand the process, we restrict our attention to the rows where Charlie has input 0, corresponding to Alice and Bob receiving the Bell basis states $|\Phi^\pm\rangle$:
\begin{center}
\begin{tabular}{l|rrrrrrrr}
\hfill
ABC & 000 & 001 & 010 & 011 & 100 & 101 & 110 & 111\\
\hline
000 & $1/4$ & $1/4$ & $0$ & $0$ & $0$ & $0$ & $1/4$ & $1/4$\\
010 & $1/8$ & $1/8$ & $1/8$ & $1/8$ & $1/8$ & $1/8$ & $1/8$ & $1/8$\\
100 & $1/8$ & $1/8$ & $1/8$ & $1/8$ & $1/8$ & $1/8$ & $1/8$ & $1/8$\\
110 & $1/4$ & $0$ & $0$ & $1/4$ & $0$ & $1/4$ & $1/4$ & $0$\\
\end{tabular}
\end{center}
When Charlie's output is 0 (left below) Alice and Bob receive the Bell basis state $|\Phi^+\rangle$: they get perfectly correlated outputs when they both measure in Z or both measure in X, and uncorrelated uniformly distributed outputs otherwise.
When Charlie's output is 1 (right below) Alice and Bob receive the Bell basis state $|\Phi^-\rangle$: they get perfectly correlated outputs when they both measure in Z, perfectly anti-correlated outputs when they both measure in X, and uncorrelated uniformly distributed outputs otherwise. 
\begin{center}
\begin{tabular}{l|rrrr}
\hfill
ABC & 000 & 010 & 100 & 110\\
\hline
000 & $1/4$ & $0$ & $0$ & $1/4$\\
010 & $1/8$ & $1/8$ & $1/8$ & $1/8$\\
100 & $1/8$ & $1/8$ & $1/8$ & $1/8$\\
110 & $1/4$ & $0$ & $0$ & $1/4$\\
\end{tabular}
\hspace{2cm}
\begin{tabular}{l|rrrr}
\hfill
ABC & 001 & 011 & 101 & 111\\
\hline
000 & $1/4$ & $0$ & $0$ & $1/4$\\
010 & $1/8$ & $1/8$ & $1/8$ & $1/8$\\
100 & $1/8$ & $1/8$ & $1/8$ & $1/8$\\
110 & $0$ & $1/4$ & $1/4$ & $0$\\
\end{tabular}
\end{center}
Rather interestingly, the empirical model for this experiment is 100\% supported by two incompatible spaces of input histories, both in equivalence class 33: the space $\Theta_{\ev{A}\vee (\ev{C}\seqcomposeSym \ev{B})}$ induced by causal order $\ev{A}\vee (\ev{C}\seqcomposeSym \ev{B})$ (left below) an the space $\Theta_{\ev{B}\vee (\ev{C}\seqcomposeSym \ev{A})}$ induced by causal order $\ev{B}\vee (\ev{C}\seqcomposeSym \ev{A})$. In other words, the empirical data is compatible both with absence of signalling from \ev{C} to \ev{A} (left below) and with absence of signalling from \ev{C} to \ev{B} (right below).
\begin{center}
    \begin{tabular}{ccc}
    \includegraphics[height=3cm]{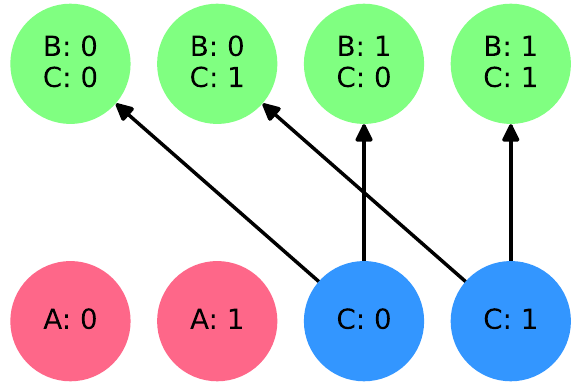}
    & \hspace{2cm} &
    \includegraphics[height=3cm]{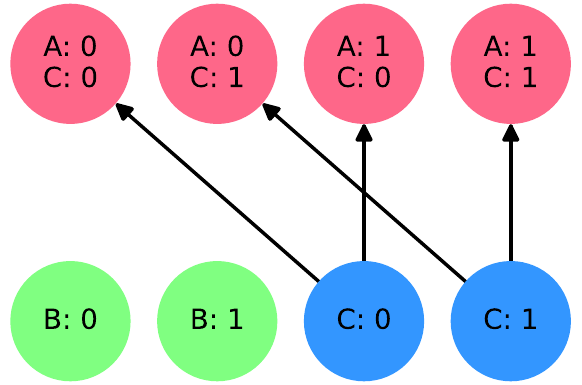}
    \\
    Causal fraction: $100\%$
    &&
    Causal fraction: $100\%$
    \end{tabular}
\end{center}
What makes this empirical model even more interesting is that its no-signalling fraction is 0\%: no part of it can be explained without signalling from \ev{C} to at least one of \ev{A} or \ev{B}.
\begin{center}
    \begin{tabular}{c}
    \includegraphics[height=3cm]{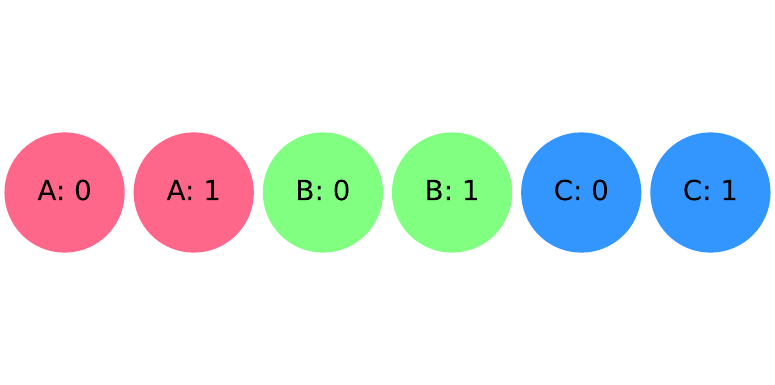}
    \\
    Causal fraction: $0\%$
    \end{tabular}
\end{center}
Since the discrete space $\Theta_{\ev{A}\vee \ev{B} \vee \ev{C}}$ is the meet of the two order-induced spaces $\Theta_{\ev{A}\vee (\ev{C}\seqcomposeSym \ev{B})}$ and $\Theta_{\ev{B}\vee (\ev{C}\seqcomposeSym \ev{A})}$, we now have an example of an empirical model which is fully supported by two spaces of input histories but not supported at all by their meet.
In particular, this shows that the intersection of two causaltopes is not necessarily the causaltope for the meet of the underlying spaces.
To better understand what's going on, we look at the systems of equations in RREF defining the two causaltopes for $\Theta_{\ev{A}\vee (\ev{C}\seqcomposeSym \ev{B})}$ and $\Theta_{\ev{B}\vee (\ev{C}\seqcomposeSym \ev{A})}$, both of dimension 32.
\begin{center}
    \begin{tabular}{ccc}
    \includegraphics[width=6.5cm]{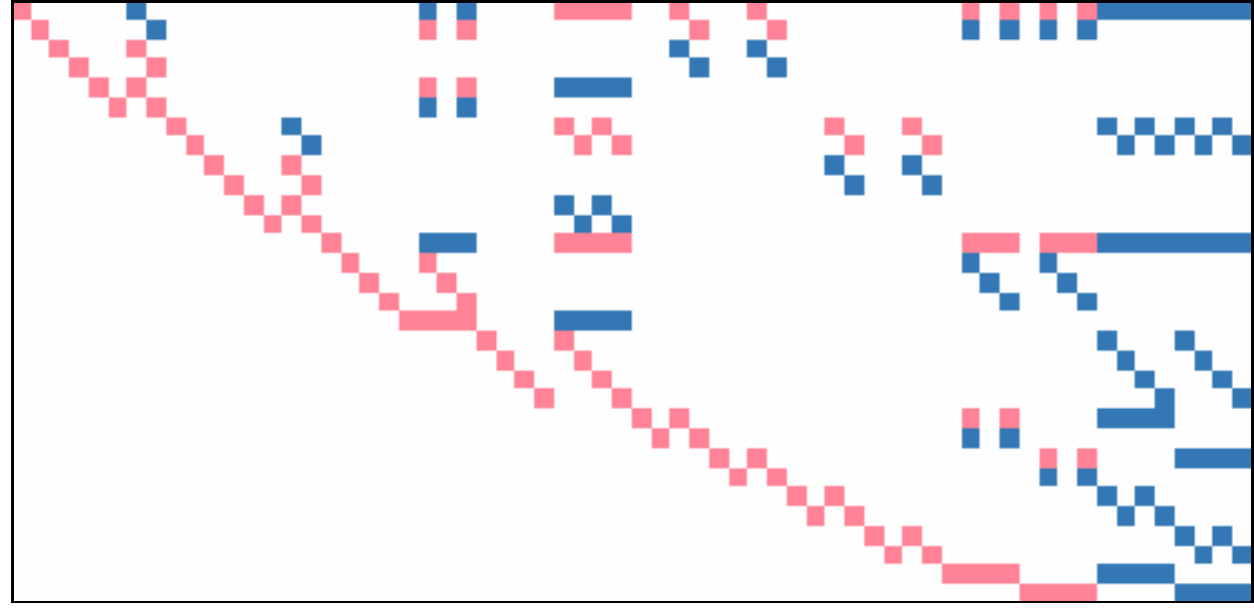}
    & \hspace{1.5cm} &
    \includegraphics[width=6.5cm]{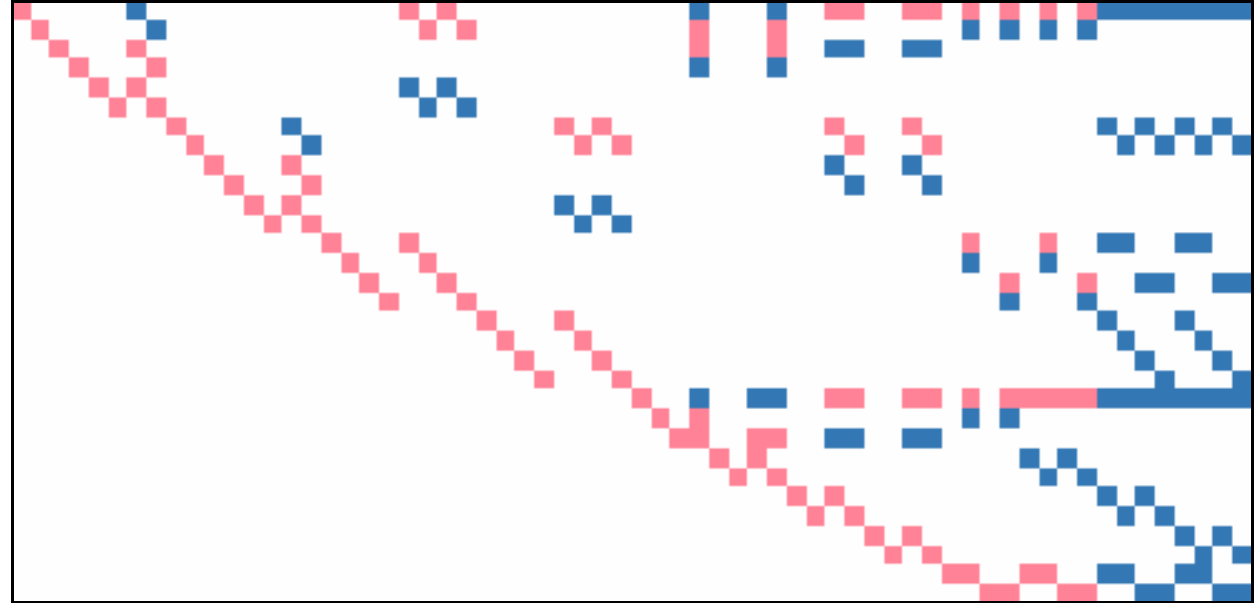}
    \\
    $\StdCausaltope{\Theta_{\ev{A}\vee (\ev{C}\seqcomposeSym \ev{B})}, [0,1]}$
    &&
    $\StdCausaltope{\Theta_{\ev{B}\vee (\ev{C}\seqcomposeSym \ev{A})}, [0,1]}$
    \\
    dimension 32
    &&
    dimension 32
    \end{tabular}
\end{center}
We then compare the system of equations in RREF defining the intersection of the two causaltopes (left below) to the system of equations in RREF defining the no-signalling causaltope (right below).
The intersection of the causaltopes has dimension 30, while the no-signalling causaltope has dimension 26, proving that they do not coincide: indeed, the empirical model presented in this subsection lies in the former, but not in the latter.
\begin{center}
    \begin{tabular}{ccc}
    \includegraphics[width=6.5cm]{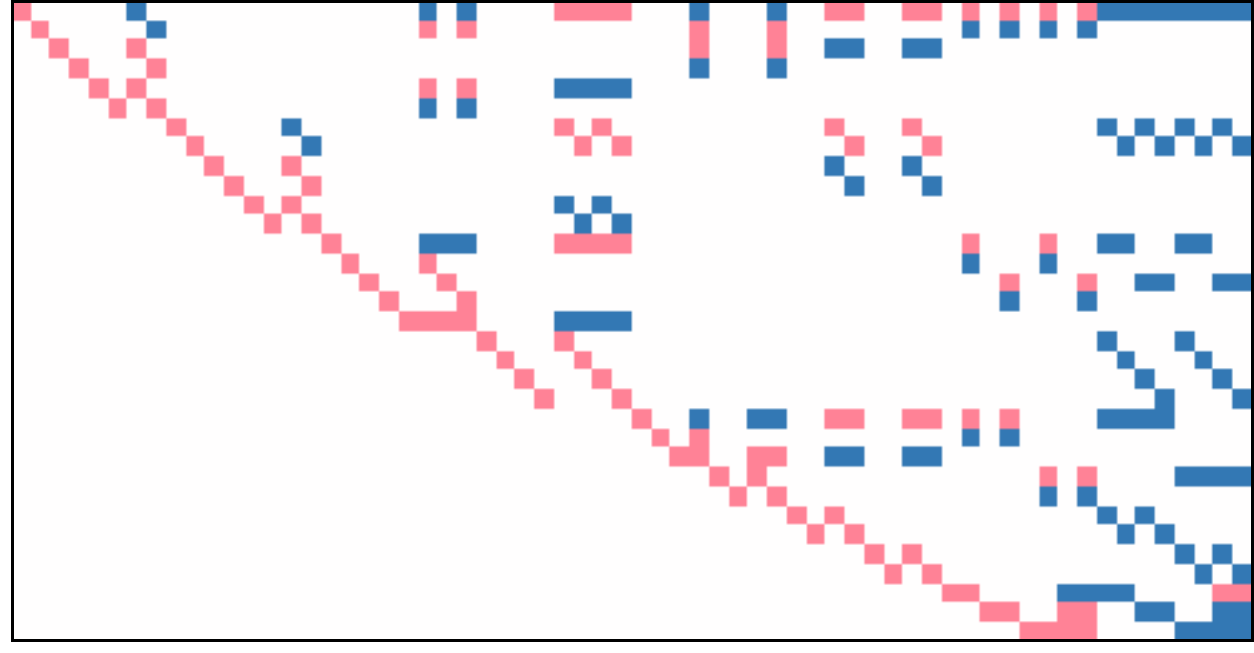}
    & \hspace{1.5cm} &
    \includegraphics[width=6.5cm]{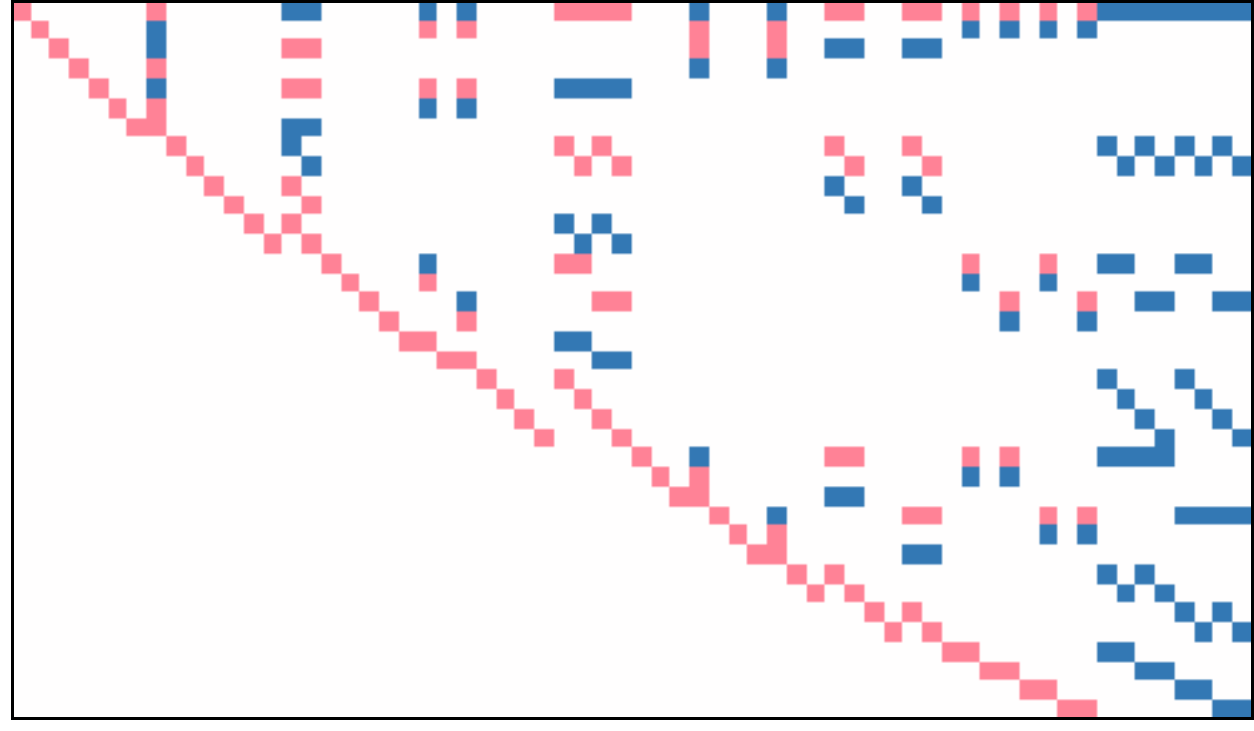}
    \\
    intersection
    &&
    $\StdCausaltope{\Theta_{\ev{A}\vee \ev{B} \vee \ev{C}}, [0,1]}$
    \\
    dimension 30
    &&
    dimension 26
    \end{tabular}
\end{center}
The intersection of causaltopes $\StdCausaltope{\Theta_{\ev{A}\vee (\ev{C}\seqcomposeSym \ev{B})}, [0,1]} \cap \StdCausaltope{\Theta_{\ev{B}\vee (\ev{C}\seqcomposeSym \ev{A})}, [0,1]}$ is not the causaltope for any causally complete space on 3 events, and in particular it isn't the causaltope for any subspace of $\Theta_{\ev{A}\vee (\ev{C}\seqcomposeSym \ev{B})}$ or $\Theta_{\ev{B}\vee (\ev{C}\seqcomposeSym \ev{A})}$.
However, the empirical model does happen to be 100\% supported by two unrelated non-tight subspaces of $\Theta_{\ev{A}\vee (\ev{C}\seqcomposeSym \ev{B})}$ and $\Theta_{\ev{B}\vee (\ev{C}\seqcomposeSym \ev{A})}$ respectively, both falling into equivalence class 2:
\begin{center}
    \begin{tabular}{ccc}
    \includegraphics[height=3cm]{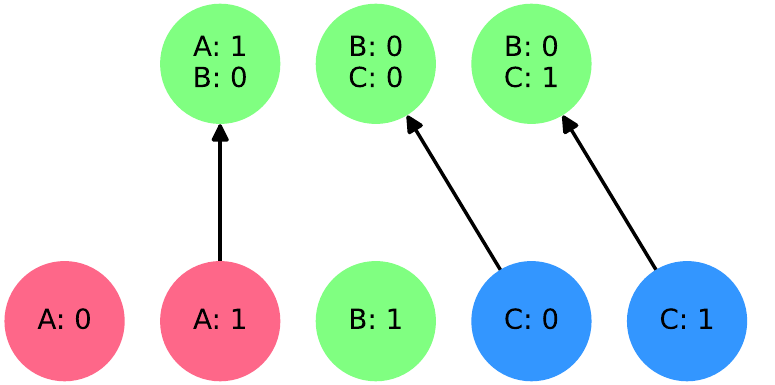}
    & \hspace{2cm} &
    \includegraphics[height=3cm]{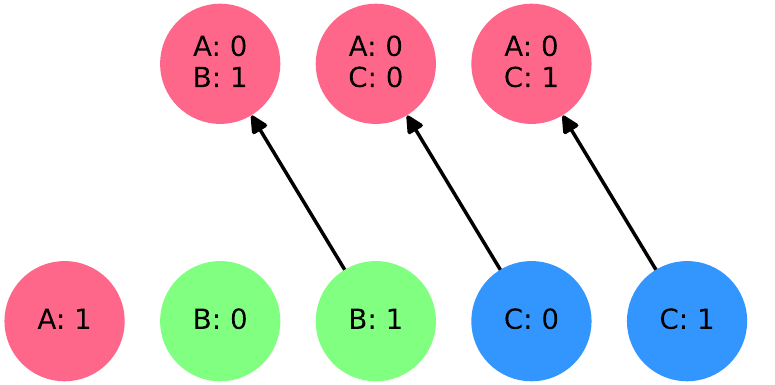}
    \\
    Causal fraction: $100\%$
    &&
    Causal fraction: $100\%$
    \end{tabular}
\end{center}
Unlike $\Theta_{\ev{A}\vee (\ev{C}\seqcomposeSym \ev{B})}$ and $\Theta_{\ev{B}\vee (\ev{C}\seqcomposeSym \ev{A})}$, these two spaces have exactly the same standard causaltope, of dimension 27.
Because it is only 1 dimension larger than the no-signalling causaltope, this is the minimal supporting causaltope for our empirical model.
\begin{center}
    \begin{tabular}{ccc}
    \includegraphics[width=6.5cm]{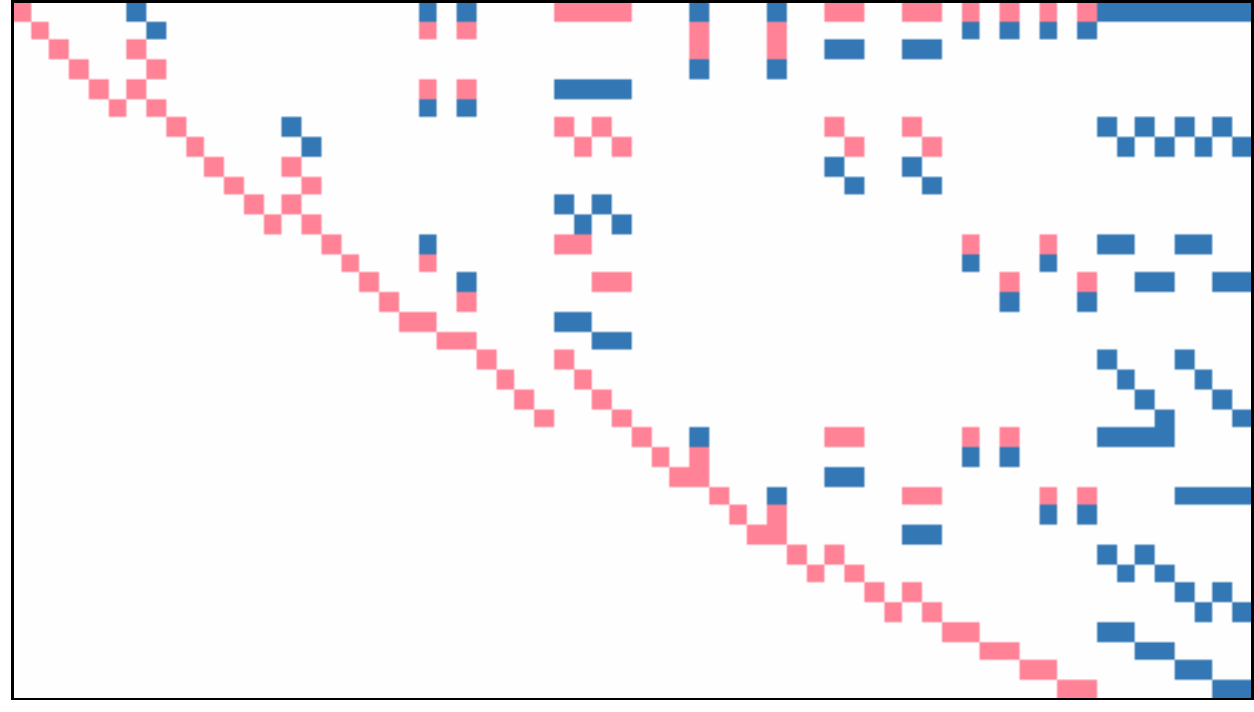}
    & \hspace{1.5cm} &
    \includegraphics[width=6.5cm]{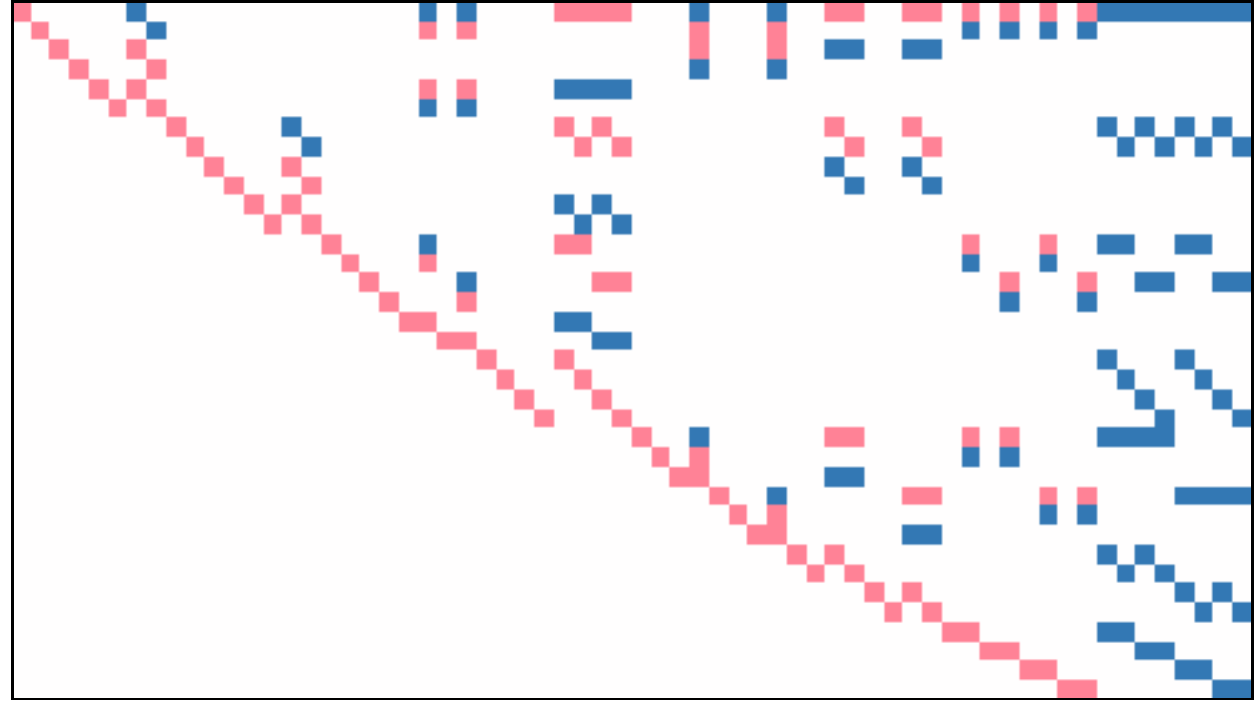}
    \\
    non-tight subspace of $\Theta_{\ev{A}\vee (\ev{C}\seqcomposeSym \ev{B})}$
    &&
    non-tight subspace of $\Theta_{\ev{B}\vee (\ev{C}\seqcomposeSym \ev{A})}$
    \\
    dimension 27
    &&
    dimension 27
    \end{tabular}
\end{center}
The empirical model is local for both space $\Theta_{\ev{A}\vee (\ev{C}\seqcomposeSym \ev{B})}$ and space $\Theta_{\ev{B}\vee (\ev{C}\seqcomposeSym \ev{A})}$: for example, below is a decomposition as a uniform mixture of 8 causal functions for $\Theta_{\ev{B}\vee (\ev{C}\seqcomposeSym \ev{A})}$.
The black dots are classical copies, the $\oplus$ dots are classical XORs and the $\wedge$ dots classical ANDs. 
\[
    \frac{1}{8} \sum_{(x,y,z) \in \{0,1\}^3}
    \scalebox{0.65}{
        \begin{tikzpicture}[scale=2.7]
	\begin{pgfonlayer}{nodelayer}
		\node [style=black dot] (0) at (-2, -3.75) {};
		\node [style=none] (3) at (-1.25, -3) {};
		\node [style=none] (4) at (-1.25, -1) {};
		\node [style=none] (5) at (0.75, -1) {};
		\node [style=none] (6) at (0.75, -3) {};
		\node [style=none] (7) at (1.5, 2.5) {};
		\node [style=none] (8) at (1.5, 5) {};
		\node [style=none] (9) at (4, 5) {};
		\node [style=none] (10) at (4, 2.5) {};
		\node [style=none] (11) at (-5, 2.5) {};
		\node [style=none] (12) at (-5, 5.75) {};
		\node [style=none] (13) at (-1.5, 5.75) {};
		\node [style=none] (14) at (-1.5, 2.5) {};
		\node [style=none] (15) at (-0.5, -3) {};
		\node [style=none] (16) at (-4.25, 2.5) {};
		\node [style=point] (17) at (-2, -4.75) {$z$};
		\node [style=none] (18) at (0, -1) {};
		\node [style=none] (19) at (0, -3) {};
		\node [style=none] (20) at (-0.5, -1) {};
		\node [style=none] (21) at (-2.5, 2.5) {};
		\node [style=none] (22) at (-3.5, 2.5) {};
		\node [style=black dot] (23) at (-2.25, 0.5) {};
		\node [style=none] (24) at (1.75, 2.5) {};
		\node [style=none] (25) at (2, 2.5) {};
		\node [style=point] (26) at (-2.25, -0.5) {$x$};
		\node [style=black dot] (27) at (0.5, 0.5) {};
		\node [style=point] (28) at (0.5, -0.5) {$y$};
		\node [style=none] (29) at (-2, 2.5) {};
		\node [style=none] (30) at (0, -5) {};
		\node [style=none] (32) at (0, 7.25) {};
		\node [style=white dot] (33) at (2.75, 3.25) {$\wedge$};
		\node [style=none] (34) at (2.75, 4.25) {$\oplus$};
		\node [style=none] (35) at (3.5, 2.5) {};
		\node [style=none] (36) at (3.5, -5) {};
		\node [style=none] (37) at (2.75, 5) {};
		\node [style=none] (38) at (2.75, 7.25) {};
		\node [style=none] (39) at (-3.25, 7.25) {};
		\node [style=none] (40) at (-4.5, 2.5) {};
		\node [style=none] (41) at (-4.5, -5) {};
		\node [style=white dot] (42) at (-4.5, 4.5) {$\wedge$};
		\node [style=black dot] (45) at (-2.5, 3) {};
		\node [style=none] (46) at (-3.25, 5.75) {};
		\node [style=none] (47) at (-3.25, 5.25) {$\oplus$};
		\node [style=none] (48) at (-3.5, 3.5) {$\oplus$};
		\node [style=label] (49) at (-2, -1.75) {$C$};
		\node [style=label] (50) at (-5.75, 4) {$A$};
		\node [style=label] (51) at (0.75, 3.5) {$B$};
	\end{pgfonlayer}
	\begin{pgfonlayer}{edgelayer}
		\draw [style=carta da zucchero thin] (5.center)
			 to (4.center)
			 to (3.center)
			 to (6.center)
			 to cycle;
		\draw [style=carta da zucchero thin] (9.center)
			 to (8.center)
			 to (7.center)
			 to (10.center)
			 to cycle;
		\draw [style=carta da zucchero thin] (13.center)
			 to (12.center)
			 to (11.center)
			 to (14.center)
			 to cycle;
		\draw [style=bold black, in=-90, out=0, looseness=1.25] (0) to (15.center);
		\draw [style=bold black, in=-90, out=180, looseness=0.75] (0) to (16.center);
		\draw [style=bold black] (17) to (0);
		\draw [style=bold black, in=-90, out=90] (15.center) to (18.center);
		\draw [style=bold black, in=-90, out=90] (19.center) to (20.center);
		\draw [style=bold black, in=-90, out=90, looseness=0.50] (20.center) to (21.center);
		\draw [style=bold black, in=-90, out=180] (23) to (22.center);
		\draw [style=bold black, in=-90, out=0, looseness=0.75] (23) to (25.center);
		\draw [style=bold black] (23) to (26);
		\draw [style=bold black] (28) to (27);
		\draw [style=bold black, in=180, out=-90] (29.center) to (27);
		\draw [style=bold black, in=-90, out=0] (27) to (24.center);
		\draw [style=bold black] (32.center) to (18.center);
		\draw [style=bold black] (30.center) to (19.center);
		\draw [style=bold black] (36.center) to (35.center);
		\draw [style=bold black, bend left=45] (25.center) to (33);
		\draw [style=bold black, bend left=45] (33) to (35.center);
		\draw [style=bold black] (33) to (34.center);
		\draw [style=bold black, in=180, out=90, looseness=1.50] (24.center) to (34.center);
		\draw [style=bold black] (34.center) to (37.center);
		\draw [style=bold black] (37.center) to (38.center);
		\draw [style=bold black] (41.center) to (40.center);
		\draw [style=bold black] (40.center) to (42);
		\draw [style=bold black] (21.center) to (45);
		\draw [style=bold black] (46.center) to (39.center);
		\draw [style=bold black, in=-90, out=15] (45) to (47.center);
		\draw [style=bold black, in=180, out=90] (16.center) to (48.center);
		\draw [style=bold black] (22.center) to (48.center);
		\draw [style=bold black, in=90, out=0] (42) to (48.center);
		\draw [style=bold black, in=0, out=-180, looseness=1.25] (45) to (48.center);
		\draw [style=bold black, in=-180, out=90] (42) to (47.center);
		\draw [style=bold black, in=90, out=0, looseness=1.25] (47.center) to (29.center);
		\draw [style=bold black] (46.center) to (47.center);
	\end{pgfonlayer}
\end{tikzpicture}
    }
\]
However, we know from Figure~5 (p.42) of \cite{gogioso2022combinatorics} that spaces in equivalence class 2 have exactly the same causal functions as the discrete space, in equivalence class 0.
Since the empirical model has a no-signalling fraction of 0\%, it immediately follows that it has a local fraction of 0\% in its minimal supporting causaltope, i.e. that it is maximally non-local there. To recap, this example bears many gifts:
\begin{itemize}
    \item It shows that there are empirical models 100\% supported by multiple spaces but 0\% supported by their meet; in particular, it shows that the intersection of causaltopes is not necessarily the causaltope for the meet of the underlying spaces.
    \item Further to the previous point, it shows that there are causaltopes whose intersection is not the causaltope for any space.
    \item It shows that there can be unrelated spaces with equal causaltopes, differing from the no-signalling causaltope.
    \item It provides an empirical model whose minimally supporting space is non-tight, providing additional evidence for the importance of non-tight spaces in the study of causality.
    \item It shows that the notions of non-locality and contextuality depend on a specific choice of causal constraints, by providing an empirical model which is local in for a space and maximally non-local for a sub-space.
\end{itemize}

\subsubsection{A Causal Cross Empirical Model.}

In this example, Charlie receives qubits from Alice and Bob and forwards them to Diane and Eve, choosing whether to forward the qubits as $A\rightarrow D, B\rightarrow E$ or as $A\rightarrow E, B\rightarrow D$.
Below is the ``cross'' causal order that naturally supports this example:
\begin{center}
\includegraphics[height=2.5cm]{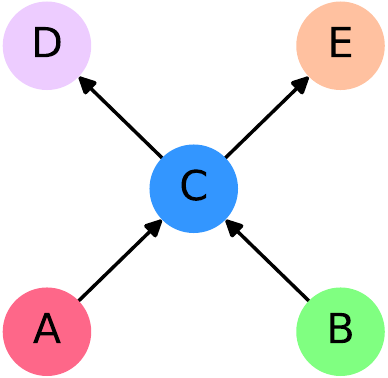}
\end{center}
More specifically, the parties act as follows:
\begin{enumerate}
    \item Alice and Bob encode their input into the Z basis of one qubit each, which they then forward to Charlie. Their output is trivial, constantly set to 0.
    \item Charlie receives the two qubits from Alice and Bob, and decides how to forward them:
    \begin{itemize}
         \item On input 0, Charlie forwards Alice's qubit to Diane and Bob's qubit to Eve.
         \item On input 1, Charlie forwards Alice's qubit to Eve and Bob's qubit to Diane.
     \end{itemize} 
    \item Diane and Eve have trivial input, with only 0 as an option. They measure the qubit they receive in the Z basis and use the outcome as their output.
\end{enumerate}
We will consider two version of this protocol: one where Charlie measures the parity of the qubits he receives, and one where he doesn't perform any measurement and trivially outputs 0.
The version where Charlie measures the parity corresponds to the following empirical model $e$; note that the outputs of Alice and Bob, as well as the inputs of Diane and Eve, are fixed to 0. 
\begin{center}
\scalebox{0.8}{
\begin{tabular}{l|rrrrrrrr}
\hfill
ABCDE & 00000 & 00001 & 00010 & 00011 & 00100 & 00101 & 00110 & 00111\\
\hline
00000 & $1$ & $0$ & $0$ & $0$ & $0$ & $0$ & $0$ & $0$\\
00100 & $1$ & $0$ & $0$ & $0$ & $0$ & $0$ & $0$ & $0$\\
01000 & $0$ & $0$ & $0$ & $0$ & $0$ & $1$ & $0$ & $0$\\
01100 & $0$ & $0$ & $0$ & $0$ & $0$ & $0$ & $1$ & $0$\\
10000 & $0$ & $0$ & $0$ & $0$ & $0$ & $0$ & $1$ & $0$\\
10100 & $0$ & $0$ & $0$ & $0$ & $0$ & $1$ & $0$ & $0$\\
11000 & $0$ & $0$ & $0$ & $1$ & $0$ & $0$ & $0$ & $0$\\
11100 & $0$ & $0$ & $0$ & $1$ & $0$ & $0$ & $0$ & $0$\\
\end{tabular}
}
\end{center}
The figures below exemplify this full scenario:
\begin{center}
    \scalebox{1}{
        \begin{tikzpicture}
	\begin{pgfonlayer}{nodelayer}
		\node [style=none] (0) at (-3.5, 5.5) {};
		\node [style=white dot] (3) at (-2.5, 3.75) {};
		\node [style=none] (7) at (-3.5, 6.5) {};
		\node [style=none] (8) at (3.475, 5.5) {};
		\node [style=white dot] (11) at (2.475, 3.75) {};
		\node [style=none] (15) at (3.475, 6.5) {};
		\node [style=none] (20) at (-3.5, -5.5) {};
		\node [style=white dot] (23) at (-2.5, -3.75) {};
		\node [style=none] (27) at (-3.5, -6.5) {};
		\node [style=none] (28) at (3.475, -5.5) {};
		\node [style=white dot] (31) at (2.475, -3.75) {};
		\node [style=none] (35) at (3.475, -6.5) {};
		\node [style=none] (39) at (0, -6.5) {};
		\node [style=none] (41) at (-0.25, 3) {};
		\node [style=none] (46) at (1.975, 6) {};
		\node [style=none] (47) at (3.975, 6) {};
		\node [style=none] (48) at (3.975, 3.25) {};
		\node [style=none] (49) at (1.975, 3.25) {};
		\node [style=none] (50) at (1.975, -3.25) {};
		\node [style=none] (51) at (3.975, -3.25) {};
		\node [style=none] (52) at (3.975, -6) {};
		\node [style=none] (53) at (1.975, -6) {};
		\node [style=none] (54) at (-4.025, 6) {};
		\node [style=none] (55) at (-2.025, 6) {};
		\node [style=none] (56) at (-2.025, 3.25) {};
		\node [style=none] (57) at (-4.025, 3.25) {};
		\node [style=none] (58) at (-4.025, -3.25) {};
		\node [style=none] (59) at (-2.025, -3.25) {};
		\node [style=none] (60) at (-2.025, -6) {};
		\node [style=none] (61) at (-4.025, -6) {};
		\node [style=red label] (62) at (4.75, 4.75) {$E$};
		\node [style=red label] (64) at (-4.75, -4.5) {$A$};
		\node [style=red label] (65) at (4.75, -4.5) {$B$};
		\node [style=red label] (66) at (-4.75, 4.75) {$D$};
		\node [style=none] (67) at (-2.5, 2.5) {};
		\node [style=none] (68) at (2.475, 2.5) {};
		\node [style=none] (69) at (-1, -1) {};
		\node [style=none] (70) at (-0.5, 1) {};
		\node [style=none] (71) at (-1, 1) {};
		\node [style=none] (72) at (-0.5, -1) {};
		\node [style=none] (73) at (-2.5, -2.5) {};
		\node [style=none] (74) at (2.475, -2.5) {};
		\node [style=none] (75) at (-0.5, -1) {};
		\node [style=none] (76) at (-1, -1) {};
		\node [style=copoint] (77) at (0.75, -1) {$0$};
		\node [style=none] (78) at (-1.75, 1.75) {};
		\node [style=none] (79) at (1.75, 1.75) {};
		\node [style=none] (80) at (1.75, -1.75) {};
		\node [style=none] (81) at (-1.75, -1.75) {};
		\node [style=red label] (82) at (2.5, 0) {$C$};
		\node [style=none] (83) at (-1, 0) {};
		\node [style=none] (84) at (-0.5, 0) {};
		\node [style=box] (85) at (-0.75, -0.5) {$P$};
		\node [style=none] (86) at (-0.25, 0) {};
	\end{pgfonlayer}
	\begin{pgfonlayer}{edgelayer}
		\draw [style=classical, in=-90, out=90] (3) to (0.center);
		\draw [style=classical] (0.center) to (7.center);
		\draw [style=classical, in=-90, out=90] (11) to (8.center);
		\draw [style=classical] (8.center) to (15.center);
		\draw [style=classical, in=90, out=-90] (23) to (20.center);
		\draw [style=classical] (20.center) to (27.center);
		\draw [style=classical, in=90, out=-90] (31) to (28.center);
		\draw [style=classical] (28.center) to (35.center);
		\draw [style=carta da zucchero thin] (46.center)
			 to (47.center)
			 to (48.center)
			 to (49.center)
			 to cycle;
		\draw [style=carta da zucchero thin] (53.center)
			 to (50.center)
			 to (51.center)
			 to (52.center)
			 to cycle;
		\draw [style=carta da zucchero thin] (54.center)
			 to (55.center)
			 to (56.center)
			 to (57.center)
			 to cycle;
		\draw [style=carta da zucchero thin] (61.center)
			 to (58.center)
			 to (59.center)
			 to (60.center)
			 to cycle;
		\draw [style=carta da zucchero thin] (78.center)
			 to (79.center)
			 to (80.center)
			 to (81.center)
			 to cycle;
		\draw [style=bold black, in=90, out=-90] (67.center) to (71.center);
		\draw [style=bold black, in=-90, out=90] (70.center) to (68.center);
		\draw [style=bold black, in=-90, out=90] (73.center) to (76.center);
		\draw [style=bold black, in=90, out=-90] (75.center) to (74.center);
		\draw [style=bold black] (76.center) to (83.center);
		\draw [style=bold black] (75.center) to (84.center);
		\draw [style=bold black] (68.center) to (11);
		\draw [style=bold black] (67.center) to (3);
		\draw [style=bold black] (23) to (73.center);
		\draw [style=bold black] (31) to (74.center);
		\draw [style=bold black] (83.center) to (71.center);
		\draw [style=bold black] (84.center) to (70.center);
		\draw [style=classical, in=-90, out=90] (39.center) to (77);
		\draw [style=classical] (86.center) to (41.center);
	\end{pgfonlayer}
\end{tikzpicture}
    }
    \hspace{0.5cm}
    \scalebox{1}{
        \begin{tikzpicture}
	\begin{pgfonlayer}{nodelayer}
		\node [style=none] (0) at (-3.5, 5.5) {};
		\node [style=white dot] (3) at (-2.5, 3.75) {};
		\node [style=none] (4) at (-2.5, 2.5) {};
		\node [style=none] (7) at (-3.5, 6.5) {};
		\node [style=none] (8) at (3.475, 5.5) {};
		\node [style=none] (10) at (3.475, 3.75) {};
		\node [style=white dot] (11) at (2.475, 3.75) {};
		\node [style=none] (12) at (2.475, 2.5) {};
		\node [style=none] (15) at (3.475, 6.5) {};
		\node [style=none] (16) at (-1, -1) {};
		\node [style=none] (17) at (-0.5, 1) {};
		\node [style=none] (18) at (-1, 1) {};
		\node [style=none] (19) at (-0.5, -1) {};
		\node [style=none] (20) at (-3.5, -5.5) {};
		\node [style=white dot] (23) at (-2.5, -3.75) {};
		\node [style=none] (24) at (-2.5, -2.5) {};
		\node [style=none] (27) at (-3.5, -6.5) {};
		\node [style=none] (28) at (3.475, -5.5) {};
		\node [style=white dot] (31) at (2.475, -3.75) {};
		\node [style=none] (32) at (2.475, -2.5) {};
		\node [style=none] (35) at (3.475, -6.5) {};
		\node [style=none] (36) at (-0.5, -1) {};
		\node [style=none] (37) at (-1, -1) {};
		\node [style=copoint] (38) at (0.75, -1) {$1$};
		\node [style=none] (39) at (0, -6.5) {};
		\node [style=none] (41) at (-0.25, 3) {};
		\node [style=none] (42) at (-1.75, 1.75) {};
		\node [style=none] (43) at (1.75, 1.75) {};
		\node [style=none] (44) at (1.75, -1.75) {};
		\node [style=none] (45) at (-1.75, -1.75) {};
		\node [style=none] (46) at (1.975, 6) {};
		\node [style=none] (47) at (3.975, 6) {};
		\node [style=none] (48) at (3.975, 3.25) {};
		\node [style=none] (49) at (1.975, 3.25) {};
		\node [style=none] (50) at (1.975, -3.25) {};
		\node [style=none] (51) at (3.975, -3.25) {};
		\node [style=none] (52) at (3.975, -6) {};
		\node [style=none] (53) at (1.975, -6) {};
		\node [style=none] (54) at (-4.025, 6) {};
		\node [style=none] (55) at (-2.025, 6) {};
		\node [style=none] (56) at (-2.025, 3.25) {};
		\node [style=none] (57) at (-4.025, 3.25) {};
		\node [style=none] (58) at (-4.025, -3.25) {};
		\node [style=none] (59) at (-2.025, -3.25) {};
		\node [style=none] (60) at (-2.025, -6) {};
		\node [style=none] (61) at (-4.025, -6) {};
		\node [style=red label] (62) at (4.75, 4.75) {$E$};
		\node [style=red label] (63) at (2.5, 0) {$C$};
		\node [style=red label] (64) at (-4.75, -4.5) {$A$};
		\node [style=red label] (65) at (4.75, -4.5) {$B$};
		\node [style=red label] (66) at (-4.75, 4.75) {$D$};
		\node [style=none] (67) at (-1, 0) {};
		\node [style=none] (68) at (-0.5, 0) {};
		\node [style=box] (69) at (-0.75, -0.5) {$P$};
		\node [style=none] (71) at (-0.25, 0) {};
	\end{pgfonlayer}
	\begin{pgfonlayer}{edgelayer}
		\draw [style=classical, in=-90, out=90] (3) to (0.center);
		\draw [style=bold black] (4.center) to (3);
		\draw [style=classical] (0.center) to (7.center);
		\draw [style=classical, in=-90, out=90] (11) to (8.center);
		\draw [style=bold black] (12.center) to (11);
		\draw [style=classical] (8.center) to (15.center);
		\draw [style=classical, in=90, out=-90] (23) to (20.center);
		\draw [style=bold black] (24.center) to (23);
		\draw [style=classical] (20.center) to (27.center);
		\draw [style=classical, in=90, out=-90] (31) to (28.center);
		\draw [style=bold black] (32.center) to (31);
		\draw [style=classical] (28.center) to (35.center);
		\draw [style=classical, in=90, out=-90] (38) to (39.center);
		\draw [style=carta da zucchero thin] (42.center)
			 to (43.center)
			 to (44.center)
			 to (45.center)
			 to cycle;
		\draw [style=carta da zucchero thin] (46.center)
			 to (47.center)
			 to (48.center)
			 to (49.center)
			 to cycle;
		\draw [style=carta da zucchero thin] (53.center)
			 to (50.center)
			 to (51.center)
			 to (52.center)
			 to cycle;
		\draw [style=carta da zucchero thin] (57.center)
			 to (54.center)
			 to (55.center)
			 to (56.center)
			 to cycle;
		\draw [style=carta da zucchero thin] (61.center)
			 to (58.center)
			 to (59.center)
			 to (60.center)
			 to cycle;
		\draw [style=bold black, in=90, out=-90] (4.center) to (18.center);
		\draw [style=bold black, in=-90, out=90] (17.center) to (12.center);
		\draw [style=bold black, in=-90, out=90] (24.center) to (37.center);
		\draw [style=bold black, in=90, out=-90] (36.center) to (32.center);
		\draw [style=bold black, in=-90, out=90] (67.center) to (17.center);
		\draw [style=bold black, in=-90, out=90] (68.center) to (18.center);
		\draw [style=bold black] (37.center) to (67.center);
		\draw [style=bold black] (36.center) to (68.center);
		\draw [style=classical] (71.center) to (41.center);
	\end{pgfonlayer}
\end{tikzpicture}
    }
\end{center}
The version where Charlie doesn't measures the parity corresponds to the following simplified empirical model $e'$; note that the outputs of Alice, Bob and Charlie, as well as the inputs of Diane and Eve, are fixed to 0. 
\begin{center}
\scalebox{0.8}{
\begin{tabular}{l|rrrr}
\hfill
ABCDE & 00000 & 00001 & 00010 & 00011\\
\hline
00000 & $1$ & $0$ & $0$ & $0$\\
00100 & $1$ & $0$ & $0$ & $0$\\
01000 & $0$ & $1$ & $0$ & $0$\\
01100 & $0$ & $0$ & $1$ & $0$\\
10000 & $0$ & $0$ & $1$ & $0$\\
10100 & $0$ & $1$ & $0$ & $0$\\
11000 & $0$ & $0$ & $0$ & $1$\\
11100 & $0$ & $0$ & $0$ & $1$\\
\end{tabular}
}
\end{center}
The figures below exemplify this latter, simplified scenario:
\begin{center}
    \scalebox{1}{
        \begin{tikzpicture}
	\begin{pgfonlayer}{nodelayer}
		\node [style=none] (0) at (-3.5, 5.5) {};
		\node [style=white dot] (3) at (-2.5, 3.75) {};
		\node [style=none] (4) at (-2.5, 2.5) {};
		\node [style=none] (7) at (-3.5, 6.5) {};
		\node [style=none] (8) at (3.475, 5.5) {};
		\node [style=white dot] (11) at (2.475, 3.75) {};
		\node [style=none] (12) at (2.475, 2.5) {};
		\node [style=none] (15) at (3.475, 6.5) {};
		\node [style=none] (16) at (-1, -1) {};
		\node [style=none] (17) at (1, 1) {};
		\node [style=none] (18) at (-1, 1) {};
		\node [style=none] (19) at (1, -1) {};
		\node [style=none] (20) at (-3.5, -5.5) {};
		\node [style=white dot] (23) at (-2.5, -3.75) {};
		\node [style=none] (24) at (-2.5, -2.5) {};
		\node [style=none] (27) at (-3.5, -6.5) {};
		\node [style=none] (28) at (3.475, -5.5) {};
		\node [style=white dot] (31) at (2.475, -3.75) {};
		\node [style=none] (32) at (2.475, -2.5) {};
		\node [style=none] (35) at (3.475, -6.5) {};
		\node [style=none] (36) at (1, -1) {};
		\node [style=none] (37) at (-1, -1) {};
		\node [style=copoint] (38) at (0, -1.25) {$0$};
		\node [style=none] (39) at (0, -6.5) {};
		\node [style=none] (42) at (-1.75, 1.75) {};
		\node [style=none] (43) at (1.75, 1.75) {};
		\node [style=none] (44) at (1.75, -1.75) {};
		\node [style=none] (45) at (-1.75, -1.75) {};
		\node [style=none] (46) at (1.975, 6) {};
		\node [style=none] (47) at (3.975, 6) {};
		\node [style=none] (48) at (3.975, 3.25) {};
		\node [style=none] (49) at (1.975, 3.25) {};
		\node [style=none] (50) at (1.975, -3.25) {};
		\node [style=none] (51) at (3.975, -3.25) {};
		\node [style=none] (52) at (3.975, -6) {};
		\node [style=none] (53) at (1.975, -6) {};
		\node [style=none] (54) at (-4.025, 6) {};
		\node [style=none] (55) at (-2.025, 6) {};
		\node [style=none] (56) at (-2.025, 3.25) {};
		\node [style=none] (57) at (-4.025, 3.25) {};
		\node [style=none] (58) at (-4.025, -3.25) {};
		\node [style=none] (59) at (-2.025, -3.25) {};
		\node [style=none] (60) at (-2.025, -6) {};
		\node [style=none] (61) at (-4.025, -6) {};
		\node [style=red label] (62) at (4.75, 4.75) {$E$};
		\node [style=red label] (63) at (2.5, 0) {$C$};
		\node [style=red label] (64) at (-4.75, -4.5) {$A$};
		\node [style=red label] (65) at (4.75, -4.5) {$B$};
		\node [style=red label] (66) at (-4.75, 4.75) {$D$};
	\end{pgfonlayer}
	\begin{pgfonlayer}{edgelayer}
		\draw [style=classical, in=-90, out=90] (3) to (0.center);
		\draw [style=classical] (0.center) to (7.center);
		\draw [style=classical, in=-90, out=90] (11) to (8.center);
		\draw [style=classical] (8.center) to (15.center);
		\draw [style=classical, in=90, out=-90] (23) to (20.center);
		\draw [style=classical] (20.center) to (27.center);
		\draw [style=classical, in=90, out=-90] (31) to (28.center);
		\draw [style=classical] (28.center) to (35.center);
		\draw [style=classical] (38) to (39.center);
		\draw [style=carta da zucchero thin] (42.center)
			 to (43.center)
			 to (44.center)
			 to (45.center)
			 to cycle;
		\draw [style=carta da zucchero thin] (46.center)
			 to (47.center)
			 to (48.center)
			 to (49.center)
			 to cycle;
		\draw [style=carta da zucchero thin] (53.center)
			 to (50.center)
			 to (51.center)
			 to (52.center)
			 to cycle;
		\draw [style=carta da zucchero thin] (54.center)
			 to (55.center)
			 to (56.center)
			 to (57.center)
			 to cycle;
		\draw [style=carta da zucchero thin] (61.center)
			 to (58.center)
			 to (59.center)
			 to (60.center)
			 to cycle;
		\draw [style=bold black] (12.center) to (11);
		\draw [style=bold black, in=-90, out=90] (17.center) to (12.center);
		\draw [style=bold black] (32.center) to (31);
		\draw [style=bold black, in=90, out=-90] (36.center) to (32.center);
		\draw [style=bold black] (36.center) to (17.center);
		\draw [style=bold black] (4.center) to (3);
		\draw [style=bold black, in=90, out=-90] (4.center) to (18.center);
		\draw [style=bold black] (24.center) to (23);
		\draw [style=bold black, in=-90, out=90] (24.center) to (37.center);
		\draw [style=bold black] (37.center) to (18.center);
	\end{pgfonlayer}
\end{tikzpicture}
    }
    \hspace{0.5cm}
    \scalebox{1}{
        \begin{tikzpicture}
	\begin{pgfonlayer}{nodelayer}
		\node [style=none] (0) at (-3.5, 5.5) {};
		\node [style=white dot] (3) at (-2.5, 3.75) {};
		\node [style=none] (4) at (-2.5, 2.5) {};
		\node [style=none] (7) at (-3.5, 6.5) {};
		\node [style=none] (8) at (3.475, 5.5) {};
		\node [style=none] (10) at (3.475, 3.75) {};
		\node [style=white dot] (11) at (2.475, 3.75) {};
		\node [style=none] (12) at (2.475, 2.5) {};
		\node [style=none] (15) at (3.475, 6.5) {};
		\node [style=none] (16) at (-1, -1) {};
		\node [style=none] (17) at (1, 1) {};
		\node [style=none] (18) at (-1, 1) {};
		\node [style=none] (19) at (1, -1) {};
		\node [style=none] (20) at (-3.5, -5.5) {};
		\node [style=white dot] (23) at (-2.5, -3.75) {};
		\node [style=none] (24) at (-2.5, -2.5) {};
		\node [style=none] (27) at (-3.5, -6.5) {};
		\node [style=none] (28) at (3.475, -5.5) {};
		\node [style=white dot] (31) at (2.475, -3.75) {};
		\node [style=none] (32) at (2.475, -2.5) {};
		\node [style=none] (35) at (3.475, -6.5) {};
		\node [style=none] (36) at (1, -1) {};
		\node [style=none] (37) at (-1, -1) {};
		\node [style=copoint] (38) at (0, -1.25) {$1$};
		\node [style=none] (39) at (0, -6.5) {};
		\node [style=none] (42) at (-1.75, 1.75) {};
		\node [style=none] (43) at (1.75, 1.75) {};
		\node [style=none] (44) at (1.75, -1.75) {};
		\node [style=none] (45) at (-1.75, -1.75) {};
		\node [style=none] (46) at (1.975, 6) {};
		\node [style=none] (47) at (3.975, 6) {};
		\node [style=none] (48) at (3.975, 3.25) {};
		\node [style=none] (49) at (1.975, 3.25) {};
		\node [style=none] (50) at (1.975, -3.25) {};
		\node [style=none] (51) at (3.975, -3.25) {};
		\node [style=none] (52) at (3.975, -6) {};
		\node [style=none] (53) at (1.975, -6) {};
		\node [style=none] (54) at (-4.025, 6) {};
		\node [style=none] (55) at (-2.025, 6) {};
		\node [style=none] (56) at (-2.025, 3.25) {};
		\node [style=none] (57) at (-4.025, 3.25) {};
		\node [style=none] (58) at (-4.025, -3.25) {};
		\node [style=none] (59) at (-2.025, -3.25) {};
		\node [style=none] (60) at (-2.025, -6) {};
		\node [style=none] (61) at (-4.025, -6) {};
		\node [style=red label] (62) at (4.75, 4.75) {$E$};
		\node [style=red label] (63) at (2.5, 0) {$C$};
		\node [style=red label] (64) at (-4.75, -4.5) {$A$};
		\node [style=red label] (65) at (4.75, -4.5) {$B$};
		\node [style=red label] (66) at (-4.75, 4.75) {$D$};
	\end{pgfonlayer}
	\begin{pgfonlayer}{edgelayer}
		\draw [style=classical, in=-90, out=90] (3) to (0.center);
		\draw [style=bold black] (4.center) to (3);
		\draw [style=classical] (0.center) to (7.center);
		\draw [style=classical, in=-90, out=90] (11) to (8.center);
		\draw [style=bold black] (12.center) to (11);
		\draw [style=classical] (8.center) to (15.center);
		\draw [style=classical, in=90, out=-90] (23) to (20.center);
		\draw [style=bold black] (24.center) to (23);
		\draw [style=classical] (20.center) to (27.center);
		\draw [style=classical, in=90, out=-90] (31) to (28.center);
		\draw [style=bold black] (32.center) to (31);
		\draw [style=classical] (28.center) to (35.center);
		\draw [style=classical] (38) to (39.center);
		\draw [style=carta da zucchero thin] (42.center)
			 to (43.center)
			 to (44.center)
			 to (45.center)
			 to cycle;
		\draw [style=carta da zucchero thin] (46.center)
			 to (47.center)
			 to (48.center)
			 to (49.center)
			 to cycle;
		\draw [style=carta da zucchero thin] (53.center)
			 to (50.center)
			 to (51.center)
			 to (52.center)
			 to cycle;
		\draw [style=carta da zucchero thin] (57.center)
			 to (54.center)
			 to (55.center)
			 to (56.center)
			 to cycle;
		\draw [style=carta da zucchero thin] (61.center)
			 to (58.center)
			 to (59.center)
			 to (60.center)
			 to cycle;
		\draw [style=bold black, in=90, out=-90] (4.center) to (18.center);
		\draw [style=bold black, in=90, out=-90, looseness=0.75] (18.center) to (19.center);
		\draw [style=bold black, in=-90, out=90, looseness=0.75] (16.center) to (17.center);
		\draw [style=bold black, in=-90, out=90] (17.center) to (12.center);
		\draw [style=bold black, in=-90, out=90] (24.center) to (37.center);
		\draw [style=bold black, in=90, out=-90] (36.center) to (32.center);
	\end{pgfonlayer}
\end{tikzpicture}
    }
\end{center}
By construction, both empirical model are 100\% supported by the space of input histories induced the cross causal order. In fact, they are both deterministic, and hence they correspond to causal functions for the space.
\begin{center}
\includegraphics[height=3cm]{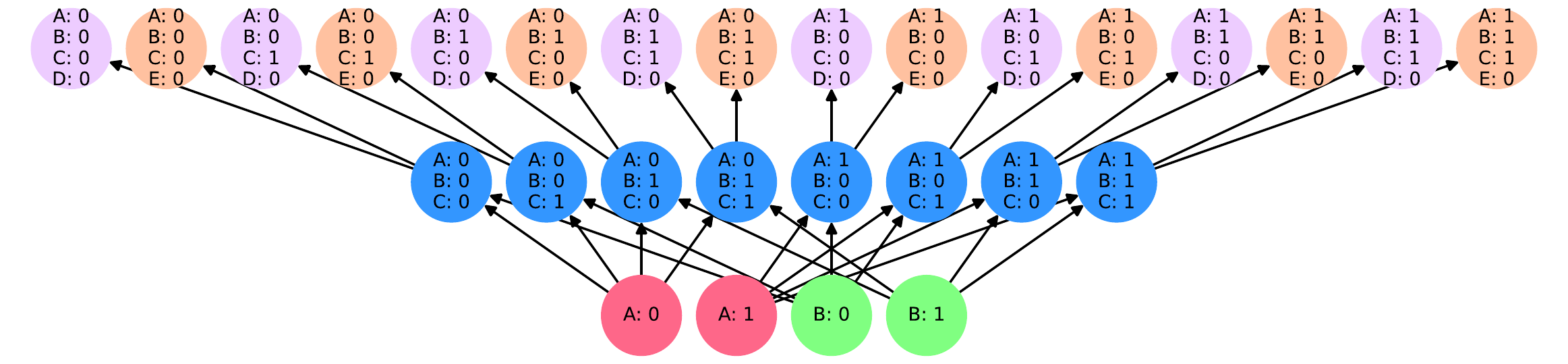}
\end{center}
In the second version of the experiment, Charlie doesn't learn anything about Alice and Bob's inputs: his trivial output can be explained without signalling from either one of Alice or Bob.
Indeed, the simplified empirical model $e'$ is 100\% supported by the space of input histories induced by the following ``$K_{3,2}$'' causal order.  
\begin{center}
\includegraphics[height=3cm]{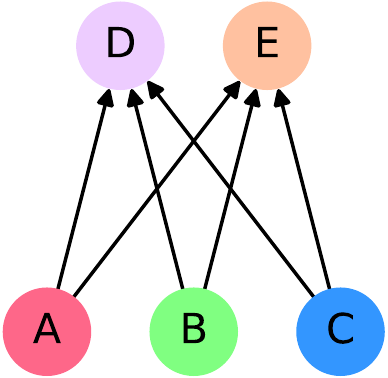}
\end{center}
The space of input histories is explicitly depicted below.
\begin{center}
\includegraphics[height=2.75cm]{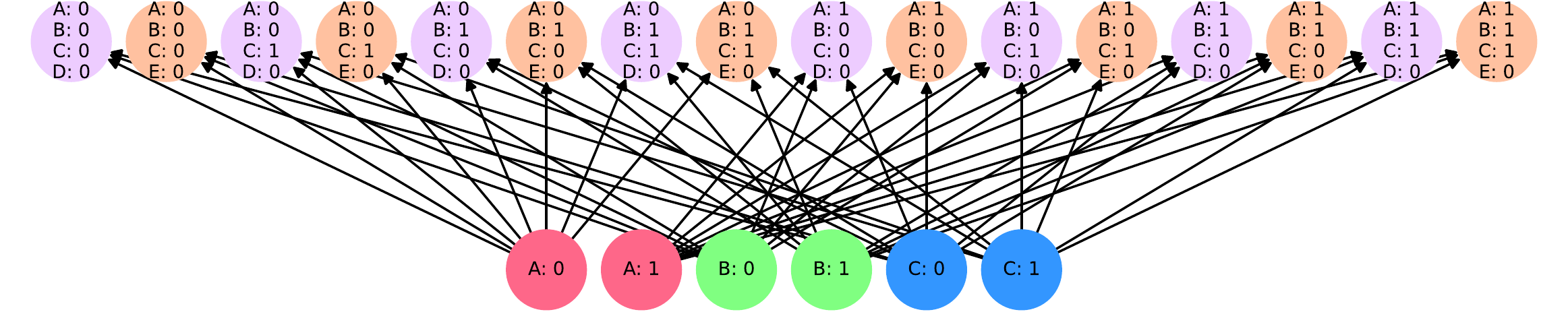}
\end{center}
In fact, the causaltopes for the spaces induced by the cross causal order and the $K_{3, 2}$ causal order coincide when the output $O_{\ev{C}} = \{0\}$ for Charlie is trivial, as demonstrated by the corresponding systems of equations in RREF: 
\begin{center}
    \begin{tabular}{ccc}
    \includegraphics[width=6cm]{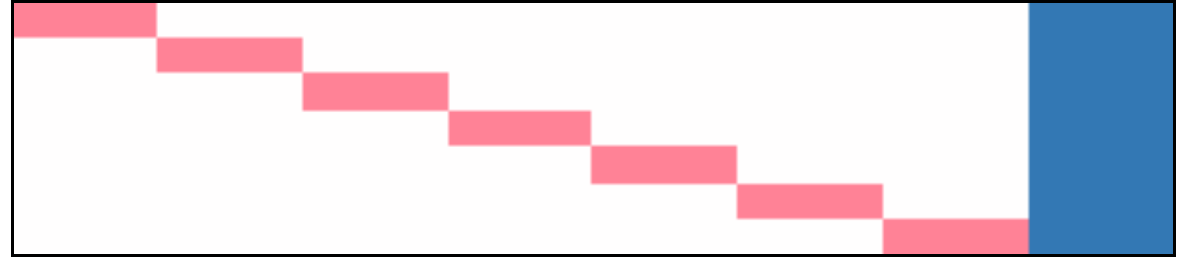}
    & \hspace{2cm} &
    \includegraphics[width=6cm]{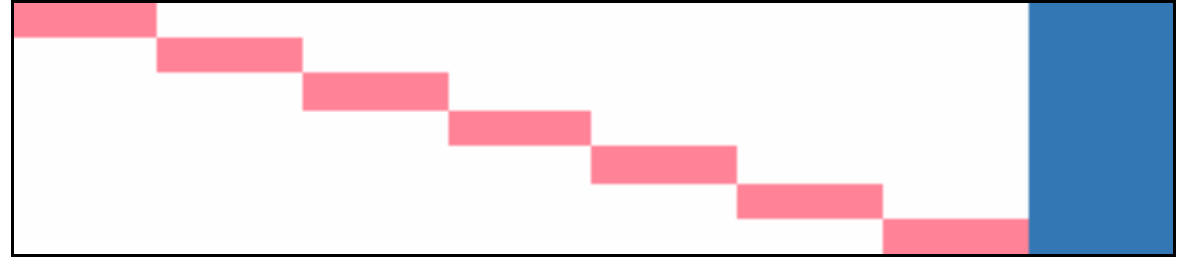}
    \\
    cross causal order
    &&
    $K_{3, 2}$ causal order
    \\
    dimension 24
    &&
    dimension 24
    \end{tabular}
\end{center} 
We can also construct an entirely different space where Charlie's output is independent of Alice and Bob's input, by exploiting the additional constraints on causal functions afforded by lack of tightness.
Indeed, the empirical model $e'$ is also 100\% supported by the non-tight space below: no-signalling from Alice to Charlie is enforced by the \hist{B/b,C/c} histories on the right, while no-signalling from Bob to Charlie is enforced by the \hist{A/a,C/c} histories on the left.
\begin{center}
\includegraphics[height=3cm]{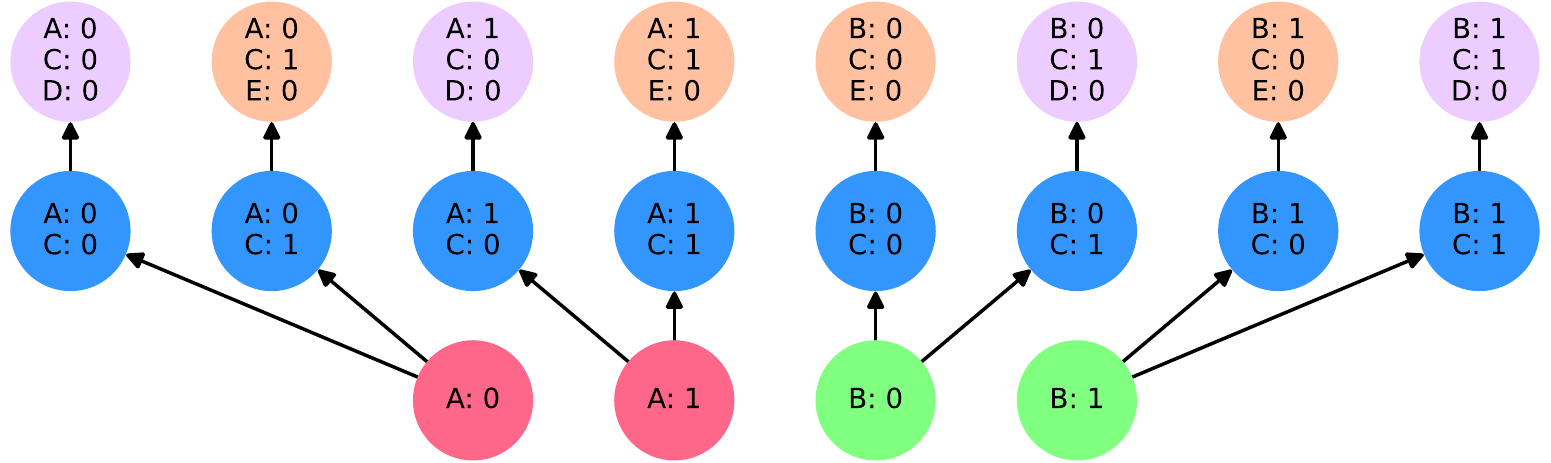}
\end{center}
This space of input histories is a subspace of the space induced by the cross causal order, but is unrelated to the one for the $K_{3,2}$ causal order.
However, the 16-dimensional causaltope for this non-tight space sits inside the causaltope for the two order-induced spaces: even though the spaces are unrelated, every standard empirical model for the non-tight space is also a standard empirical model for the space induced by the $K_{3, 2}$ order.
For completeness, below is the system of equations in RREF defining the causaltope for the non-tight space:
\begin{center}
    \begin{tabular}{c}
    \includegraphics[width=6cm]{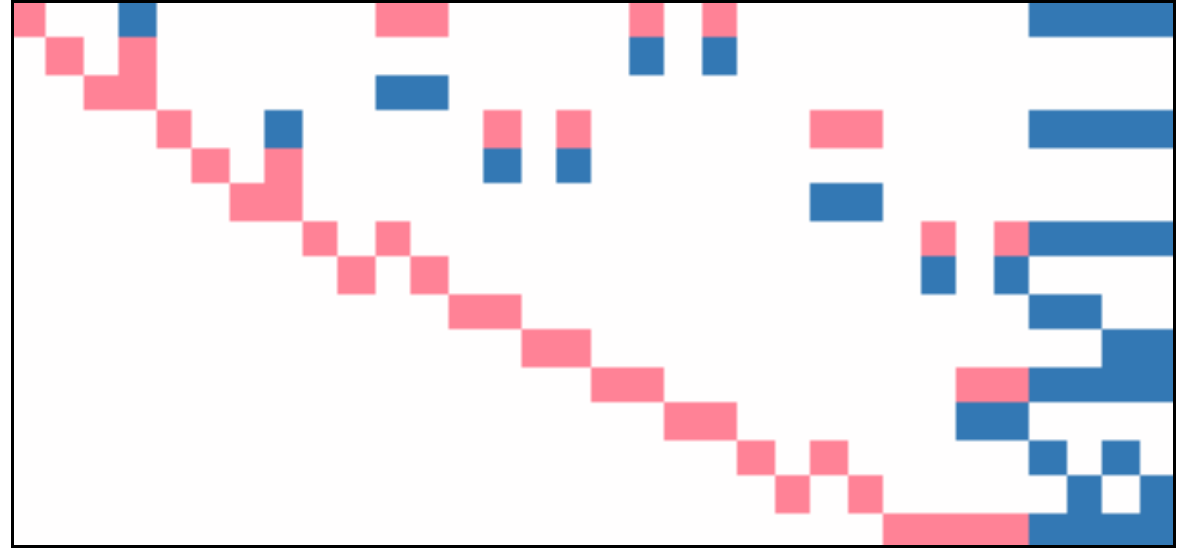}
    \\
    non-tight space
    \\
    dimension 16
    \end{tabular}
\end{center}

\subsubsection{The Leggett-Garg Empirical Model.}

To disprove macro-realistic explanations of quantum mechanical phenomena, the authors of \cite{leggett1985quantum} propose the following experiment on a 2-level quantum system, i.e. a qubit, which evolves in time by rotating about the Y axis at a constant angular rate.
Writing $\Delta t > 0$ for the minimum time over which qubit evolution performs a $\frac{2\pi}{3}$ Y rotation, the experiment proceeds as follows:
\begin{enumerate}
    \item The qubit is prepared in the $|+\rangle$ state at time $t_0$ and left alone to evolve.
    \item At time $t_1 := t_0 + \Delta t$, known to us as event $\ev{A}$, the qubit is either left alone (input 0 at $\ev{A}$, with output fixed to $0$) or a non-demolition measurement in the Z basis is performed on it (input 1 at $\ev{A}$, with meas. outcome as output). The qubit is again left alone to evolve.
    \item At time $t_2 := t_1 + \Delta t$, known to us as event $\ev{B}$, the qubit is either left alone (input 0 at $\ev{B}$, with output fixed to $0$) or a non-demolition measurement in the Z basis is performed on it (input 1 at $\ev{B}$, with meas. outcome as output). The qubit is again left alone to evolve.
    \item At time $t_3 := t_2 + \Delta t$, known to us as event $\ev{C}$, the qubit is either discarded (input 0 at $\ev{C}$, with output fixed to $0$) or a demolition measurement in the Z basis is performed on it (input 1 at $\ev{C}$, with meas. outcome as output).
\end{enumerate}
The figure below exemplifies the scenario we have just described:
\begin{center}
\scalebox{1}{
    \begin{tikzpicture}
	\begin{pgfonlayer}{nodelayer}
		\node [style=small box] (0) at (-3, -7.25) {$R_y(\frac{2\pi}{3})$};
		\node [style=small box] (1) at (-3, -5.5) {$M$};
		\node [style=small box] (2) at (-3, -3.5) {$R_y(\frac{2\pi}{3})$};
		\node [style=point] (4) at (-3, -8.75) {$+$};
		\node [style=small box] (5) at (-3, 0.5) {$R_y(\frac{2\pi}{3})$};
		\node [style=small box] (6) at (-3, -1.5) {$M$};
		\node [style=small box] (7) at (-3, 2.5) {$M$};
		\node [style=upground] (8) at (-3, 4) {};
		\node [style=none] (9) at (-2.75, 2.25) {};
		\node [style=none] (10) at (-2.75, 2.75) {};
		\node [style=none] (11) at (-1, 4) {};
		\node [style=none] (12) at (-1, 1) {};
		\node [style=none] (13) at (-5, 0) {};
		\node [style=none] (14) at (-5, -3) {};
		\node [style=none] (15) at (-3.25, -1.25) {};
		\node [style=none] (16) at (-3.25, -1.75) {};
		\node [style=none] (17) at (-5, 1) {};
		\node [style=none] (18) at (-5, -4) {};
		\node [style=none] (19) at (-5, -3) {};
		\node [style=none] (20) at (-1, 0) {};
		\node [style=none] (21) at (-1, 5) {};
		\node [style=none] (22) at (-1, -4) {};
		\node [style=none] (23) at (-1, -7) {};
		\node [style=none] (24) at (-2.5, -5.25) {};
		\node [style=none] (25) at (-2.5, -5.75) {};
		\node [style=none] (26) at (-1, -3) {};
		\node [style=none] (27) at (-1, -8) {};
		\node [style=none] (28) at (-1, -7) {};
		\node [style=none] (29) at (-3.75, 3.5) {};
		\node [style=none] (30) at (-2, 3.5) {};
		\node [style=none] (31) at (-2, 1.5) {};
		\node [style=none] (32) at (-3.75, 1.5) {};
		\node [style=none] (33) at (-3.75, -0.5) {};
		\node [style=none] (34) at (-2, -0.5) {};
		\node [style=none] (35) at (-2, -2.5) {};
		\node [style=none] (36) at (-3.75, -2.5) {};
		\node [style=none] (37) at (-3.75, -4.5) {};
		\node [style=none] (38) at (-2, -4.5) {};
		\node [style=none] (39) at (-2, -6.5) {};
		\node [style=none] (40) at (-3.75, -6.5) {};
		\node [style=small box] (41) at (1.5, 0.5) {$M$};
		\node [style=none] (42) at (1.75, 0.25) {};
		\node [style=none] (43) at (1.75, 0.75) {};
		\node [style=none] (44) at (2.25, 2) {};
		\node [style=point] (45) at (2.75, -1) {$0$};
		\node [style=none] (46) at (1.5, -1) {};
		\node [style=none] (47) at (1.5, 2) {};
		\node [style=small box] (48) at (1.5, -5.5) {$M$};
		\node [style=none] (49) at (1.75, -5.75) {};
		\node [style=none] (50) at (1.75, -5.25) {};
		\node [style=none] (51) at (2.25, -4) {};
		\node [style=point] (52) at (2.75, -7) {$1$};
		\node [style=none] (53) at (1.5, -7) {};
		\node [style=none] (54) at (1.5, -4) {};
		\node [style=none] (55) at (3.5, 0.5) {$=$};
		\node [style=none] (56) at (3.5, -5.5) {$=$};
		\node [style=none] (57) at (5, 0.5) {};
		\node [style=point] (59) at (5.75, 0.75) {$0$};
		\node [style=none] (60) at (5.75, 2) {};
		\node [style=none] (61) at (5, -1) {};
		\node [style=none] (62) at (5, 2) {};
		\node [style=white dot] (63) at (5, -5.5) {};
		\node [style=none] (65) at (5.75, -5) {};
		\node [style=none] (66) at (5.75, -4) {};
		\node [style=none] (67) at (5, -7) {};
		\node [style=none] (68) at (5, -4) {};
	\end{pgfonlayer}
	\begin{pgfonlayer}{edgelayer}
		\draw [style=bold black] (0) to (1);
		\draw [style=bold black] (1) to (2);
		\draw [style=bold black] (4) to (0);
		\draw [style=bold black] (2) to (6);
		\draw [style=bold black] (6) to (5);
		\draw [style=bold black] (5) to (7);
		\draw [style=bold black] (7) to (8);
		\draw [style=classical, in=270, out=90] (10.center) to (11.center);
		\draw [style=classical, in=90, out=-90] (9.center) to (12.center);
		\draw [style=classical, in=-90, out=90] (15.center) to (13.center);
		\draw [style=classical, in=90, out=-90, looseness=1.25] (16.center) to (14.center);
		\draw [style=classical] (18.center) to (19.center);
		\draw [style=classical] (13.center) to (17.center);
		\draw [style=classical] (20.center) to (12.center);
		\draw [style=classical] (11.center) to (21.center);
		\draw [style=classical, in=-90, out=90] (24.center) to (22.center);
		\draw [style=classical, in=90, out=-90, looseness=1.25] (25.center) to (23.center);
		\draw [style=classical] (27.center) to (28.center);
		\draw [style=classical] (22.center) to (26.center);
		\draw [style=carta da zucchero thin] (29.center)
			 to (30.center)
			 to (31.center)
			 to (32.center)
			 to cycle;
		\draw [style=carta da zucchero thin] (33.center)
			 to (34.center)
			 to (35.center)
			 to (36.center)
			 to cycle;
		\draw [style=carta da zucchero thin] (37.center)
			 to (38.center)
			 to (39.center)
			 to (40.center)
			 to cycle;
		\draw [style=bold black] (41) to (47.center);
		\draw [style=bold black] (41) to (46.center);
		\draw [style=classical, in=-90, out=90] (43.center) to (44.center);
		\draw [style=classical, in=90, out=-90] (42.center) to (45);
		\draw [style=bold black] (48) to (54.center);
		\draw [style=bold black] (48) to (53.center);
		\draw [style=classical, in=-90, out=90] (50.center) to (51.center);
		\draw [style=classical, in=90, out=-90] (49.center) to (52);
		\draw [style=bold black] (57.center) to (62.center);
		\draw [style=bold black] (57.center) to (61.center);
		\draw [style=classical, in=-90, out=90] (59) to (60.center);
		\draw [style=bold black] (63) to (68.center);
		\draw [style=bold black] (63) to (67.center);
		\draw [style=classical, in=-90, out=90] (65.center) to (66.center);
		\draw [style=classical, in=-90, out=0, looseness=0.75] (63) to (65.center);
	\end{pgfonlayer}
\end{tikzpicture}
}
\end{center}
The description above results in the following empirical model on 3 events:
\begin{center}
\scalebox{0.9}{
\begin{tabular}{l|rrrrrrrr}
\hfill
ABC & 000 & 001 & 010 & 011 & 100 & 101 & 110 & 111\\
\hline
000 & 1.000 & 0.000 & 0.000 & 0.000 & 0.000 & 0.000 & 0.000 & 0.000\\
001 & 0.933 & 0.067 & 0.000 & 0.000 & 0.000 & 0.000 & 0.000 & 0.000\\
010 & 0.067 & 0.000 & 0.933 & 0.000 & 0.000 & 0.000 & 0.000 & 0.000\\
011 & 0.017 & 0.050 & 0.700 & 0.233 & 0.000 & 0.000 & 0.000 & 0.000\\
100 & 0.500 & 0.000 & 0.000 & 0.000 & 0.500 & 0.000 & 0.000 & 0.000\\
101 & 0.125 & 0.375 & 0.000 & 0.000 & 0.375 & 0.125 & 0.000 & 0.000\\
110 & 0.125 & 0.000 & 0.375 & 0.000 & 0.375 & 0.000 & 0.125 & 0.000\\
111 & 0.031 & 0.094 & 0.281 & 0.094 & 0.094 & 0.281 & 0.094 & 0.031\\
\end{tabular}
}
\end{center}
The Leggett-Garg inequalities provides bounds, valid in macro-realistic interpretations, for the sum of the expected $\pm1$-valued parity of outputs when the $\pm1$-valued parity of inputs is $+1$:
\[
-1 \leq 
\mathbb{E}(-1^{o_A\!\oplus\! o_B \!\oplus\! o_C}|011)
+\mathbb{E}(-1^{o_A\!\oplus\! o_B \!\oplus\! o_C}|101)
+\mathbb{E}(-1^{o_A\!\oplus\! o_B \!\oplus\! o_C}|110)
\leq 3
\]
The authors then observe that, in the experiment they propose, the sum of such expected parities is $-\frac{3}{2}$, violating the lower bound and thus excluding a macro-realistic explanation.
Indeed, we can restrict ourselves to the relevant rows of the empirical model:
\begin{center}
\scalebox{0.9}{
\begin{tabular}{l|rrrrrrr}
\hfill
ABC & 000 & 001 & 010 & 011 & 100 & 101 & 110\\
\hline
011 & 0.017 & 0.050 & 0.700 & 0.233 & 0.000 & 0.000 & 0.000\\
101 & 0.125 & 0.375 & 0.000 & 0.000 & 0.375 & 0.125 & 0.000\\
110 & 0.125 & 0.000 & 0.375 & 0.000 & 0.375 & 0.000 & 0.125\\
\end{tabular}
}
\end{center}
The sum of the expected parity of outputs is then computed as follows: 
\[
\begin{array}{rl}
&\mathbb{E}(-1^{o_A\!\oplus\! o_B \!\oplus\! o_C}|011)
+\mathbb{E}(-1^{o_A\!\oplus\! o_B \!\oplus\! o_C}|101)
+\mathbb{E}(-1^{o_A\!\oplus\! o_B \!\oplus\! o_C}|110)
\\
=& (0.017-0.050-0.700+0.233)
\\
 &+(0.125-0.375-0.375+0.125)
\\
 &+(0.125-0.375-0.375+0.125)
\\
=& \frac{1}{2}+\frac{1}{2}+\frac{1}{2} = \frac{3}{2}
\end{array}
\]
By construction, this empirical model is 100\% supported by the space of inputs histories for the total order $\total{A, B, C}$.
As a consequence, it is necessarily non-contextual/local for this space: for an explicit decomposition as a convex combination of 12 causal functions, see Subsection 4.5.4 of \cite{gogioso2022topology}.

The constraints specified by Equation (1) of \cite{leggett1985quantum} are in fact causal constraint, stating that a macro-realist model has to be supported by the following 3 indefinite causal orders.
\begin{center}
    \begin{tabular}{ccc}
    \includegraphics[height=4cm]{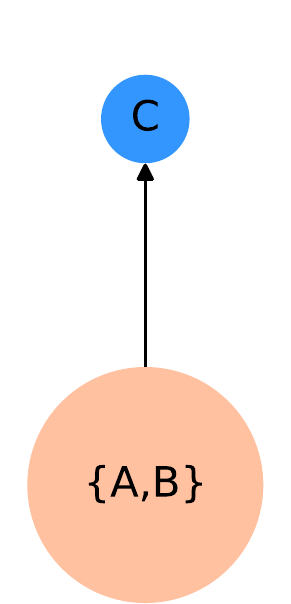}
    &\hspace{2cm}
    \includegraphics[height=4cm]{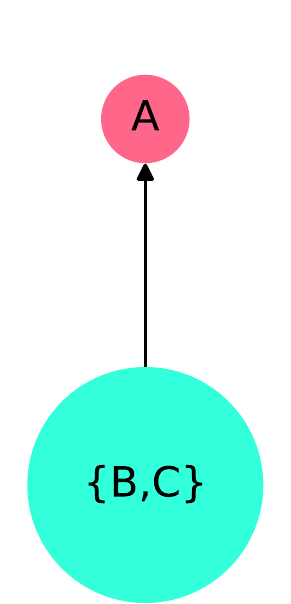}
    &\hspace{2cm}
    \includegraphics[height=4cm]{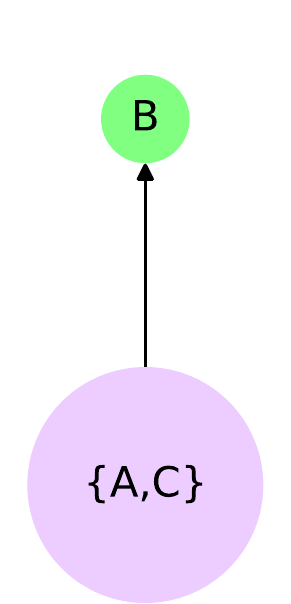}
    \end{tabular}
\end{center}
By construction, the Leggett-Garg empirical model is 100\% supported by the leftmost causal order (of which \total{A, B, C} is a sub-order).
However, it is only 56.7\% supported by the middle causal order and 62.5\% supported by the right causal order, showing that it violates the causal constraints imposed by the macro-realist assumption.
The meet of the spaces of input histories induced by the three indefinite causal orders above is the non-causally-complete, non-tight space depicted below:
\begin{center}
\includegraphics[height=2.5cm]{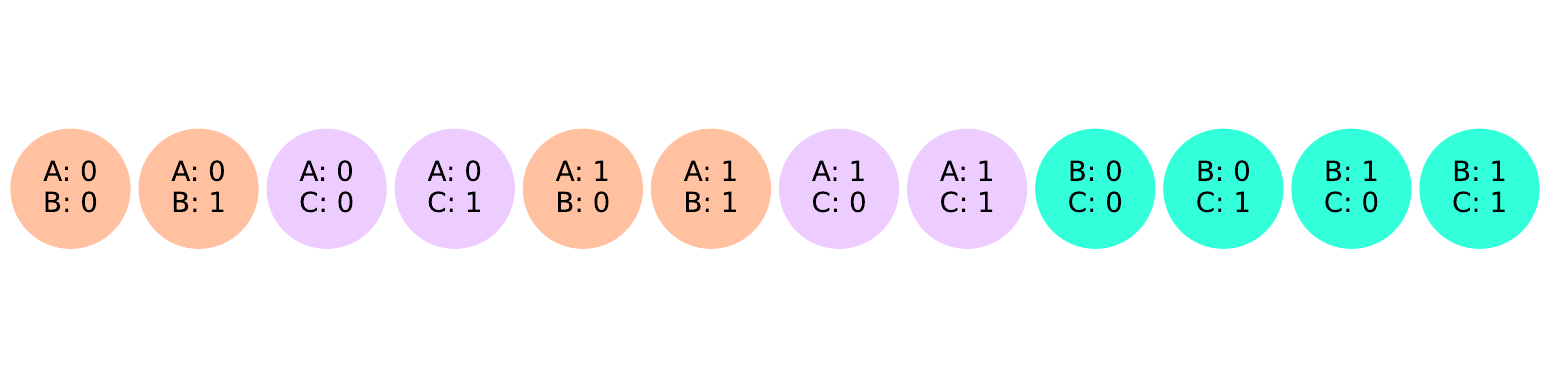}
\end{center}
The meet of the causaltopes for the spaces induced by the three indefinite causal orders is the same as the causaltope for the meet space: they both have dimension 26, and hence coincide with the no-signalling causaltope.
Imposing the three constraints together is thus the same as imposing no-signalling, and the Leggett-Garg empirical model has a 30.25\% no-signalling fraction.
It also has varying causal fractions over other total orders: 62.50\% over $\total{A, C, B}$, 56.70\% over $\total{B, A, C}$ and $\total{B, C, A}$, 45.87\% over $\total{C, B, A}$ and 37.95\% over $\total{C, A, B}$.
The (unique) minimal supporting space is depicted below, with a 35-dim causaltope. 
\begin{center}
\includegraphics[height=3cm]{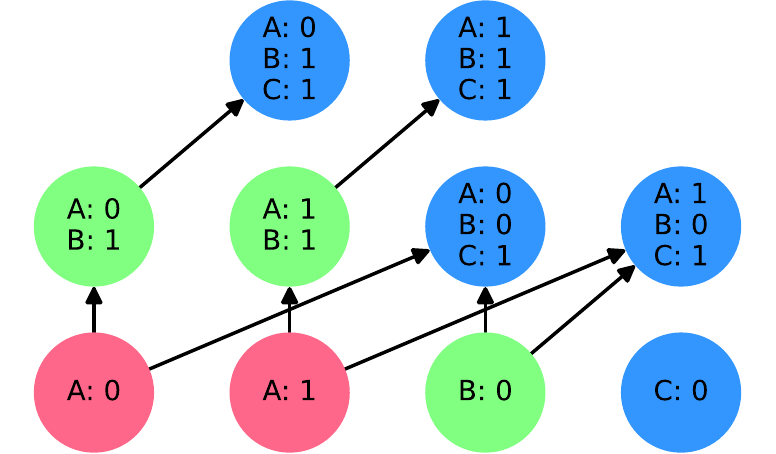}
\end{center}
This space captures the causal constraints---additional with respect to the total order \total{A,B,C}---associated with the absence of measurement on input 0:
\begin{itemize}
    \item The presence of history \hist{C/0} states that there is no signalling from \ev{A} nor \ev{B} to \ev{C} when no measurement is performed at \ev{C}.
    \item The presence of history \hist{B/0} states that there is no signalling from \ev{A} to \ev{B} when no measurement is performed at \ev{B}.
\end{itemize}
The causal functions involved in the deterministic causal HVM for the empirical model over \total{A, B, C} are also causal for the minimal supporting space, hence the empirical model is non-contextual/local there as well.

\subsubsection{An OCB Empirical Model.}

In this example, we look consider two agents, Alice and Bob, acting within the context of the process matrix described by Oreshkov, Costa and Brukner in \cite{oreshkov2012quantum}:
\[
    W^{A_1A_2B_1B_2} = \frac{1}{4}[1^{A_1A_2B_1B_2} + \frac{1}{\sqrt{2}}(\sigma_z^{A_2} \sigma_z^{B_1} + \sigma_z^{A_1} \sigma_x^{B_1} \sigma_z^{B_2})]
\]
The two agents perform the following local instruments:
\begin{itemize}
    \item Alice measures the incoming qubit in the Z basis, producing an output $x \in \{0, 1\}$. She then encodes her input $a \in \{0,1\}$ into the Z basis of the outgoing qubit.
    \item Bob has input $(b, b') \in \{0,1\}^2$:
    \begin{itemize}
        \item If $b' = 0$, Bob measures the incoming qubit in the X basis, obtaining a measurement outcome $z \in \{+,-\}$: if $z=+$, Bob prepares the outgoing qubit in $|b\rangle$; if $z=-$, Bob prepares the outgoing qubit in $|1-b\rangle$ instead. Regardless of the value of $z$, the output $y \in \{0,1\}$ of Bob is set constantly to $y=0$.
        \item If $b' = 1$, Bob measures the incoming qubit in the Z basis and uses the measurement outcome as his output $y \in \{0,1\}$. He prepares the outgoing qubit in $|0\rangle$, regardless of the value of $b$.
    \end{itemize}
\end{itemize}
The description above results in the following empirical model on 2 events:
\begin{center}
\scalebox{0.9}{
\begin{tabular}{l|rrrr}
\hfill
ABB & 00 & 01 & 10 & 11\\
\hline
000 & $1/4+\sqrt{2}/8$ & $1/4+\sqrt{2}/8$ & $1/4-\sqrt{2}/8$ & $1/4-\sqrt{2}/8$\\
001 & $1/4-\sqrt{2}/8$ & $1/4-\sqrt{2}/8$ & $1/4+\sqrt{2}/8$ & $1/4+\sqrt{2}/8$\\
010 & $1/4+\sqrt{2}/16$ & $1/4-\sqrt{2}/16$ & $1/4+\sqrt{2}/16$ & $1/4-\sqrt{2}/16$\\
011 & $1/4+\sqrt{2}/16$ & $1/4-\sqrt{2}/16$ & $1/4+\sqrt{2}/16$ & $1/4-\sqrt{2}/16$\\
100 & $1/4+\sqrt{2}/8$ & $1/4+\sqrt{2}/8$ & $1/4-\sqrt{2}/8$ & $1/4-\sqrt{2}/8$\\
101 & $1/4-\sqrt{2}/8$ & $1/4-\sqrt{2}/8$ & $1/4+\sqrt{2}/8$ & $1/4+\sqrt{2}/8$\\
110 & $1/4-\sqrt{2}/16$ & $1/4+\sqrt{2}/16$ & $1/4-\sqrt{2}/16$ & $1/4+\sqrt{2}/16$\\
111 & $1/4-\sqrt{2}/16$ & $1/4+\sqrt{2}/16$ & $1/4-\sqrt{2}/16$ & $1/4+\sqrt{2}/16$\\
\end{tabular}
}
\end{center}
This is our first example of causally inseparable empirical model: the maximum causal fraction achieved over the causaltopes for the spaces induced by \total{A, B} and \total{B, A} (the maximal causally complete spaces on 2 events $\{\ev{A}, \ev{B}\}$) is around 93.9\%.
The particular decomposition achieving this fraction in the convex hull has components with the following causal fractions over the two individual causaltopes:  
\begin{center}
    \begin{tabular}{ccc}
    \includegraphics[height=3cm]{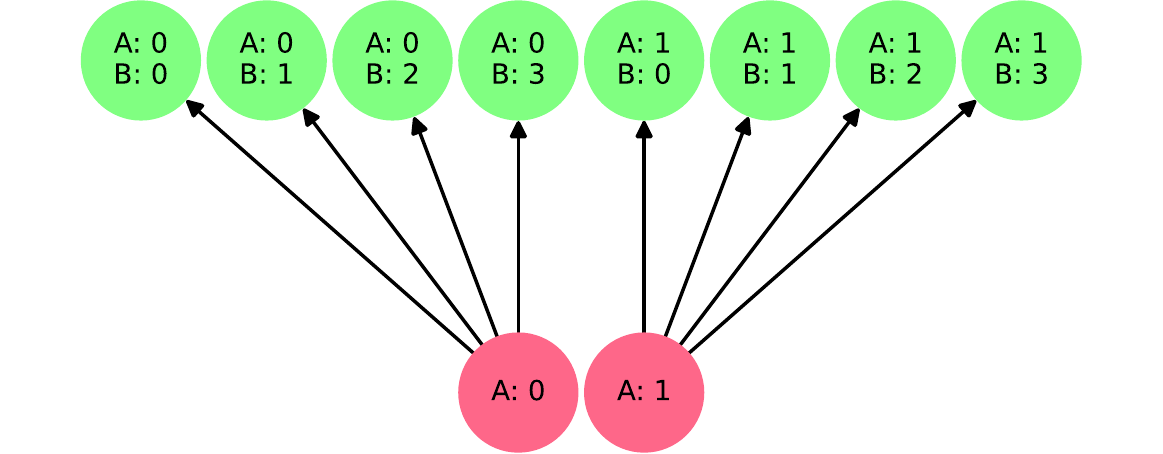}
    &
    \includegraphics[height=3cm]{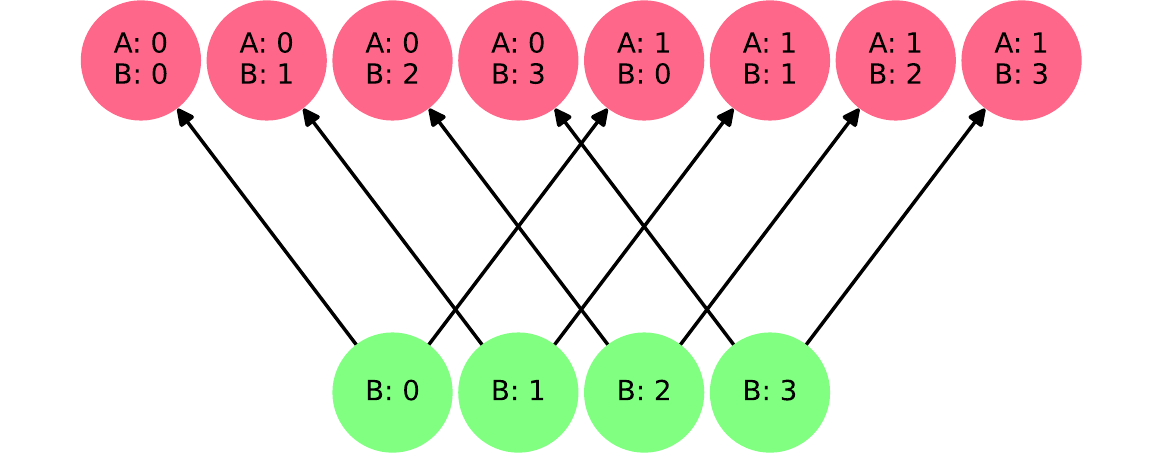}
    \\
    $\Hist{\total{A,B}, \{0,1\}}$
    &
    $\Hist{\total{B,A}, \{0,1\}}$
    \\
    Causal fraction: 29.3\%
    &
    Causal fraction: 64.6\%
    \end{tabular}
\end{center}
There is no ambiguity in the allocation above: the component for each causaltope has causal fraction 0\% in the other causaltope.
The two unnormalised components are shown below:
\begin{center}
\scalebox{0.9}{
    \begin{tabular}{l|rrrr}
    \hfill
    ABB & 00 & 01 & 10 & 11\\
    \hline
    000 & 0.086 & 0.061 & 0.073 & 0.073\\
    001 & 0.073 & 0.073 & 0.061 & 0.086\\
    010 & 0.146 & 0.000 & 0.146 & 0.000\\
    011 & 0.146 & 0.000 & 0.146 & 0.000\\
    100 & 0.146 & 0.000 & 0.073 & 0.073\\
    101 & 0.073 & 0.073 & 0.000 & 0.146\\
    110 & 0.000 & 0.146 & 0.000 & 0.146\\
    111 & 0.000 & 0.146 & 0.000 & 0.146\\
    \end{tabular}
}
\hspace{2cm}
\scalebox{0.9}{
    \begin{tabular}{l|rrrr}
    \hfill
    ABB & 00 & 01 & 10 & 11\\
    \hline
    000 & 0.280 & 0.366 & 0.000 & 0.000\\
    001 & 0.000 & 0.000 & 0.366 & 0.280\\
    010 & 0.131 & 0.162 & 0.192 & 0.162\\
    011 & 0.131 & 0.162 & 0.192 & 0.162\\
    100 & 0.280 & 0.366 & 0.000 & 0.000\\
    101 & 0.000 & 0.000 & 0.366 & 0.280\\
    110 & 0.162 & 0.131 & 0.162 & 0.192\\
    111 & 0.162 & 0.192 & 0.162 & 0.131\\
    \end{tabular}
}
\end{center}
The maximum causal fraction supported by \total{A, B} is 29.3\% and the maximum causal fraction supported by \total{B, A} is 64.6\%.
The maximum causal fraction supported by the no-signalling polytope is also 29.3\%, but it is necessarily different from the 29.3\% component shown on the left above, because it is supported by both \total{A, B} and \total{B, A}.

\newpage
\subsubsection{The BFW Empirical Model.}

We now look at the empirical model introduced by Baumeler, Feix and Wolf in \cite{baumeler2014maximal,baumeler2014perfect}, described by the authors as the 50\%-50\% mixture of a circular ``identity'' classical process and a circular ``bitflip'' classical process, for three agents Alice, Bob and Charlie.
\begin{center}
\scalebox{0.9}{
\begin{tabular}{l|rrrrrrrr}
\hfill
ABC & 000 & 001 & 010 & 011 & 100 & 101 & 110 & 111\\
\hline
000 & $1/2$ & $0$ & $0$ & $0$ & $0$ & $0$ & $0$ & $1/2$\\
001 & $0$ & $0$ & $0$ & $1/2$ & $1/2$ & $0$ & $0$ & $0$\\
010 & $0$ & $1/2$ & $0$ & $0$ & $0$ & $0$ & $1/2$ & $0$\\
011 & $0$ & $0$ & $1/2$ & $0$ & $0$ & $1/2$ & $0$ & $0$\\
100 & $0$ & $0$ & $1/2$ & $0$ & $0$ & $1/2$ & $0$ & $0$\\
101 & $0$ & $1/2$ & $0$ & $0$ & $0$ & $0$ & $1/2$ & $0$\\
110 & $0$ & $0$ & $0$ & $1/2$ & $1/2$ & $0$ & $0$ & $0$\\
111 & $1/2$ & $0$ & $0$ & $0$ & $0$ & $0$ & $0$ & $1/2$\\
\end{tabular}
}
\end{center}
Specifically, the empirical model above is the 50\%-50\% mixture of the following causally inseparable functions (cf. Subsubsection 4.5.5 of \cite{gogioso2022topology}) for the causally incomplete space induced by the causal order \indiscrete{A, B, C}:
\begin{center}
\scalebox{0.7}{
    \begin{tabular}{l|rrrrrrrr}
    \hfill
    ABC & 000 & 001 & 010 & 011 & 100 & 101 & 110 & 111\\
    \hline
    000 & $1$ & $0$ & $0$ & $0$ & $0$ & $0$ & $0$ & $0$\\
    001 & $0$ & $0$ & $0$ & $0$ & $1$ & $0$ & $0$ & $0$\\
    010 & $0$ & $1$ & $0$ & $0$ & $0$ & $0$ & $0$ & $0$\\
    011 & $0$ & $0$ & $0$ & $0$ & $0$ & $1$ & $0$ & $0$\\
    100 & $0$ & $0$ & $1$ & $0$ & $0$ & $0$ & $0$ & $0$\\
    101 & $0$ & $0$ & $0$ & $0$ & $0$ & $0$ & $1$ & $0$\\
    110 & $0$ & $0$ & $0$ & $1$ & $0$ & $0$ & $0$ & $0$\\
    111 & $0$ & $0$ & $0$ & $0$ & $0$ & $0$ & $0$ & $1$\\
    \end{tabular}
}
\hspace{5mm}
\scalebox{0.7}{
    \begin{tabular}{l|rrrrrrrr}
    \hfill
    ABC & 000 & 001 & 010 & 011 & 100 & 101 & 110 & 111\\
    \hline
    000 & $0$ & $0$ & $0$ & $0$ & $0$ & $0$ & $0$ & $1$\\
    001 & $0$ & $0$ & $0$ & $1$ & $0$ & $0$ & $0$ & $0$\\
    010 & $0$ & $0$ & $0$ & $0$ & $0$ & $0$ & $1$ & $0$\\
    011 & $0$ & $0$ & $1$ & $0$ & $0$ & $0$ & $0$ & $0$\\
    100 & $0$ & $0$ & $0$ & $0$ & $0$ & $1$ & $0$ & $0$\\
    101 & $0$ & $1$ & $0$ & $0$ & $0$ & $0$ & $0$ & $0$\\
    110 & $0$ & $0$ & $0$ & $0$ & $1$ & $0$ & $0$ & $0$\\
    111 & $1$ & $0$ & $0$ & $0$ & $0$ & $0$ & $0$ & $0$\\
    \end{tabular}
}
\end{center}
This is an example of a maximally causally inseparable empirical model: it has 0\% support over all causally complete spaces on 3 events.

Interestingly, however, the BFW empirical model is 100\% supported by each of the following 3 indefinite causal orders: this shows that either one of the 3 parties can be taken to act first, as long as the other two parties remain in indefinite causal order.
In contrast, the two individual causally inseparable functions have 0\% support over each of the 3 indefinite causal orders. 
\begin{center}
    \begin{tabular}{ccc}
    \includegraphics[height=4cm]{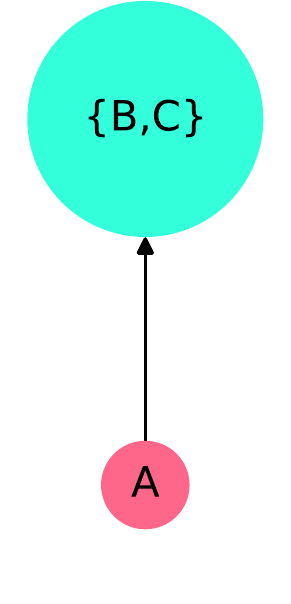}
    &\hspace{2cm}
    \includegraphics[height=4cm]{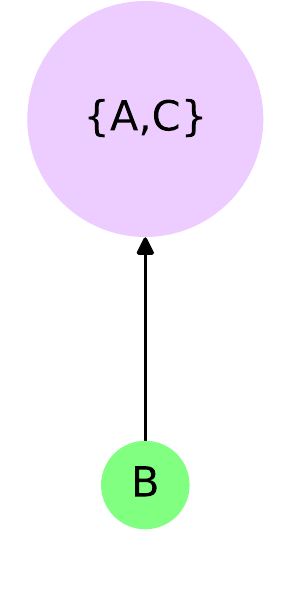}
    &\hspace{2cm}
    \includegraphics[height=4cm]{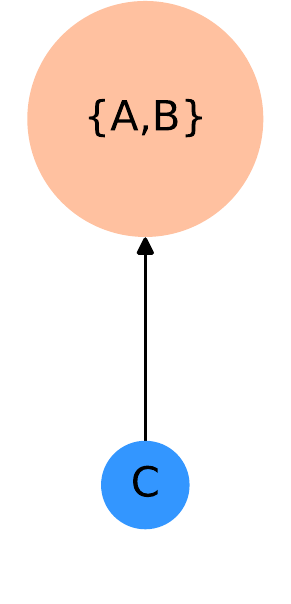}
    \end{tabular}
\end{center}
The meet of the associated spaces of input histories is the discrete space on 3 events: the associated no-signalling causaltope is 26-dimensional and supports 0\% of the BFW empirical model.
The intersection of the associated causaltopes, on the other hand, has dimension 38, and it supports 100\% of the BFW empirical model.

We now look more in detail at the 3 indefinite causal order explanations: without loss of generality, we take Alice to act first.
The absence of any support by the 2 total orders where Alice acts first means that no part of the indefinite causal order between Bob and Charlie is explainable by a fixed causal structure.
Furthermore, the absence of any support by the 2 non-trivial switch orders where Alice acts first means that no part of the indefinite causal order between Bob and Charlie is controlled by Alice's input.
Therefore, the only remaining explanation is that the indefinite causal order between Bob and Charlie is somehow correlated to Alice's outputs.
To verify that this is indeed the case, we consider the scenarios corresponding to a fixed input choice by Alice (input 0 left below, input 1 right below), where the output of Alice has been discarded.  
\begin{center}
\scalebox{0.9}{
    \begin{tabular}{l|rrrr}
    \hfill
    BC & 00 & 01 & 10 & 11\\
    \hline
    00 & $1/2$ & $0$ & $0$ & $1/2$\\
    01 & $1/2$ & $0$ & $0$ & $1/2$\\
    10 & $0$ & $1/2$ & $1/2$ & $0$\\
    11 & $0$ & $1/2$ & $1/2$ & $0$\\
    \end{tabular}
}
\hspace{2cm}
\scalebox{0.9}{
    \begin{tabular}{l|rrrr}
    \hfill
    BC & 00 & 01 & 10 & 11\\
    \hline
    00 & $0$ & $1/2$ & $1/2$ & $0$\\
    01 & $0$ & $1/2$ & $1/2$ & $0$\\
    10 & $1/2$ & $0$ & $0$ & $1/2$\\
    11 & $1/2$ & $0$ & $0$ & $1/2$\\
    \end{tabular}
}
\end{center}
Unsurprisingly, the two restricted empirical models above are both causally separable.
Perhaps surprisingly, they are both 100\% supported by the no-signalling causaltope for Bob and Charlie.
To understand whether Alice's output determines a fixed causal order between Bob and Charlie, we look at the empirical models obtained by conditioning on each of Alice's outputs (output 0 left below, output 1 right below), which have 50\%-50\% distribution independently of her input.
\begin{center}
\scalebox{0.9}{
    \begin{tabular}{l|rrrr}
    \hfill
    ABC & 000 & 001 & 010 & 011\\
    \hline
    000 & $1$ & $0$ & $0$ & $0$\\
    001 & $0$ & $0$ & $0$ & $1$\\
    010 & $0$ & $1$ & $0$ & $0$\\
    011 & $0$ & $0$ & $1$ & $0$\\
    100 & $0$ & $0$ & $1$ & $0$\\
    101 & $0$ & $1$ & $0$ & $0$\\
    110 & $0$ & $0$ & $0$ & $1$\\
    111 & $1$ & $0$ & $0$ & $0$\\
    \end{tabular}
}
\hspace{2cm}
\scalebox{0.9}{
    \begin{tabular}{l|rrrr}
    \hfill
    ABC & 100 & 101 & 110 & 111\\
    \hline
    000 & $0$ & $0$ & $0$ & $1$\\
    001 & $1$ & $0$ & $0$ & $0$\\
    010 & $0$ & $0$ & $1$ & $0$\\
    011 & $0$ & $1$ & $0$ & $0$\\
    100 & $0$ & $1$ & $0$ & $0$\\
    101 & $0$ & $0$ & $1$ & $0$\\
    110 & $1$ & $0$ & $0$ & $0$\\
    111 & $0$ & $0$ & $0$ & $1$\\
    \end{tabular}
}
\end{center}
The empirical models above are deterministic and correspond to causal functions on the space of input histories determined by the order $\total{A, \evset{B,C}}$.
However, the two functions are causally inseparable, and hence so are the two empirical models.
When Alice's input and output are both 0, i.e. in the first 4 rows of the empirical model left above, Bob and Charlie's outputs are related to their inputs in a way which requires bi-directional signalling:
\[
\begin{array}{rcl}
o_B &=& i_C\\
o_C &=& \neg (i_B \wedge i_C)
\end{array}
\]
This evidence supports the intuition that causal inseparability for the BFW model stems from a cyclic order structure, as its very definition seems to suggest.

\subsection{Prologue to Quantum Switches with Entangled Controls.}

The next four examples provide evidence for the phenomenon of \emph{contextual causality}, where causal structure is correlated to contextual information in such a way that non-locality/contextuality implies causal inseparability.
All four examples are based on quantum switches with entangled, or otherwise contextual, control.
In all four examples, we consider causal separability as a function of experimental parameters, and compare it to both the (causally separable) local fraction of the example itself and the local fraction of the non-locality experiment that they each example is inspired by (either on the Bell state or on the 3-partite GHZ state):
\begin{itemize}
    \item In the first example, the causally separable fraction closely tracks both the local fraction of the example itself and the local fraction of the underlying non-locality experiment, but with differences at certain angles.
    \item In the second example, the causally separable fraction has similar qualitative features as the local fraction, but is significantly higher, while the local fraction of the example still closely tracks the local fraction of the underlying non-locality experiment.
    \item In the third example, the causally separable fraction coincides with the local fraction of the example---saturating the lower bound provided by Proposition \ref{proposition:csep-frac-bounded-below-by-csep-local-frac}---but is significantly lower than the local fraction of the underlying non-locality experiment at most angles.
    \item In the fourth example, the locality of the underlying experiment perfectly correlates with causal separability, making it so that the causally separable fraction, the local fraction of the example and the local fraction of the underlying non-locality experiment all coincide.
\end{itemize}
These examples highlight the importance of defining causal separability relative to an ambient causal order (or, more generally, space of input histories): if no-signalling constraints between different groups of agents are enforced---as could very much be the case in a real-world scenario---then such experiments can be used to witness indefinite causal order in quantum theory.
If no constraints are enforced, on the other hand, then the experiments all become causally separable, and cannot be used as witnesses of indefinite causality.

\subsubsection{A Quantum Switch with Bell Entangled Control.}
\label{subsubsection:ghzswitch-2-1}

In this example we consider a single quantum switch between Alice and Bob, where the switch control is the first qubit of a maximally entangled 2-qubit pair.
Charlie measures the control qubit, after the switch, while Diane measures the second qubit in the entangled pair.
Specifically, for angles $\gamma_0, \gamma_1 \in [0, \pi)$:
\begin{enumerate}
    \item A 2-qubit Bell state $|\Phi^+\rangle$ is created: one qubit is sent to the control of a quantum switch between Alice and Bob, the other qubit is sent to Diane.
    \item Alice and Bob are in a quantum switch, with one of the two $|\Phi^+\rangle$ qubits as its control and the $|+\rangle$ state as its input.
    \item Alice and Bob do the same thing: they performs an X measurement on the incoming qubit, using the measurement outcome as their individual output, and they encode their individual input into the X basis of the outgoing qubit.
    \item The output qubit of the switch is discarded. Charlie receives the control qubit, which he measures in a basis chosen as follows: on input 0, he first applies a X rotation by $-\gamma_0$ and then measures in the Z basis; on input 1, he first applies a X rotation by $-\gamma_1$ and then measures in the Z basis.
    \item Diane receives the second qubit of the entangled state $|\Phi^+\rangle$, which she measures in a basis chosen in the same manner as Charlie.
\end{enumerate}
The figure below exemplifies the scenario we have just described:
\begin{center}
\scalebox{1.25}{
    \begin{tikzpicture}
	\begin{pgfonlayer}{nodelayer}
		\node [style=none] (22) at (2.5, -2.25) {};
		\node [style=none] (37) at (-1, -2.25) {};
		\node [style=none] (38) at (-7.25, -2.25) {};
		\node [style=none] (39) at (-7.25, 2) {};
		\node [style=none] (40) at (-1, 2) {};
		\node [style=none] (41) at (-3.5, -2.25) {};
		\node [style=none] (42) at (-6.5, -0.75) {};
		\node [style=black dot] (43) at (-5.25, -0.75) {};
		\node [style=none] (44) at (-6.5, 0.5) {};
		\node [style=black dot] (45) at (-5.25, 0.5) {};
		\node [style=none] (46) at (-5.25, 1.25) {};
		\node [style=none] (47) at (-6.5, 3) {};
		\node [style=none] (48) at (-6.5, -3.25) {};
		\node [style=none] (54) at (-2, -0.75) {};
		\node [style=black dot] (55) at (-3.25, -0.75) {};
		\node [style=none] (56) at (-2, 0.5) {};
		\node [style=black dot] (57) at (-3.25, 0.5) {};
		\node [style=none] (58) at (-3.25, 1.25) {};
		\node [style=none] (59) at (-3.25, -1.5) {};
		\node [style=none] (60) at (-2, 3) {};
		\node [style=none] (61) at (-2, -3.25) {};
		\node [style=none] (62) at (-5.25, -1.5) {};
		\node [style=none] (67) at (-3.75, 1.25) {};
		\node [style=none] (68) at (-3.75, -1.5) {};
		\node [style=none] (69) at (-1.75, -1.5) {};
		\node [style=none] (70) at (-1.75, 1.25) {};
		\node [style=none] (71) at (-6.75, 1.25) {};
		\node [style=none] (72) at (-6.75, -1.5) {};
		\node [style=none] (73) at (-4.75, -1.5) {};
		\node [style=none] (74) at (-4.75, 1.25) {};
		\node [style=none] (80) at (-3.5, -3.25) {};
		\node [style=none] (81) at (-3.75, 3) {};
		\node [style=none] (82) at (-3.75, 2) {};
		\node [style=none] (85) at (3.5, -2.25) {};
		\node [style=none] (86) at (3.5, -3.25) {};
		\node [style=none] (87) at (-1.5, -5) {};
		\node [style=none] (88) at (1.5, -5) {};
		\node [style=small box] (89) at (-3.75, 4.75) {$R_x(-\gamma_b)$};
		\node [style=none] (90) at (3.5, -1.25) {};
		\node [style=none] (91) at (3.5, -2.25) {};
		\node [style=small box] (92) at (3.5, 0.5) {$R_x(-\gamma_b)$};
		\node [style=none] (93) at (-3.5, 4.5) {};
		\node [style=none] (94) at (3.25, 0.25) {};
		\node [style=none] (95) at (-2.75, 3) {};
		\node [style=none] (96) at (2.5, -1.25) {};
		\node [style=eslabel] (97) at (-2.75, 2.75) {$b$};
		\node [style=eslabel] (98) at (2.5, -1.5) {$b$};
		\node [style=white dot] (99) at (-3.75, 5.75) {};
		\node [style=white dot] (100) at (3.5, 1.5) {};
		\node [style=none] (101) at (-5, -2.25) {};
		\node [style=point] (102) at (-5, -3.25) {$+$};
		\node [style=none] (103) at (-5.925, 2) {};
		\node [style=upground] (104) at (-5.925, 3) {};
		\node [style=none] (109) at (2, 2) {};
		\node [style=none] (110) at (2, -1.75) {};
		\node [style=none] (111) at (5, -1.75) {};
		\node [style=none] (112) at (5, 2) {};
		\node [style=none] (113) at (-5.25, 6.25) {};
		\node [style=none] (114) at (-5.25, 2.5) {};
		\node [style=none] (115) at (-2.25, 2.5) {};
		\node [style=none] (116) at (-2.25, 6.25) {};
		\node [style=red label] (117) at (-5.875, 4.75) {$C$};
		\node [style=red label] (118) at (5.5, 0.5) {$D$};
		\node [style=red label] (119) at (-7.75, 0) {};
		\node [style=red label] (120) at (-0.5, 0) {};
		\node [style=red label] (124) at (-7.75, 0) {$A$};
		\node [style=red label] (125) at (-0.5, 0) {$B$};
		\node [style=none] (126) at (-3.75, 7.25) {};
		\node [style=none] (127) at (3.5, 6.75) {};
	\end{pgfonlayer}
	\begin{pgfonlayer}{edgelayer}
		\draw [style=pink thin] (39.center)
			 to (38.center)
			 to (37.center)
			 to (40.center)
			 to cycle;
		\draw [style=classical, in=90, out=-90] (44.center) to (43);
		\draw [style=classical, in=90, out=-90] (45) to (42.center);
		\draw (46.center) to (45);
		\draw [style=classical] (42.center) to (48.center);
		\draw [style=classical] (44.center) to (47.center);
		\draw [style=classical, in=90, out=-90] (56.center) to (55);
		\draw [style=classical, in=90, out=-90] (57) to (54.center);
		\draw (58.center) to (57);
		\draw (55) to (59.center);
		\draw [style=classical] (54.center) to (61.center);
		\draw [style=classical] (56.center) to (60.center);
		\draw [style=carta da zucchero thin] (68.center)
			 to (67.center)
			 to (70.center)
			 to (69.center)
			 to cycle;
		\draw [style=carta da zucchero thin] (72.center)
			 to (71.center)
			 to (74.center)
			 to (73.center)
			 to cycle;
		\draw [style=classical, in=90, out=-90] (44.center) to (43);
		\draw [style=classical, in=90, out=-90] (45) to (42.center);
		\draw [style=classical, in=90, out=-90] (56.center) to (55);
		\draw [style=classical, in=90, out=-90] (57) to (54.center);
		\draw [style=bold black] (46.center) to (45);
		\draw [style=bold black] (55) to (59.center);
		\draw [style=bold black] (58.center) to (57);
		\draw [style=bold black] (62.center) to (43);
		\draw [style=bold black] (80.center) to (41.center);
		\draw [style=bold black] (82.center) to (81.center);
		\draw [style=bold black] (86.center) to (85.center);
		\draw [style=classical] (42.center) to (48.center);
		\draw [style=classical] (44.center) to (47.center);
		\draw [style=classical] (54.center) to (61.center);
		\draw [style=classical] (56.center) to (60.center);
		\draw [style=bold black, in=-90, out=90, looseness=1.25] (87.center) to (80.center);
		\draw [style=bold black, in=-90, out=-90] (87.center) to (88.center);
		\draw [style=bold black, in=-90, out=90, looseness=1.25] (88.center) to (86.center);
		\draw [style=bold black] (81.center) to (89);
		\draw [style=bold black] (91.center) to (90.center);
		\draw [style=bold black] (90.center) to (92);
		\draw [style=classical, in=90, out=-90, looseness=1.25] (93.center) to (95.center);
		\draw [style=classical, in=90, out=-90] (94.center) to (96.center);
		\draw [style=bold black] (89) to (99);
		\draw [style=bold black] (92) to (100);
		\draw [style=bold black] (103.center) to (104);
		\draw [style=bold black] (102) to (101.center);
		\draw [style=carta da zucchero thin] (110.center)
			 to (109.center)
			 to (112.center)
			 to (111.center)
			 to cycle;
		\draw [style=carta da zucchero thin] (116.center)
			 to (115.center)
			 to (114.center)
			 to (113.center)
			 to cycle;
		\draw [style=classical] (99) to (126.center);
		\draw [style=classical] (100) to (127.center);
	\end{pgfonlayer}
\end{tikzpicture}
}
\end{center}
For $\gamma_0 = \frac{\pi}{5}, \gamma_1 = \frac{3\pi}{5}$, the description above results in the following empirical model:
\begin{center}
    \includegraphics[height=10cm]{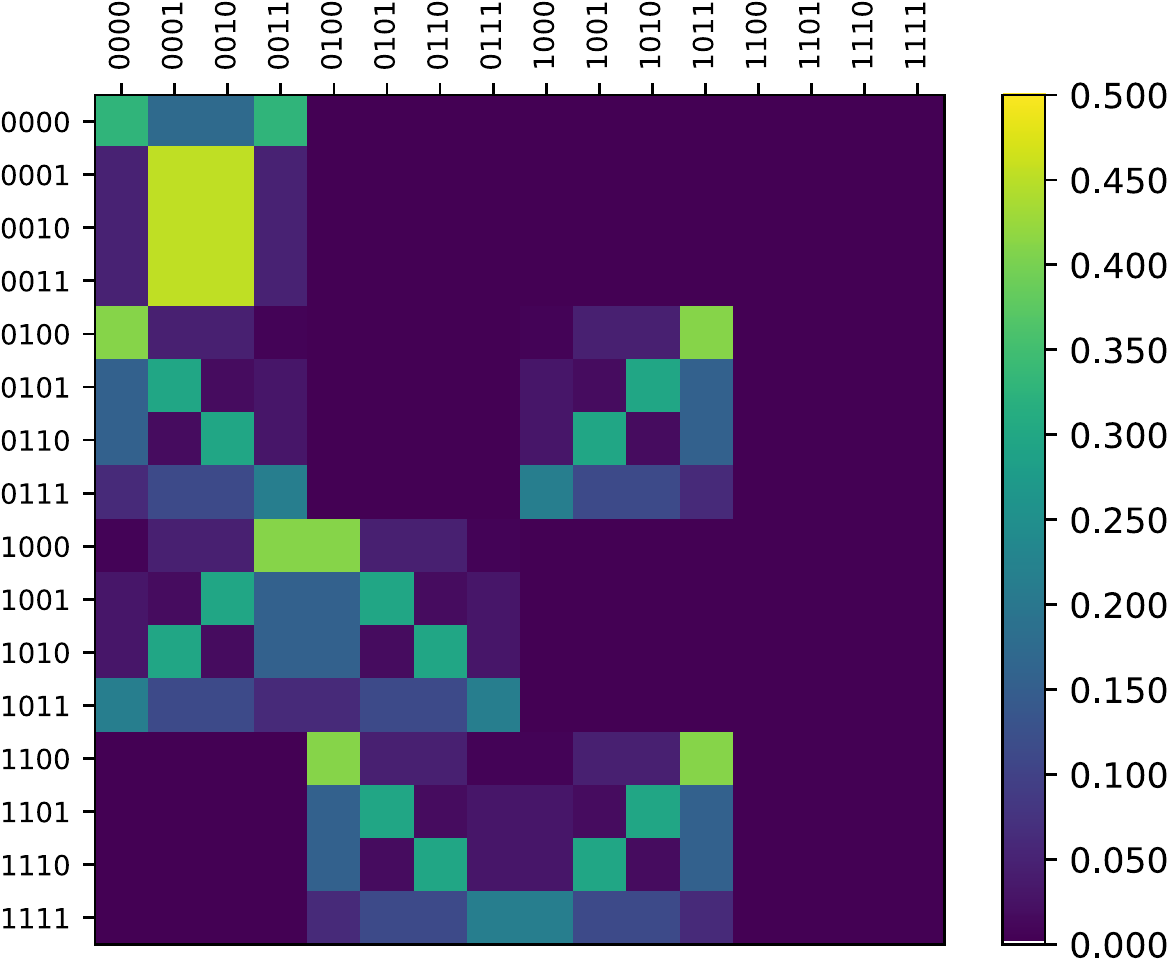}
\end{center}
By construction, the empirical model is 100\% supported by the space of input histories induced by the following indefinite causal order (whose standard causaltope is 134-dimensional):
\begin{center}
        \includegraphics[height=3.5cm]{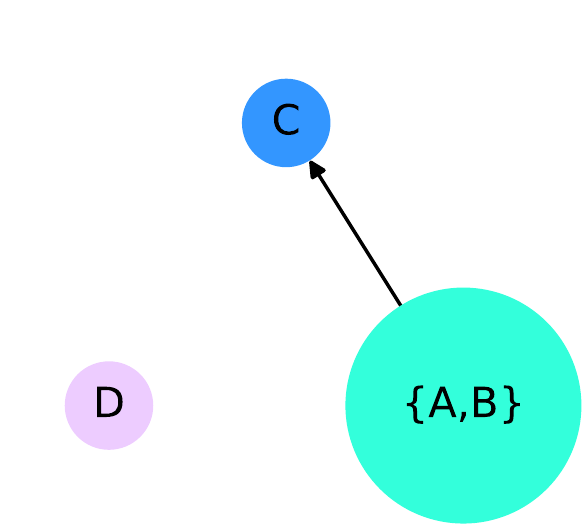}
\end{center}
The empirical model is causally inseparable for the space right above, with a causally separable fraction of \char`\~63.2\% over its two causal completions.
The two completions are induced by the causal orders $\total{A, B, C}\vee\discrete{D}$ and $\total{B, A, C}\vee\discrete{D}$, and their causaltopes are 128-dimensional:
\begin{center}
    \begin{tabular}{ccc}
    \includegraphics[height=3.5cm]{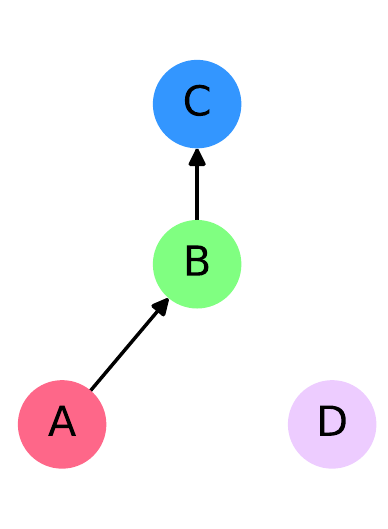}
    &
    \hspace{1cm}
    &
    \includegraphics[height=3.5cm]{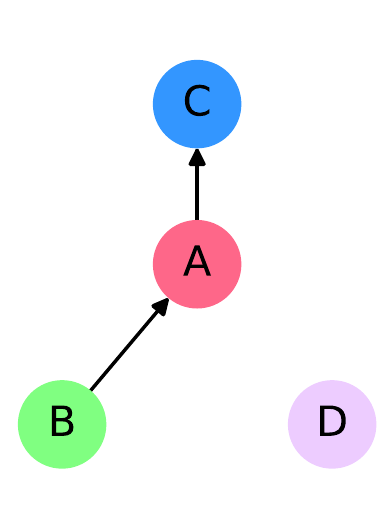}
    \\
    Causal fraction: \char`\~31.6\%
    &
    &
    Causal fraction: \char`\~31.6\%
    \end{tabular}
\end{center}
Below we plot the causally separable fraction (left) and the local fraction (middle) for this empirical model as a function of the $\gamma_0$ and $\gamma_1$ measurement angles used by Charlie and Diane.
We compare it to the local fraction for a Bell experiment with the same measurements (right), as originally shown by Figure 1(a) of \cite{abramsky2017contextual}.
\begin{center}
    \begin{tabular}{ccc}
    \includegraphics[height=4.cm]{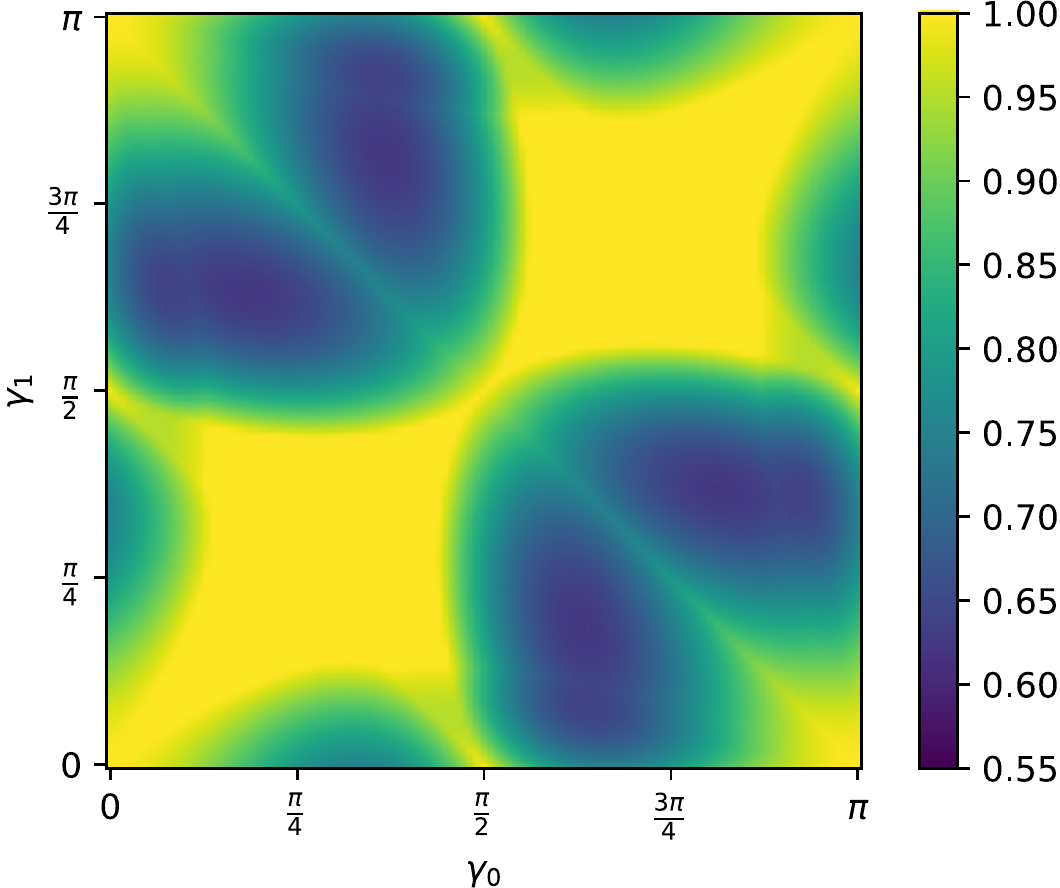}
    &
    \includegraphics[height=4cm]{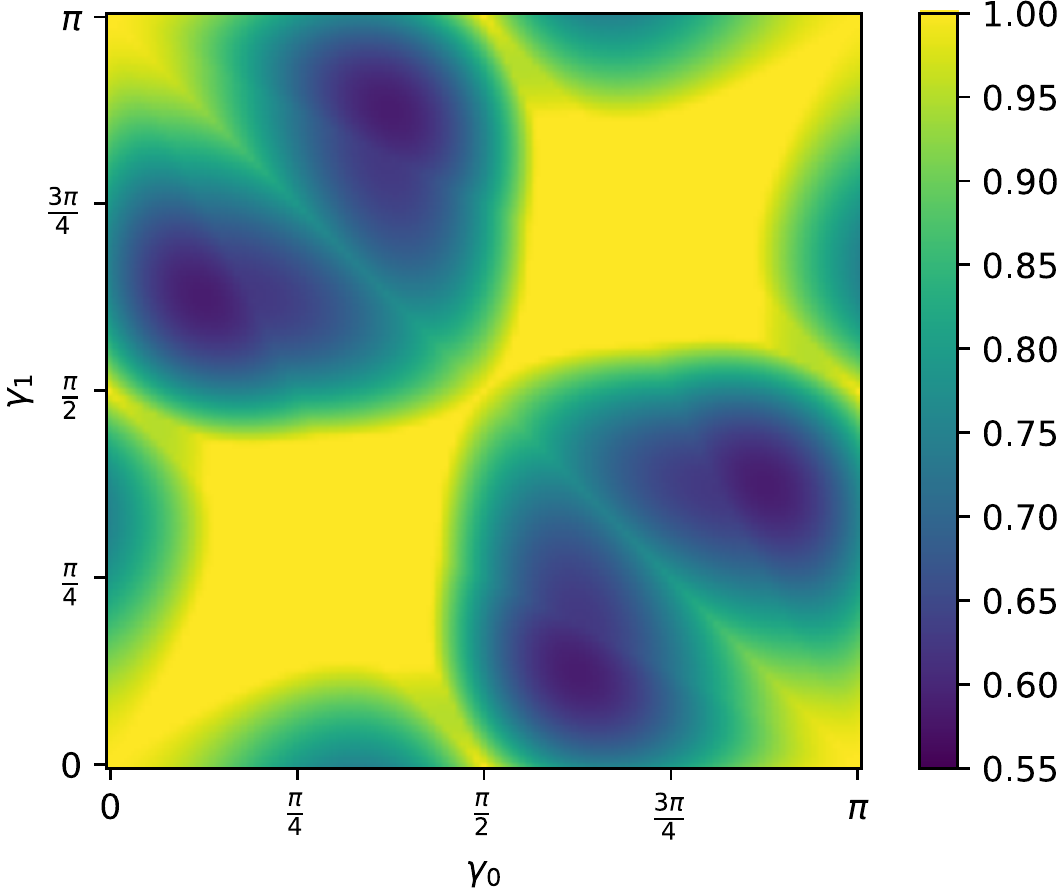}
    &
    \includegraphics[height=4cm]{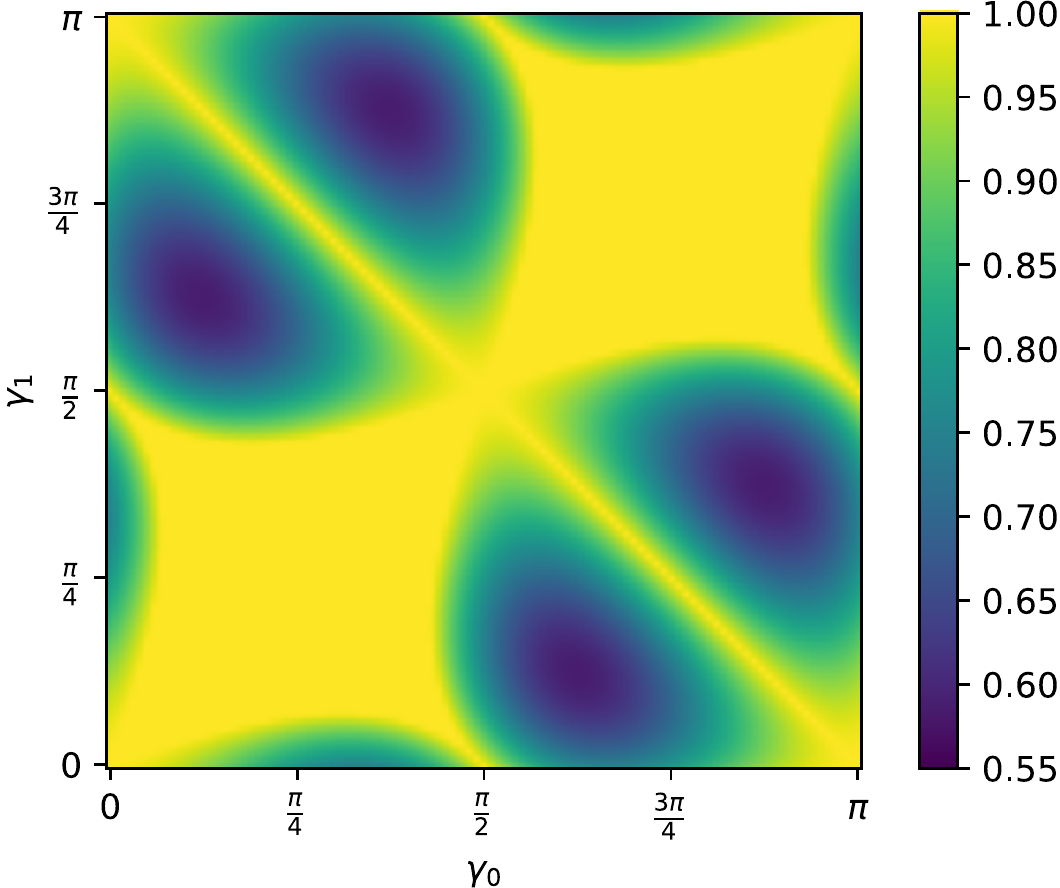}
    \\
    causally separable fraction
    &
    local fraction
    &
    local fraction (Bell exp.)
    \end{tabular}
\end{center}
At the angular resolution of $\frac{\pi}{100}$ used by the plots, the minimum causally separable fraction of \char`\~62.3\% is achieved at $\gamma_0 = \frac{18\pi}{100}$ and $\gamma_1 = \frac{63\pi}{100}$, close to the $\gamma_0 = \frac{\pi}{5}, \gamma_1 = \frac{3\pi}{5}$ values used in the empirical model above.
For comparison, the minimum local fraction for the empirical model is \char`\~58.6\%, and it coincides with the minimum local fraction of the corresponding Bell experiment.
The local fraction for the empirical model, computed over all 262144 causal functions (reduced to 8192 compatible with the empirical model support), coincides with the causally separable local fraction, computed over the 114688 causally separable functions (reduced to 4096 compatible with the empirical model support).

To better understand the differences between the three landscapes above, below we plot the difference between causally separable fraction and local fraction for the empirical model (left), as well as the difference between local fraction for the Bell experiment and local fraction for the empirical model (right).
\begin{center}
    \begin{tabular}{ccc}
    \includegraphics[height=4.5cm]{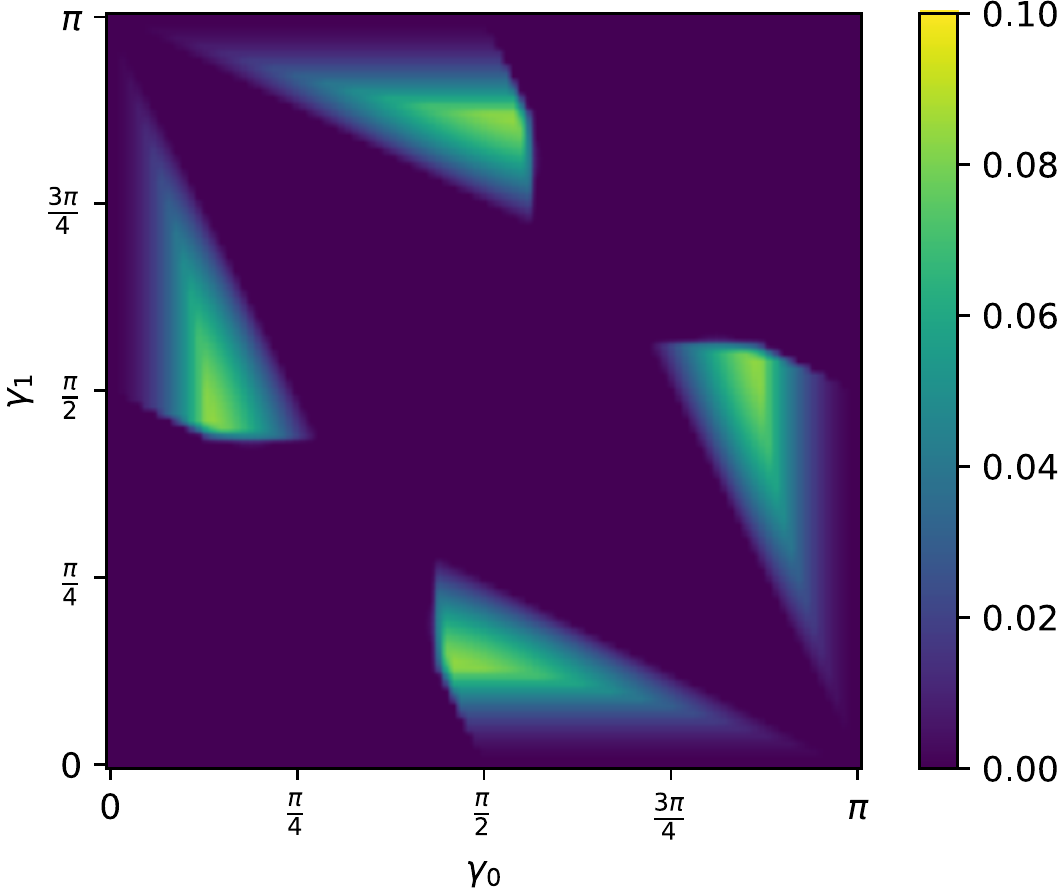}
    &
    \hspace{8mm}
    &
    \includegraphics[height=4.5cm]{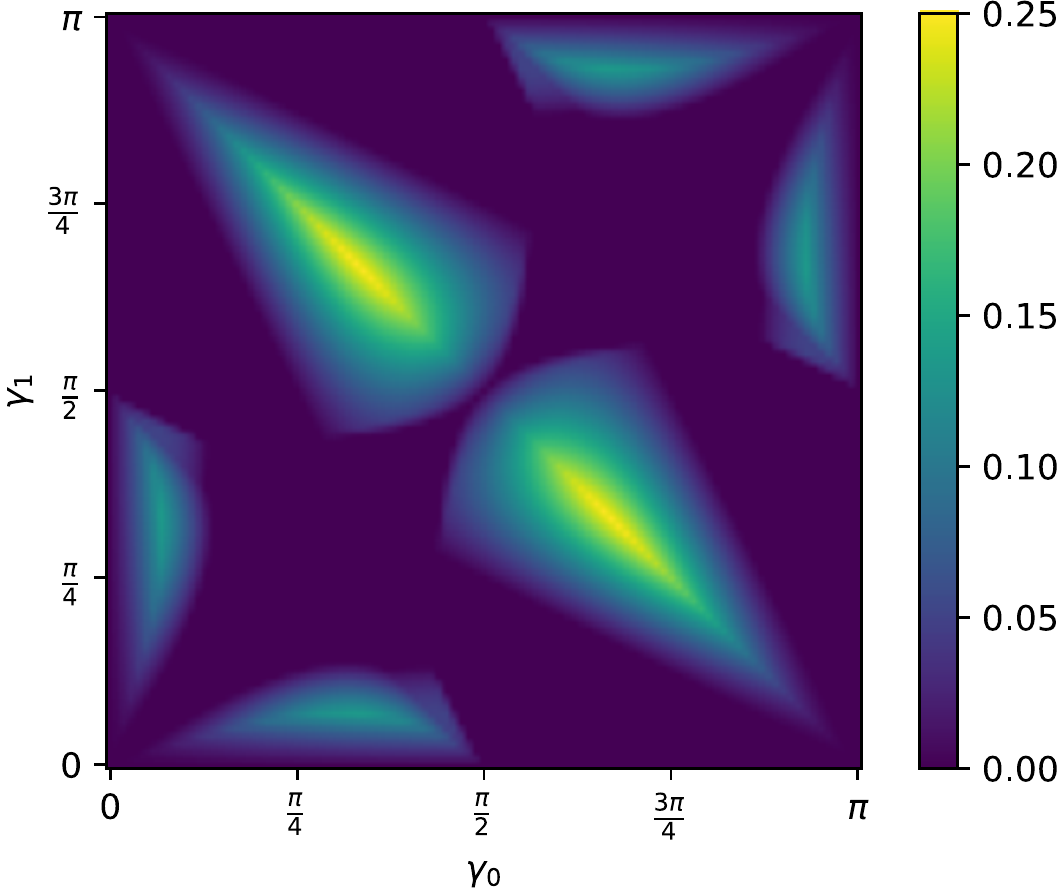}
    \\
    causally separable fraction - local fraction
    &
    &
    local fraction (Bell exp.) - local fraction
    \end{tabular}
\end{center}
In the left difference plot, we observe that the local fraction for the empirical model provides a lower bound for its causally separable fraction at all angles, as expected from Proposition \ref{proposition:csep-frac-bounded-below-by-csep-local-frac}.
However, we note that the bound is not tight for all angles, showing that the two fractions need not coincide even in the case of simple switch-based models.
In the right difference plot, we observe that the local fraction for the empirical model also provides a lower bound for the local fraction of the corresponding Bell experiment; in particular, we observe that the empirical model is non-local on the diagonal $\gamma_0 = \pi-\gamma_1$, where the corresponding Bell empirical model is always local.

The empirical model becomes causally separable if the no-signalling constraint to Diane is dropped: the model is 100\% supported by the spaces of input histories induced by causal orders \total{A,B,C,D} and \total{B,A,C,D}, the two causal completions of \total{\evset{A,B},C,D}, with each causal completion supporting exactly 50\% of the empirical model.

\subsubsection{A Quantum Switch with GHZ Entangled Control.}

In this example we analyse a variation of the previous example where the control qubit for the quantum switch is now part of a 3-partite state, namely a GHZ state in the X basis:
\[
\frac{1}{\sqrt{2}}\left(|\!+\!++\rangle+|\!-\!--\rangle\right)
= \frac{1}{2}\sum_{b_0 \oplus b_1 \oplus b_2 = 0} |b_0 b_1 b_2\rangle
\]
The setup is identical to that of the previous example, except that now there are two qubits entangled with the control qubit, one measured by Diane---as in the previous example---and one measured by Eve.
Note that the choice of creating the GHZ state in the X basis is consistent with the entangled state used in the previous example:
\[
|\Phi^+\rangle
=
\frac{1}{\sqrt{2}}\left(|\!+\!+\rangle+|\!-\!-\rangle\right)
=
\frac{1}{\sqrt{2}}\sum_{b_0 \oplus b_1 = 0} |b_0 b_1\rangle
\]
The figure below exemplifies the scenario we have just described:
\begin{center}
\scalebox{1.25}{
    \begin{tikzpicture}
	\begin{pgfonlayer}{nodelayer}
		\node [style=none] (22) at (1.25, -2.5) {};
		\node [style=none] (37) at (-1, -2.25) {};
		\node [style=none] (38) at (-7.25, -2.25) {};
		\node [style=none] (39) at (-7.25, 2) {};
		\node [style=none] (40) at (-1, 2) {};
		\node [style=none] (41) at (-3.5, -2.25) {};
		\node [style=none] (42) at (-6.5, -0.75) {};
		\node [style=black dot] (43) at (-5.25, -0.75) {};
		\node [style=none] (44) at (-6.5, 0.5) {};
		\node [style=black dot] (45) at (-5.25, 0.5) {};
		\node [style=none] (46) at (-5.25, 1.25) {};
		\node [style=none] (47) at (-6.5, 3) {};
		\node [style=none] (48) at (-6.5, -3.25) {};
		\node [style=none] (54) at (-2, -0.75) {};
		\node [style=black dot] (55) at (-3.25, -0.75) {};
		\node [style=none] (56) at (-2, 0.5) {};
		\node [style=black dot] (57) at (-3.25, 0.5) {};
		\node [style=none] (58) at (-3.25, 1.25) {};
		\node [style=none] (59) at (-3.25, -1.5) {};
		\node [style=none] (60) at (-2, 3) {};
		\node [style=none] (61) at (-2, -3.25) {};
		\node [style=none] (62) at (-5.25, -1.5) {};
		\node [style=none] (67) at (-3.75, 1.25) {};
		\node [style=none] (68) at (-3.75, -1.5) {};
		\node [style=none] (69) at (-1.75, -1.5) {};
		\node [style=none] (70) at (-1.75, 1.25) {};
		\node [style=none] (71) at (-6.75, 1.25) {};
		\node [style=none] (72) at (-6.75, -1.5) {};
		\node [style=none] (73) at (-4.75, -1.5) {};
		\node [style=none] (74) at (-4.75, 1.25) {};
		\node [style=none] (80) at (-3.5, -3.25) {};
		\node [style=none] (81) at (-3.75, 3) {};
		\node [style=none] (82) at (-3.75, 2) {};
		\node [style=none] (85) at (2.25, -2.5) {};
		\node [style=none] (86) at (2.25, -3.5) {};
		\node [style=none] (87) at (-1.5, -5) {};
		\node [style=none] (88) at (2.25, -5) {};
		\node [style=small box] (89) at (-3.75, 4.75) {$R_x(-\gamma_b)$};
		\node [style=none] (90) at (2.25, -1.5) {};
		\node [style=none] (91) at (2.25, -2.5) {};
		\node [style=small box] (92) at (2.25, 0.25) {$R_x(-\gamma_b)$};
		\node [style=none] (93) at (-3.5, 4.5) {};
		\node [style=none] (94) at (2, 0) {};
		\node [style=none] (95) at (-2.75, 3) {};
		\node [style=none] (96) at (1.25, -1.5) {};
		\node [style=eslabel] (97) at (-2.75, 2.75) {$b$};
		\node [style=eslabel] (98) at (1.25, -1.75) {$b$};
		\node [style=white dot] (99) at (-3.75, 5.75) {};
		\node [style=white dot] (100) at (2.25, 1.25) {};
		\node [style=none] (101) at (-5, -2.25) {};
		\node [style=point] (102) at (-5, -3.25) {$+$};
		\node [style=none] (103) at (-5.925, 2) {};
		\node [style=upground] (104) at (-5.925, 3) {};
		\node [style=none] (109) at (0.75, 1.75) {};
		\node [style=none] (110) at (0.75, -2) {};
		\node [style=none] (111) at (3.75, -2) {};
		\node [style=none] (112) at (3.75, 1.75) {};
		\node [style=none] (113) at (-5.25, 6.25) {};
		\node [style=none] (114) at (-5.25, 2.5) {};
		\node [style=none] (115) at (-2.25, 2.5) {};
		\node [style=none] (116) at (-2.25, 6.25) {};
		\node [style=red label] (117) at (-5.875, 4.75) {$C$};
		\node [style=red label] (118) at (4.25, 0.25) {$D$};
		\node [style=red label] (119) at (-7.75, 0) {};
		\node [style=red label] (120) at (-0.5, 0) {};
		\node [style=red label] (124) at (-7.75, 0) {$A$};
		\node [style=red label] (125) at (-0.5, 0) {$B$};
		\node [style=none] (126) at (-3.75, 7.25) {};
		\node [style=none] (127) at (2.25, 6.5) {};
		\node [style=none] (128) at (6.5, -2.5) {};
		\node [style=none] (129) at (7.5, -2.5) {};
		\node [style=none] (130) at (7.5, -3.5) {};
		\node [style=none] (131) at (6.5, -5) {};
		\node [style=none] (132) at (7.5, -1.5) {};
		\node [style=none] (133) at (7.5, -2.5) {};
		\node [style=small box] (134) at (7.5, 0.25) {$R_x(-\gamma_b)$};
		\node [style=none] (135) at (7.25, 0) {};
		\node [style=none] (136) at (6.5, -1.5) {};
		\node [style=eslabel] (137) at (6.5, -1.75) {$b$};
		\node [style=white dot] (138) at (7.5, 1.25) {};
		\node [style=none] (139) at (6, 1.75) {};
		\node [style=none] (140) at (6, -2) {};
		\node [style=none] (141) at (9, -2) {};
		\node [style=none] (142) at (9, 1.75) {};
		\node [style=red label] (143) at (9.5, 0.25) {$E$};
		\node [style=none] (144) at (7.5, 6.5) {};
		\node [style=black ddot] (145) at (2.25, -7) {};
	\end{pgfonlayer}
	\begin{pgfonlayer}{edgelayer}
		\draw [style=pink thin] (39.center)
			 to (38.center)
			 to (37.center)
			 to (40.center)
			 to cycle;
		\draw [style=classical, in=90, out=-90] (44.center) to (43);
		\draw [style=classical, in=90, out=-90] (45) to (42.center);
		\draw (46.center) to (45);
		\draw [style=classical] (42.center) to (48.center);
		\draw [style=classical] (44.center) to (47.center);
		\draw [style=classical, in=90, out=-90] (56.center) to (55);
		\draw [style=classical, in=90, out=-90] (57) to (54.center);
		\draw (58.center) to (57);
		\draw (55) to (59.center);
		\draw [style=classical] (54.center) to (61.center);
		\draw [style=classical] (56.center) to (60.center);
		\draw [style=carta da zucchero thin] (68.center)
			 to (67.center)
			 to (70.center)
			 to (69.center)
			 to cycle;
		\draw [style=carta da zucchero thin] (72.center)
			 to (71.center)
			 to (74.center)
			 to (73.center)
			 to cycle;
		\draw [style=classical, in=90, out=-90] (44.center) to (43);
		\draw [style=classical, in=90, out=-90] (45) to (42.center);
		\draw [style=classical, in=90, out=-90] (56.center) to (55);
		\draw [style=classical, in=90, out=-90] (57) to (54.center);
		\draw [style=bold black] (46.center) to (45);
		\draw [style=bold black] (55) to (59.center);
		\draw [style=bold black] (58.center) to (57);
		\draw [style=bold black] (62.center) to (43);
		\draw [style=bold black] (80.center) to (41.center);
		\draw [style=bold black] (82.center) to (81.center);
		\draw [style=bold black] (86.center) to (85.center);
		\draw [style=classical] (42.center) to (48.center);
		\draw [style=classical] (44.center) to (47.center);
		\draw [style=classical] (54.center) to (61.center);
		\draw [style=classical] (56.center) to (60.center);
		\draw [style=bold black, in=-90, out=90, looseness=1.25] (87.center) to (80.center);
		\draw [style=bold black, in=-90, out=90, looseness=1.25] (88.center) to (86.center);
		\draw [style=bold black] (81.center) to (89);
		\draw [style=bold black] (91.center) to (90.center);
		\draw [style=bold black] (90.center) to (92);
		\draw [style=classical, in=90, out=-90, looseness=1.25] (93.center) to (95.center);
		\draw [style=classical, in=90, out=-90] (94.center) to (96.center);
		\draw [style=bold black] (89) to (99);
		\draw [style=bold black] (92) to (100);
		\draw [style=bold black] (103.center) to (104);
		\draw [style=bold black] (102) to (101.center);
		\draw [style=carta da zucchero thin] (110.center)
			 to (109.center)
			 to (112.center)
			 to (111.center)
			 to cycle;
		\draw [style=carta da zucchero thin] (116.center)
			 to (115.center)
			 to (114.center)
			 to (113.center)
			 to cycle;
		\draw [style=classical] (99) to (126.center);
		\draw [style=classical] (100) to (127.center);
		\draw [style=bold black] (130.center) to (129.center);
		\draw [style=bold black, in=-90, out=90, looseness=1.25] (131.center) to (130.center);
		\draw [style=bold black] (133.center) to (132.center);
		\draw [style=bold black] (132.center) to (134);
		\draw [style=classical, in=90, out=-90] (135.center) to (136.center);
		\draw [style=bold black] (134) to (138);
		\draw [style=carta da zucchero thin] (140.center)
			 to (139.center)
			 to (142.center)
			 to (141.center)
			 to cycle;
		\draw [style=classical] (138) to (144.center);
		\draw [style=bold black] (145) to (88.center);
		\draw [style=bold black, in=-90, out=0] (145) to (131.center);
		\draw [style=bold black, in=-90, out=180] (145) to (87.center);
	\end{pgfonlayer}
\end{tikzpicture}
}
\end{center}
For $\gamma_0 = \frac{7\pi}{25}, \gamma_1 = \frac{3\pi}{5}$, the description above results in the following empirical model:
\begin{center}
    \includegraphics[height=10cm]{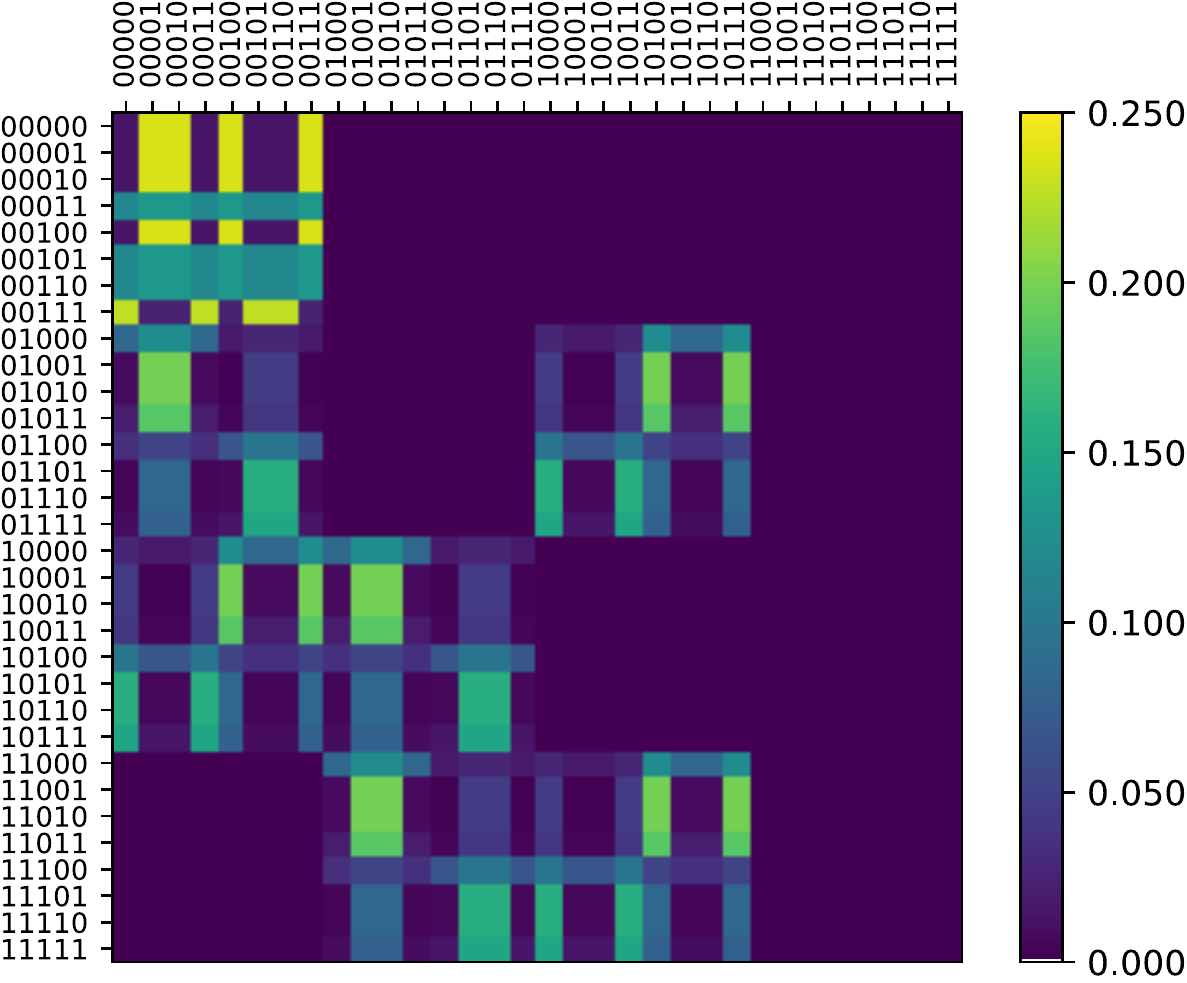}
\end{center}
By construction, the empirical model is 100\% supported by the space of input histories induced by the following indefinite causal order (whose standard causaltope is 404-dimensional):
\begin{center}
        \includegraphics[height=3.5cm]{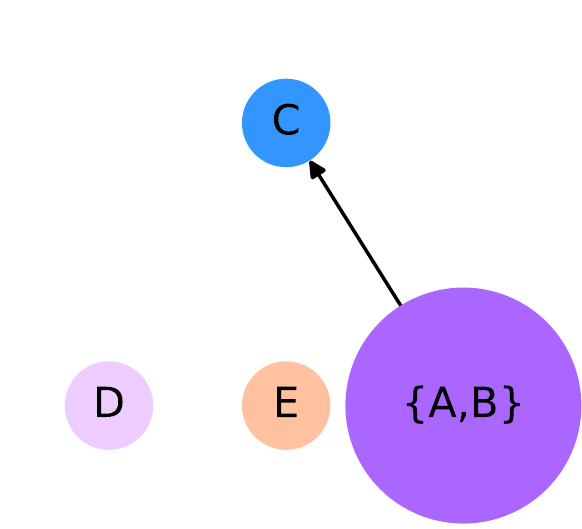}
\end{center}
The empirical model is causally inseparable for the space right above, with a causally separable fraction of \char`\~41.2\% over its two causal completions.
The two completions are induced by the causal orders $\total{A, B, C}\vee\discrete{D}\vee\discrete{E}$ and $\total{B, A, C}\vee\discrete{D}\vee\discrete{E}$, and their causaltopes are 386-dimensional:
\begin{center}
    \begin{tabular}{ccc}
    \includegraphics[height=3.5cm]{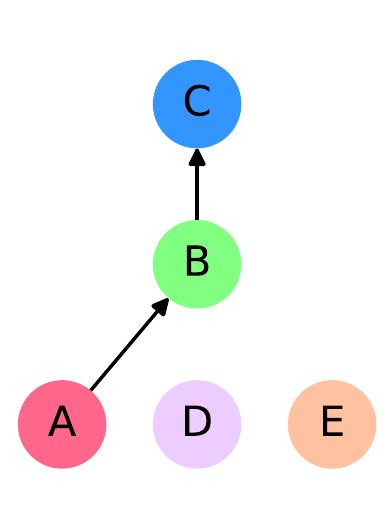}
    &
    \hspace{1cm}
    &
    \includegraphics[height=3.5cm]{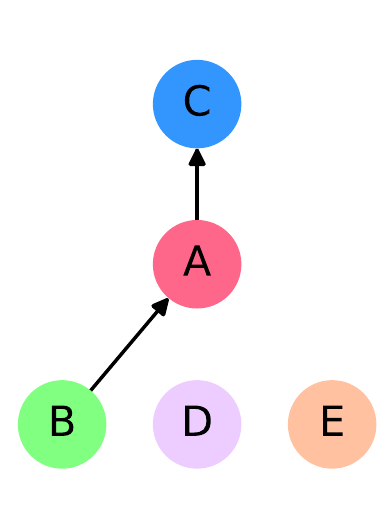}
    \\
    Causal fraction: \char`\~20.6\%
    &
    &
    Causal fraction: \char`\~20.6\%
    \end{tabular}
\end{center}
The empirical model becomes causally separable if the no-signalling constraint to Diane and Eve is dropped: the model is 100\% supported by the spaces of input histories induced by causal orders \total{A,B,C,D,E} and \total{B,A,C,D,E}, the two causal completions of \total{\evset{A,B},C,D,E}, with each causal completion supporting exactly 50\% of the empirical model.

Below we plot the causally separable fraction (left) and the local fraction (middle) for this empirical model as a function of the $\gamma_0$ and $\gamma_1$ measurement angles used by Charlie, Diane and Eve.
We compare it to the local fraction for a GHZ experiment with the same measurements (right), as originally shown by Figure 1(b) of \cite{abramsky2017contextual}.
\begin{center}
    \begin{tabular}{ccc}
    \includegraphics[height=4.cm]{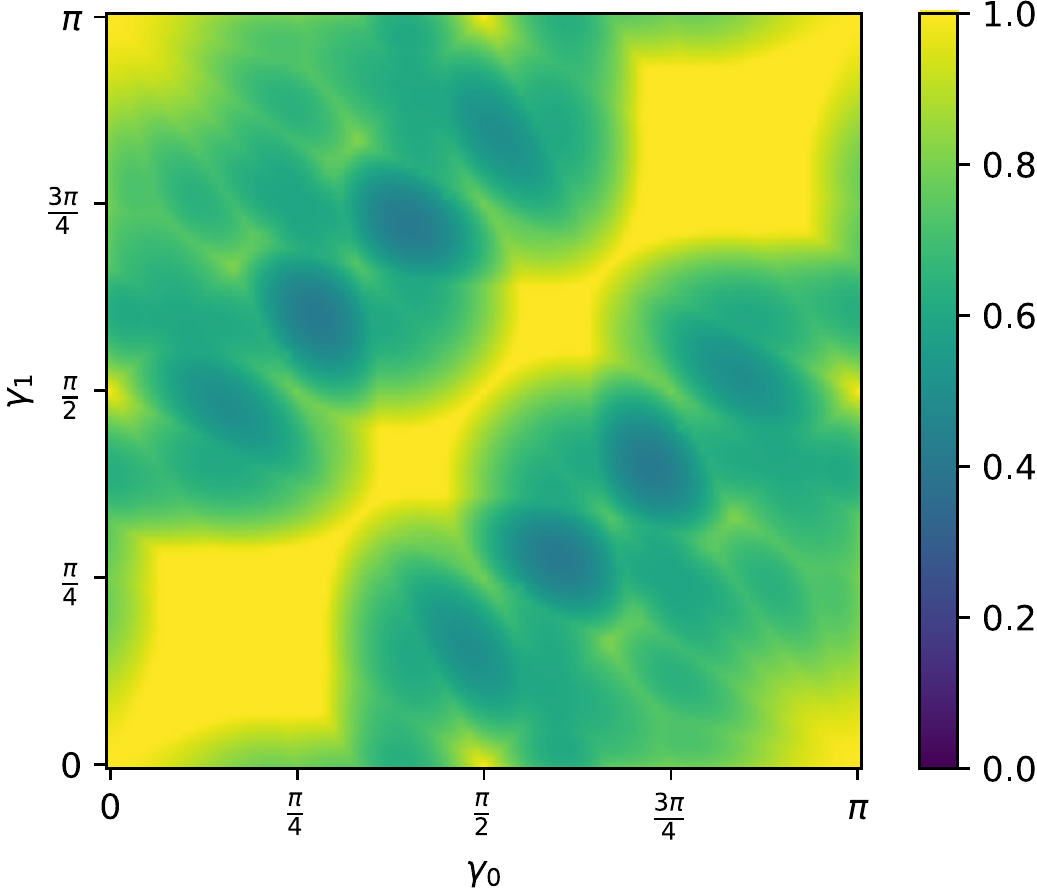}
    &
    \includegraphics[height=4cm]{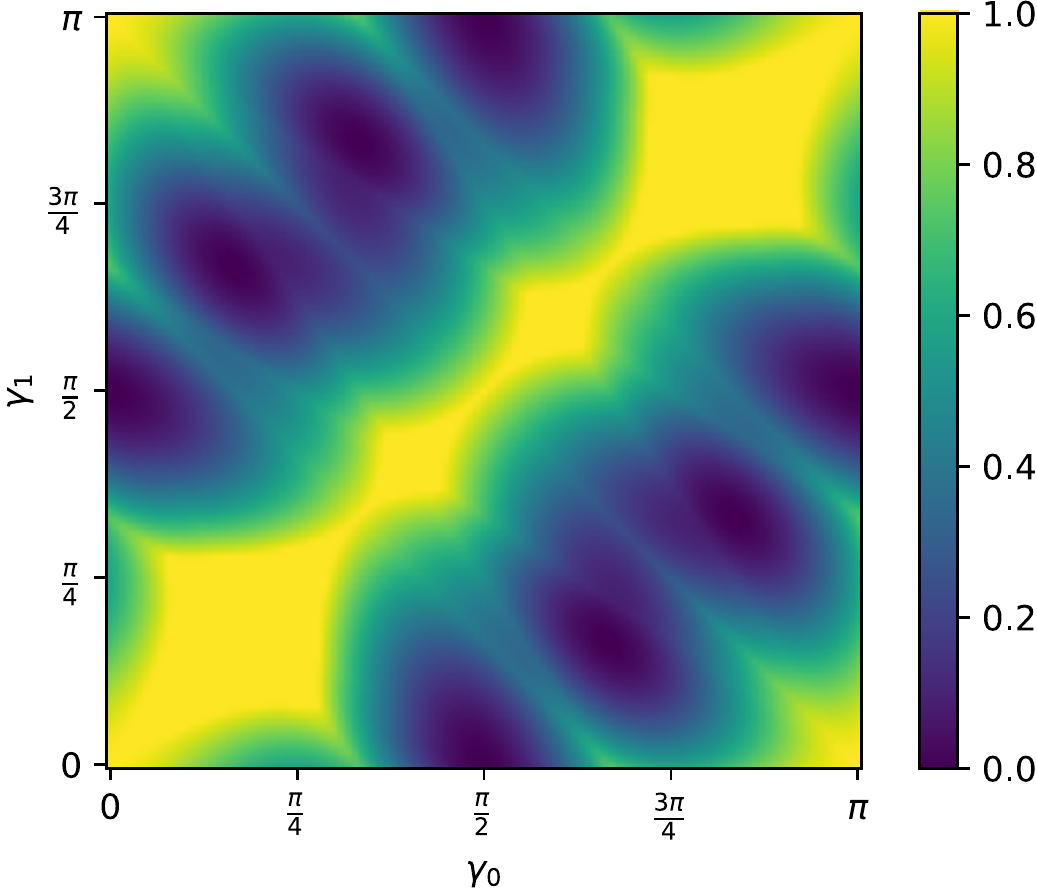}
    &
    \includegraphics[height=4cm]{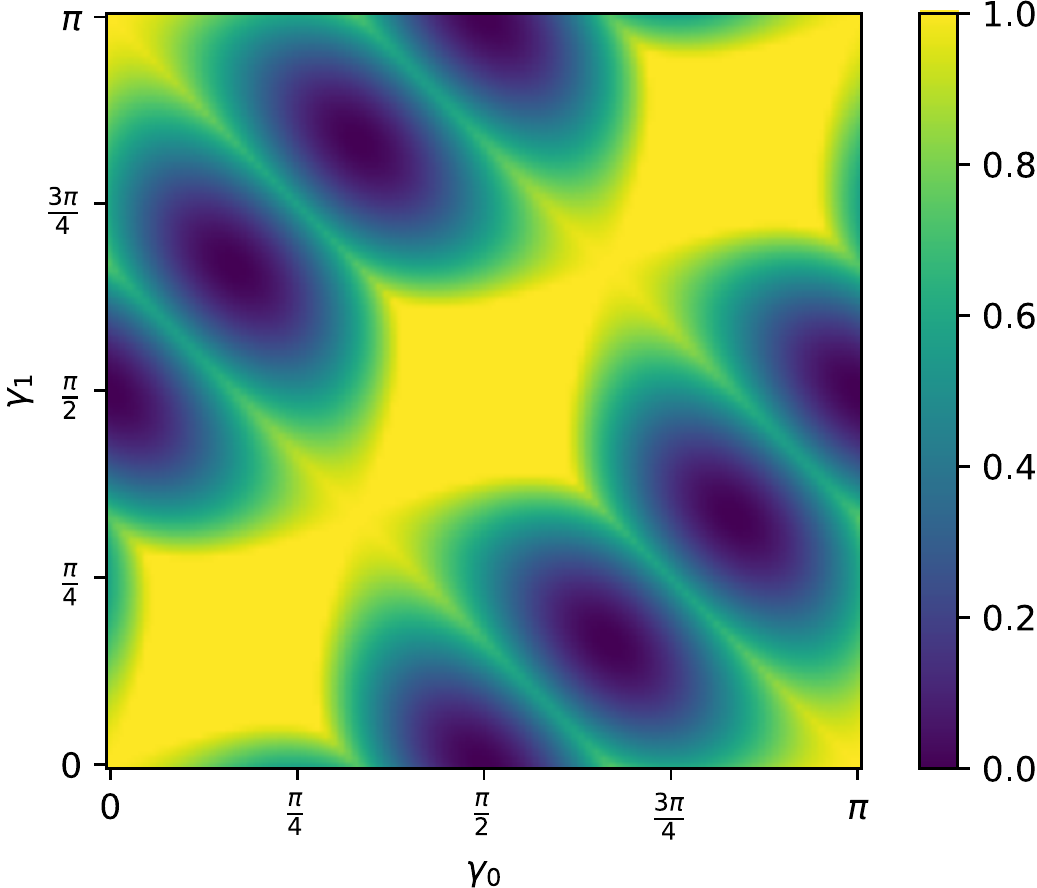}
    \\
    causally separable fraction
    &
    local fraction
    &
    local fraction (GHZ exp.)
    \end{tabular}
\end{center}
At the angular resolution of $\frac{\pi}{100}$ used by the plots, the minimum causally separable fraction of \char`\~41.2\% is achieved at $\gamma_0 = \frac{7\pi}{25}$ and $\gamma_1 = \frac{3\pi}{5}$.
For comparison, the minimum local fraction for the empirical model is 0\%, and it coincides with the minimum local fraction for the corresponding GHZ experiment.
The local fraction for the empirical model, computed over all 1048576 causal functions (reduced to 32768 compatible with the empirical model support), coincides with the causally separable local fraction, computed over the 458752 causally separable functions (reduced to 16384 compatible with the empirical model support).

To better understand the differences between the three landscapes above, below we plot the difference between causally separable fraction and local fraction for the empirical model (left), as well as the difference between local fraction for the GHZ experiment and local fraction for the empirical model (right).
\begin{center}
    \begin{tabular}{ccc}
    \includegraphics[height=4.5cm]{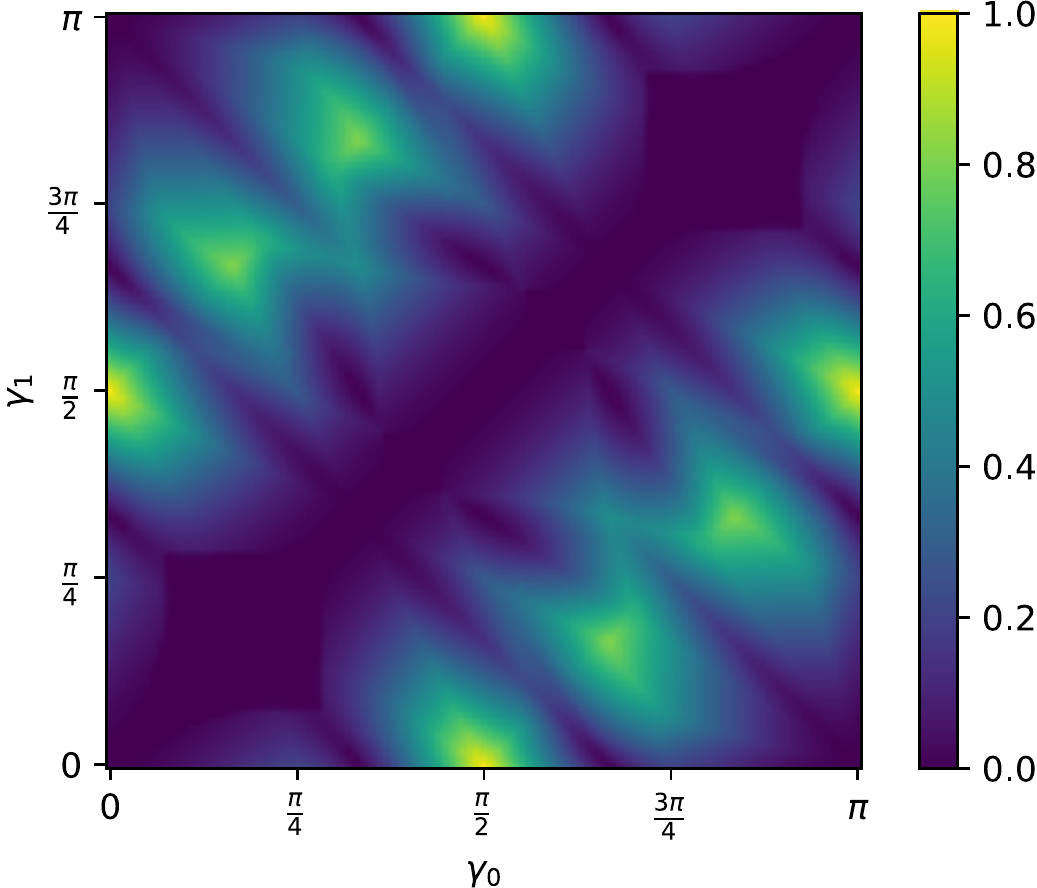}
    &
    \hspace{8mm}
    &
    \includegraphics[height=4.5cm]{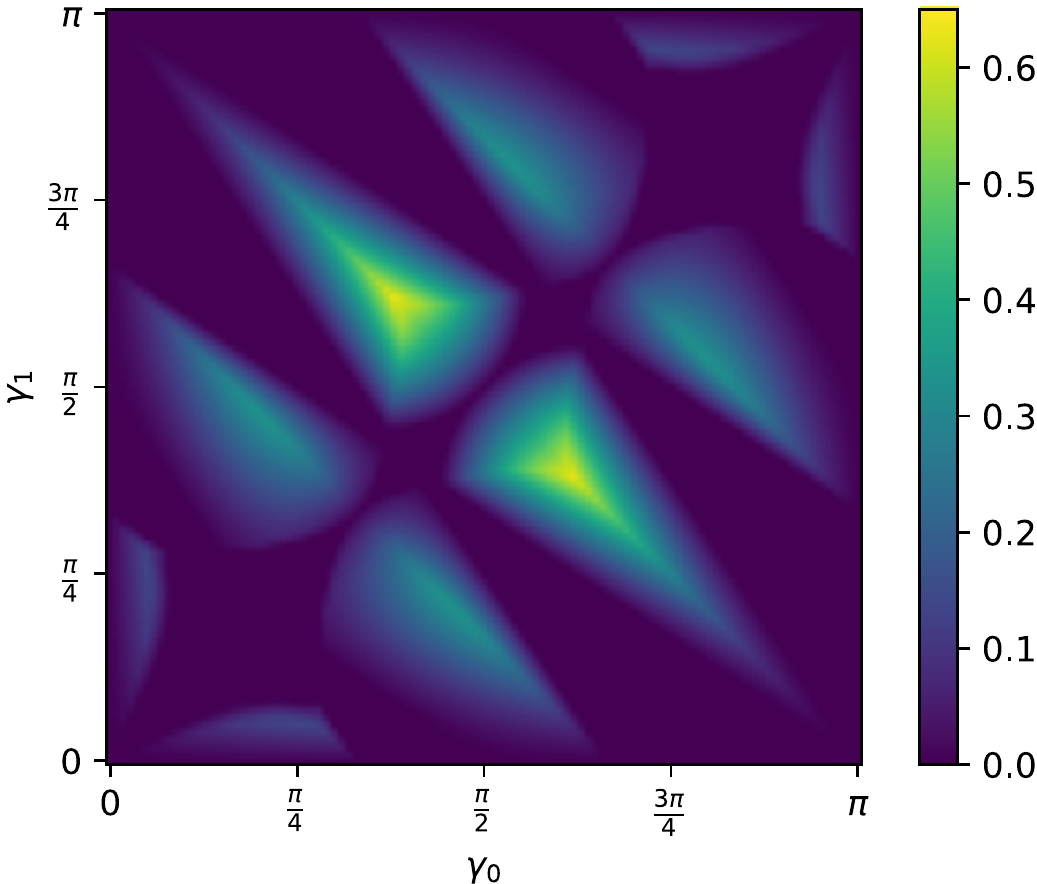}
    \\
    causally separable fraction - local fraction
    &
    &
    local fraction (Bell exp.) - local fraction
    \end{tabular}
\end{center}
In the left difference plot, we observe that the local fraction for the empirical model provides a lower bound for its causally separable fraction at all angles, as expected from Proposition \ref{proposition:csep-frac-bounded-below-by-csep-local-frac}.
As with our previous example, we note that the bound is not tight for all angles.
In the right difference plot, we observe that the local fraction for the empirical model provides a lower bound for the local fraction of the corresponding Bell experiment.

\subsubsection{Two Quantum Switches with Bell Entangled Control.}

In this example we consider two quantum switches---one between Alice and Bob, the other between Charlie and Diane---with entangled controls, where Eve measures the control of the Alice/Bob switch and Felix measures the control of the Charlie/Diane switch.
Specifically, for angles $\gamma_0, \gamma_1 \in [0, \pi)$:
\begin{enumerate}
    \item A 2-qubit Bell state $|\Phi^+\rangle$ is created: one qubit is sent to the control of the Alice/Bob switch, the other is sent to the control of the Charlie/Diane switch.
    \item Alice and Bob are in the left quantum switch, with the first of the two $|\Phi^+\rangle$ qubits as its control and the $|+\rangle$ state as its input.
    \item Charlie and Diane are in the right quantum switch, with second of the two $|\Phi^+\rangle$ qubits as its control and the $|+\rangle$ state as its input.
    \item Alice, Bob, Charlie and Diane all do the same thing: they perform an X measurement on the incoming qubit, using the measurement outcome as their individual output, and then encode their individual input into the X basis of the outgoing qubit.
    \item The output qubit of each switch is discarded. Eve and Felix receive the control qubit of the left and right switch respectively, and do the same thing: on input 0, they apply a X rotation by $-\gamma_0$ and then measure in the Z basis; on input 1, they first apply a X rotation by $-\gamma_1$ and then measure in the Z basis.
\end{enumerate}
The figure below exemplifies the scenario we have just described:
\begin{center}
\scalebox{1.25}{
    \begin{tikzpicture}
	\begin{pgfonlayer}{nodelayer}
		\node [style=none] (37) at (-1, -2.25) {};
		\node [style=none] (38) at (-7.25, -2.25) {};
		\node [style=none] (39) at (-7.25, 2) {};
		\node [style=none] (40) at (-1, 2) {};
		\node [style=none] (41) at (-3.5, -2.25) {};
		\node [style=none] (42) at (-6.5, -0.75) {};
		\node [style=black dot] (43) at (-5.25, -0.75) {};
		\node [style=none] (44) at (-6.5, 0.5) {};
		\node [style=black dot] (45) at (-5.25, 0.5) {};
		\node [style=none] (46) at (-5.25, 1.25) {};
		\node [style=none] (47) at (-6.5, 3) {};
		\node [style=none] (48) at (-6.5, -3.25) {};
		\node [style=none] (54) at (-2, -0.75) {};
		\node [style=black dot] (55) at (-3.25, -0.75) {};
		\node [style=none] (56) at (-2, 0.5) {};
		\node [style=black dot] (57) at (-3.25, 0.5) {};
		\node [style=none] (58) at (-3.25, 1.25) {};
		\node [style=none] (59) at (-3.25, -1.5) {};
		\node [style=none] (60) at (-2, 3) {};
		\node [style=none] (61) at (-2, -3.25) {};
		\node [style=none] (62) at (-5.25, -1.5) {};
		\node [style=none] (67) at (-3.75, 1.25) {};
		\node [style=none] (68) at (-3.75, -1.5) {};
		\node [style=none] (69) at (-1.75, -1.5) {};
		\node [style=none] (70) at (-1.75, 1.25) {};
		\node [style=none] (71) at (-6.75, 1.25) {};
		\node [style=none] (72) at (-6.75, -1.5) {};
		\node [style=none] (73) at (-4.75, -1.5) {};
		\node [style=none] (74) at (-4.75, 1.25) {};
		\node [style=none] (80) at (-3.5, -3.25) {};
		\node [style=none] (86) at (3.5, -3.25) {};
		\node [style=none] (87) at (-1.5, -5) {};
		\node [style=none] (88) at (1.5, -5) {};
		\node [style=none] (101) at (-5, -2.25) {};
		\node [style=point] (102) at (-5, -3.25) {$+$};
		\node [style=none] (103) at (-5.75, 2) {};
		\node [style=upground] (104) at (-5.75, 3) {};
		\node [style=none] (113) at (-5.2, 6.25) {};
		\node [style=none] (114) at (-5.2, 2.5) {};
		\node [style=none] (115) at (-2.25, 2.5) {};
		\node [style=none] (116) at (-2.25, 6.25) {};
		\node [style=red label] (117) at (-5.75, 4.75) {$E$};
		\node [style=red label] (119) at (-7.75, 0) {};
		\node [style=red label] (120) at (-0.5, 0) {};
		\node [style=red label] (124) at (-7.75, 0) {$A$};
		\node [style=red label] (125) at (-0.5, 0) {$B$};
		\node [style=none] (126) at (1, -2.25) {};
		\node [style=none] (127) at (7.25, -2.25) {};
		\node [style=none] (128) at (7.25, 2) {};
		\node [style=none] (129) at (1, 2) {};
		\node [style=none] (130) at (3.5, -2.25) {};
		\node [style=none] (131) at (6.5, -0.75) {};
		\node [style=black dot] (132) at (5.25, -0.75) {};
		\node [style=none] (133) at (6.5, 0.5) {};
		\node [style=black dot] (134) at (5.25, 0.5) {};
		\node [style=none] (135) at (5.25, 1.25) {};
		\node [style=none] (136) at (6.5, 3) {};
		\node [style=none] (137) at (6.5, -3.25) {};
		\node [style=none] (138) at (2, -0.75) {};
		\node [style=black dot] (139) at (3.25, -0.75) {};
		\node [style=none] (140) at (2, 0.5) {};
		\node [style=black dot] (141) at (3.25, 0.5) {};
		\node [style=none] (142) at (3.25, 1.25) {};
		\node [style=none] (143) at (3.25, -1.5) {};
		\node [style=none] (144) at (2, 3) {};
		\node [style=none] (145) at (2, -3.25) {};
		\node [style=none] (146) at (5.25, -1.5) {};
		\node [style=none] (147) at (3.75, 1.25) {};
		\node [style=none] (148) at (3.75, -1.5) {};
		\node [style=none] (149) at (1.75, -1.5) {};
		\node [style=none] (150) at (1.75, 1.25) {};
		\node [style=none] (151) at (6.75, 1.25) {};
		\node [style=none] (152) at (6.75, -1.5) {};
		\node [style=none] (153) at (4.75, -1.5) {};
		\node [style=none] (154) at (4.75, 1.25) {};
		\node [style=none] (155) at (3.5, -3.25) {};
		\node [style=none] (156) at (3.75, 3) {};
		\node [style=none] (157) at (3.75, 2) {};
		\node [style=small box] (158) at (3.75, 4.75) {$R_x(-\gamma_b)$};
		\node [style=none] (159) at (3.5, 4.5) {};
		\node [style=none] (160) at (2.75, 3) {};
		\node [style=eslabel] (161) at (2.75, 2.75) {$b$};
		\node [style=white dot] (162) at (3.75, 5.75) {};
		\node [style=none] (163) at (5, -2.25) {};
		\node [style=point] (164) at (5, -3.25) {$+$};
		\node [style=none] (165) at (5.75, 2) {};
		\node [style=upground] (166) at (5.75, 3) {};
		\node [style=none] (167) at (5.225, 6.25) {};
		\node [style=none] (168) at (5.225, 2.5) {};
		\node [style=none] (169) at (2.25, 2.5) {};
		\node [style=none] (170) at (2.25, 6.25) {};
		\node [style=red label] (171) at (5.75, 4.75) {$F$};
		\node [style=red label] (172) at (7.75, 0) {};
		\node [style=red label] (173) at (0.5, 0) {};
		\node [style=red label] (174) at (7.75, 0) {$D$};
		\node [style=red label] (175) at (0.5, 0) {$C$};
		\node [style=none] (176) at (3.75, 6.75) {};
		\node [style=none] (177) at (-3.75, 3) {};
		\node [style=none] (178) at (-3.75, 2) {};
		\node [style=small box] (179) at (-3.75, 4.75) {$R_x(-\gamma_b)$};
		\node [style=none] (180) at (-3.5, 4.5) {};
		\node [style=none] (181) at (-2.75, 3) {};
		\node [style=eslabel] (182) at (-2.75, 2.75) {$b$};
		\node [style=white dot] (183) at (-3.75, 5.75) {};
		\node [style=none] (184) at (-3.75, 6.75) {};
	\end{pgfonlayer}
	\begin{pgfonlayer}{edgelayer}
		\draw [style=pink thin] (39.center)
			 to (38.center)
			 to (37.center)
			 to (40.center)
			 to cycle;
		\draw [style=classical, in=90, out=-90] (44.center) to (43);
		\draw [style=classical, in=90, out=-90] (45) to (42.center);
		\draw (46.center) to (45);
		\draw [style=classical] (42.center) to (48.center);
		\draw [style=classical] (44.center) to (47.center);
		\draw [style=classical, in=90, out=-90] (56.center) to (55);
		\draw [style=classical, in=90, out=-90] (57) to (54.center);
		\draw (58.center) to (57);
		\draw (55) to (59.center);
		\draw [style=classical] (54.center) to (61.center);
		\draw [style=classical] (56.center) to (60.center);
		\draw [style=carta da zucchero thin] (68.center)
			 to (67.center)
			 to (70.center)
			 to (69.center)
			 to cycle;
		\draw [style=carta da zucchero thin] (72.center)
			 to (71.center)
			 to (74.center)
			 to (73.center)
			 to cycle;
		\draw [style=classical, in=90, out=-90] (44.center) to (43);
		\draw [style=classical, in=90, out=-90] (45) to (42.center);
		\draw [style=classical, in=90, out=-90] (56.center) to (55);
		\draw [style=classical, in=90, out=-90] (57) to (54.center);
		\draw [style=bold black] (46.center) to (45);
		\draw [style=bold black] (55) to (59.center);
		\draw [style=bold black] (58.center) to (57);
		\draw [style=bold black] (62.center) to (43);
		\draw [style=bold black] (80.center) to (41.center);
		\draw [style=classical] (42.center) to (48.center);
		\draw [style=classical] (44.center) to (47.center);
		\draw [style=classical] (54.center) to (61.center);
		\draw [style=classical] (56.center) to (60.center);
		\draw [style=bold black, in=-90, out=90, looseness=1.25] (87.center) to (80.center);
		\draw [style=bold black, in=-90, out=-90] (87.center) to (88.center);
		\draw [style=bold black, in=-90, out=90, looseness=1.25] (88.center) to (86.center);
		\draw [style=bold black] (103.center) to (104);
		\draw [style=bold black] (102) to (101.center);
		\draw [style=carta da zucchero thin] (116.center)
			 to (115.center)
			 to (114.center)
			 to (113.center)
			 to cycle;
		\draw [style=pink thin] (129.center)
			 to (128.center)
			 to (127.center)
			 to (126.center)
			 to cycle;
		\draw [style=classical, in=90, out=-90] (133.center) to (132);
		\draw [style=classical, in=90, out=-90] (134) to (131.center);
		\draw (135.center) to (134);
		\draw [style=classical] (131.center) to (137.center);
		\draw [style=classical] (133.center) to (136.center);
		\draw [style=classical, in=90, out=-90] (140.center) to (139);
		\draw [style=classical, in=90, out=-90] (141) to (138.center);
		\draw (142.center) to (141);
		\draw (139) to (143.center);
		\draw [style=classical] (138.center) to (145.center);
		\draw [style=classical] (140.center) to (144.center);
		\draw [style=carta da zucchero thin] (149.center)
			 to (148.center)
			 to (147.center)
			 to (150.center)
			 to cycle;
		\draw [style=carta da zucchero thin] (151.center)
			 to (154.center)
			 to (153.center)
			 to (152.center)
			 to cycle;
		\draw [style=classical, in=90, out=-90] (133.center) to (132);
		\draw [style=classical, in=90, out=-90] (134) to (131.center);
		\draw [style=classical, in=90, out=-90] (140.center) to (139);
		\draw [style=classical, in=90, out=-90] (141) to (138.center);
		\draw [style=bold black] (135.center) to (134);
		\draw [style=bold black] (139) to (143.center);
		\draw [style=bold black] (142.center) to (141);
		\draw [style=bold black] (146.center) to (132);
		\draw [style=bold black] (155.center) to (130.center);
		\draw [style=classical] (131.center) to (137.center);
		\draw [style=classical] (133.center) to (136.center);
		\draw [style=classical] (138.center) to (145.center);
		\draw [style=classical] (140.center) to (144.center);
		\draw [style=classical, in=90, out=-90, looseness=1.25] (159.center) to (160.center);
		\draw [style=bold black] (165.center) to (166);
		\draw [style=bold black] (164) to (163.center);
		\draw [style=carta da zucchero thin] (170.center)
			 to (169.center)
			 to (168.center)
			 to (167.center)
			 to cycle;
		\draw [style=bold black] (156.center) to (158);
		\draw [style=bold black] (157.center) to (156.center);
		\draw [style=bold black] (158) to (162);
		\draw [style=classical] (162) to (176.center);
		\draw [style=classical, in=90, out=-90, looseness=1.25] (180.center) to (181.center);
		\draw [style=bold black] (177.center) to (179);
		\draw [style=bold black] (178.center) to (177.center);
		\draw [style=bold black] (179) to (183);
		\draw [style=classical] (183) to (184.center);
	\end{pgfonlayer}
\end{tikzpicture}
}
\end{center}
For $\gamma_0 = \frac{\pi}{5}, \gamma_1 = \frac{3\pi}{5}$, the description above results in the following empirical model:
\begin{center}
    \includegraphics[height=12cm]{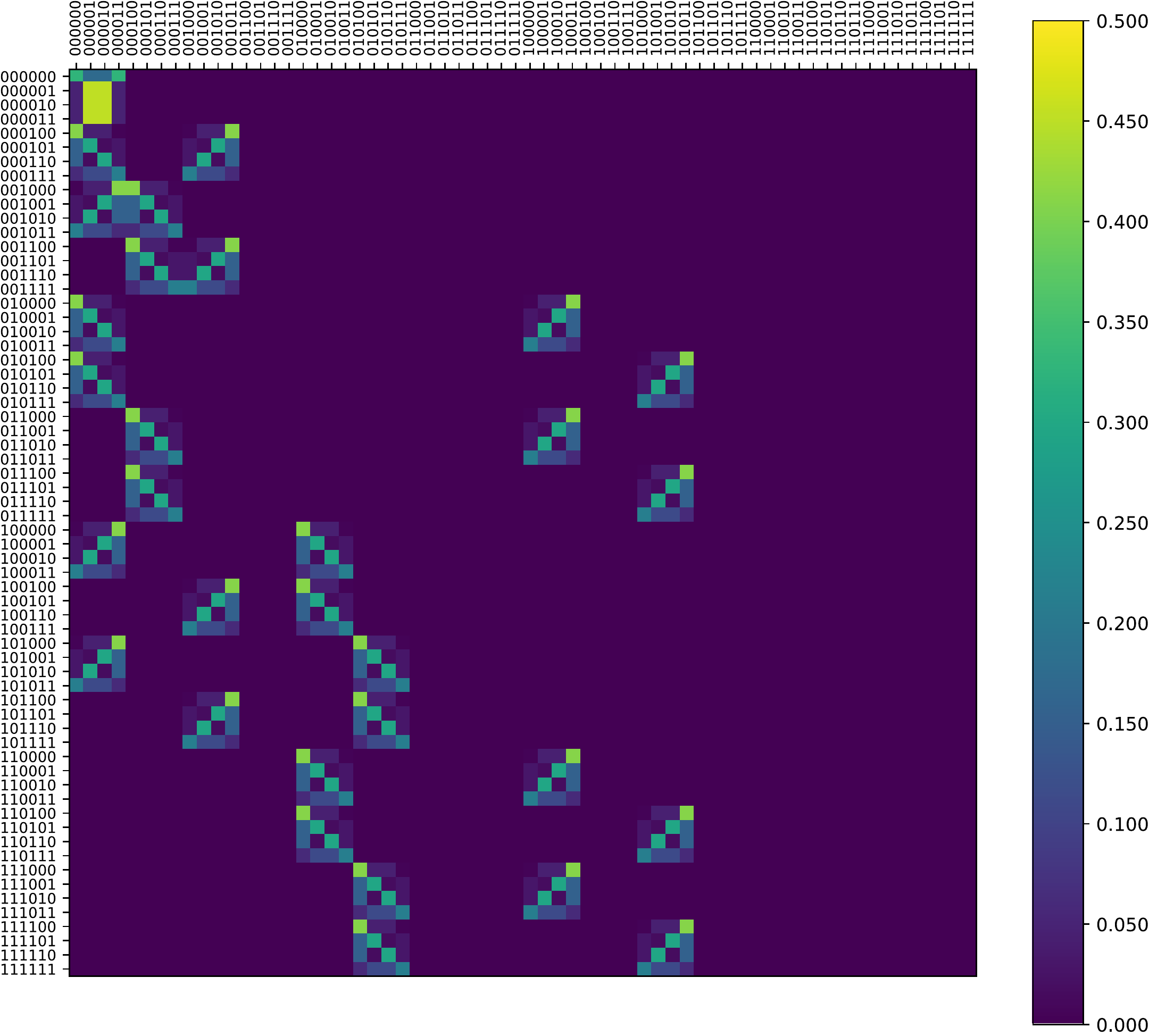}
\end{center}
By construction, the empirical model is 100\% supported by the space of input histories induced by the following indefinite causal order (whose standard causaltope is 2024-dimensional):
\begin{center}
        \includegraphics[height=3.5cm]{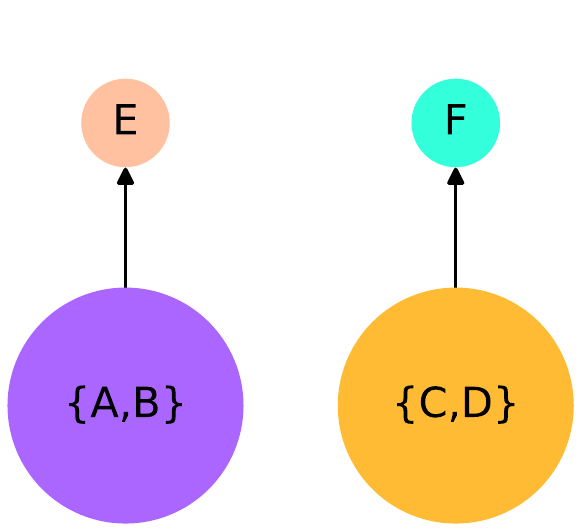}
\end{center}
The empirical model is causally inseparable for the space right above, with a causally separable fraction of \char`\~47.74\% over its two causal completions.
The two completions are induced by the causal orders $\total{A, B, E}\vee\total{C, D, F}$ and $\total{B, A, E}\vee\total{D, C, F}$, and their causaltopes are 1848-dimensional:
\begin{center}
    \begin{tabular}{ccc}
    \includegraphics[height=3.5cm]{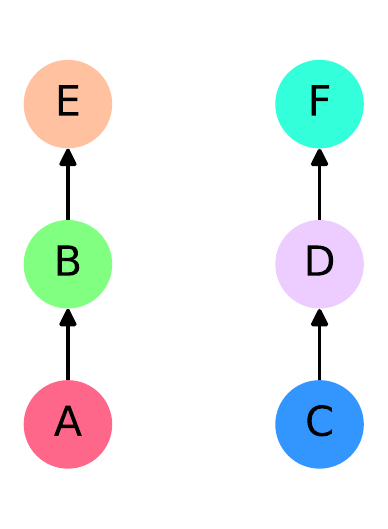}
    &
    \hspace{1cm}
    &
    \includegraphics[height=3.5cm]{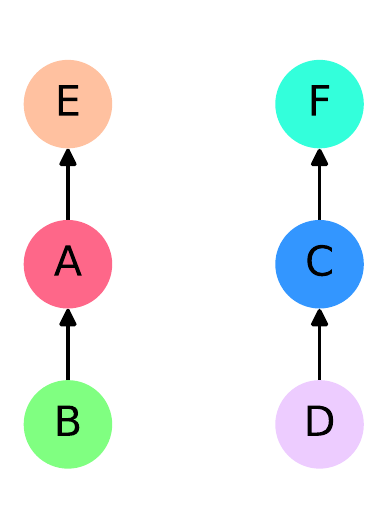}
    \\
    Causal fraction: \char`\~23.87\%
    &
    &
    Causal fraction: \char`\~23.87\%
    \end{tabular}
\end{center}
The empirical model becomes causally separable if the no-signalling constraint between the two 3-party groups are dropped: the model is 100\% supported by the spaces of input histories induced by causal orders \total{C,D,F,A,B,E} and \total{D, C, F, B, A, E}, two of the 512 causal completions of $\total{\evset{C,D},F,\evset{A,B},E}$, with each causal completion supporting exactly 50\% of the empirical model.

Below we plot the causally separable fraction for this empirical model as a function of the $\gamma_0$ and $\gamma_1$ measurement angles used by Eve and Felix.
We compare it to the local fraction for a Bell experiment with the same measurements.
\begin{center}
    \begin{tabular}{ccc}
    \includegraphics[height=4cm]{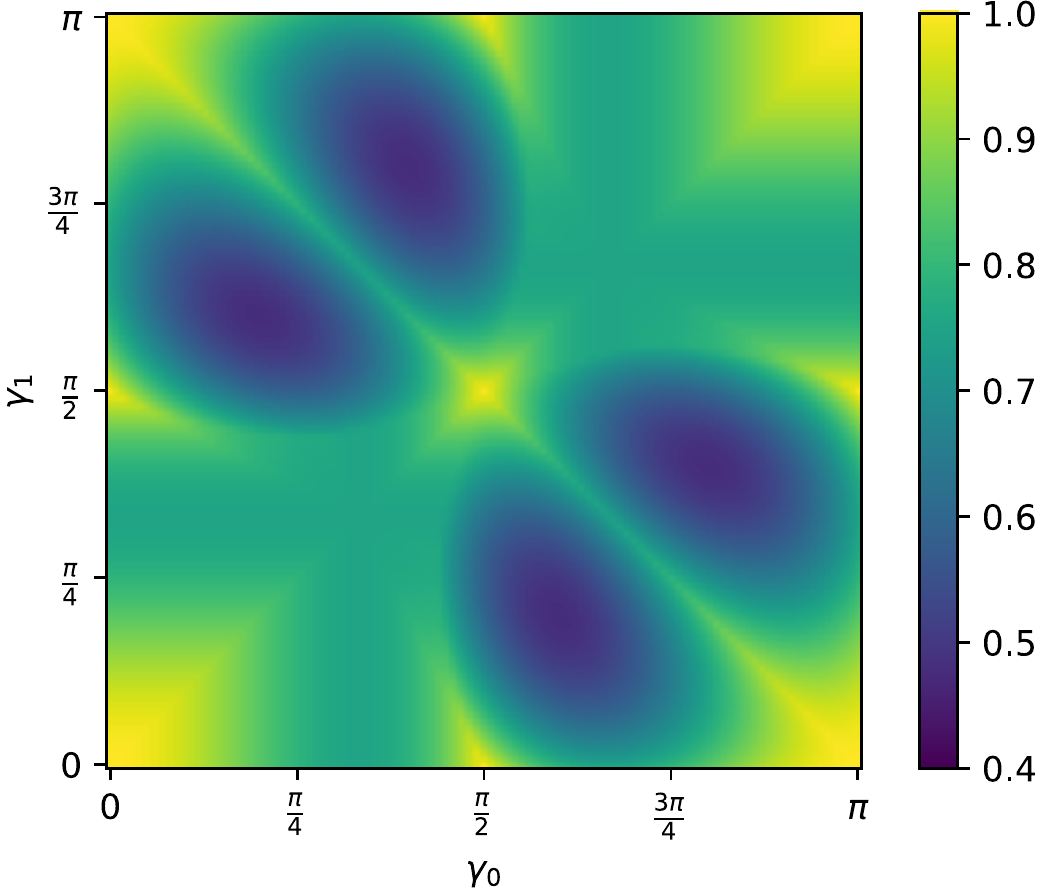}
    &
    \includegraphics[height=4cm]{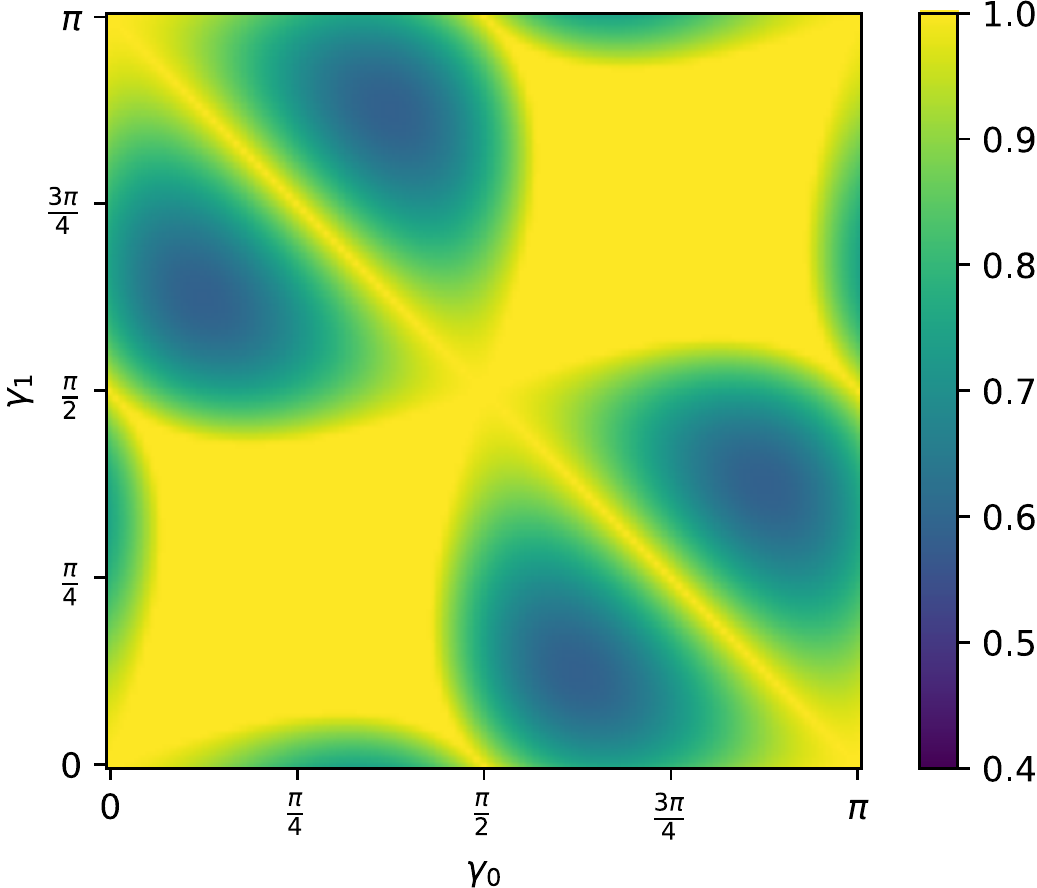}
    &
    \includegraphics[height=4cm]{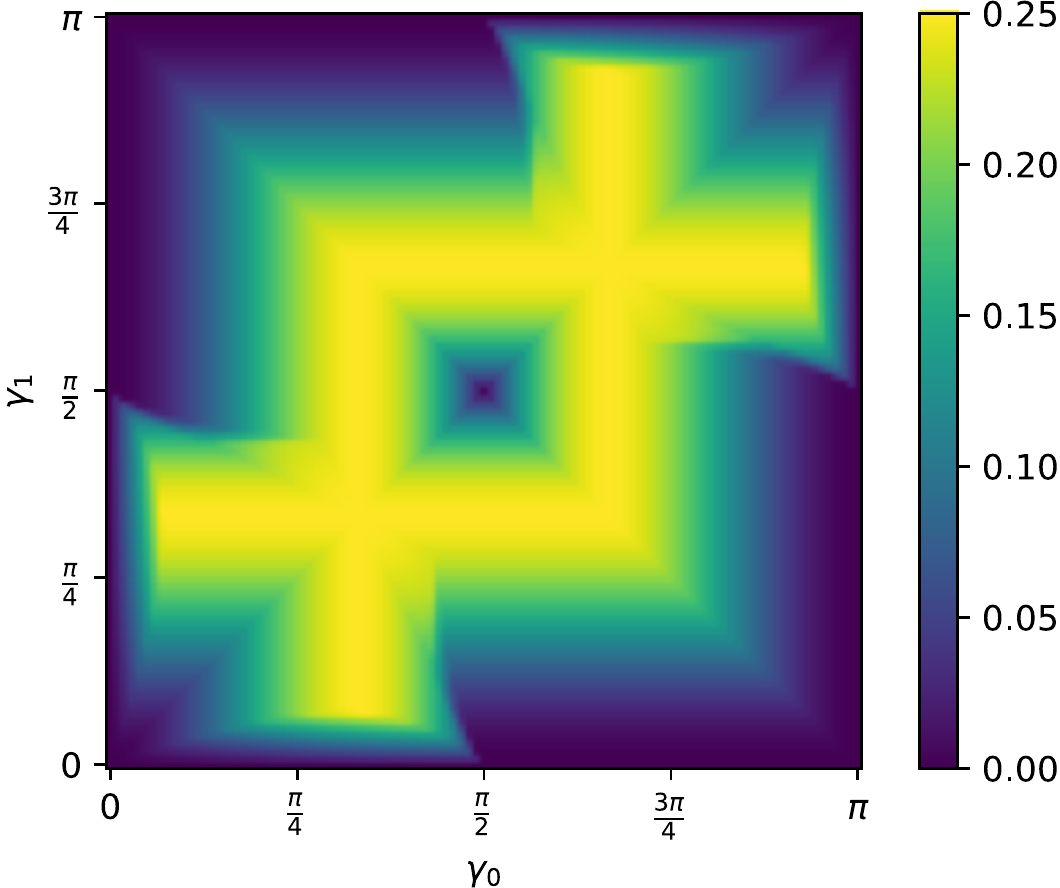}
    \\
    causally separable fraction
    &
    local fraction (Bell exp.)
    &
    difference
    \end{tabular}
\end{center}
At the angular resolution of $\frac{\pi}{100}$ used by the plots, the minimum causally separable fraction of \char`\~47.74\% is achieved at $\gamma_0 = \frac{\pi}{5}$ and $\gamma_1 = \frac{3\pi}{5}$.
For comparison, the minimum local fraction for the corresponding Bell experiment is \char`\~58.2\%.
The causally separable fraction coincides with the (causally separable) local fraction, saturating the lower bound provided by Proposition \ref{proposition:csep-frac-bounded-below-by-csep-local-frac}.
For further discussion of how the local fraction was practically calculated, see ``The Topology of Causality'' \cite{gogioso2022topology}.
\begin{center}
    \begin{tabular}{ccc}
    \includegraphics[height=4cm]{svg-inkscape/example-ghzswitch-2-2-X-plus-X-sepfrac-100_svg-tex.pdf}
    &
    \includegraphics[height=4cm]{svg-inkscape/example-ghzswitch-2-2-X-plus-X-sepfrac-100_svg-tex.pdf}
    &
    \includegraphics[height=4cm]{svg-inkscape/bell-local-fraction-res100-entangled-switch-comparison_svg-tex.pdf}
    \\
    causally separable fraction
    &
    local fraction
    &
    local fraction (Bell exp.)
    \end{tabular}
\end{center}

\subsubsection{Two Contextually Controlled Classical Switches.}
\label{subsubsection:contextually-control-switches}

In this example we consider two quantum switches---one between Charlie and Diane, the other between Eve and Felix, both with the $|+\rangle$ state as their input---with ``contextual control'', determined by Alice and Bob's non-demolition measurements of a Bell state.
Specifically, for angles $\gamma_0, \gamma_1 \in [0, \pi)$:
\begin{enumerate}
    \item Two qubits in a Bell state $|\Phi^+\rangle$ are shared by Alice and Bob at the start of the protocol.
    \item On input 0, Alice applies an X rotation by $-\gamma_0$ to her qubit, then decoheres it in the Z basis; on input 1, Alice applies an X rotation by $-\gamma_1$ to her qubit, then decoheres it in the Z basis. Alice forwards the qubit to the control of the switch between Charlie and Diane. Alice's output is fixed to 0.
    \item Bob does the same as Alice to his qubit before forwarding it to the control of the switch between Eve and Felix. Bob's output is fixed to 0.
    \item Inside the switches, Charlie, Diane, Eve and Felix all do the same thing: they perform an X measurement on the incoming qubit, using the measurement outcome as their individual output, and then encode their individual input into the X basis of the outgoing qubit.
    \item Both the outgoing qubit and the control qubit of the two switches are discarded.
\end{enumerate}
The figure below exemplifies the scenario we have just described:
\begin{center}
\scalebox{1}{
    \input{bell-before-switches-implicit-var.tikz}
}
\end{center}
For $\gamma_0 = 0, \gamma_1 = \frac{2\pi}{3}$, the description above results in the following empirical model:
\begin{center}
    \includegraphics[height=12cm]{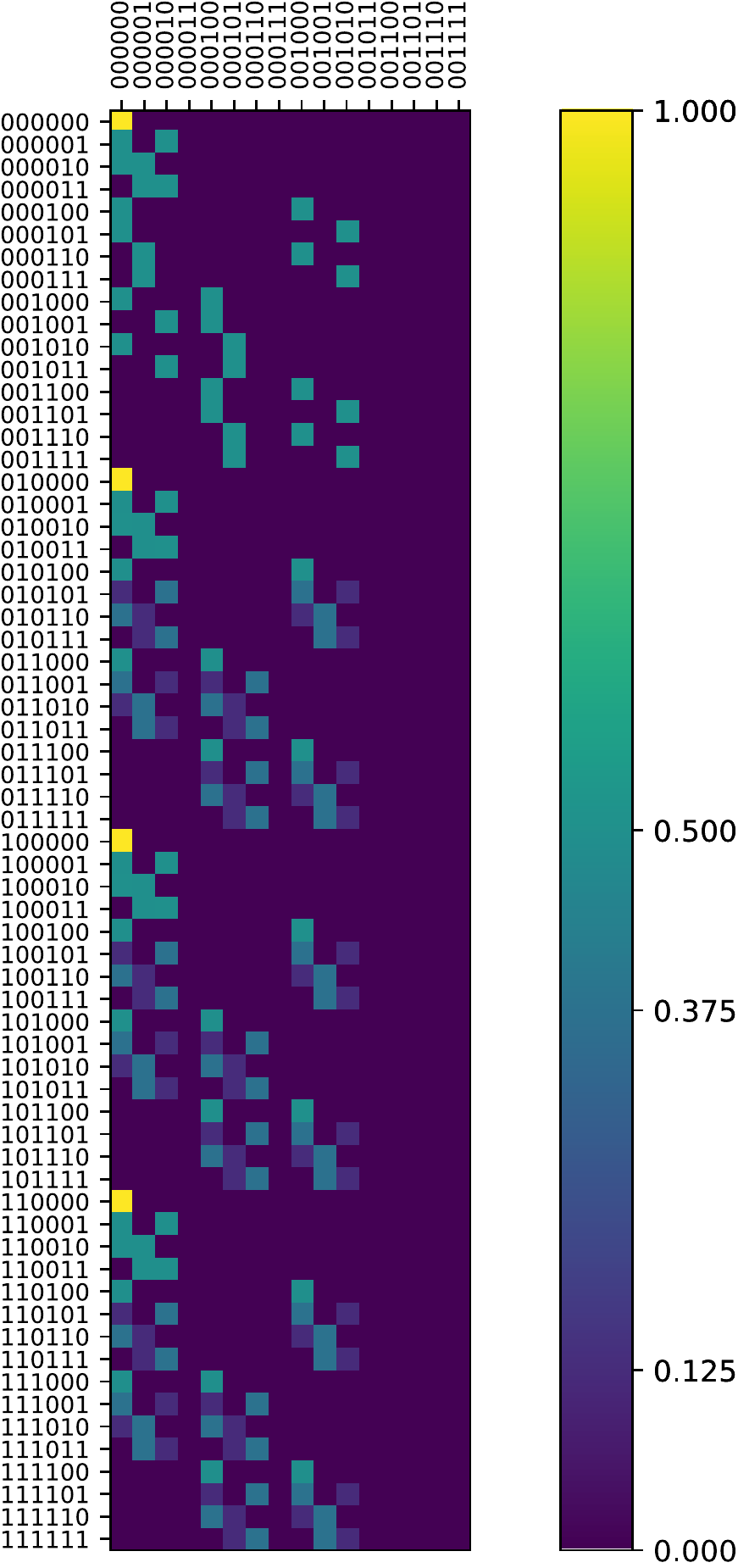}
\end{center}
By construction, the empirical model is 100\% supported by the space of input histories induced by the following indefinite causal order:
\begin{center}
        \includegraphics[height=3.25cm]{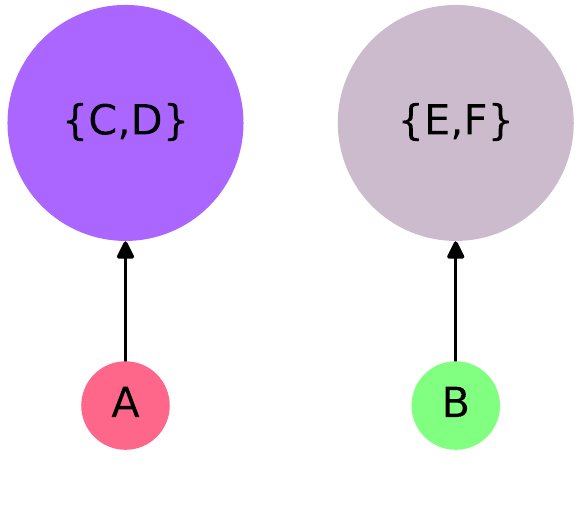}
\end{center}
The empirical model is causally inseparable for the space right above: it has a causally separable fraction of 75\%, coinciding with the local fraction for the Bell empirical model used to control the switches.
The space has $2^2 \cdot 2^2 = 16$ causal completions: these are all possible combinations of the $2^2=4$ switch spaces on events \evset{A,C,D} having \ev{A} as first event and the $2^2=4$ switch spaces on events \evset{B,E,F} having \ev{B} as first event.
Of the 16 causal completions, only the 6 spaces shown below support a non-zero fraction of the empirical model: each supports exactly 12.5\%, with no fraction in common between spaces, for a total of $6\cdot12.5\% = 75\%$.
They correspond to the 6 no-signalling functions appearing in the 75\% local decomposition of the Bell empirical model used to contextually control the causal order.

In the first space, Charlie precedes Diane when Alice's input is 0 and succeeds her when Alice's input is 1, while Eve always precedes Felix.
\begin{center}
\includegraphics[height=2cm]{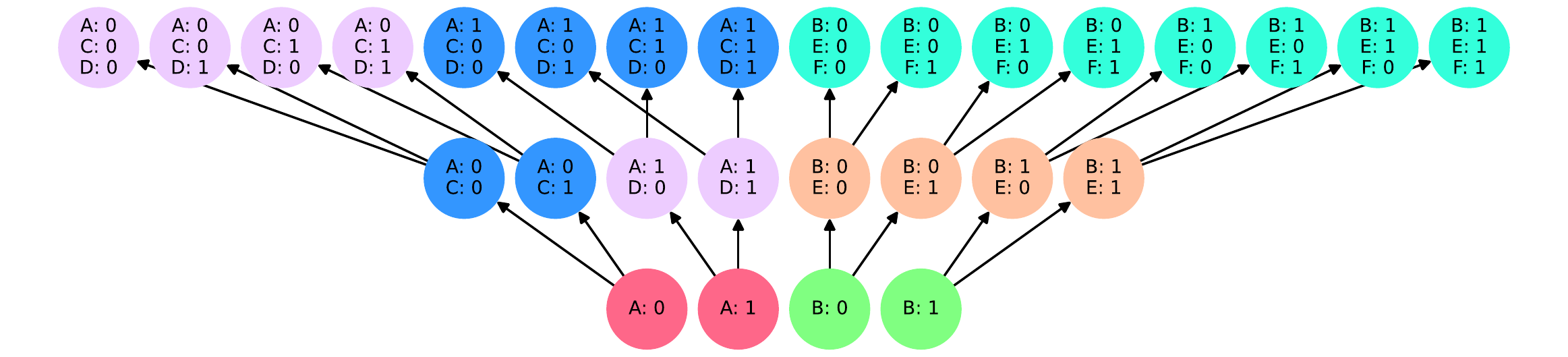}
\end{center}
In the second space, Charlie always precedes Diane, while Eve precedes Felix when Bob's input is 0 and succeeds him when Bob's input is 1.
\begin{center}
\includegraphics[height=2cm]{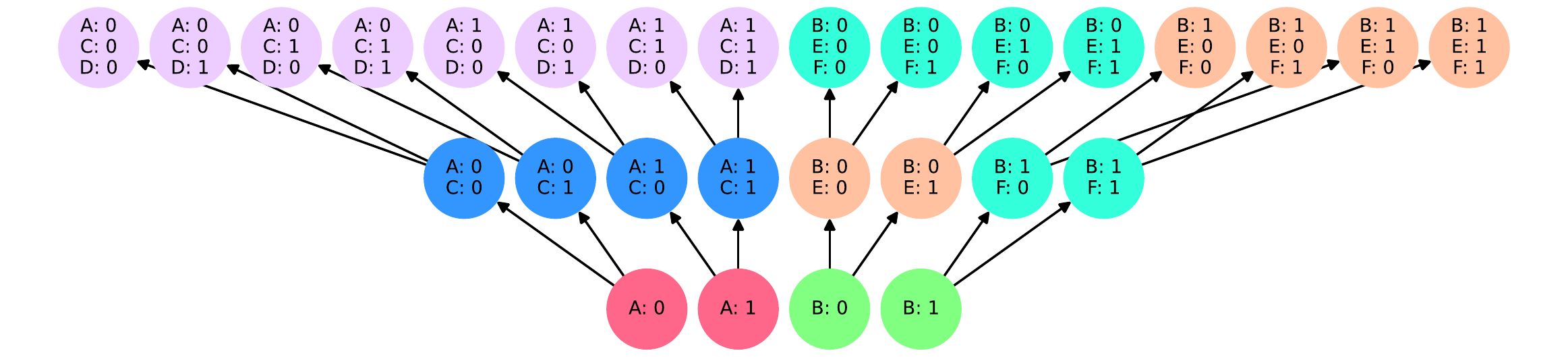}
\end{center}
In the third space, Charlie succeeds Diane when Alice's input is 0 and precedes her when Alice's input is 1, while while Eve succeeds Felix when Bob's input is 0 and precedes him when Bob's input is 1:
\begin{center}
\includegraphics[height=2cm]{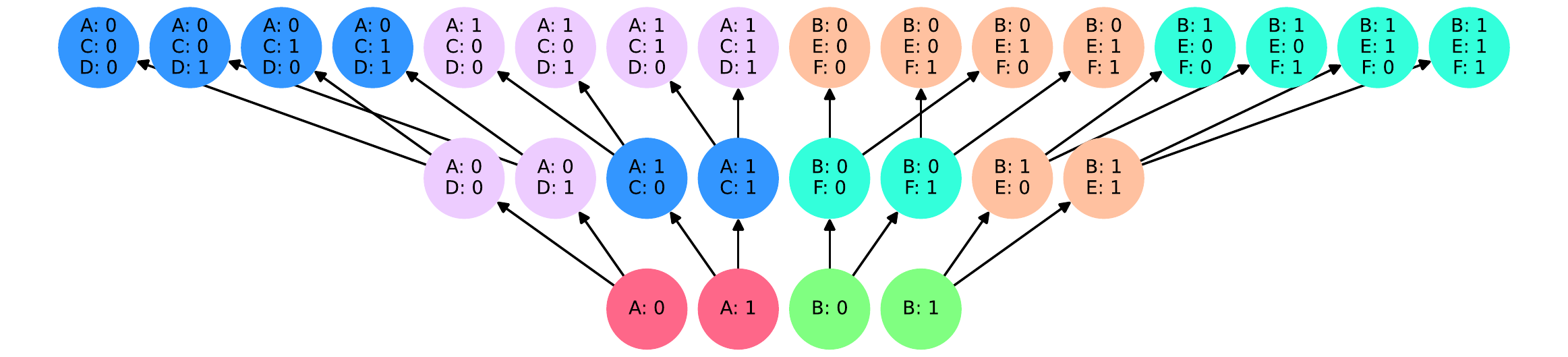}
\end{center}
In the fourth space, Charlie precedes Diane when Alice's input is 0 and succeeds her when Alice's input is 1, while while Eve precedes Felix when Bob's input is 0 and succeeds him when Bob's input is 1:
\begin{center}
\includegraphics[height=2cm]{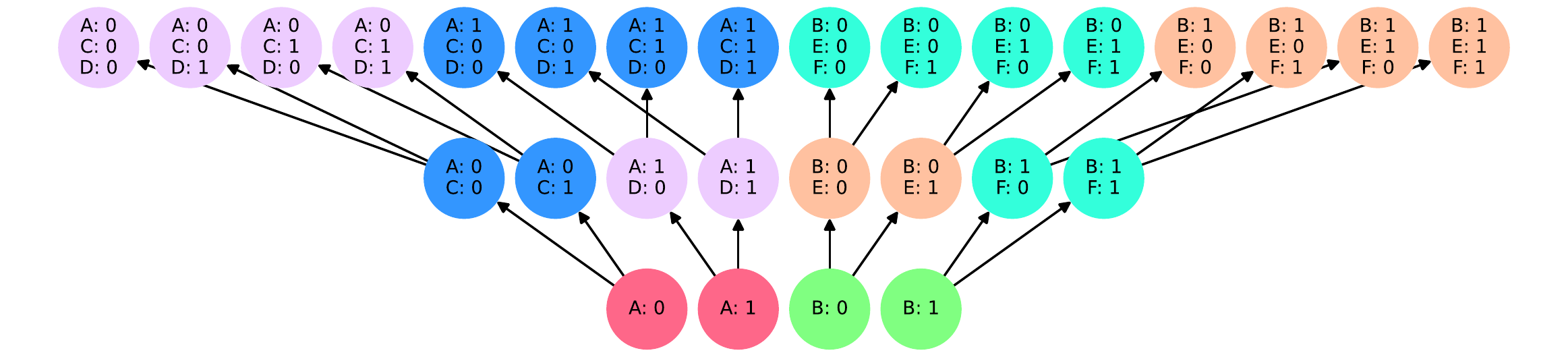}
\end{center}
In the fifth space, Charlie always succeeds Diane, while while Eve succeeds Felix when Bob's input is 0 and precedes him when Bob's input is 1:
\begin{center}
\includegraphics[height=2cm]{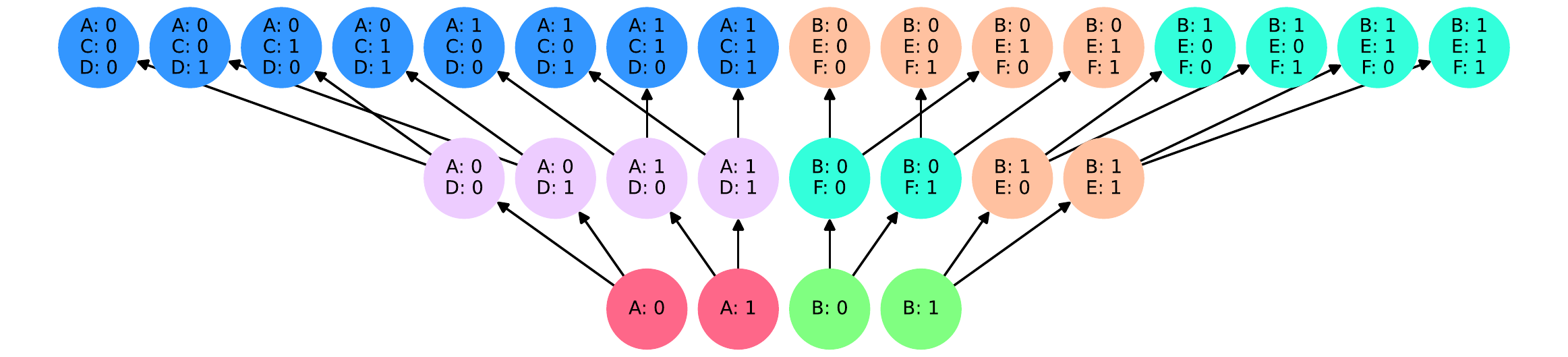}
\end{center}
In the sixth and final space, Charlie succeeds Diane when Alice's input is 0 and precedes her when Alice's input is 1, while while Eve always succeeds Felix:
\begin{center}
\includegraphics[height=2cm]{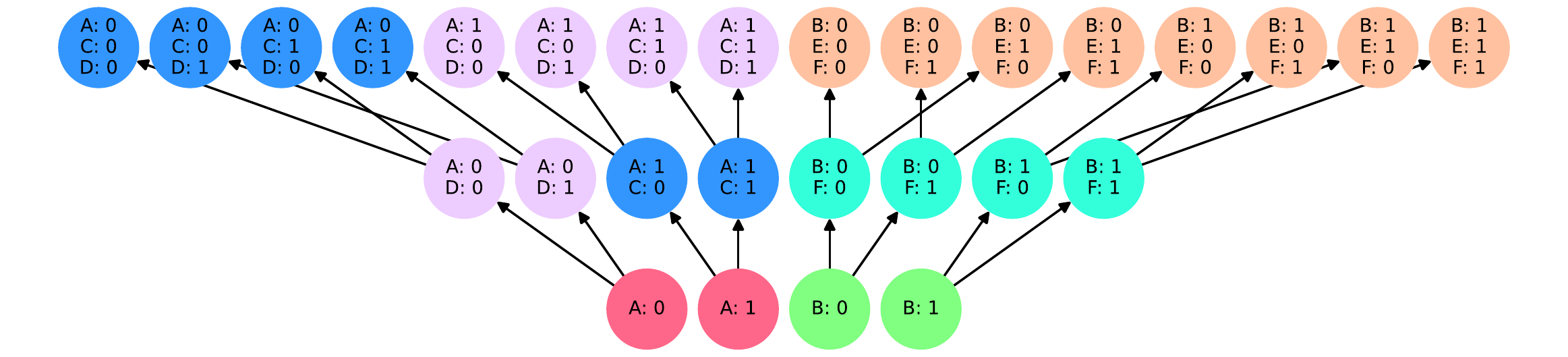}
\end{center}
Because its construction relies on contextuality of the inputs to the switches, we expect the empirical model to be causally separable once the no-signalling constraint between Alice and Bob is removed.
Unfortunately, we cannot simply iterate through all switch spaces to demonstrate causal separability: at 16511297126400 of them, the task would take millions of years.
Luckily, we can restrict our attention to a selected few: we know that signalling from \ev{A} to \ev{B} will be enough to explain the contextuality of the switch inputs, and that classically controlled switches are causally separable.
Hence, it suffices to look at the $4^4 = 256$ switch spaces in the following form:
\[
\begin{array}{l}
    \Hist{\total{A,B}, \{0,1\}}
    \seqcomposeSym
    \underline{\Theta}
    \\
    \underline{\Theta}
    \in
    \suchthat{
    \begin{array}{l}
    \Hist{\Omega_{CD}, \{0,1\}}
    \\\cup\;
    \Hist{\Omega_{EF}, \{0,1\}}
    \end{array}
    }{
    \begin{array}{l}
    \Omega_{CD} \in \{\total{C,D}, \total{D,C}\},\\
    \Omega_{EF} \in \{\total{E,F}, \total{F,E}\}
    \end{array}
    }^{K}
    \\
    \text{where }
    K := \max{\ExtHist{\total{A,B}, \{0,1\}}}
    \end{array}
\]
The remaining 25\% of the empirical model is explained by exactly two out of the 256 spaces, each one explaining exactly 12.5\%: in the two spaces, Charlie and Diane are in fixed causal order, while the causal order between Eve and Felix is dictated by the logical OR of Alice and Bob's outcome.
In the first space, Charlie always precedes Diane, while Eve precedes Felix when Alice and Bob's outcomes are both 0 and otherwise succeeds him:
\begin{center}
\includegraphics[height=1.75cm]{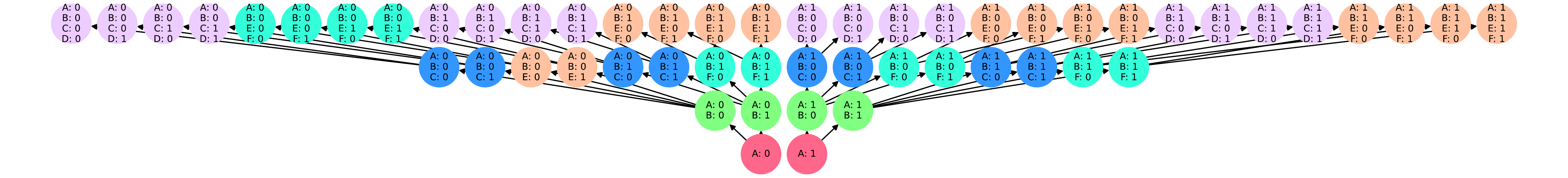}
\end{center}
In the second space, Charlie always succeeds Diane, while Eve succeeds Felix when Alice and Bob's outcomes are both 0 and otherwise precedes him.
\begin{center}
\includegraphics[height=1.75cm]{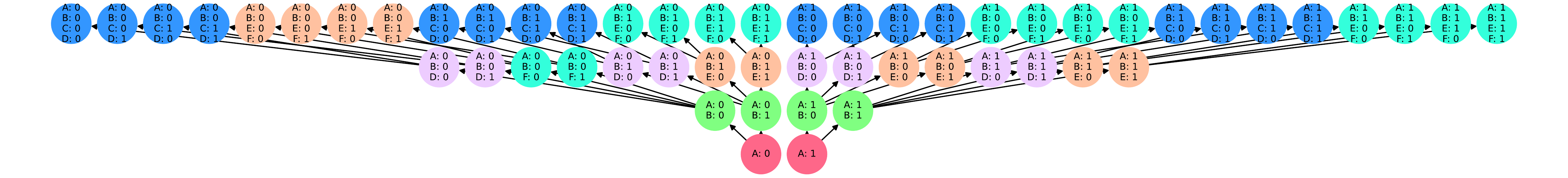}
\end{center}

Below we plot the causally separable fraction (left) and the local fraction (middle) for this empirical model as a function of the $\gamma_0$ and $\gamma_1$ measurement angles used by Alice and Bob.
We compare it to the local fraction for a Bell experiment with the same measurements (right).
\begin{center}
    \begin{tabular}{ccc}
    \includegraphics[height=3.25cm]{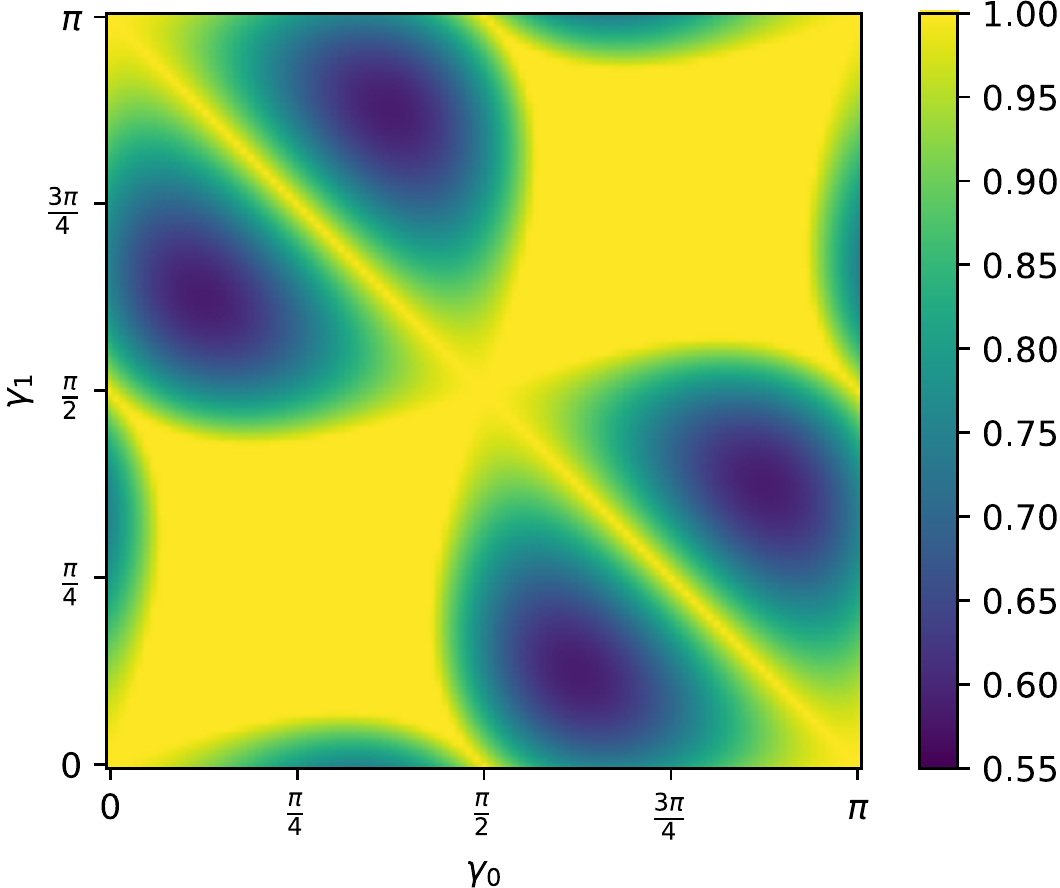}
    &
    \hspace{5mm}
    \includegraphics[height=3.25cm]{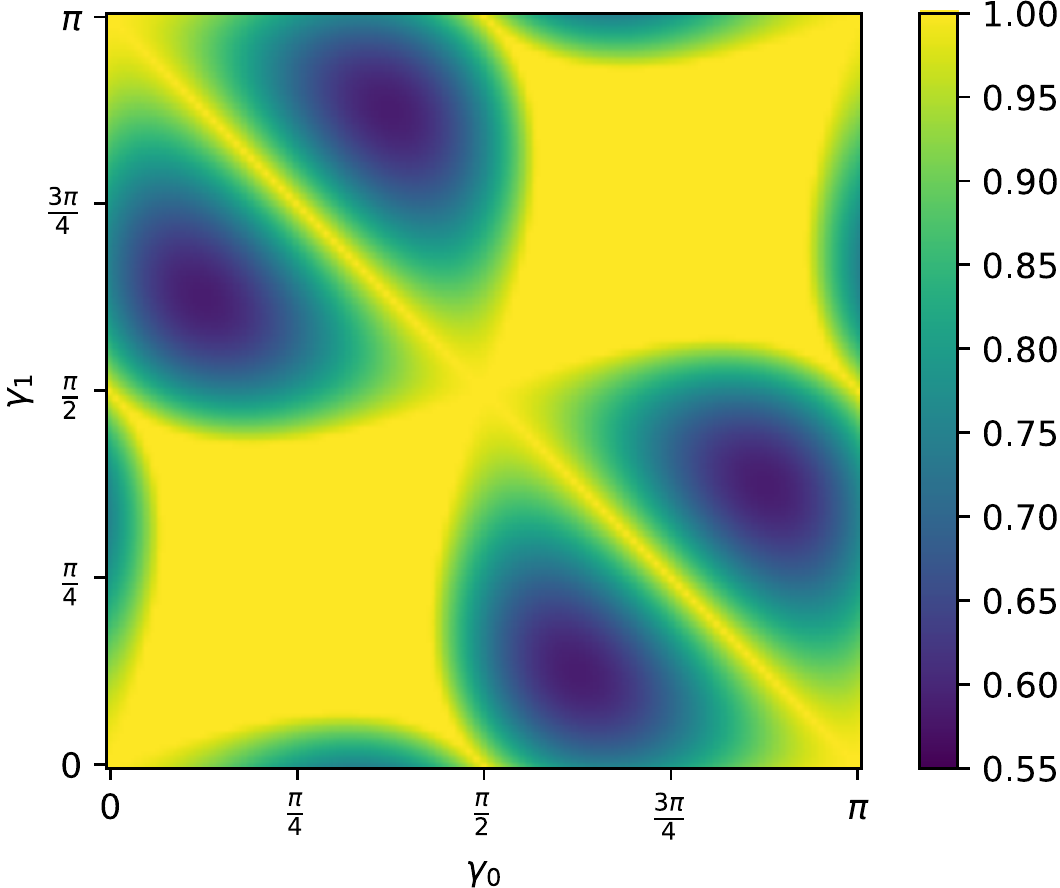}
    \hspace{5mm}
    &
    \includegraphics[height=3.25cm]{svg-inkscape/bell-local-fraction-res100_svg-tex.pdf}
    \\
    causally separable fraction
    &
    local fraction
    &
    local fraction (Bell exp.)
    \end{tabular}
\end{center}
By construction, the causally separable fraction, local fraction and causally separable local fraction for this empirical model all coincide, being perfectly correlated with the local fraction for a Bell experiment with the same measurements.
Explicit evaluation of the causally separable fraction landscape for this example is computationally intensive, requiring 15 hours on 30 cores (2.5GHz), with 260GB RAM used at peak.

\subsection{Causal equations for arbitrary causaltopes}
\label{subsection:nonstd-causal-eqs}

Expanding on the presentation of standard causaltopes from Subsection \ref{subsection:std-causal-eqs}, we now investigate how causality equations generalise to covers other than the standard cover.
Recall the definition of causality equations from Definition \ref{definition:caus-eqs} (p.\pageref{definition:caus-eqs}) and Proposition \ref{proposition:caus-eqs-chain} (p.\pageref{proposition:caus-eqs-chain}):
\[
\begin{array}{rl}
    \CausEqs{\mathcal{C}, \underline{O}}_{\mu, \lambda, \lambda'}
    &:=
    \suchthat{
        \underline{u}
        \in
        \PsEmpModelsVec{\mathcal{C}, \underline{O}}
    }{
        \restrict{\underline{u}^{(\lambda)}}{\mu}
        =\restrict{\underline{u}^{(\lambda')}}{\mu}
    }
    \\
    \restrict{\underline{u}^{(\lambda)}}{\mu}
    &:= \Dist{\rho_{\lambda, \mu}}\left(\underline{u}^{(\lambda)}\right)
    \\
    \CausEqs{\mathcal{C}, \underline{O}}
    &:=
    \bigcap_{\mu \in \Lsets{\Theta}, n_\mu \geq 1}
    \bigcap_{i=1}^{n_\mu-1}
    \CausEqs{\mathcal{C}, \underline{O}}_{\mu, \lambda_{\mu,i}, \lambda_{\mu,i+1}}
\end{array}
\]
where $\mu \in \Lsets{\Theta}$ and $\lambda_{\mu,1}, ..., \lambda_{\mu, n_\mu}$ is a total order on the $\lambda \in \mathcal{C}$ such that $\mu \subseteq \lambda$.

In Subsection \ref{subsection:std-causal-eqs}, we employed a succinct indexing scheme for causality equation, based on the observation that all contexts $\lambda \in \mathcal{C}$ in a standard cover $\mathcal{C}$ have exactly one input for each event, and all lowersets $\mu \subseteq \lambda$ have at most one input for each event.
Once the events are consistently sorted, rows can be labelled using the sequence of input/output values associated to each event by a lowerset, with a special character---the underscore \_, in our case---used to indicate the absence of an event from a lowerset.
Similarly, columns can be labelled using the input/output values associated to each event by a context.
This is shown below on the left, for the standard cover of a space on two events.

Moving away from the standard cover, we no longer have this luxury: the same event may appear multiple times in each context, and there is no natural way to consistently order the histories within contexts.
As a consequence, we will index causality equations for arbitrary covers by explicitly indicating the input histories that appear in each lowerset (for rows) and in each context (for columns).
The output values associated to the tip events will be still indicated succinctly, based on the order of the explicitly listed input histories (at least in those cases where each history $h$ in a context has a separate tip historysets $\histconstreqcls{h}{\omega}$).
This is shown below on the right, for the same standard cover as the left example: the explicit notation has the disadvantage of being verbose, but the advantage of being unambiguous.

\begin{center}
    \begin{tabular}{cc}
    \includegraphics[width=6cm]{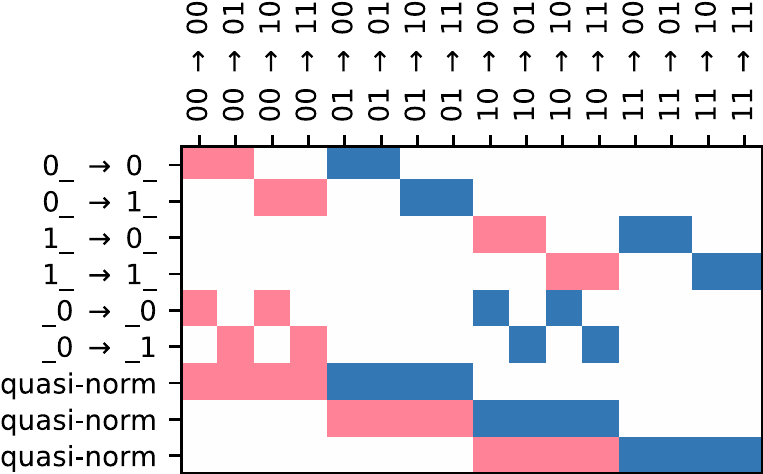}
    \hspace{1cm}
    &
    \includegraphics[width=6cm]{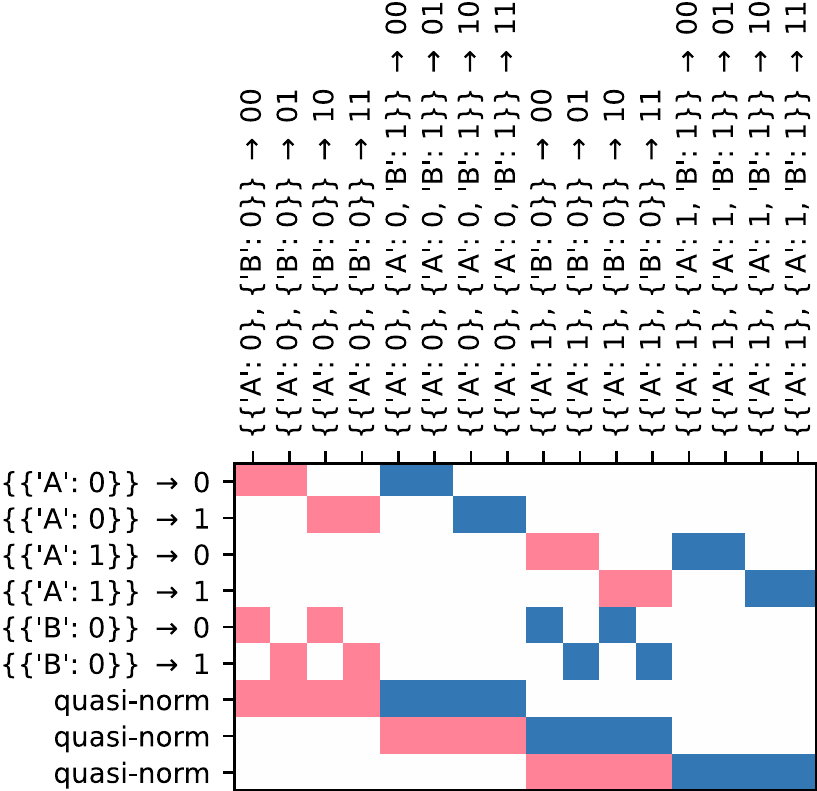}
    \end{tabular}
\end{center}

As our first example, we consider the discrete space on a single event, with ternary inputs and outputs.
This space has 9 covers, arranged in the lattice shown below: the finest cover is the standard cover \#0, on the left, while the coarsest cover is the classical cover \#8, on the right.
\begin{center}
    \begin{tabular}{cc}
    \includegraphics[width=4cm]{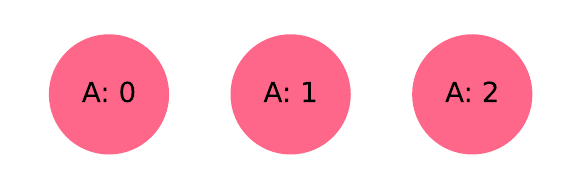}
    \hspace{2cm}
    &
    \includegraphics[width=8cm]{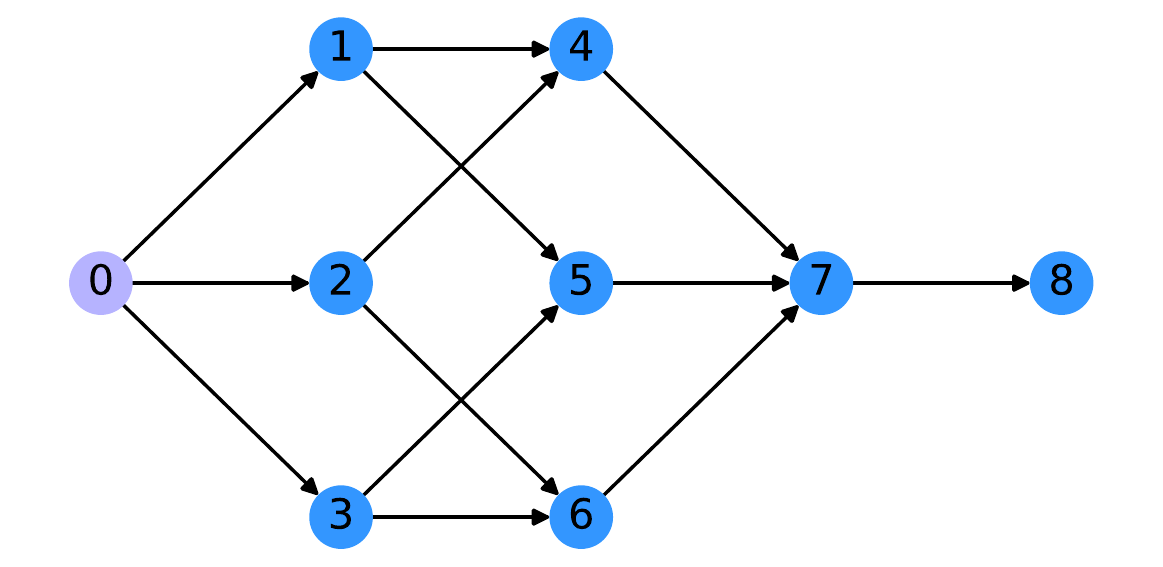}
    \\
    space of input histories
    \hspace{2cm}
    &
    hierarchy of covers for the space
    \end{tabular}
\end{center}
The standard cover for this space has three contexts: $\left\{\hist{A/0}\right\}$, $\left\{\hist{A/1}\right\}$ and $\left\{\hist{A/2}\right\}$.
There are no causality equations, because the contexts don't have any non-trivial lowerset $\mu$.
Causality and quasi-normalisation equations for this cover are depicted below.
\begin{center}
    \begin{tabular}{c}
    \includegraphics[width=6cm]{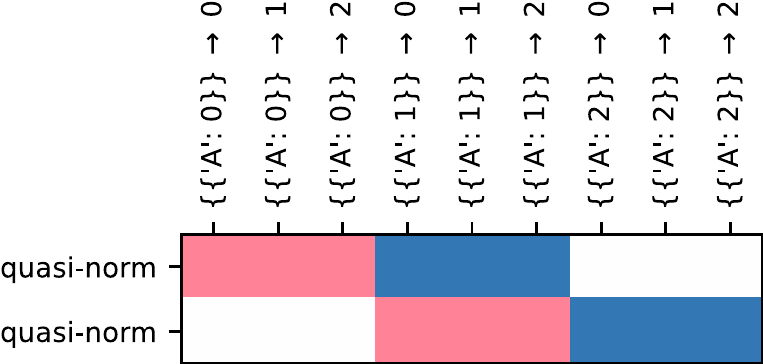}
    \end{tabular}
\end{center}
Covers \#1, \#2 and \#3 have 2 contexts each: $\left\{\hist{A/i_A}\right\}$ and $\suchthat{\hist{A/j_A}}{j_A \in \{0, 1, 2\}, j_A \neq i_A}$.
There are no causality equations in either case, because the two contexts in each cover don't have non-trivial lowerset $\mu$ in common.
Causality and quasi-normalisation equations for all three covers are depicted below.
\begin{center}
    \begin{tabular}{ccc}
    \includegraphics[width=5cm]{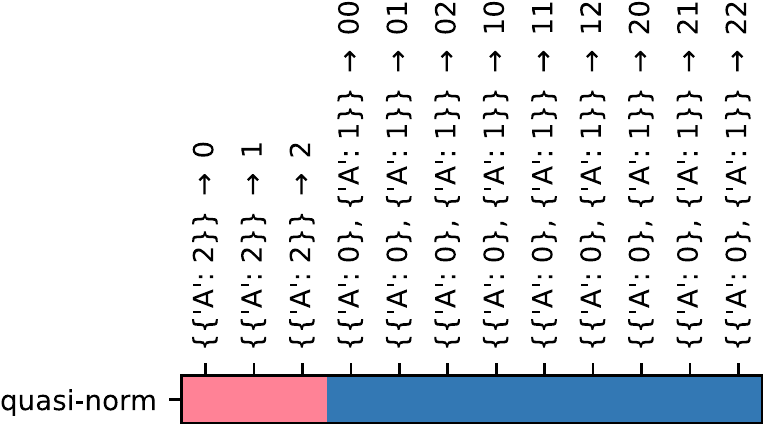}
    &
    \includegraphics[width=5cm]{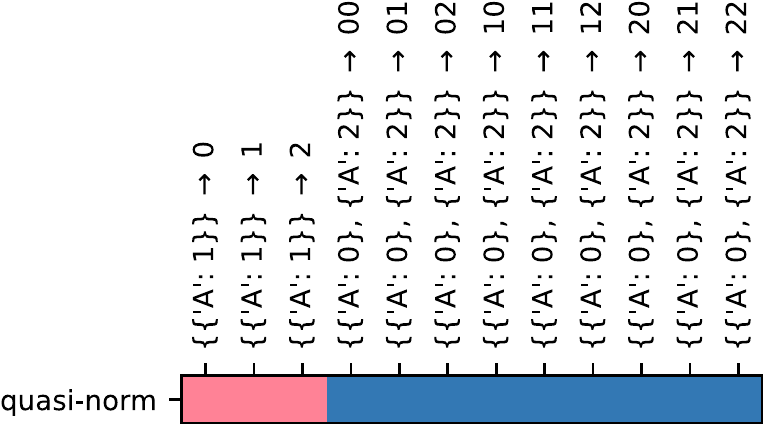}
    &
    \includegraphics[width=5cm]{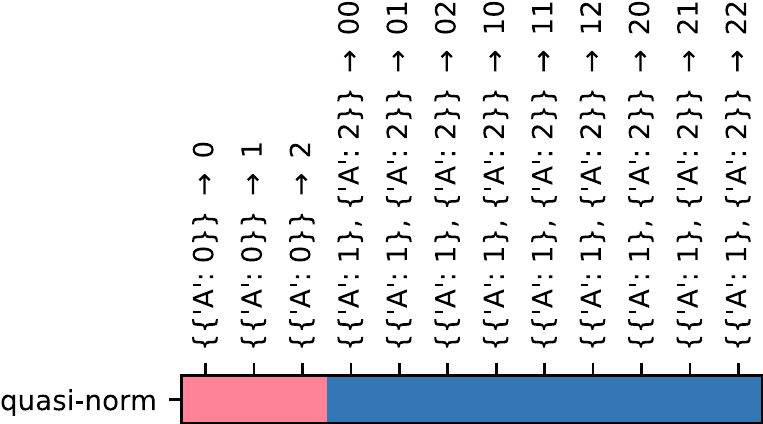}
    \end{tabular}
\end{center}
Covers \#4, \#5 and \#6 have 2 contexts each: $\left\{\hist{A/i_0}, \hist{A/i_1}\right\}$ and $\left\{\hist{A/i_0}, \hist{A/i_2}\right\}$, where $\{i_0, i_1, i_2\} = \{0, 1, 2\}$.
There are 3 causality equations in each case: the two contexts in each cover have the lowerset $\left\{\hist{A/i_0}\right\}$ in common, and there are three possible outputs for its tip event \ev{A}.
Causality and quasi-normalisation equations for all three covers are depicted below.
\begin{center}
    \begin{tabular}{ccc}
    \includegraphics[width=5cm]{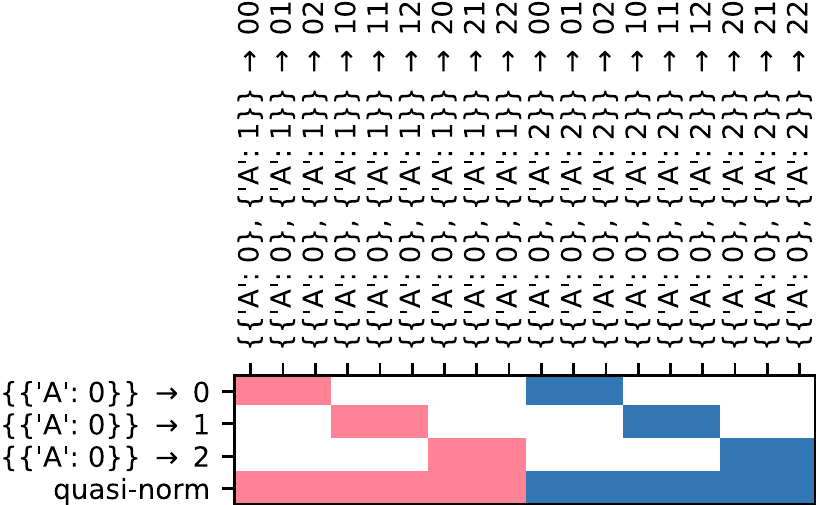}
    &
    \includegraphics[width=5cm]{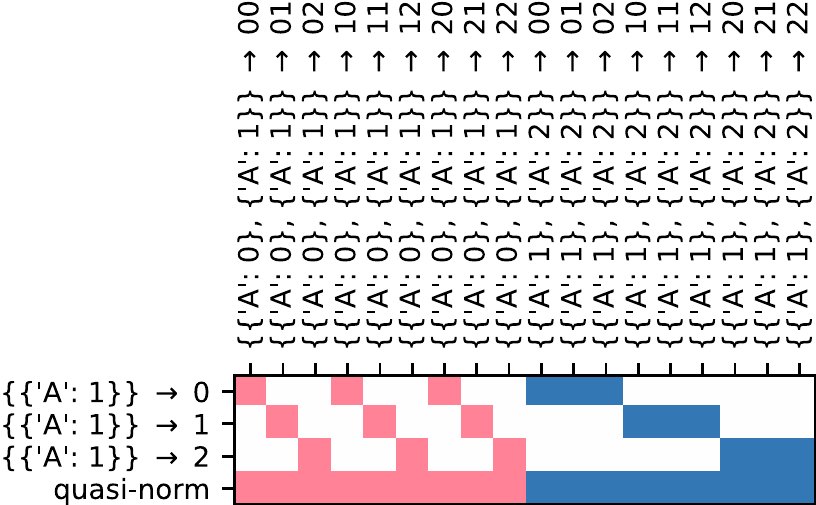}
    &
    \includegraphics[width=5cm]{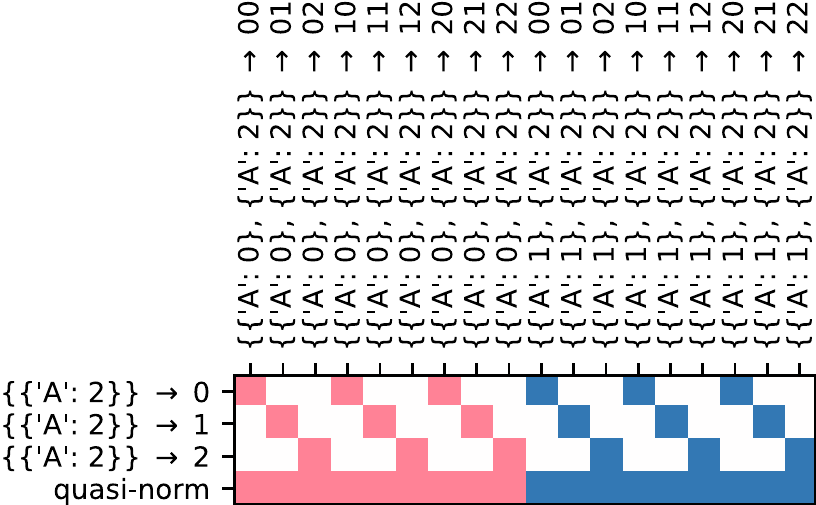}
    \end{tabular}
\end{center}
Cover \#7 has 3 contexts: $\left\{\hist{A/0}, \hist{A/1}\right\}$, $\left\{\hist{A/0}, \hist{A/2}\right\}$ and $\left\{\hist{A/1}, \hist{A/2}\right\}$.
There are 9 causality equations: each one of the three pairs of contexts has a different lowerset $\left\{\hist{A/i_A}\right\}$ in common, and there are three possible outputs for its tip event \ev{A}.
Causality and quasi-normalisation equations for this cover are depicted below.
\begin{center}
    \begin{tabular}{c}
    \includegraphics[width=8cm]{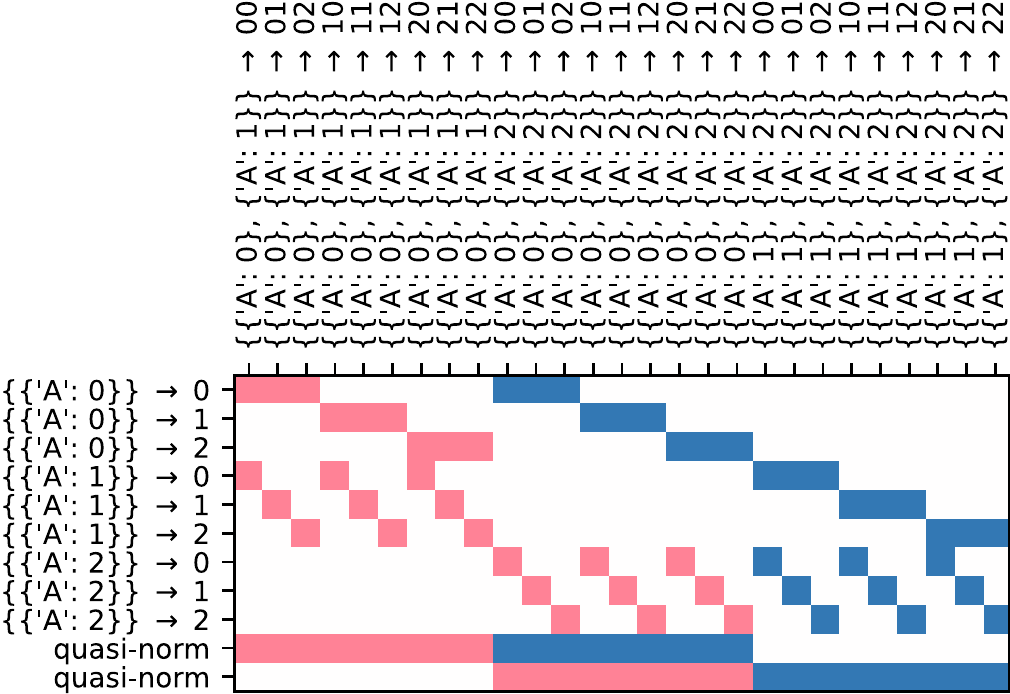}
    \end{tabular}
\end{center}
Cover \#8, finally, is the classical cover: it has 1 context, no causality equations and no quasi-normalisation equations.

As our second example, we consider one of the four non-discrete, non-total causally complete spaces on two events with binary inputs.
This space has 89 covers, arranged in the lattice shown below: the finest cover is the fully solipsistic cover \#0, on the left, the coarsest cover is the classical cover \#88, on the right, and the standard cover is cover \#4, in the middle.
\begin{center}
    \begin{tabular}{cc}
    \includegraphics[width=4cm]{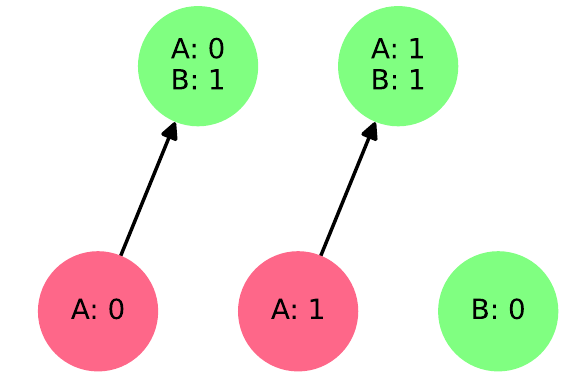}
    \hspace{0cm}
    &
    \includegraphics[width=10cm]{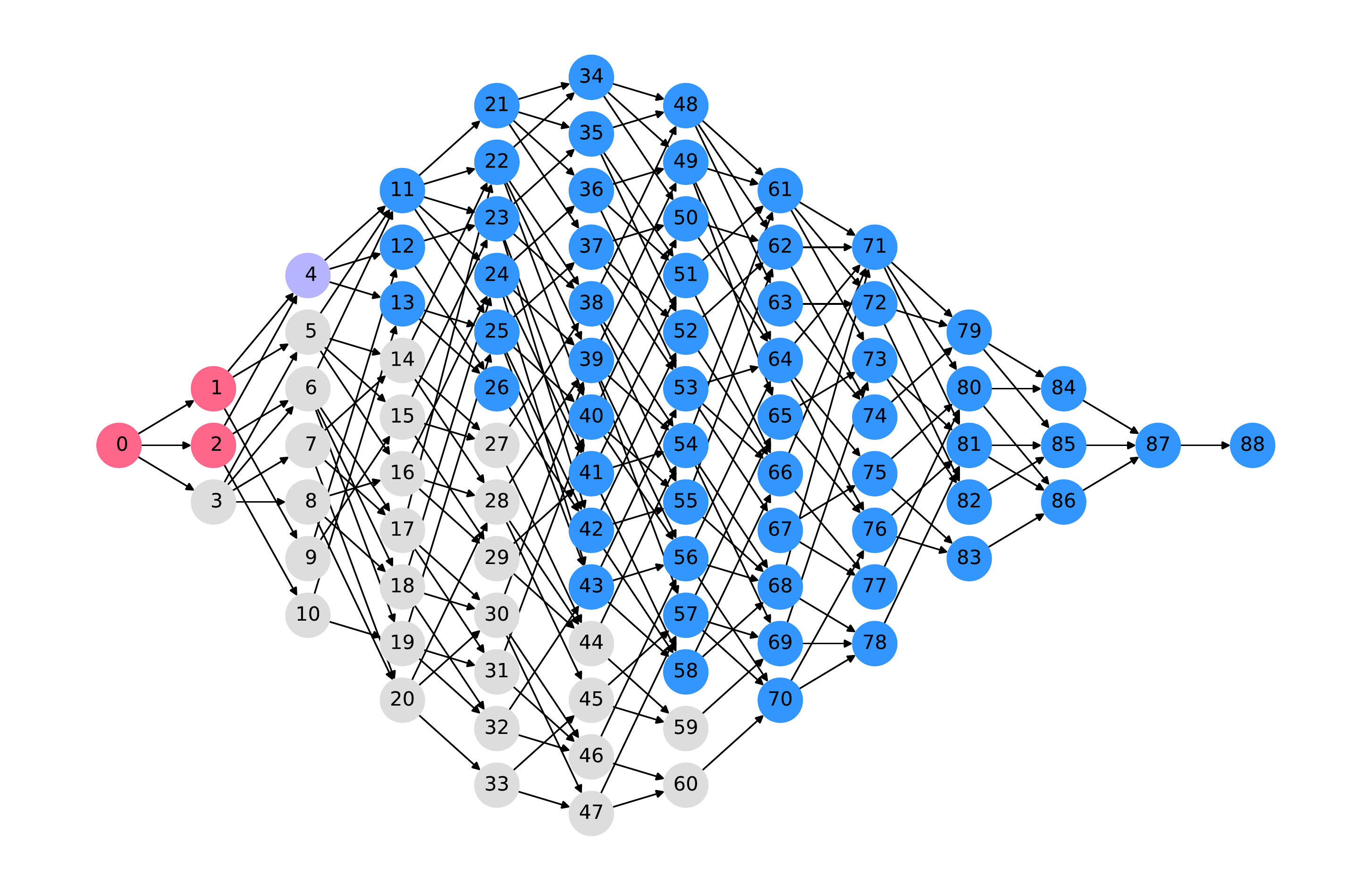}
    \\
    space of input histories $\Theta$
    \hspace{0cm}
    &
    hierarchy of covers for the space $\Theta$
    \end{tabular}
\end{center}
The fully solipsistic cover has the downsets $\downset{h}$ of input histories $h \in \Theta$ as its contexts.
There are no causality equations, because the contexts don't have non-trivial lowerset $\mu$ in common.
\begin{center}
    \begin{tabular}{c}
    \includegraphics[width=5cm]{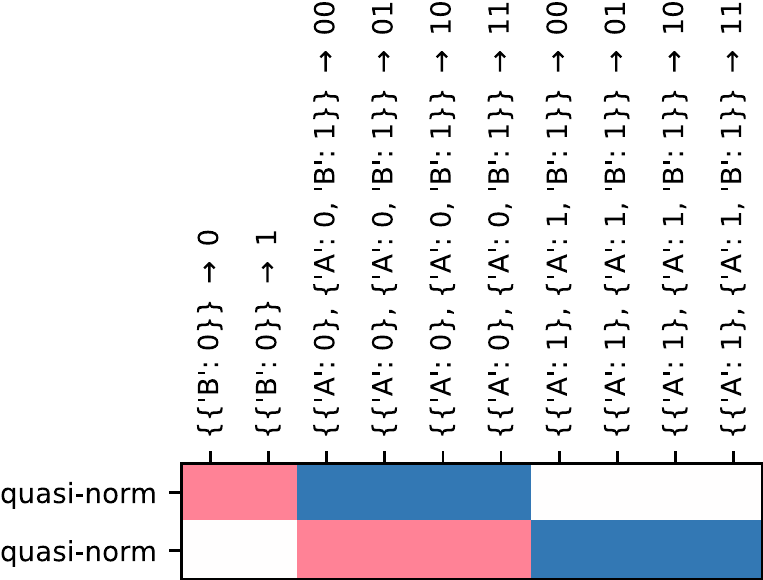}
    \end{tabular}
\end{center}
Covers \#1 and \#2 lie between the fully solipsistic cover \#0 and the standard cover \#4.
They have 3 contexts each, where $i_A = 0$ for cover \#1 and $i_A = 1$ for cover \#2:
\[
    \left\{\hist{A/i_A}, \hist{B/0}\right\}
    \hspace{5mm}
    \left\{\hist{A/0}, \hist{A/0,B/1}\right\}
    \hspace{5mm}
    \left\{\hist{A/1}, \hist{A/1,B/1}\right\}
\]
There are 2 causality equations for each one of covers \#1 and \#2:
\begin{itemize}
    \item lowerset $\left\{\hist{A/i_A}\right\}$ is common to contexts $\left\{\hist{A/i_A}, \hist{B/0}\right\}$ and $\left\{\hist{A/i_A}, \hist{A/i_A,B/1}\right\}$, with two outputs for tip event $\ev{A}$
\end{itemize}
Cover \#3 lies above the fully solipsistic cover, but not below the standard cover.
There are 4 context in this cover:
\[
\begin{array}{l}
    \left\{\hist{B/0}\right\}
    \\
    \left\{\hist{A/0}, \hist{A/1}\right\}
    \\
    \left\{\hist{A/0}, \hist{A/0,B/1}\right\}
    \\
    \left\{\hist{A/1}, \hist{A/1,B/1}\right\}
\end{array}
\]
There are 4 causality equations for cover \#3:
\begin{itemize}
    \item lowerset $\left\{\hist{A/0}\right\}$ is common to contexts $\left\{\hist{A/0}, \hist{A/1}\right\}$ and $\left\{\hist{A/0}, \hist{A/0,B/1}\right\}$, with two outputs for tip event $\ev{A}$
    \item lowerset $\left\{\hist{A/1}\right\}$ is common to contexts $\left\{\hist{A/0}, \hist{A/1}\right\}$ and $\left\{\hist{A/1}, \hist{A/1,B/1}\right\}$, with two outputs for tip event $\ev{A}$
\end{itemize}
Causality and quasi-normalisation equations for all three covers are depicted below.
\begin{center}
    \begin{tabular}{ccc}
    \includegraphics[width=5cm]{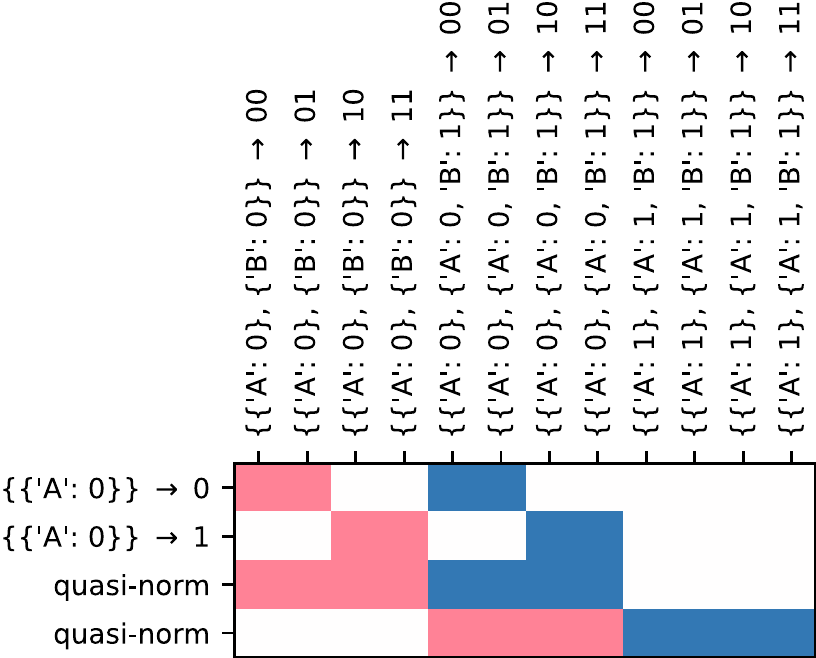}
    &
    \includegraphics[width=5cm]{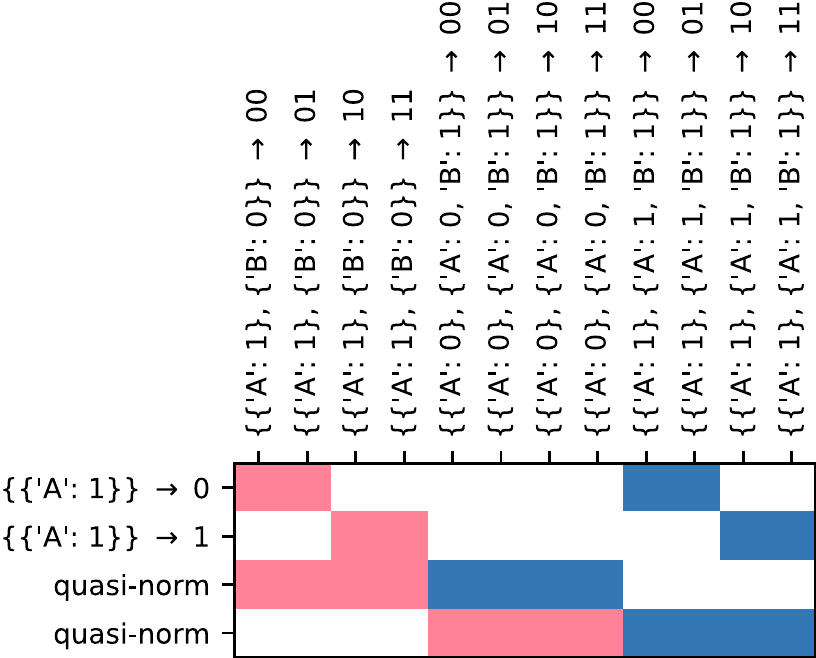}
    &
    \includegraphics[width=5cm]{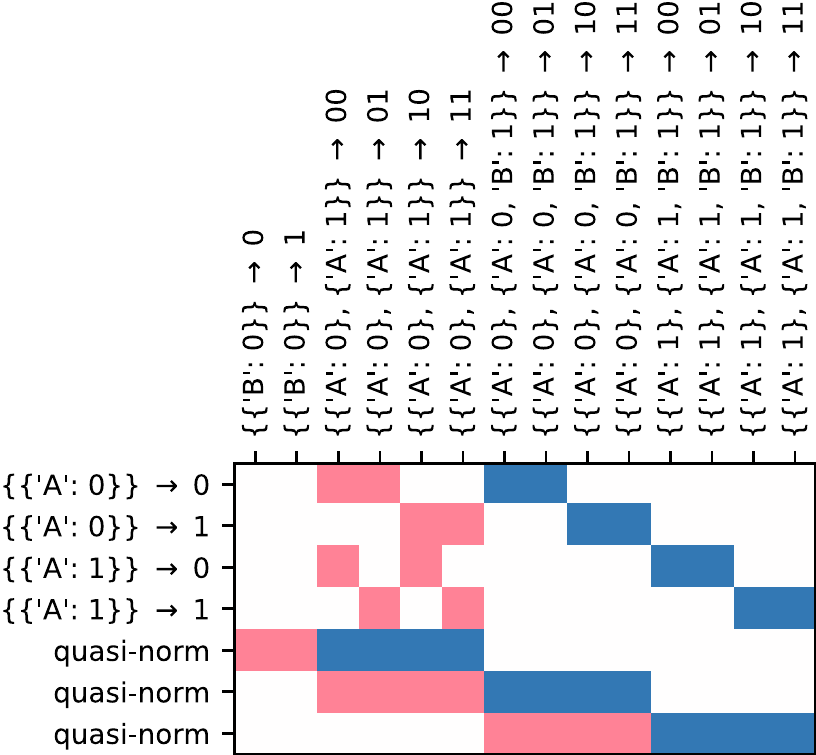}
    \\
    cover \#1
    &
    cover \#2
    &
    cover \#3
    \end{tabular}
\end{center}
The standard cover \#4 has the downsets $\downset{k}$ of extended input histories $k \in \Ext{\Theta}$ as its contexts:
\[
    \left\{\hist{A/i_A,B/0}\right\}
    \\
    \left\{\hist{A/i_A}, \hist{A/i_A,B/0}\right\}
\]
where $i_A$ ranges over $i_A \in \{0, 1\}$.
There are 6 causality equations:
\begin{itemize}
    \item lowerset $\left\{\hist{A/0}\right\}$ is common to contexts $\left\{\hist{A/0}, \hist{B/0}\right\}$ and $\left\{\hist{A/0}\right\}$, with two outputs for tip event $\ev{A}$
    \item lowerset $\left\{\hist{A/1}\right\}$ is common to contexts $\left\{\hist{A/1}, \hist{B/0}\right\}$ and $\left\{\hist{A/1}\right\}$, with two outputs for tip event $\ev{A}$
    \item lowerset $\left\{\hist{B/0}\right\}$ is common to contexts $\left\{\hist{A/0}, \hist{B/0}\right\}$ and $\left\{\hist{A/1}, \hist{B/0}\right\}$, with two outputs for tip event $\ev{B}$
\end{itemize}
Causality and quasi-normalisation equations for this cover are depicted below.
\begin{center}
    \begin{tabular}{c}
    \includegraphics[width=6cm]{svg-inkscape/2303-causeqs-space-e2idx1-causeqs-cover-4.pdf}
    \end{tabular}
\end{center}
Covers \#11, \#12 and \#13 lie just above the standard cover.
There are 5 contexts in cover \#11:
\[
\begin{array}{l}
    \left\{\hist{A/0}, \hist{A/1}\right\}
    \\
    \left\{\hist{A/i_A}, \hist{B/0}\right\}
    \\
    \left\{\hist{A/i_A}, \hist{A/i_A,B/0}\right\}
\end{array}
\]
where $i_A$ ranges over $i_A \in \{0, 1\}$.
There are 10 causality equations for cover \#11:
\begin{itemize}
    \item lowerset $\left\{\hist{A/0}\right\}$ is common to contexts $\left\{\hist{A/0}, \hist{A/1}\right\}$, $\left\{\hist{A/0}, \hist{B/0}\right\}$ and $\left\{\hist{A/0}, \hist{A/0,B/0}\right\}$, with two outputs for tip event $\ev{A}$
    \item lowerset $\left\{\hist{A/1}\right\}$ is common to contexts $\left\{\hist{A/0}, \hist{A/1}\right\}$, $\left\{\hist{A/1}, \hist{B/0}\right\}$ and $\left\{\hist{A/1}, \hist{A/1,B/0}\right\}$, with two outputs for tip event $\ev{A}$
    \item lowerset $\left\{\hist{B/0}\right\}$ is common to contexts $\left\{\hist{A/i_A}, \hist{B/0}\right\}$ and $\left\{\hist{A/i_A}, \hist{A/i_A,B/0}\right\}$, with two outputs for tip event $\ev{B}$
\end{itemize}
Causality and quasi-normalisation equations for all three covers are depicted below.
\begin{center}
    \begin{tabular}{ccc}
    \includegraphics[width=5cm]{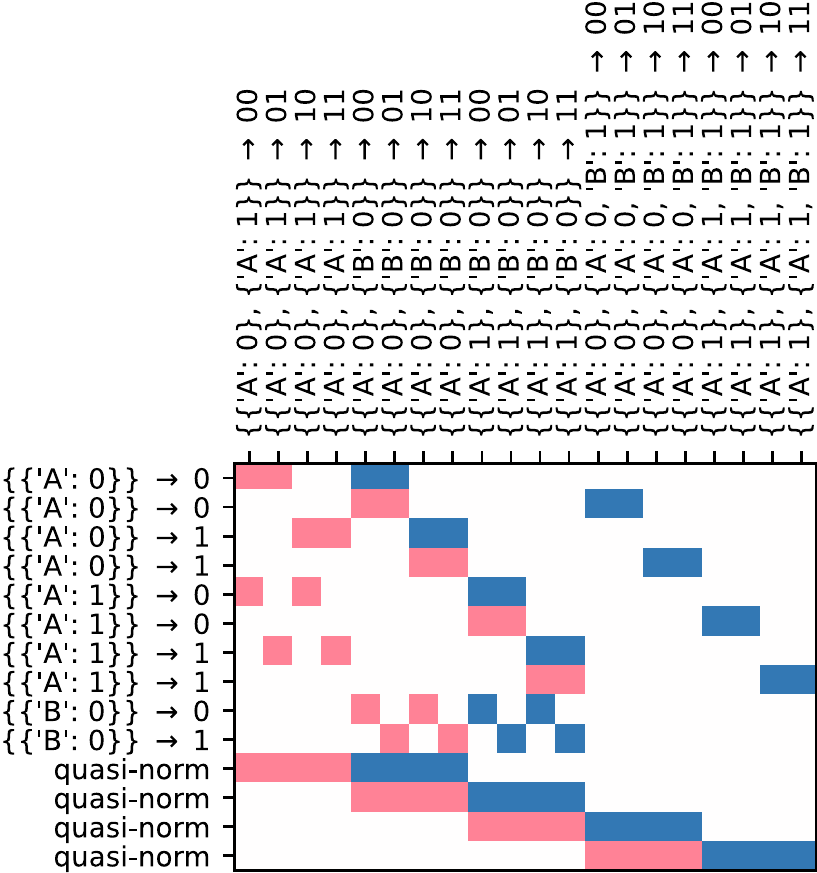}
    &
    \includegraphics[width=5cm]{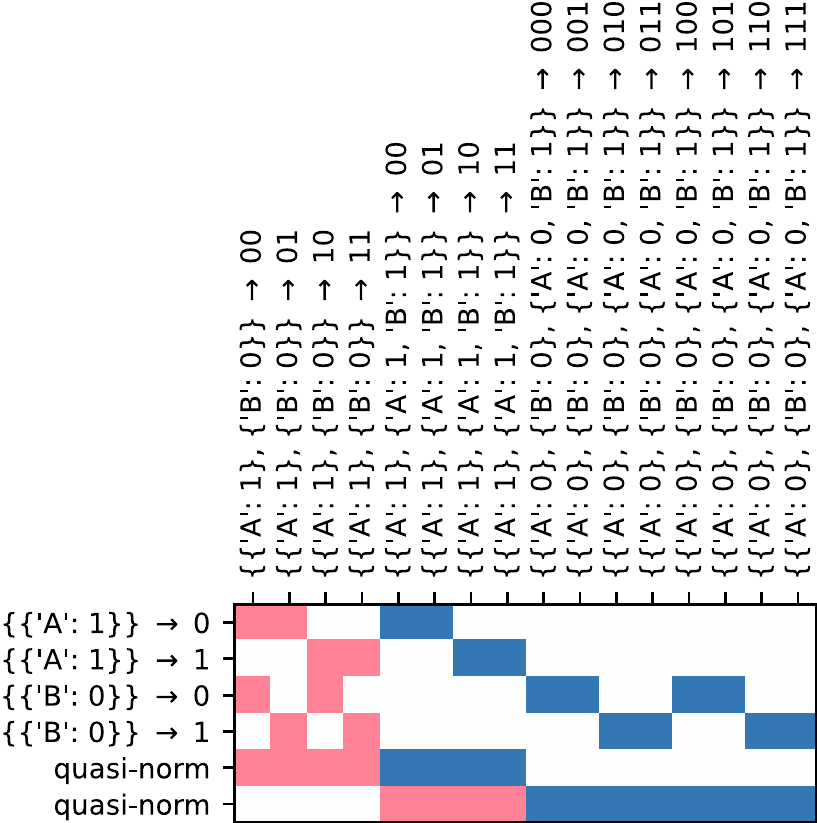}
    &
    \includegraphics[width=5cm]{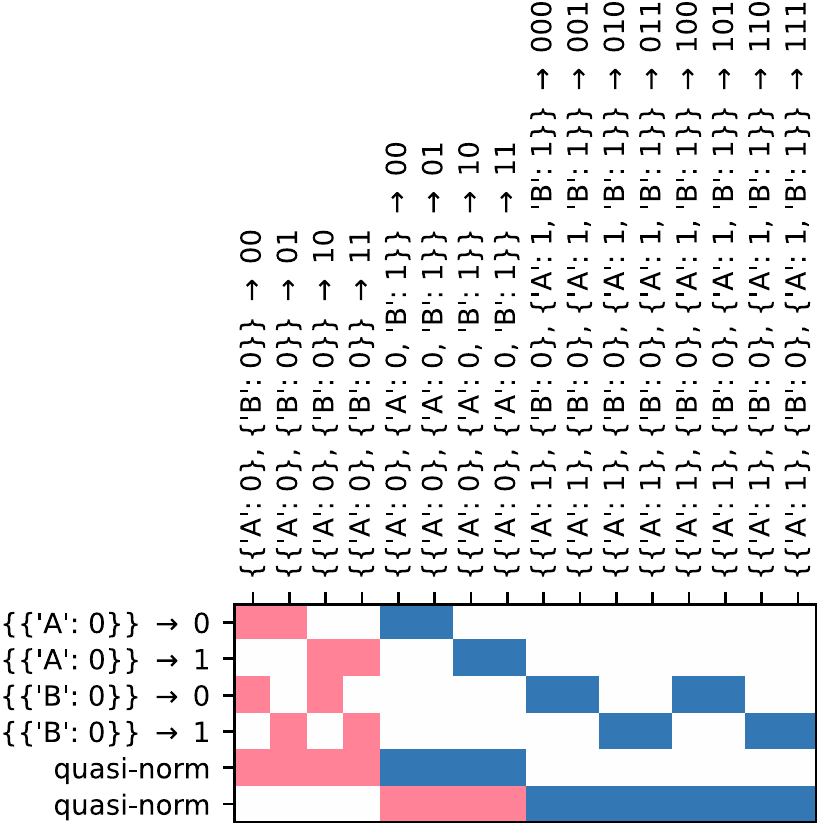}
    \\
    cover \#11
    &
    cover \#12
    &
    cover \#13
    \end{tabular}
\end{center}
There are 3 contexts in each one of covers \#12 and \#13:
\[
\begin{array}{l}
    \left\{\hist{A/i_A}, \hist{B/0}\right\}
    \\
    \left\{\hist{A/i_A}, \hist{A/i_A,B/1}\right\}
    \\
    \left\{\hist{A/1-i_A}, \hist{B/0}, \hist{A/1-i_A,B/1}\right\}
\end{array}
\]
where $i_A = 1$ for cover \#12 and $i_A = 0$ for cover \#13.
There are 4 causality equations for each:
\begin{itemize}
    \item lowerset $\left\{\hist{A/i_A}\right\}$ is common to contexts $\left\{\hist{A/i_A}, \hist{B/0}\right\}$ and $\left\{\hist{A/i_A}, \hist{A/i_A,B/1}\right\}$, with two outputs for tip event $\ev{A}$
    \item lowerset $\left\{\hist{B/0}\right\}$ is common to contexts $\left\{\hist{A/1-i_A}, \hist{B/0}, \hist{A/1-i_A,B/1}\right\}$ and $\left\{\hist{A/i_A}, \hist{B/0}\right\}$, with two outputs for tip event $\ev{B}$
\end{itemize}
Causality and quasi-normalisation equations for all three covers are depicted at the bottom of the previous page.
As our final example for this space, below are the causality and quasi-normalisation equations for cover \#87, the coarsest non-classical cover.
\begin{center}
    \begin{tabular}{c}
    \includegraphics[width=12cm]{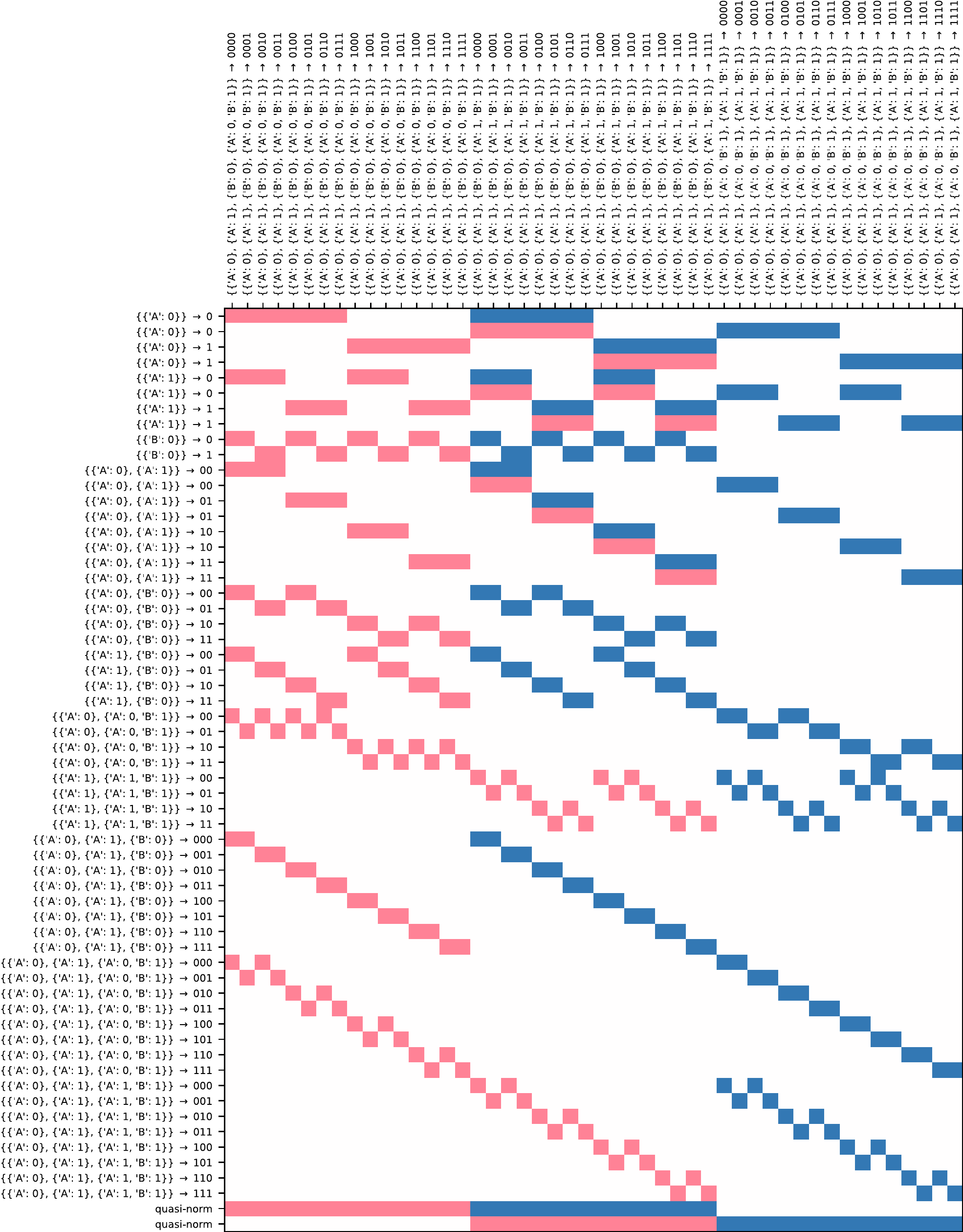}
    \end{tabular}
\end{center}
Cover \#87 has 3 contexts:
\[
\begin{array}{l}
    \left\{
        \hist{A/0},
        \hist{A/1},
        \hist{B/0},
        \hist{A/0,B/1}
    \right\}
    \\
    \left\{
        \hist{A/0},
        \hist{A/1},
        \hist{B/0},
        \hist{A/1,B/1}
    \right\}
    \\
    \left\{
        \hist{A/0},
        \hist{A/1},
        \hist{A/0,B/1},
        \hist{A/1,B/1}
    \right\}
\end{array}
\]
There are 58 causality equations:
\begin{itemize}
    \item lowerset $\left\{\hist{A/0}\right\}$ is common to all three contexts, with two outputs for $\ev{A}$
    \item lowerset $\left\{\hist{A/1}\right\}$ is common to all three contexts, with two outputs for $\ev{A}$
    \item lowerset $\left\{\hist{B/0}\right\}$ is common to the first two contexts, with two outputs for $\ev{B}$
    \item lowerset $\left\{\hist{A/0},\hist{A/1}\right\}$ is common to all three contexts, with two pairs of outputs for $\ev{A}$ (one output when input is 0, one output when input is 1)
    \item lowerset $\left\{\hist{A/0},\hist{B/0}\right\}$ is common to the first two contexts, with two outputs for $\ev{A}$ and two outputs for $\ev{B}$
    \item lowerset $\left\{\hist{A/1},\hist{B/0}\right\}$ is common to the first two contexts, with two outputs for $\ev{A}$ and two outputs for $\ev{B}$
    \item lowerset $\left\{\hist{A/0},\hist{A/0,B/1}\right\}$ is common to the first and third contexts, with two outputs for $\ev{A}$ and two outputs for $\ev{B}$
    \item lowerset $\left\{\hist{A/1},\hist{A/1,B/1}\right\}$ is common to the second and third contexts, with two outputs for $\ev{A}$ and two outputs for $\ev{B}$
    \item lowerset $\left\{\hist{A/0}, \hist{A/1},\hist{B/0}\right\}$ is common to the first and second contexts, with two pairs of outputs for $\ev{A}$ and two outputs for $\ev{B}$
    \item lowerset $\left\{\hist{A/0}, \hist{A/1},\hist{A/0,B/1}\right\}$ is common to the first and third contexts, with two pairs of outputs for $\ev{A}$ and two outputs for $\ev{B}$
    \item lowerset $\left\{\hist{A/0}, \hist{A/1},\hist{A/1,B/1}\right\}$ is common to the second and third contexts, with two pairs of outputs for $\ev{A}$ and two outputs for $\ev{B}$
\end{itemize}
Cover \#88, finally, is the classical cover: it has 1 context, no causality equations and no quasi-normalisation equations.

As our final example, we consider the following causally incomplete space on 3 events, which will later be used to support an empirical model describing a classical switch controlled by the output of a contextual triangle.
\begin{center}
    \begin{tabular}{c}
    \includegraphics[width=8cm]{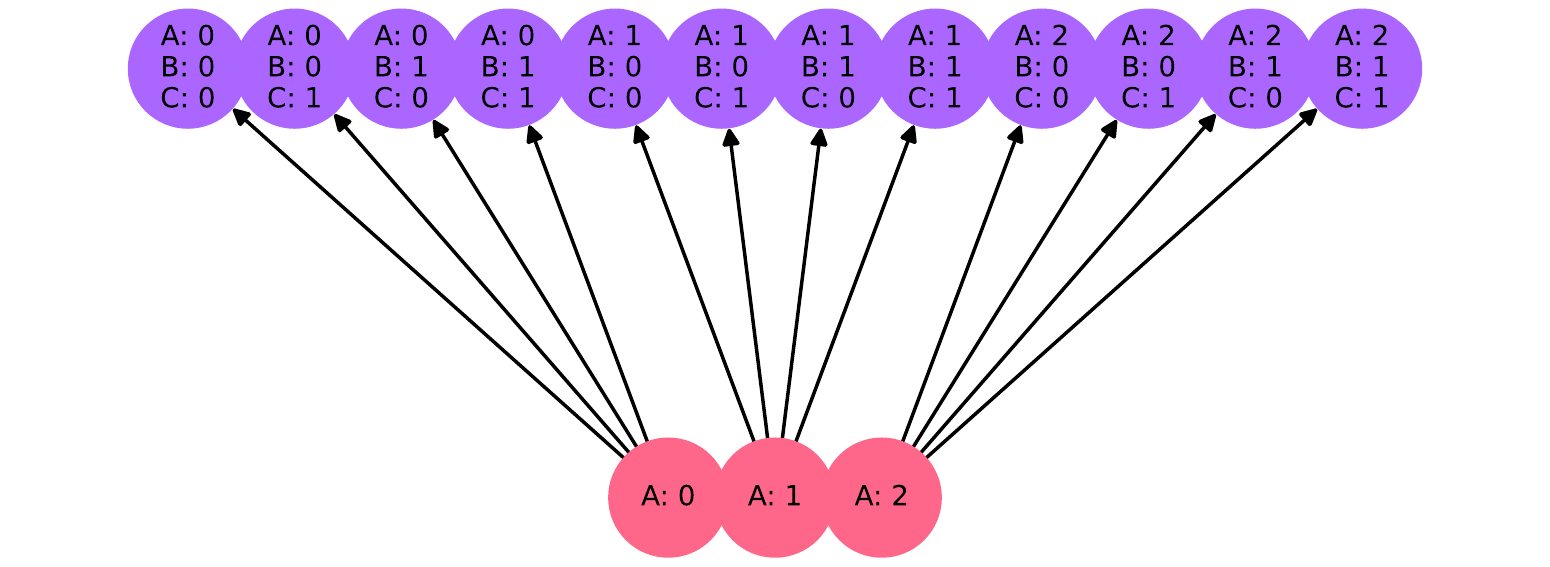}
    \end{tabular}
\end{center}
On this space, we consider the following cover, consisting of the 12 contexts indexed by all possible combinations of $i_A, j_A \in \{0, 1, 2\}$ and $i_B, i_C \in \{0,1\}$, where $i_A \neq j_A$:
\[
\scalebox{0.8}{$
    \lambda_{\{i_A,j_A\},i_B,i_C}:=
    \left\{
    \hist{A/i_A},
    \hist{A/j_A},
    \hist{A/i_A,B/i_B,C/i_C},
    \hist{A/j_A,B/i_B,C/i_C}
    \right\}
$}
\]
Below is a depiction of the causality and quasi-normalisation equations for this cover, where $\ev{A}$ has output fixed to $0$ (for simplicity), while $\ev{B}$ and $\ev{C}$ have binary output $\{0, 1\}$.
\begin{center}
    \includegraphics[width=\textwidth]{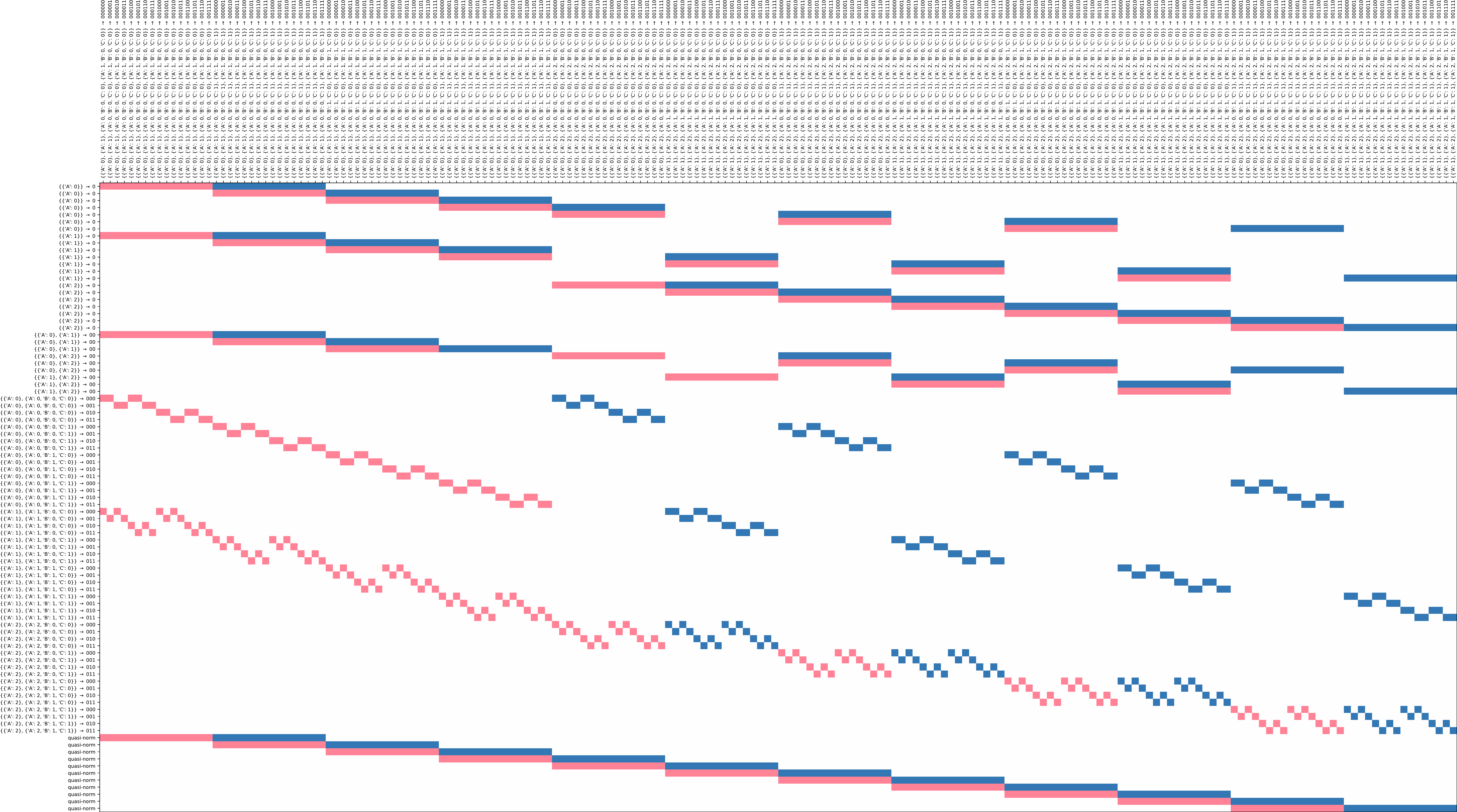}
\end{center}
There are 78 causality equations:
\begin{itemize}
    \item for each $i_A \in \{0, 1, 2\}$, lowerset $\left\{\hist{A/i_A}\right\}$ is common to the 8 contexts in the form $\lambda_{\{i_A,j_A\},i_B,i_C}$, with output for $\ev{A}$ fixed to 0
    \item for each $i_A, j_A \in \{0, 1, 2\}$ such that $i_A \neq j_A$, lowerset $\left\{\hist{A/i_A}, \hist{A/j_A}\right\}$ is common to the 4 contexts in the form $\lambda_{\{i_A,j_A\},i_B,i_C}$, with both outputs for $\ev{A}$ fixed to 0
    \item for each $i_A \in \{0, 1, 2\}$ and $i_B, i_C \in \{0, 1\}$, lowerset $\left\{\hist{A/i_A}, \hist{A/i_A,B/i_B,C/i_C}\right\}$ is common to the 2 contexts in the form $\lambda_{\{i_A,j_A\},i_B,i_C}$, with output for $\ev{A}$ fixed to 0 and binary output for $\ev{B}$ and $\ev{C}$
\end{itemize}
Columns are indexed with the following pattern, where $o_{B,i}$ (resp. $o_{C,i}$) is the output for $\ev{B}$ (resp. $\ev{C}$) when the input at $\ev{A}$ is $i \in \{0, 1, 2\}$:
\[
\lambda_{\{i_A,j_A\},i_B,i_C} \rightarrow 00o_{B,i_A}o_{B,j_A}o_{C,i_A}o_{C,j_A}
\]
The first two entries $o_{A,i}$ are fixed to $0$ because the output at $\ev{A}$ is fixed to $0$.

\subsection{Causal inseparability for non-standard empirical models}
\label{subsection:causal-inseparability-nonstd}

An important aspect of the definition of causal separability is the requirement for empirical models living in different causaltopes to be put into contact: the components $\underline{v}^{(z)}$ appearing in the causal decomposition $\underline{u} = \sum_{z \in Z}\underline{v}^{(z)}$ of an empirical model according to Definition \ref{definition:decomposition-supported-fraction-causal} (p.\pageref{definition:decomposition-supported-fraction-causal}) are empirical models which live in sub-spaces $\Theta^{(z)}$ of the space upon which $\underline{u}$ is defined, and it is the job of Proposition \ref{proposition:subspace-hierarchy-causaltopes} (p.\pageref{proposition:subspace-hierarchy-causaltopes}) to guarantee that the causaltope for each sub-space $\Theta^{(z)}$ is a sub-causaltope of the causaltope for the main space $\Theta$.

Except, we now have a problem: Proposition \ref{proposition:subspace-hierarchy-causaltopes} only mentions the standard cover, and all subsequent definitions explicitly require that the standard cover is used for both the main space $\Theta$ and each sub-space $\Theta^{(z)}$.
In order to define, and study, the causal (in)separability of empirical models over non-standard covers, we must generalise Proposition \ref{proposition:subspace-hierarchy-causaltopes}, and this requires the introduction of some new mathematical machinery.

As a concrete example, we consider the non-standard empirical model for a scenario where Alice controls the causal order between Bob and Charlie by using a device which implements the contextual triangle. Specifically:
\begin{enumerate}
    \item Alice is in possession of a device implementing the empirical model for the contextual triangle: her input $i_A \in \{0, 1, 2\}$ is fed into the device, which spits out a binary value.
    \item The value yielded Alice's device is encoded into the Z basis of a qubit and fed to the control of a quantum switch with the state $|+\rangle$ as its input. As her output, Alice returns $0$, regardless of the binary value yielded by her device.
    \item Bob and Charlie are in the quantum switch, and they do the same thing: they perform an X measurement on the incoming qubit, using the measurement outcome as their individual output, and encode their individual input into the X basis of the outgoing qubit.
\end{enumerate}
The resulting empirical model is as follows:
\begin{center}
    \includegraphics[width=14cm]{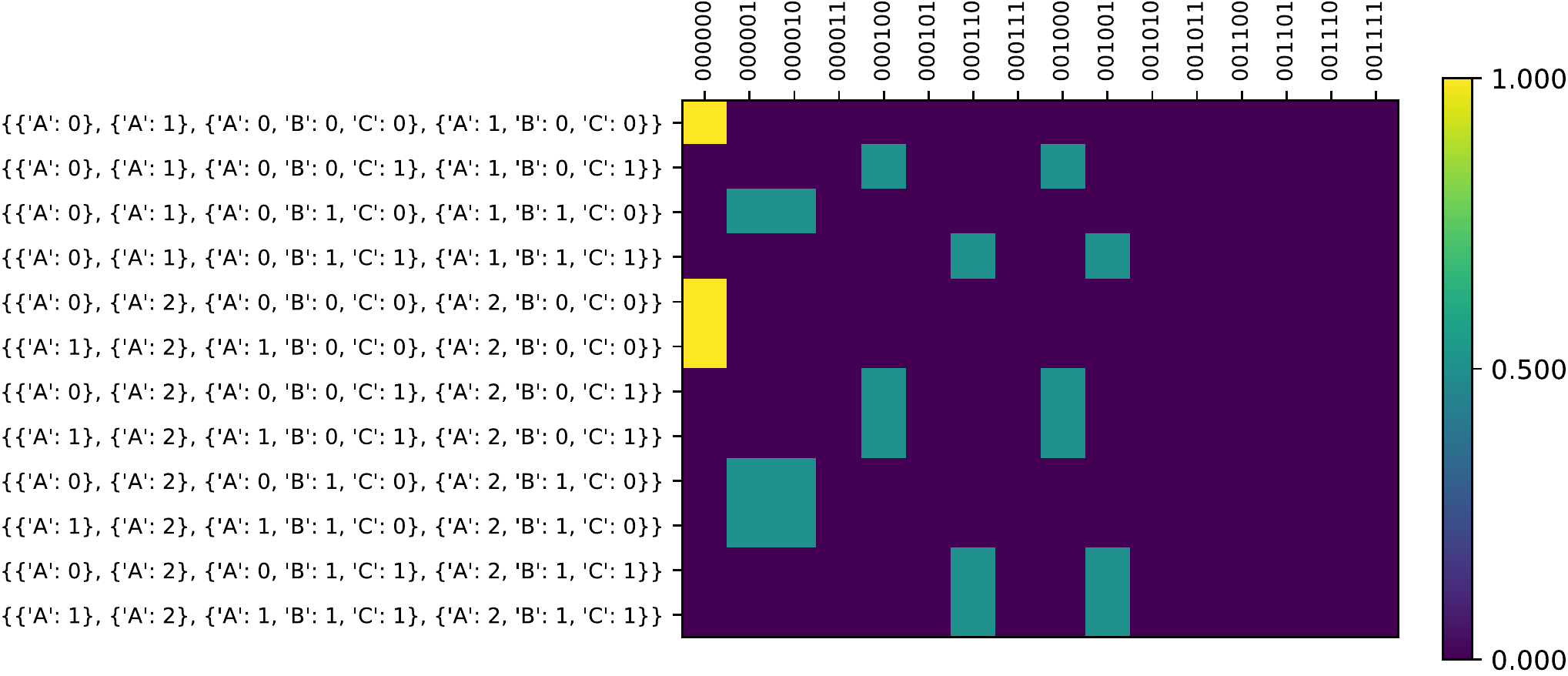}
\end{center}
A figure depicting the scenario is shown below on the left, while the causally incomplete space of input histories $\Theta$ supporting the empirical model is shown below on the right:
\begin{center}
    \begin{tabular}{cc}
    \raisebox{15mm}{\scalebox{1.25}{
    \begin{tikzpicture}
	\begin{pgfonlayer}{nodelayer}
		\node [style=none] (37) at (3.25, 4.5) {};
		\node [style=none] (38) at (-3, 4.5) {};
		\node [style=none] (39) at (-3, 0.25) {};
		\node [style=none] (40) at (3.25, 0.25) {};
		\node [style=none] (41) at (0.75, 4.5) {};
		\node [style=none] (42) at (-2.25, 3) {};
		\node [style=black dot] (43) at (-1, 3) {};
		\node [style=none] (44) at (-2.25, 1.75) {};
		\node [style=black dot] (45) at (-1, 1.75) {};
		\node [style=none] (46) at (-1, 1) {};
		\node [style=none] (47) at (-2.25, -0.75) {};
		\node [style=none] (48) at (-2.25, 5.5) {};
		\node [style=none] (54) at (2.25, 3) {};
		\node [style=black dot] (55) at (1, 3) {};
		\node [style=none] (56) at (2.25, 1.75) {};
		\node [style=black dot] (57) at (1, 1.75) {};
		\node [style=none] (58) at (1, 1) {};
		\node [style=none] (59) at (1, 3.75) {};
		\node [style=none] (60) at (2.25, -0.75) {};
		\node [style=none] (61) at (2.25, 5.5) {};
		\node [style=none] (62) at (-1, 3.75) {};
		\node [style=none] (67) at (0.5, 1) {};
		\node [style=none] (68) at (0.5, 3.75) {};
		\node [style=none] (69) at (2.5, 3.75) {};
		\node [style=none] (70) at (2.5, 1) {};
		\node [style=none] (71) at (-2.5, 1) {};
		\node [style=none] (72) at (-2.5, 3.75) {};
		\node [style=none] (73) at (-0.5, 3.75) {};
		\node [style=none] (74) at (-0.5, 1) {};
		\node [style=upground] (80) at (0.75, 5.5) {};
		\node [style=none] (101) at (-0.75, 4.5) {};
		\node [style=upground] (102) at (-0.75, 5.5) {};
		\node [style=none] (103) at (-0.875, 0.25) {};
		\node [style=point] (104) at (-0.875, -0.75) {$+$};
		\node [style=none] (113) at (0.05, -4) {};
		\node [style=none] (114) at (0.05, -0.25) {};
		\node [style=none] (115) at (2, -0.25) {};
		\node [style=none] (116) at (2, -4) {};
		\node [style=red label] (117) at (-0.5, -2.5) {$A$};
		\node [style=red label] (120) at (3.75, 2.25) {};
		\node [style=red label] (124) at (-3.75, 2.25) {$B$};
		\node [style=red label] (125) at (3.75, 2.25) {$C$};
		\node [style=none] (177) at (1, -0.75) {};
		\node [style=none] (178) at (1, 0.25) {};
		\node [style=none] (179) at (1, -2.5) {};
		\node [style=none] (181) at (2, -0.75) {};
		\node [style=box] (183) at (1, -2.5) {};
		\node [style=none] (184) at (1, -4.5) {};
		\node [style=eslabel] (185) at (-2.25, 6) {$\{0,1\}$};
		\node [style=eslabel] (186) at (2.25, 6) {$\{0,1\}$};
		\node [style=eslabel] (187) at (1, -5) {$\{0,1,2\}$};
		\node [style=small point] (188) at (1, -2.5) {};
	\end{pgfonlayer}
	\begin{pgfonlayer}{edgelayer}
		\draw [style=pink thin] (39.center)
			 to (38.center)
			 to (37.center)
			 to (40.center)
			 to cycle;
		\draw [style=classical, in=-90, out=90] (44.center) to (43);
		\draw [style=classical, in=-90, out=90] (45) to (42.center);
		\draw (46.center) to (45);
		\draw [style=classical] (42.center) to (48.center);
		\draw [style=classical] (44.center) to (47.center);
		\draw [style=classical, in=-90, out=90] (56.center) to (55);
		\draw [style=classical, in=-90, out=90] (57) to (54.center);
		\draw (58.center) to (57);
		\draw (55) to (59.center);
		\draw [style=classical] (54.center) to (61.center);
		\draw [style=classical] (56.center) to (60.center);
		\draw [style=carta da zucchero thin] (68.center)
			 to (67.center)
			 to (70.center)
			 to (69.center)
			 to cycle;
		\draw [style=carta da zucchero thin] (72.center)
			 to (71.center)
			 to (74.center)
			 to (73.center)
			 to cycle;
		\draw [style=classical, in=-90, out=90] (44.center) to (43);
		\draw [style=classical, in=-90, out=90] (45) to (42.center);
		\draw [style=classical, in=-90, out=90] (56.center) to (55);
		\draw [style=classical, in=-90, out=90] (57) to (54.center);
		\draw [style=bold black] (46.center) to (45);
		\draw [style=bold black] (55) to (59.center);
		\draw [style=bold black] (58.center) to (57);
		\draw [style=bold black] (62.center) to (43);
		\draw [style=bold black] (80) to (41.center);
		\draw [style=classical] (42.center) to (48.center);
		\draw [style=classical] (44.center) to (47.center);
		\draw [style=classical] (54.center) to (61.center);
		\draw [style=classical] (56.center) to (60.center);
		\draw [style=bold black] (103.center) to (104);
		\draw [style=bold black] (102) to (101.center);
		\draw [style=carta da zucchero thin] (116.center)
			 to (115.center)
			 to (114.center)
			 to (113.center)
			 to cycle;
		\draw [style=bold black] (177.center) to (179.center);
		\draw [style=bold black] (178.center) to (177.center);
		\draw [style=classical] (183) to (184.center);
	\end{pgfonlayer}
\end{tikzpicture}
    }}
    &
    \includegraphics[width=8cm]{svg-inkscape/2303-causeqs-space-ctxtri.pdf}
    \end{tabular}
\end{center}
The cover $\mathcal{C}$ upon which the empirical model is defined consists of 12 contexts, indexed by all possible combinations of $i_A, j_A \in \{0, 1, 2\}$ and $i_B, i_C \in \{0,1\}$, where $i_A \neq j_A$:
\[
\scalebox{0.8}{$
    \lambda_{\{i_A,j_A\},i_B,i_C}:=
    \left\{
    \hist{A/i_A},
    \hist{A/j_A},
    \hist{A/i_A,B/i_B,C/i_C},
    \hist{A/j_A,B/i_B,C/i_C}
    \right\}
$}
\]
Columns of the empirical model are indexed with the following pattern, where $o_{B,i}$ (resp. $o_{C,i}$) is the output for $\ev{B}$ (resp. $\ev{C}$) when the input at $\ev{A}$ is $i \in \{0, 1, 2\}$:
\[
00o_{B,i_A}o_{B,j_A}o_{C,i_A}o_{C,j_A}
\]
The first two entries $o_{A,i}$ are fixed to $0$ because the output at $\ev{A}$ is fixed to $0$.
The 78 causality equations for this cover were discussed at the end of the previous Subsection.

The first step in our generalisation is the observation that empirical models for the space $\Theta$ and the subspace $\Theta'$ will be defined, in general, on different covers.
Given a cover on $\Theta$, the first step of our journey is to define what the corresponding ``induced'' cover on $\Theta'$ should be.

\begin{definition}
\label{definition:cover-restriction}
Let $\Theta' \leq \Theta$ be a space of input histories such that $\Events{\Theta} = \Events{\Theta'}$ and $\Inputs{\Theta} = \Inputs{\Theta'}$.
Let $\underline{O} = (O_\omega)_{\omega \in \Events{\Theta}}$ be a family of non-empty sets of outputs and let $\mathcal{C} \in \Covers{\Theta}$ be any cover.
The lowerset $\restrict{\lambda}{\Theta'}$ \emph{induced onto subspace $\Theta'$} by a lowerset $\lambda \in \Lsets{\Theta}$ is defined as follows:
\begin{equation}
    \restrict{\lambda}{\Theta'}
    := \bigcup_{k \in \lambda} \downset{k} \!\cap \,\Theta'
    = \suchthat{
        h \in \Theta'
    }{
        \exists k \in \lambda.\, h \leq k
    }
\end{equation}
where we exploited the fact that $\Ext{\Theta'} \supseteq \Ext{\Theta}$ to view each $k \in \lambda \subseteq \Ext{\Theta}$ as a $k \in \Ext{\Theta'}$.
The cover $\restrict{\mathcal{C}}{\Theta'}$ \emph{induced onto subspace $\Theta'$} by the cover $\mathcal{C}$ is then defined as follows:
\begin{equation}
    \restrict{\mathcal{C}}{\Theta'}
    := \max\suchthat{
        \restrict{\lambda}{\Theta'}
    }{
        \lambda \in \mathcal{C}
    }
\end{equation}
where $\max$ indicates that we only keep the maximal induced lowersets under inclusion, ensuring that $\restrict{\mathcal{C}}{\Theta'}$ is a cover of $\Theta'$ (i.e. an antichain of lowersets).
\end{definition}

\begin{lemma}
\label{lemma:context-restriction-subspace}
Let $\Theta' \leq \Theta$ be a space of input histories such that $\Events{\Theta} = \Events{\Theta'}$ and $\Inputs{\Theta} = \Inputs{\Theta'}$.
For every lowerset $\lambda \in \Lsets{\Theta}$, the induced lowerset $\restrict{\lambda}{\Theta'}$ is a subspace of $\lambda$:
\[
    \begin{array}{rcl}
        \restrict{\lambda}{\Theta'} &\leq& \lambda
        \\
        \Ext{\restrict{\lambda}{\Theta'}} &\supseteq& \Ext{\lambda}
    \end{array}
\]
Furthermore, $\dom{\lambda} = \dom{\restrict{\lambda}{\Theta'}}$.
\end{lemma}
\begin{proof}
See \ref{proof:lemma:context-restriction-subspace}
\end{proof}

To exemplify our newly defined notion of induced cover, we consider the following subspace $\Theta'$, one of the 8 causal completions of $\Theta$:
\begin{center}
    \begin{tabular}{c}
    \includegraphics[width=8cm]{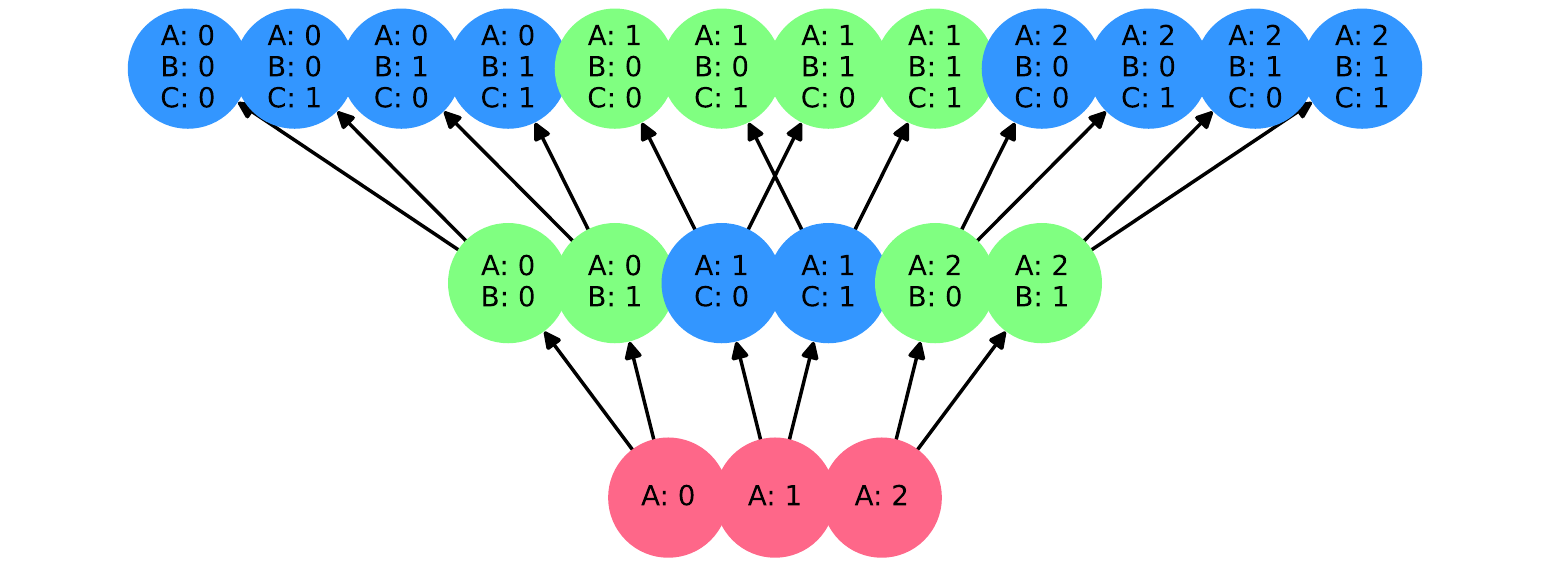}
    \end{tabular}
\end{center}
For each context $\lambda_{\{i_A,j_A\},i_B,i_C}$, the induced context $\restrict{\lambda_{\{i_A,j_A\},i_B,i_C}}{\Theta'}$ additionally contains $\hist{A/k,B/i_B}$ when $k \in \{0, 2\}$ and $\hist{A/k,C/i_C}$ when $k = 1$, for both $k:=i_A$ and $k:=j_A$. The induced cover $\restrict{\mathcal{C}}{\Theta'}$ then takes the following form:
\[
\scalebox{0.6}{$
\begin{array}{rcl}
    \restrict{\lambda_{\{0,1\},i_B,i_C}}{\Theta'} &=&
    \left\{
    \hist{A/0},
    \hist{A/1},
    \hist{A/0,B/i_B},
    \hist{A/1,C/i_C},
    \hist{A/0,B/i_B,C/i_C},
    \hist{A/1,B/i_B,C/i_C}
    \right\}
    \\
    \restrict{\lambda_{\{0,2\},i_B,i_C}}{\Theta'} &=&
    \left\{
    \hist{A/0},
    \hist{A/2},
    \hist{A/0,B/i_B},
    \hist{A/2,B/i_B},
    \hist{A/0,B/i_B,C/i_C},
    \hist{A/2,B/i_B,C/i_C}
    \right\}
    \\
    \restrict{\lambda_{\{1,2\},i_B,i_C}}{\Theta'} &=&
    \left\{
    \hist{A/1},
    \hist{A/2},
    \hist{A/1,C/i_C},
    \hist{A/2,B/i_B},
    \hist{A/1,B/i_B,C/i_C},
    \hist{A/2,B/i_B,C/i_C}
    \right\}
\end{array}
$}
\]
This is a particularly fortuitous case, where the contexts in the induced cover correspond bijectively to the contexts of the original cover, but note that this may not hold in general.

The assumptions of Proposition \ref{proposition:subspace-hierarchy-causaltopes} (p.\pageref{proposition:subspace-hierarchy-causaltopes}) include the requirement that $\Theta$ and $\Theta'$ have the same maximal input histories.
This assumption implies that the two spaces have exactly the same polytope of pseudo-empirical models, significantly simplifying the proof of the result:
\[
    \PsEmpModels{\StdCov{\Theta}, \underline{O}}
    \simeq
    \prod_{k \in \max\Ext{\Theta}}
    \Dist{
        \prod_{\omega \in \dom{k}}
        O_\omega
    }
\]
The assumption also implies that the cover induced onto $\Theta'$ by the standard cover $\StdCov{\Theta}$ on $\Theta$ is the standard cover $\StdCov{\Theta'}$ on $\Theta'$.
This implies backward-compatibility of Proposition \ref{proposition:subspace-hierarchy-causaltopes-general} (p.\pageref{proposition:subspace-hierarchy-causaltopes-general}) below, stating inclusions of causaltopes for general covers, with the original Proposition \ref{proposition:subspace-hierarchy-causaltopes}, restricted to the special case of standard covers.

For covers other than the standard cover, a relationship between the pseudo-empirical models of the original cover $\mathcal{C}$ and those of the induced cover $\restrict{\mathcal{C}}{\Theta'}$ still exists, but turns out to be somewhat more complicated.
Recall the definition of pseudo-empirical models for a generic cover $\mathcal{C}$:
\[
    \PsEmpModels{\mathcal{C}, \underline{O}}
    =
    \prod_{\lambda \in \mathcal{C}}
    \Dist{\prod\limits_{\omega \in \dom{\lambda}} \left(O_\omega\right)^{\TipEqCls{\lambda}{\omega}}}
\]
where $\TipEqCls{\lambda}{\omega}$ is the set of equivalence classes of tip histories for $\omega$ under the equivalence relation $\histconstrSym{\omega}$, constraining histories to yield the same output value at $\omega$ for all extended functions on $\lambda$:
\[
\begin{array}{rl}
    \TipEqCls{\lambda}{\omega}
    &=
    \suchthat{\histconstreqcls{h}{\omega}}{h \in \TipHists{\lambda}{\omega}}\vspace{1mm}
    \\
    &=
    \suchthat{\histconstreqcls{h}{\omega}}{h \in \lambda \,\wedge\, \omega \in \tips{\lambda}{h}}
\end{array}
\]
Pseudo-empirical models are conditional probability distributions: the $\prod_{\lambda \in \mathcal{C}}$ part of the definition of $\PsEmpModels{\mathcal{C}, \underline{O}}$ above is the ``conditional'' part, indexing the distributions $\Dist{J^{(\lambda)}}$ on joint outputs corresponding to each individual context $\lambda \in \mathcal{C}$ in the cover.
Expanding a bit on the definition, the set $J^{(\lambda)}$ of joint outputs for a given context $\lambda \in \mathcal{C}$ is seen to take the following form:
\[
    J^{(\lambda)}
    :=
    \hspace{-8mm}
    \prod\limits_{
        \scriptsize
        \begin{array}{c}
            \omega \in \dom{\lambda}
            \\
            \hspace{1mm}\histconstreqcls{h}{\omega}\hspace{-1.5mm}\in \TipEqCls{\lambda}{\omega}
        \end{array}
    }
    \hspace{-9mm}
    O_\omega
\]
It is convenient to think of the $\omega \in \dom{\lambda}$ and $\histconstreqcls{h}{\omega}\hspace{-1.5mm}\in \TipEqCls{\lambda}{\omega}$ as ``indexes'' of individual output values within a vector $o$ of joint output values, each individual component $o_{(\omega, \histconstreqcls{h}{\omega})}$ freely chosen within the set $O_\omega$.

Our goal is now to connect the pseudo-empirical models for a cover $\mathcal{C}$ on a space $\Theta$ to those for the cover induced on a sub-space $\Theta'$:
\[
    \PsEmpModels{\restrict{\mathcal{C}}{\Theta'}, \underline{O}}
    \lhook\joinrel\longrightarrow
    \PsEmpModels{\mathcal{C}, \underline{O}}
\]
We used an inclusion symbol $\lhook\joinrel\longrightarrow$, rather than the subset symbol $\subseteq$, to highlight that the indexing sets for individual outputs in pseudo-empirical models for the original cover $\mathcal{C}$ and for the induced cover $\restrict{\mathcal{C}}{\Theta'}$ are distinct. (Note: in the case of the standard cover, they were the same.) This is evident when we expand the two polytopes into their definition:
\[
    \prod_{\restrict{\lambda}{\Theta'} \text{ maximal}}
    \Dist{
        \hspace{-3mm}
        \prod\limits_{
            \scriptsize
            \begin{array}{c}
                \omega \in \dom{\lambda}
                \\
                \hspace{1mm}\histconstreqcls{k}{\omega}\hspace{-1.5mm}\in \TipEqCls{\restrict{\lambda}{\Theta'}}{\omega}
            \end{array}
        }
        \hspace{-9mm}
        O_\omega
        \hspace{2mm}
    }
    \hspace{8mm}
    \lhook\joinrel\longrightarrow
    \hspace{5mm}
    \prod_{\lambda \in \mathcal{C}}
    \Dist{
        \hspace{-3mm}
        \prod\limits_{
            \scriptsize
            \begin{array}{c}
                \omega \in \dom{\lambda}
                \\
                \hspace{1mm}\histconstreqcls{h}{\omega}\hspace{-1.5mm}\in \TipEqCls{\lambda}{\omega}
            \end{array}
        }
        \hspace{-9mm}
        O_\omega
        \hspace{2mm}
    }
\]
where we exploited the fact that $\dom{\lambda}=\dom{\restrict{\lambda}{\Theta'}}$.
We already have a mapping from the contexts $\lambda \in \mathcal{C}$ of the original cover to lowersets $\restrict{\lambda}{\Theta'}$; the maxima of the latter, under inclusion, are the contexts of the induced cover $\restrict{\mathcal{C}}{\Theta'}$.
One more ingredient is needed to describe our desired connection: for a given event $\omega \in \dom{\lambda}$ and ``tip historyset'' $\histconstreqcls{h}{\omega}\hspace{-1.5mm}\in \TipEqCls{\lambda}{\omega}$, we need to define a correspoding ``induced tip historyset'' on $\restrict{\lambda}{\Theta'}$, and prove that all tip historysets on $\restrict{\lambda}{\Theta'}$ arise this way.

\begin{proposition}
\label{proposition:induced-tip-historysets}
    Let $\Theta' \leq \Theta$ be a space of input histories such that $\Events{\Theta} = \Events{\Theta'}$ and $\Inputs{\Theta} = \Inputs{\Theta'}$.
    Let $\underline{O} = (O_\omega)_{\omega \in \Events{\Theta}}$ be a family of non-empty sets of outputs and let $\mathcal{C} \in \Covers{\Theta}$ be any cover.
    For each event $\omega \in \dom{\lambda}$ and each \emph{tip historyset} $\histconstreqcls{h}{\omega}\hspace{-1.5mm}\in \TipEqCls{\lambda}{\omega}$, the \emph{induced tip historyset} $\restrict{\histconstreqcls{h}{\omega}\!\!}{\Theta'} \in \TipEqCls{\restrict{\lambda}{\Theta'}}{\omega}$ is defined as follows:
    \begin{equation}
        \restrict{\histconstreqcls{h}{\omega}\!\!}{\Theta'}
        :=
        \histconstreqcls{k}{\omega} \text{ for any } k \in \Theta' \text{ s.t. } k \leq h
    \end{equation}
    Induced tip historyset are well-defined, and every tip historyset in $\TipEqCls{\restrict{\lambda}{\Theta'}}{\omega}$ arises this way.
\end{proposition}
\begin{proof}
See \ref{proof:proposition:induced-tip-historysets}
\end{proof}

Having established a connection from both the contexts (the ``rows'' of an empirical model ``table'') and events-tip historyset pairs (the ``columns'' of an empirical mode ``table''), we can now explicitly define the embedding of pseudo-empirical models for the induced cover into those for the original cover.

\begin{definition}
\label{definition:pseudo-empmodels-embedding-general}
    Let $\Theta' \leq \Theta$ be a space of input histories such that $\Events{\Theta'} = \Events{\Theta}$ and $\Inputs{\Theta'} = \Inputs{\Theta}$.
    Let $\underline{O} = (O_\omega)_{\omega \in \Events{\Theta}}$ be a family of non-empty sets of outputs and let $\mathcal{C} \in \Covers{\Theta}$ be any cover.
    First, we define a function $s_\lambda$, mapping the indexing set for output values of a context $\lambda$ surjectively onto the indexing set for output values of the corresponding induced lowerset $\restrict{\lambda}{\Theta'}$:
    \begin{equation}
    \begin{array}{rccc}
        s_\lambda:
        & \sum\limits_{\omega \in \dom{\lambda}} \TipEqCls{\lambda}{\omega}
        & \twoheadrightarrow
        & \sum\limits_{\omega \in \dom{\lambda}} \TipEqCls{\restrict{\lambda}{\Theta'}}{\omega}
        \\
        & \left(\omega, \histconstreqcls{h}{\omega}\right)
        & \mapsto
        & \left(\omega, \restrict{\histconstreqcls{h}{\omega}\!\!}{\Theta'}\right)
    \end{array}
    \end{equation}
    From $s_\lambda$, we derive a second function $\sigma_\lambda$, mapping the joint output values of a context $\restrict{\lambda}{\Theta'} \in \restrict{\mathcal{C}}{\Theta'}$ in the induced cover injectively into the joint output values of a corresponding context $\lambda$ in the original cover:
    \begin{equation}
        \begin{array}{rccc}
        \sigma_\lambda:
        &
            \hspace{-3mm}
            \prod\limits_{
                \scriptsize
                \begin{array}{c}
                    \omega \in \dom{\lambda}
                    \\
                    \hspace{1mm}\histconstreqcls{k}{\omega}\hspace{-1.5mm}\in \TipEqCls{\restrict{\lambda}{\Theta'}}{\omega}
                \end{array}
            }
            \hspace{-9mm}
            O_\omega
            \hspace{2mm}
        & \hookrightarrow
        &
            \hspace{-3mm}
            \prod\limits_{
                \scriptsize
                \begin{array}{c}
                    \omega \in \dom{\lambda}
                    \\
                    \hspace{1mm}\histconstreqcls{h}{\omega}\hspace{-1.5mm}\in \TipEqCls{\lambda}{\omega}
                \end{array}
            }
            \hspace{-9mm}
            O_\omega
            \hspace{2mm}
        \\
        & o
        & \mapsto
        & o \circ s_\lambda
    \end{array}
    \end{equation}
    Using $\sigma_\lambda$, we finally define the embedding $i_{\mathcal{C},\Theta',\Theta}$ of the polytope of pseudo-empirical models for the induced cover $\restrict{\mathcal{C}}{\Theta'}$ into the polytope of pseudo-empirical models for the original cover $\mathcal{C}$:
    \begin{equation}
    \begin{array}{rccc}
        i_{\mathcal{C},\Theta',\Theta}:
        & \PsEmpModels{\restrict{\mathcal{C}}{\Theta'}, \underline{O}}
        & \hookrightarrow
        & \PsEmpModels{\mathcal{C}, \underline{O}}
        \\
        & \underline{u}
        & \mapsto
        & \left(
            \Dist{\sigma_\lambda}\left(\underline{u}^{\left(\restrict{\lambda}{\Theta'}\right)}\right)
          \right)_{\lambda \in \mathcal{C}}
    \end{array}
    \end{equation}
    Explicitly, the embedding acts as follows on the pseudo-empirical models for the induced cover:
    \[
    \begin{array}{rccc}
        i_{\mathcal{C},\Theta',\Theta}:
        & \PsEmpModels{\restrict{\mathcal{C}}{\Theta'}, \underline{O}}
        & \hookrightarrow
        & \PsEmpModels{\mathcal{C}, \underline{O}}
        \\
        & \underline{u}
        & \mapsto
        & \left(
            \left(
                u^{(\restrict{\lambda}{\Theta'})}_{\omega, \restrict{\histconstreqcls{h}{\omega}\!\!}{\Theta'}}
            \right)_{\omega, \histconstreqcls{h}{\omega}}
          \right)_{\lambda}
    \end{array}
    \]
\end{definition}

\begin{remark}
\label{remark:pseudo-empmodels-embedding-vec-general}
    The embedding $i_{\mathcal{C},\Theta',\Theta}$ extends to a linear embedding between the vector spaces spaned by the pseudo-empirical models:
    \[
        \PsEmpModelsVec{\restrict{\mathcal{C}}{\Theta'}, \underline{O}}
        \stackrel{i_{\mathcal{C},\Theta',\Theta}}{\lhook\joinrel\longrightarrow}
        \PsEmpModelsVec{\mathcal{C}, \underline{O}}
    \]
\end{remark}

We are now, finally, in a position to generalise Proposition \ref{proposition:subspace-hierarchy-causaltopes-general} (p.\pageref{proposition:subspace-hierarchy-causaltopes-general}) from standard covers to arbitrary covers.
With that in hand, we can then proceed to generalise the relevant parts of Definition \ref{definition:component-supported-fraction-causal} (p.\pageref{definition:component-supported-fraction-causal}) and Definition \ref{definition:decomposition-supported-fraction-causal} (p.\pageref{definition:decomposition-supported-fraction-causal}), i.e. the notions of \emph{component} and \emph{causal decomposition} of an empirical model, to arbitrary covers.
The remaining machinery introduced by the two definitions, as well as that introduced by Definition \ref{definition:causal-separability} (p.\pageref{definition:causal-separability}), entirely depend on those two notions, and generalise accordingly.
Proposition \ref{proposition:csep-frac-bounded-below-by-csep-local-frac} (p.\pageref{proposition:csep-frac-bounded-below-by-csep-local-frac}) also straightforwardly generalises.

\begin{proposition}
\label{proposition:subspace-hierarchy-causaltopes-general}
    Let $\Theta' \leq \Theta$ be a space of input histories such that $\Events{\Theta'} = \Events{\Theta}$ and $\Inputs{\Theta'} = \Inputs{\Theta}$.
    Let $\underline{O} = (O_\omega)_{\omega \in \Events{\Theta}}$ be a family of non-empty sets of outputs and let $\mathcal{C} \in \Covers{\Theta}$ be any cover.
    The embedding $i_{\mathcal{C},\Theta',\Theta}$ of the polytope of pseudo-empirical models for the induced cover $\restrict{\mathcal{C}}{\Theta'}$ into the polytope of pseudo-empirical models for the original cover $\mathcal{C}$ induces an embedding of causality equations, and hence an embedding of causaltopes:
    \[
    \begin{array}{rcl}
        \CausEqs{\restrict{\mathcal{C}}{\Theta'}, \underline{O}}_{\restrict{\mu}{\Theta'}, \restrict{\lambda}{\Theta'}, \restrict{\lambda'}{\Theta'}}
        &\stackrel{i_{\mathcal{C},\Theta',\Theta}}{\lhook\joinrel\longrightarrow}&
        \CausEqs{\mathcal{C}, \underline{O}}_{\mu, \lambda, \lambda'}
        \\
        \CausEqs{\restrict{\mathcal{C}}{\Theta'}, \underline{O}}
        &\stackrel{i_{\mathcal{C},\Theta',\Theta}}{\lhook\joinrel\longrightarrow}&
        \CausEqs{\mathcal{C}, \underline{O}}
        \\
        \Causaltope{\restrict{\mathcal{C}}{\Theta'}, \underline{O}}
        &\stackrel{i_{\mathcal{C},\Theta',\Theta}}{\lhook\joinrel\longrightarrow}&
        \Causaltope{\mathcal{C}, \underline{O}}
    \end{array}
    \]
\end{proposition}
\begin{proof}
See \ref{proof:proposition:subspace-hierarchy-causaltopes-general}
\end{proof}

\begin{definition}
\label{definition:component-decomposition-general-covers}
    Let $\Theta$ be a space of input histories and let $\underline{O} = (O_\omega)_{\omega \in \Events{\Theta}}$ be a family of non-empty sets of outputs and let $\mathcal{C} \in \Covers{\Theta}$ be any cover.
    Let $\underline{u} \in \Causaltope{\mathcal{C}, \underline{O}}$ be an empirical model for cover $\mathcal{C}$.
    For any $\Theta' \leq \Theta$ such that $\Events{\Theta'} = \Events{\Theta}$ and $\Inputs{\Theta'} = \Inputs{\Theta}$, we give the following definitions:
    \begin{itemize}
        \item A \emph{component} of $\underline{u}$ in $\Theta'$ is a component of $\underline{u}$ in the sub-polytope of constrained conditional probability distributions $\Causaltope{\restrict{\mathcal{C}}{\Theta'}, \underline{O}}$ according to Definition \ref{definition:component-supported-fraction}.
        \item A \emph{(causal) decomposition} of $\underline{u}$ over the sub-spaces $\left(\Theta^{(z)}\right)_{z \in Z}$ is a decomposition of $\underline{u}$ in $\left(\Causaltope{\restrict{\mathcal{C}}{\Theta^{(z)}}, \underline{O}}\right)_{z \in Z}$ according to Definition \ref{definition:decomposition-supported-fraction}.
    \end{itemize}
\end{definition}

\begin{proposition}
\label{proposition:csep-frac-bounded-below-by-csep-local-frac-general}
    Let $\Theta$ be a space of input histories, let $\mathcal{C} \in \Covers{\Theta}$ be any cover and let $e$ be am empirical model for $\mathcal{C}$.
    The causally separable fraction of $e$ is bounded below by its separable non-contextual fraction.
    In particular, if $e$ is non-contextual in a causally separable way, then $e$ is causally separable.
\end{proposition}
\begin{proof}
See \ref{proof:proposition:csep-frac-bounded-below-by-csep-local-frac-general}
\end{proof}

Perhaps unsurprisingly, the empirical model for a quantum switch controlled by the output of a contextual triangle is causally inseparable.
In fact, it's causally separable fraction is zero: the causal order between Bob and Charlie can be inferred from relationship between their inputs and outputs, and the same causal order is a function of the output of a contextual triangle, which itself has zero non-contextual fraction.
This effect is analogous to the one previously observed on the two contextually controlled classical switches, where the causally separable fraction coincided with the local fraction of the controlling Bell scenario, but this time observed on a cover different from the standard one.

\newpage
\subsection{Proofs for Section \ref{section:geometry-causality}}

\subsubsection{Proof of Proposition \ref{proposition:slices-are-polytopes}}
\label{proof:proposition:slices-are-polytopes}
\begin{proof}
By definition of the polytope $K$, we have that $\underline{x} \in K$ if and only if $A \underline{x} = \underline{b}$ and $C \underline{x} = \underline{d}$.
By definition of the affine subspace $W$, we have that $\underline{x} \in W$ if and only if $A' \underline{x} = \underline{b}'$.
As a consequence, $\underline{x} \in K \cap W$ if and only if $A \underline{x} = \underline{b}$,  $A' \underline{x} = \underline{b}'$ and $C \underline{x} = \underline{d}$.
\end{proof}

\subsubsection{Proof of Proposition \ref{proposition:slicing-closed-under-iteration}}
\label{proof:proposition:slicing-closed-under-iteration}
\begin{proof}
This is simply associativity of intersection:
\[
\Slice{V}{\Slice{W}{K}} = V \cap \left(W \cap K\right)
= \left(V \cap W\right)\cap K = \Slice{V\cap W}{K}
\]
\end{proof}

\subsubsection{Proof of Proposition \ref{proposition:probdist-from-slicing-hypercube}}
\label{proof:proposition:probdist-from-slicing-hypercube}
\begin{proof}
Taking the normalisation equations together with the defining inequalities for $[0,1]^{\sqcup_{y \in Y} J^{(y)}} \cap \NormEqs{\underline{J}}$ yields the following system of equations and inequalities:
\[
    \forall y \in Y.\;
    \sum_{j \in J^{(y)}} x^{(y)}_j = 1
    \hspace{2cm}
    \forall y \in Y.\;
    \forall j \in J^{(y)}.\;
    0 \leq x^{(y)}_j \leq 1
\]
The normalisation equation $\sum_{j \in J^{(y)}} x^{(y)}_j = 1$ together with the inequalities $0 \leq x^{(y)}_j$ for all $j \in J^{(y)}$ implies the inequalities $x^{(y)}_j \leq 1$ for all $j \in J^{(y)}$, making them redundant.
We are thus left with the defining system of equations and inequalities for the polytope of conditional probability distributions $\prod_{y \in Y} \Dist{J^{(y)}}$, as claimed.
\end{proof}

\subsubsection{Proof of Proposition \ref{proposition:qnorm-eqs-independent-of-order}}
\label{proof:proposition:qnorm-eqs-independent-of-order}
\begin{proof}
By reflexive-transitive closure, the quasi-normalisation equations are equivalent to the following set of equations, which is independent of the choice of total order on $Y$
\[
\forall y, y' \in Y.\;
\sum_{j \in J^{(y)}} x^{(y)}_j = \sum_{j \in J^{(y')}} x^{(y')}_j
\]
\end{proof}

\subsubsection{Proof of Proposition \ref{proposition:mass-of-qnorm-dist}}
\label{proof:proposition:mass-of-qnorm-dist}
\begin{proof}
Let $y_0 \in Y$ be any element and define $m := \sum_{j \in J^{(y_0)}} x^{(y_0)}_j$.
The quasi-normalisation equations imply that:
\[
\forall y \in Y.\;
\sum_{j \in J^{(y)}} x^{(y)}_j = m
\]
If $m = 0$, then $\underline{u} = m \underline{e}$ for all conditional probability distributions $\underline{e} \in \prod_{y \in Y} \Dist{J^{(y)}}$.
If $m > 0$, then the following $\underline{e}$ is a conditional probability distribution:
\[
    \underline{e} := \frac{1}{m} \underline{u}
\]
We have $m\underline{e} = \underline{u}$ by definition, and $m\underline{e'} = \underline{u} m\underline{e}$ implies $\underline{e'} = \underline{e}$, because $m \neq 0$.
Setting $\mass{\underline{u}} := m$ completes the proof.
\end{proof}

\subsubsection{Proof of Proposition \ref{proposition:hierarchy-of-ccpd-polytopes}}
\label{proof:proposition:hierarchy-of-ccpd-polytopes}
\begin{proof}
We prove point by point:
\begin{enumerate}
    \item   If $V \subseteq U$ then:
            \[
            V \cap \prod\limits_{y \in Y} \Dist{J^{(y)}}
            \subseteq U \cap \prod\limits_{y \in Y} \Dist{J^{(y)}}
            \]
            That is, $\CCPD{V,\underline{J}} \subseteq \CCPD{U,\underline{J}}$.
    \item   The previous point, together with $\langle\CCPD{V,\underline{J}}\rangle \subseteq V$, proves that:
            \[
            \CCPD{\langle\CCPD{V,\underline{J}}\rangle,\underline{J}} \subseteq \CCPD{V,\underline{J}}
            \]
            The equality then follows from the observation that $\CCPD{V,\underline{J}} \subseteq \langle\CCPD{V,\underline{J}}\rangle$.
    \item   If $\QNCCPD{V,\underline{J}} \subseteq \QNCCPD{U,\underline{J}}$, then
            \[
            \QNCCPD{V,\underline{J}}\cap\NormEqs{\underline{J}}
            \subseteq \QNCCPD{U,\underline{J}}\cap\NormEqs{\underline{J}}
            \]
            That is, again, $\CCPD{V,\underline{J}} \subseteq \CCPD{U,\underline{J}}$.
    \item   If $\CCPD{V,\underline{J}} \subseteq \CCPD{U,\underline{J}}$ then:
            \[
            \langle \CCPD{V,\underline{J}}\rangle \subseteq \langle\CCPD{U,\underline{J}} \rangle
            \]
    \item   Hence, if $\CCPD{V,\underline{J}} \subseteq \CCPD{U,\underline{J}}$ then:
            \[
            \QNCCPD{\langle\CCPD{V,\underline{J}}\rangle,\underline{J}}
            \subseteq \QNCCPD{\langle\CCPD{U,\underline{J}}\rangle,\underline{J}}
            \]
            That is, $\QNCCPD{V,\underline{J}} \subseteq \QNCCPD{U,\underline{J}}$.
\end{enumerate}
\end{proof}

\subsubsection{Proof of Proposition \ref{proposition:ccpd-polytopes-meet}}
\label{proof:proposition:ccpd-polytopes-meet}
\begin{proof}
By basic set algebra:
\[
\begin{array}{rcl}
&&\CCPD{V,\underline{J}} \cap \CCPD{U,\underline{J}}
\\
&=& \left(V \cap \prod\limits_{y \in Y} \Dist{J^{(y)}}\right) \cap \left(U \cap \prod\limits_{y \in Y} \Dist{J^{(y)}}\right)
\\
&=& \left(V \cap U \cap \prod\limits_{y \in Y} \Dist{J^{(y)}}\right)
= \CCPD{V\cap U,\underline{J}}
\end{array}
\]
\end{proof}

\subsubsection{Proof of Proposition \ref{proposition:ccpd-from-norm-of-qnccpd}}
\label{proof:proposition:ccpd-from-norm-of-qnccpd}
\begin{proof}
Both claims follow from closure of slicing under iteration.
For the first claim, we use the following observation:
\[
    \QNormEqs{\underline{J}}\cap \NormEqs{\underline{J}} = \NormEqs{\underline{J}}
\]
Now observe the $\sum_{j \in J^{(y)}} x_j^{(y)} = 1$ together with $\sum_{j \in J^{(y)}} x_j^{(y)} - \sum_{j \in J^{(y')}} x_j^{(y')} = 0$ implies $\sum_{j \in J^{(y')}} x_j^{(y')} = 1$, for every $y' \in Y$.
Hence we get the following, which in turn proves the second claim:
\[
    \QNormEqs{\underline{J}}\cap \NormEqs{\underline{J}}^{(y)} = \NormEqs{\underline{J}}
\]
\end{proof}

\subsubsection{Proof of Proposition \ref{proposition:difference-in-super-polytope}}
\label{proof:proposition:difference-in-super-polytope}
\begin{proof}
Because $\underline{u}, \underline{v} \in \prod_{y \in Y} \Dist{J^{(y)}}$ and $\underline{v} \leq \underline{u}$, we have that the difference $\underline{u}-\underline{v} \in \prod_{y \in Y} \Dist{J^{(y)}}$ is itself a conditional probability distribution.
Because $\underline{u} \in \CCPD{U,\underline{J}}$ and $\underline{v} \in \CCPD{V,\underline{J}} \subseteq \CCPD{U,\underline{J}}$, then $\underline{u}, \underline{v} \in U$ and hence the difference $\underline{u}-\underline{v}$ satisfies the constraints imposed by $U$.
We conclude that:
\[
\underline{u}-\underline{v} \in \prod_{y \in Y} \Dist{J^{(y)}} \cap U
= \CCPD{U,\underline{J}}
\]
\end{proof}

\subsubsection{Proof of Proposition \ref{corollary:difference-in-super-polytope-decomp}}
\label{proof:corollary:difference-in-super-polytope-decomp}
\begin{proof}
This follows by iterating Proposition \ref{proposition:difference-in-super-polytope} for each component in the decomposition.
\end{proof}

\subsubsection{Proof of Theorem \ref{theorem:topdist-extdist}}
\label{proof:theorem:topdist-extdist}
\begin{proof}
By the consistency condition, $\Ext{f}(h) = \restrict{\Ext{f}(k)}{\dom{h}}$ for all $h \in \downset{k}$, so the function $\Ext{f} \mapsto \Ext{f}(k)$ is injective.
For every $o \in \prod_{\omega \in \dom{k}}O_\omega$, setting $f(h) := \restrict{o}{\dom{h}}$ defines a causal function, since we automatically have $f(h)_\omega = o_\omega = f(h')_\omega$ whenever $\omega \in \dom{h} \cap \dom{h'}$: since $\Ext{f}(k) = o$, the function $\Ext{f} \mapsto \Ext{f}(k)$ is surjective.
Since $\Ext{f} \mapsto \Ext{f}(k)$ is a bijection, by Lemma 4.37 p.60 of ``The Topology of Causality'' \cite{gogioso2022topology} so is $\Dist{\Ext{f} \mapsto \Ext{f}(k)}$:
\[
\Dist{\Ext{f} \mapsto \Ext{f}(k)}
=
d \mapsto \sum\limits_{\Ext{f}} d(\Ext{f}) \delta_{\Ext{f}(k)} = \topdist{d}
\]
\end{proof}

\subsubsection{Proof of Theorem \ref{theorem:topdist-extdist-generalised}}
\label{proof:theorem:topdist-extdist-generalised}
\begin{proof}
By definition, $\Ext{f}(h')_\omega$ takes a constant value for all $h' \in \histconstreqcls{h}{\omega}$, making $\topdist{\Ext{f}}$ well-defined. By Observation 4.12 p.48 of ``The Topology of Causality'' \cite{gogioso2022topology}, the correspondence is bijective (because extended causal functions are in bijection with causal functions).
Analogously to Theorem \ref{theorem:topdist-extdist}, the bijection lifts to a bijection between the corresponding spaces of distributions: the latter bijection is defined by taking convex-linear combinations, and hence it is a convex-linear function.
\end{proof}

\subsubsection{Proof of Lemma \ref{lemma:output-history-injection}}
\label{proof:lemma:output-history-injection}
\begin{proof}
If $\histconstr{\omega}{h}{h'}$ in $\mu$, and $\mu \subseteq \lambda$, then necessarily $\histconstr{\omega}{h}{h'}$ in $\lambda$.
\end{proof}

\subsubsection{Proof of Proposition \ref{proposition:caus-eqs-chain}}
\label{proof:proposition:caus-eqs-chain}
\begin{proof}
Fix $\mu \in \Lsets{\Theta}$.
If $n_\mu = 0$, then there are no equations associated with $\mu$, so we can restrict our attention to the $\mu$ s.t. $n_\mu \geq 1$.
Consider the following subspace:
\[
    \bigcap_{\lambda \in \mathcal{C}\cap\upset{\mu} }
    \bigcap_{\lambda' \in \mathcal{C}\cap\upset{\mu} }
    \CausEqs{\mathcal{C}, \underline{O}}_{\mu, \lambda, \lambda'}
\]
The linear constraints are exactly those enforcing $\restrict{\underline{u}^{(\lambda)}}{\mu}=\restrict{\underline{u}^{(\lambda')}}{\mu}$ for all $\lambda, \lambda' \in \mathcal{C} \cap \upset{\mu}$.
If we impose a total order on $\mathcal{C} \cap \upset{\mu}$, the exact same constraints can be enforced by a chain of $n_\mu-1$ equations, as follows:
\[
    \bigcap_{i=1}^{n_\mu-1}
    \CausEqs{\mathcal{C}, \underline{O}}_{\mu, \lambda_{\mu,i}, \lambda_{\mu,i+1}}
\]
This concludes our proof.
\end{proof}

\subsubsection{Proof of Proposition \ref{proposition:caus-eqs-chain-std}}
\label{proof:proposition:caus-eqs-chain-std}
\begin{proof}
We build upon the result of Proposition \ref{proposition:caus-eqs-chain}.
Consider $\mu \in \Lsets{\Theta}$ with $n_\mu \geq 1$ and note that the associated $\lambda_{\mu, i} \in \mathcal{C}$ take the form $\downset{k_{\mu,i}}$ for some $k_{\mu,i} \in \Ext{\Theta}$.
Consider any $i \in \{1,...,n_\mu-1\}$, so that $\mu \subseteq \downset{k_{\mu,i}}\cap\downset{k_{\mu,i+1}}$: because the intersection of a downset is a downset, we have the following, for some $h_{\mu,i} \in \Ext{\Theta}$:
\[
\mu \subseteq \downset{h_{\mu,i}}
\subseteq \downset{k_{\mu,i}}\cap\downset{k_{\mu,i+1}}
\]
Since we included the relevant equations for all such $h_{\mu,i}$, we can infer the equations for $\mu$ by composing restrictions:
\[
    \restrict{\underline{u}^{(\lambda)}}{\mu}
    = \restrict{\left(\restrict{\underline{u}^{(\lambda)}}{\downset{h_{\mu, i}}}\right)}{\mu}
    = \restrict{\left(\restrict{\underline{u}^{(\lambda')}}{\downset{h_{\mu, i}}}\right)}{\mu}
    = \restrict{\underline{u}^{(\lambda')}}{\mu}
\]
This concludes our proof.
\end{proof}

\subsubsection{Proof of Theorem \ref{theorem:causaltopes-emp-models}}
\label{proof:theorem:causaltopes-emp-models}
\begin{proof}
    By applying Theorem \ref{theorem:topdist-extdist-generalised} to the individual components $e_\lambda \mapsto \topdist{e_\lambda}$ we conclude that the map is a convex-linear injection, where convex-linearity follows from the fact that all components of the empirical model are weighted equally in convex combinations.
    To prove that the map is surjective, we consider any arbitrary $\underline{u} \in \Causaltope{\mathcal{C}, \underline{O}}$ and define:
    \[
        e_\lambda := \extdist{\underline{u}^{\lambda}}{\lambda}
    \]
    The causality equations guarantee that restrictions from cover lowersets to arbitrary lowersets coincide:
    \[
    \restrict{\topdist{e_\lambda}}{\mu}
    =\restrict{\topdist{e_{\lambda'}}}{\mu}
    \]
    We now show that the restrictions of probability distributions above are exactly the same as the restrictions of empirical model components.
    To do so, it suffices to expand the top-element distribution and its restriction into their definitions:
    \[
    \begin{array}{rcl}
    \restrict{\topdist{e_\lambda}}{\mu}
    &=&
    \Dist{\rho_{\lambda, \mu}}
    \left(
    \sum\limits_{\Ext{f}} e_\lambda\left(\Ext{f}\right) \delta_{\topdist{\Ext{f}}}
    \right)
    \\
    &=&
    \sum\limits_{\Ext{f}} e_\lambda\left(\Ext{f}\right)
    \Dist{\rho_{\lambda, \mu}}
    \left(
    \delta_{\topdist{\Ext{f}}}
    \right)
    \\
    &=&
    \sum\limits_{\Ext{f}} e_\lambda\left(\Ext{f}\right)
    \delta_{\rho_{\lambda, \mu}\left(\topdist{\Ext{f}}\right)}
    \\
    &=&
    \sum\limits_{\Ext{f}} e_\lambda\left(\Ext{f}\right)
    \delta_{\restrict{\topdist{\Ext{f}}}{\mu}}
    \\
    &=&\Ext{f'}
    \mapsto
    \sum\limits_{f \text{ s.t. } \restrict{f}{\mu}=f'} e_\lambda\left(\Ext{f}\right)
    \end{array}
    \]
    The last line is the definition of restriction for empirical model components from $\lambda$ to $\mu$, completing our proof.
\end{proof}

\subsubsection{Proof of Proposition \ref{proposition:subspace-hierarchy-causaltopes}}
\label{proof:proposition:subspace-hierarchy-causaltopes}
\begin{proof}
Because $\Theta$ and $\Theta'$ have the same events and inputs and because $\max\Ext{\Theta}=\max\Ext{\Theta'}$, the two spaces have the same pseudo-empirical models:
\[
\prod_{k \in \max\Ext{\Theta}}
\Dist{
    \prod_{\omega \in \dom{k}}
    O_\omega
}
\]
where we used the fact that $\TipEqCls{\downset{k}}{\omega}$ is always a singleton.
Hence, it makes sense to compare the associated linear sub-spaces of causality equations.
By Proposition \ref{proposition:caus-eqs-chain-std}, the causality equations for the standard cover are generated by extended input histories: since $\Theta' \leq \Theta$ is defined to mean $\Ext{\Theta'} \supseteq \Ext{\Theta}$, the causality equations for $\Theta'$ are a superset of those for $\Theta$.
This concludes our proof.
\end{proof}

\subsubsection{Proof of Proposition \ref{proposition:csep-frac-bounded-below-by-csep-local-frac}}
\label{proof:proposition:csep-frac-bounded-below-by-csep-local-frac}
\begin{proof}
Recall from Definition 4.38 p.70 \cite{gogioso2022topology} that the separable local fraction for $e$ is the largest $p \in [0, 1]$ such that $e$ can be decomposed as $e = p \cdot e^{NC} + (1-p) \cdot e'$, where $e^{NC}$ is an empirical model which is separably local.
Recall from Definition 4.37 p.70 \cite{gogioso2022topology} that an empirical model $e^{NC}$ is separably local if it arises as the restriction to the standard cover of a classical empirical $\hat{e}$ which is entirely supported by causal functions which are separable.

Let $e = p \cdot e^{NC} + (1-p) \cdot e'$ be an empirical model for a space of input histories $\Theta$, where $p \in [0, 1]$ is the separable local fraction and $e^{NC}$ is separably local.
We can fully decompose $e^{NC}$ over the causal completions of $\Theta$: it is a convex mixture of separable causal functions, and every separable causal function for $\Theta$ is also a causal function for at least one of its causal completions.
Rescaling the decomposition of $e^{NC}$ by $p$ yields a causal decomposition of $e$ over the causal completions of $\Theta$ of mass $p$, proving that the causally separable fraction is bounded below by $p$ (the separable local fraction).
\end{proof}

\subsubsection{Proof of Lemma \ref{lemma:context-restriction-subspace}}
\label{proof:lemma:context-restriction-subspace}
\begin{proof}
If $k \in \Ext{\lambda}$, then we can write:
\[
k = \bigvee_{
    \scriptsize
    \begin{array}{c}
        h \in \lambda
        \\
        h \leq k
    \end{array}
} \hspace{-3mm} h
= \bigvee_{
    \scriptsize
    \begin{array}{c}
        h \in \lambda
        \\
        h \leq k
    \end{array}
} \hspace{-3mm} \bigvee_{
    \scriptsize
    \begin{array}{c}
        h' \in \lambda'
        \\
        h' \leq h
    \end{array}
} \hspace{-3mm} h' \in \Ext{\restrict{\lambda}{\Theta'}}
\]
This proves that $\Ext{\restrict{\lambda}{\Theta'}} \supseteq \Ext{\lambda}$, which is the definition of $\restrict{\lambda}{\Theta'} \leq \lambda$. It also shows that $\dom{\lambda} = \dom{\restrict{\lambda}{\Theta'}}$.
\end{proof}

\subsubsection{Proof of Proposition \ref{proposition:induced-tip-historysets}}
\label{proof:proposition:induced-tip-historysets}
\begin{proof}
Let $h, h' \in \histconstreqcls{h}{\omega}$, where the equivalence class is taken with respect to the equivalence relation $\histconstrSym{\omega}$ in $\lambda$.
Let $k, k' \in \Theta'$ be such that $k \leq h$ and $k' \leq h'$: because $\Ext{\restrict{\lambda}{\Theta'}} \supseteq \Ext{\lambda}$, we have that $\histconstr{\omega}{h}{h'}$ in $\restrict{\lambda}{\Theta'}$, and hence that $\histconstr{\omega}{k}{k'}$ in $\restrict{\lambda}{\Theta'}$. This means that $\restrict{\histconstreqcls{h}{\omega}\!\!}{\Theta'}$ is well-defined.
Observing that every $k \in \restrict{\lambda}{\Theta'}$ arises as $k \leq h$ for some $h \in \lambda$ allows us to conclude that $\histconstreqcls{k}{\omega} = \restrict{\histconstreqcls{h}{\omega}\!\!}{\Theta'}$, so that every tip historyset in $\restrict{\lambda}{\Theta'}$ arises as an induced tip history from $\lambda$.
\end{proof}

\subsubsection{Proof of Proposition \ref{proposition:subspace-hierarchy-causaltopes-general}}
\label{proof:proposition:subspace-hierarchy-causaltopes-general}
\begin{proof}
Recall the definition of the causality equations:
\[
    \CausEqs{\mathcal{C}, \underline{O}}_{\mu, \lambda, \lambda'}
    :=
    \suchthat{
        \underline{u}
        \in
        \PsEmpModelsVec{\mathcal{C}, \underline{O}}
    }{
        \restrict{\underline{u}^{(\lambda)}}{\mu}
        =\restrict{\underline{u}^{(\lambda')}}{\mu}
    }
\]
where the restriction $\restrict{\underline{u}^{(\lambda)}}{\mu}$ of distributions is defined as follows:
\[
    \restrict{\underline{u}^{(\lambda)}}{\mu}
    := \Dist{\rho_{\lambda, \mu}}\left(\underline{u}^{(\lambda)}\right)
\]
and the map $\rho_{\lambda, \mu}$ is defined as follows:
\[
\begin{array}{rccc}
    \rho_{\lambda, \mu}:
    &\prod\limits_{\omega \in \dom{\lambda}} \left(O_\omega\right)^{\TipEqCls{\lambda}{\omega}}
    &\longrightarrow
    &\prod\limits_{\omega \in \dom{\mu}} \left(O_\omega\right)^{\TipEqCls{\mu}{\omega}}
    \\
    &o
    &\mapsto
    &\left(
        \left(\omega, \histconstreqcls{h}{\omega}\right)
        \mapsto
        o_{\omega,\histconstreqcls{h}{\omega}}
    \right)
\end{array}
\]
The corresponding causality equations under the embedding of pseudo-empirical models take the following form:
\[
\scalebox{0.8}{$
    \CausEqs{\restrict{\mathcal{C}}{\Theta'}, \underline{O}}_{\restrict{\mu}{\Theta'}, \restrict{\lambda}{\Theta'}, \restrict{\lambda'}{\Theta'}}
    =
    \suchthat{
        \underline{u}
        \in
        \PsEmpModelsVec{\mathcal{C}, \underline{O}}
    }{
        \restrict{\underline{u}^{\left(\restrict{\lambda}{\Theta'}\right)}}{\restrict{\mu}{\Theta'}}
        =\restrict{\underline{u}^{\left(\restrict{\lambda'}{\Theta'}\right)}}{\restrict{\mu}{\Theta'}}
    }
$}
\]
We start by observing that, analogously to $\sigma_{\lambda}$, $\rho_{\lambda, \mu}$ can be defined in a contravariantly functorial way, in terms of pre-composition of $o$ by the following $r_{\lambda, \mu}$:
\[
\begin{array}{rccc}
    r_{\lambda, \mu}:
    & \sum\limits_{\omega \in \dom{\mu}} \TipEqCls{\mu}{\omega}
    & \hookrightarrow
    & \sum\limits_{\omega \in \dom{\lambda}} \TipEqCls{\lambda}{\omega}
    \\
    & \left(\omega, \histconstreqcls{h}{\omega}\right)
    & \mapsto
    & \left(\omega, \histconstreqcls{h}{\omega}\right)
\end{array}
\]
Using $r_{\lambda, \mu}$, we get the following alternative formulation for $\rho_{\lambda, \mu}$:
\[
    \begin{array}{rccc}
    \rho_{\lambda, \mu}:
    &
        \hspace{-3mm}
        \prod\limits_{
            \scriptsize
            \begin{array}{c}
                \omega \in \dom{\lambda}
                \\
                \hspace{1mm}\histconstreqcls{k}{\omega}\hspace{-1.5mm}\in \TipEqCls{\lambda}{\omega}
            \end{array}
        }
        \hspace{-9mm}
        O_\omega
        \hspace{2mm}
    & \twoheadrightarrow
    &
        \hspace{-3mm}
        \prod\limits_{
            \scriptsize
            \begin{array}{c}
                \omega \in \dom{\mu}
                \\
                \hspace{1mm}\histconstreqcls{h}{\omega}\hspace{-1.5mm}\in \TipEqCls{\mu}{\omega}
            \end{array}
        }
        \hspace{-9mm}
        O_\omega
        \hspace{2mm}
    \\
    & o
    & \mapsto
    & o \circ r_{\lambda, \mu}
\end{array}
\]
We now show that $r$ and $s$ commute, in the following sense:
\[
r_{\lambda, \mu}\circ s_{\lambda}
= s_{\mu} \circ r_{\restrict{\lambda}{\Theta'}, \restrict{\mu}{\Theta'}}
\]
For the LHS, we have:
\[
\begin{array}{rcl}
r_{\lambda, \mu}
\left(s_{\lambda}\left(\omega, \histconstreqcls{h}{\omega}\right)\right)
&=& r_{\lambda, \mu}
   \left(\omega, \restrict{\histconstreqcls{h}{\omega}\!\!}{\Theta'}\right) \\
&=& \left(\omega, \restrict{\histconstreqcls{h}{\omega}\!\!}{\Theta'}\right)
\end{array}
\]
For the RHS, we have:
\[
\begin{array}{rcl}
s_{\mu}
\left(r_{\restrict{\lambda}{\Theta'}, \restrict{\mu}{\Theta'}}\left(\omega, \histconstreqcls{h}{\omega}\right)\right)
&=& s_{\mu}
   \left(\omega, \histconstreqcls{h}{\omega}\right) \\
&=& \left(\omega, \restrict{\histconstreqcls{h}{\omega}\!\!}{\Theta'}\right)
\end{array}
\]
Because $r$ and $s$ commute, so do $\rho$ and $\sigma$, because they are obtained in a controvariantly functorial way from $r$ and $s$ respectively:
\[
\rho_{\lambda, \mu}\circ \sigma_{\lambda}
= \sigma_{\mu} \circ \rho_{\restrict{\lambda}{\Theta'}, \restrict{\mu}{\Theta'}}
\]
From the commutation above, we conclude that restriction of distributions commutes with embedding of pseudo-empirical models:
\[
\scalebox{0.9}{$
\begin{array}{rcl}
\restrict{
    i_{\mathcal{C},\Theta',\Theta}\left(\underline{u}\right)^{(\lambda)}
}{\mu}
&=& \Dist{\rho_{\lambda, \mu}}
    \Dist{\sigma_{\lambda}}
    \left(\underline{u}^{\left(\restrict{\lambda}{\Theta'}\right)}\right)\\
&=& \Dist{\rho_{\lambda, \mu}\circ \sigma_{\lambda}}
    \left(\underline{u}^{\left(\restrict{\lambda}{\Theta'}\right)}\right)\\
&=& \Dist{\sigma_{\mu} \circ \rho_{\restrict{\lambda}{\Theta'}, \restrict{\mu}{\Theta'}}}
    \left(\underline{u}^{\left(\restrict{\lambda}{\Theta'}\right)}\right)\\
&=& \Dist{\sigma_{\mu}}
    \Dist{\rho_{\restrict{\lambda}{\Theta'}, \restrict{\mu}{\Theta'}}}
    \left(\underline{u}^{\left(\restrict{\lambda}{\Theta'}\right)}\right)\\
&=& \Dist{\sigma_{\mu}}
    \left(\restrict{\left(\underline{u}^{\left(\restrict{\lambda}{\Theta'}\right)}\right)}{\restrict{\mu}{\Theta'}}\right)\\
&=& i_{\mathcal{C},\restrict{\mu}{\Theta'},\Theta}
    \left(\restrict{\left(\underline{u}^{\left(\restrict{\lambda}{\Theta'}\right)}\right)}{\restrict{\mu}{\Theta'}}\right)
\end{array}
$}
\]
Hence, we get the following implications:
\[
\begin{array}{rl}
&\restrict{\left(\underline{u}^{\left(\restrict{\lambda}{\Theta'}\right)}\right)}{\restrict{\mu}{\Theta'}}
= \restrict{\left(\underline{u}^{\left(\restrict{\lambda'}{\Theta'}\right)}\right)}{\restrict{\mu}{\Theta'}}
\\
\Rightarrow &
i_{\mathcal{C},\restrict{\mu}{\Theta'},\Theta}
\left(\restrict{\left(\underline{u}^{\left(\restrict{\lambda}{\Theta'}\right)}\right)}{\restrict{\mu}{\Theta'}}\right)
= i_{\mathcal{C},\restrict{\mu}{\Theta'},\Theta}
\left(\restrict{\left(\underline{u}^{\left(\restrict{\lambda'}{\Theta'}\right)}\right)}{\restrict{\mu}{\Theta'}}\right)
\\
\Rightarrow &
\restrict{
    i_{\mathcal{C},\Theta',\Theta}\left(\underline{u}\right)^{(\lambda)}
}{\mu}
= \restrict{
    i_{\mathcal{C},\Theta',\Theta}\left(\underline{u}\right)^{(\lambda')}
}{\mu}
\end{array}
\]
The above implies the desired inclusion for the subspace of a single causality equation:
\[
    \CausEqs{\restrict{\mathcal{C}}{\Theta'}, \underline{O}}_{\restrict{\mu}{\Theta'}, \restrict{\lambda}{\Theta'}, \restrict{\lambda'}{\Theta'}}
    \stackrel{i_{\mathcal{C},\Theta',\Theta}}{\lhook\joinrel\longrightarrow}
    \CausEqs{\mathcal{C}, \underline{O}}_{\mu, \lambda, \lambda'}
\]
By taking intersection of subspaces for causality equations, and observing that causality equations in the form $\CausEqs{\restrict{\mathcal{C}}{\Theta'}, \underline{O}}_{\restrict{\mu}{\Theta'}, \restrict{\lambda}{\Theta'}, \restrict{\lambda'}{\Theta'}}$ are a subset of all causality equations for the induced cover $\restrict{\mathcal{C}}{\Theta'}$, we conclude that:
\[
    \CausEqs{\restrict{\mathcal{C}}{\Theta'}, \underline{O}}
    \stackrel{i_{\mathcal{C},\Theta',\Theta}}{\lhook\joinrel\longrightarrow}
    \CausEqs{\mathcal{C}, \underline{O}}
\]
From the above, and using Remark \ref{remark:pseudo-empmodels-embedding-vec-general} (p.\pageref{remark:pseudo-empmodels-embedding-vec-general}) on inclusion of spaces for pseudo-empirical models, we conclude that:
\[
    \Causaltope{\restrict{\mathcal{C}}{\Theta'}, \underline{O}}
    \stackrel{i_{\mathcal{C},\Theta',\Theta}}{\lhook\joinrel\longrightarrow}
    \Causaltope{\mathcal{C}, \underline{O}}
\]
\end{proof}

\subsubsection{Proof of Proposition \ref{proposition:csep-frac-bounded-below-by-csep-local-frac-general}}
\label{proof:proposition:csep-frac-bounded-below-by-csep-local-frac-general}
\begin{proof}
The argument laid out in Proof \ref{proof:proposition:csep-frac-bounded-below-by-csep-local-frac} for Proposition \ref{proposition:csep-frac-bounded-below-by-csep-local-frac} (p.\pageref{proposition:csep-frac-bounded-below-by-csep-local-frac}) only depends on the notion of causal decomposition, and immediately generalises to arbitrary covers.
\end{proof}








\ack
Financial support from EPSRC, the Pirie-Reid Scholarship and Hashberg Ltd is gratefully acknowledged.
This publication was made possible through the support of the ID\#62312 grant from the John Templeton Foundation, as part of the project `The Quantum Information Structure of Spacetime' (QISS), https://www.templeton.org/grant/the-quantum-information-structure-ofspacetime-qiss-second-phase.
The opinions expressed in this project/publication are those of the author(s) and do not necessarily reflect the views of the John Templeton Foundation.

\section*{Bibliography}

\bibliographystyle{unsrt}
\bibliography{biblio}

\end{document}